\documentclass[11pt]{article}
\usepackage{times}
\usepackage{xspace,comment,calc}
\usepackage{amssymb,amsmath,amsthm,amsfonts,amstext,stmaryrd}
\usepackage[margin=1cm]{caption}
\usepackage[shortlabels,inline]{enumitem}
\usepackage{multirow,makecell}
\usepackage{graphicx}
\usepackage[dvipsnames]{xcolor}
\usepackage{bbm}
\usepackage{comment}
\usepackage{mdframed}
\usepackage{mathrsfs}
\usepackage{pict2e,graphpap}
\usepackage{caption}
\usepackage{thm-restate}

\usepackage{tikz}         % For arrow and dots in \xvec
\usepackage{hyperref}

\hypersetup{colorlinks = true,linkcolor = blue, anchorcolor = blue, citecolor = blue, filecolor = red,urlcolor = red}
\makeatletter
\newcommand{\clonelabel}[2]{\@bsphack
  \expandafter\ifx\csname r@#2\endcsname\relax
  \else\protected@write\@auxout{}{\string\newlabel{#1}%
    {\csname r@#2\endcsname}}%
  \fi
  \expandafter\ifx\csname r@#2@cref\endcsname\relax
  \else\protected@write\@auxout{}{\string\newlabel{#1@cref}%
    {\csname r@#2@cref\endcsname}}%
  \fi
  \@esphack}

\newlength\xvec@height%
\newlength\xvec@depth%
\newlength\xvec@width%
\newcommand{\xvec}[2][]{%
  \ifmmode%
    \settoheight{\xvec@height}{$#2$}%
    \settodepth{\xvec@depth}{$#2$}%
    \settowidth{\xvec@width}{$#2$}%
  \else%
    \settoheight{\xvec@height}{#2}%
    \settodepth{\xvec@depth}{#2}%
    \settowidth{\xvec@width}{#2}%
  \fi%
  \def\xvec@arg{#1}%
  \def\xvec@dd{:}%
  \def\xvec@d{.}%
  \raisebox{.2ex}{\raisebox{\xvec@height}{\rlap{%
    \kern.05em%  (Because left edge of drawing is at .05em)
    \begin{tikzpicture}[scale=1]
    \pgfsetroundcap
    \draw (.05em,0em)--(\xvec@width-.05em,0em);
    \draw (\xvec@width-.05em,0em)--(\xvec@width-.25em,.15em);
    \draw (\xvec@width-.05em,0em)--(\xvec@width-.25em,-.15em);
    \ifx\xvec@arg\xvec@d%
      \fill(\xvec@width*.45,.5ex) circle (.5pt);%
    \else\ifx\xvec@arg\xvec@dd%
      \fill(\xvec@width*.30,.5ex) circle (.5pt);%
      \fill(\xvec@width*.65,.5ex) circle (.5pt);%
    \fi\fi%
    \end{tikzpicture}%
  }}}%
  #2%
}

\newcommand{\leqnomode}{\tagsleft@true\let\veqno\@@leqno}
\newcommand{\reqnomode}{\tagsleft@false\let\veqno\@@eqno}
\makeatother

\let\stdvec\vec
\renewcommand{\vec}[1]{\xvec[]{#1}}

\newcommand{\np}{{\em NP}\xspace} 
\newcommand{\nphard}{\np-hard\xspace}
 
\newcommand{\apx}{{\em APX}\xspace}
\newcommand{\apxhard}{\apx-hard\xspace}

\DeclareMathOperator*{\Exp}{E}

\newcommand{\E}[2][{}]{\ensuremath{{\textstyle{\boldsymbol \Exp}_{#1}}\bigl[#2\bigr]}}

\newtheorem{theorem}{Theorem}[section]
\newtheorem{lemma}[theorem]{Lemma}
\newtheorem{claim}[theorem]{Claim}
\newtheorem{corollary}[theorem]{Corollary}

{\theoremstyle{definition} 
\newtheorem{definition}[theorem]{Definition}
\newtheorem{defn}[theorem]{Definition}}

{\theoremstyle{remark} \newtheorem{remark}[theorem]{Remark} }

\newenvironment{proofof}[1]{\begin{proof}[Proof of #1]}{\end{proof}}

\newcommand{\R}{\ensuremath{\mathbb R}}
\newcommand{\Rp}{\ensuremath{\mathbb R_{+}}}

\newcommand{\Z}{\ensuremath{\mathbb Z}}
\newcommand{\Zp}{\ensuremath{\mathbb Z_{+}}}

\newcommand{\A}{\ensuremath{\mathcal{A}}}
\newcommand{\B}{\ensuremath{\mathcal{B}}}
\newcommand{\C}{\ensuremath{\mathcal{C}}}
\newcommand{\D}{\ensuremath{\mathcal{D}}}
\newcommand{\I}{\ensuremath{\mathcal I}}
\newcommand{\F}{\ensuremath{\mathcal F}}
\newcommand{\G}{\ensuremath{\mathcal{G}}}
\newcommand{\J}{\ensuremath{\mathcal J}}
\newcommand{\K}{\ensuremath{\mathcal K}}

\newcommand{\M}{\ensuremath{\mathcal M}}
\newcommand{\Nc}{\ensuremath{\mathcal N}}

\newcommand{\Sc}{\ensuremath{\mathcal S}}
\newcommand{\T}{\ensuremath{\mathcal T}}

\newcommand{\Pc}{\ensuremath{\mathcal P}}

\newcommand{\Rc}{\ensuremath{\mathcal R}}

\newcommand{\OPT}{\ensuremath{\mathit{OPT}}}

\newcommand{\val}{\ensuremath{\mathit{value}}}
\newcommand{\frall}{\ensuremath{\text{ for all }}}

\newcommand{\load}{\ensuremath{\mathsf{load}}}
\newcommand{\ub}{\ensuremath{\mathsf{ub}}}

\newcommand{\LP}[1]{\ensuremath{\mathsf{LP}}}
\newcommand{\alg}{\mathcal{A}}

\newcommand{\sm}{\ensuremath{\setminus}} 
\newcommand{\es}{\ensuremath{\emptyset}}
\newcommand{\sse}{\subseteq} 

\newcommand{\ceil}[1]{\ensuremath{\bigl\lceil#1\bigr\rceil}}
\newcommand{\floor}[1]{\ensuremath{\bigl\lfloor#1\bigr\rfloor}}

\newcommand{\assign}{\leftarrow}

\newcommand{\poly}{\operatorname{poly}}

\newcommand{\topl}[1][\ell]{\ensuremath{\mathsf{Top}_{#1}}}

\newcommand{\pos}{\ensuremath{\mathsf{POS}}}
\newcommand{\POS}{\pos}

\newcommand{\e}{\ensuremath{\epsilon}} 
\newcommand{\ve}{\ensuremath{\varepsilon}} 
\newcommand{\gm}{\ensuremath{\gamma}} 
\newcommand{\Gm}{\ensuremath{\Gamma}} 
\newcommand{\ld}{\ensuremath{\lambda}} 
\newcommand{\Ld}{\ensuremath{\Lambda}} 
\newcommand{\tld}{\ensuremath{\widetilde\ld}}
\newcommand{\kp}{\ensuremath{\kappa}}

\newcommand{\al}{\ensuremath{\alpha}} 

\newcommand{\tht}{\ensuremath{\theta}} 
\newcommand{\dt}{\ensuremath{\delta}}
\newcommand{\Dt}{\ensuremath{\Delta}} 
\newcommand{\sg}{\ensuremath{\sigma}}
\newcommand{\w}{\ensuremath{\omega}} 
\newcommand{\Om}{\ensuremath{\Omega}}

\newcommand{\ta}{\ensuremath{\widetilde a}}

\newcommand{\tp}{\ensuremath{\widetilde p}}
\newcommand{\tx}{\ensuremath{\widetilde x}}
\newcommand{\ty}{\ensuremath{\widetilde y}}

\newcommand{\tdh}{\ensuremath{\widetilde h}} 
\newcommand{\tn}{\ensuremath{\widetilde n}} 
\newcommand{\tS}{\ensuremath{\widetilde S}}
\newcommand{\tX}{\ensuremath{\widetilde X}}

\newcommand{\tA}{\ensuremath{\widetilde A}}

\newcommand{\hx}{\ensuremath{\widehat x}}

\newcommand{\hJ}{\ensuremath{\widehat J}}
\newcommand{\hL}{\ensuremath{\widehat L}}
\newcommand{\hA}{\ensuremath{\widehat A}}

\newcommand{\bx}{\ensuremath{\overline x}} 
\newcommand{\by}{\ensuremath{\overline y}}

\newcommand{\bi}{\ensuremath{\overline i}}
\newcommand{\bI}{\ensuremath{\overline I}}
\newcommand{\bL}{\ensuremath{\overline L}} 
\newcommand{\bell}{\ensuremath{\overline{\ell}}}
\newcommand{\bA}{\ensuremath{\overline A}} 
\newcommand{\qt}{\theta}

\newcommand{\bon}{\ensuremath{\mathbbm{1}}}
\newcommand{\down}{\ensuremath{{\:\!\downarrow}}}

\newcommand{\sz}{\ensuremath{w}\xspace}
\newcommand{\svec}[2][{\sz}]{\vec{{#1}({#2})}}
\newcommand{\rewd}{\ensuremath{\mathsf{rwd}}\xspace}
\newcommand{\sols}{\ensuremath{\Sc}}

\newcommand{\normbudg}[1]{\ensuremath{\mathsf{NormBudget{#1}}}\xspace}
\newcommand{\nb}[1]{\ensuremath{\mathsf{NB{#1}}}\xspace}
\newcommand{\normknap}{\ensuremath{\mathsf{NormBudgKnap}}\xspace}
\newcommand{\normmatch}{\ensuremath{\mathsf{NormBudgMatch}}\xspace}
\newcommand{\normmwis}{\ensuremath{\mathsf{NormBudgMWIS}}\xspace}
\newcommand{\normsched}{\ensuremath{\mathsf{NormBudgMaxGAP}}\xspace}
\newcommand{\normsap}{\ensuremath{\mathsf{NormBudgSAP}}\xspace}
\newcommand{\normfl}{\ensuremath{\mathsf{NormBudg}k\mathsf{FL}}\xspace}
\newcommand{\normkfl}{\normfl}
\newcommand{\normsepfl}{\ensuremath{\mathsf{NormBudgSepFL}}\xspace}
\newcommand{\normkmatch}{\ensuremath{\mathsf{NormBudg}k\mathsf{Match}}\xspace}
\newcommand{\knap}{\ensuremath{\mathsf{Knap}}\xspace}
\newcommand{\gap}{\ensuremath{\mathsf{GAP}}\xspace}
\newcommand{\maxgap}{\ensuremath{\mathsf{Max}\gap}\xspace}
\newcommand{\match}{\ensuremath{\mathsf{Match}}\xspace}
\newcommand{\kmatch}{\ensuremath{k\mathsf{Match}}\xspace}
\newcommand{\mwis}{\ensuremath{\mathsf{MWIS}}\xspace}
\newcommand{\sap}{\ensuremath{\mathsf{SAP}}\xspace}
\newcommand{\kfl}{\ensuremath{k\mathsf{FL}}\xspace}
\newcommand{\sepfl}{\ensuremath{\mathsf{SepFL}}\xspace}

\newcommand{\optset}{O^*}
\newcommand{\optind}{o}
\newcommand{\ovec}{\svec{\optset}}

\newcommand{\optval}{\ensuremath{\widetilde{\mathsf{opt}}}\xspace}
\newcommand{\rest}{\widetilde\rewd}
\newcommand{\trewd}{\widetilde\rewd}
\newcommand{\num}{\widetilde N}
\newcommand{\rnum}{\widetilde R}
\newcommand{\pnum}{\ensuremath{\widetilde{\mathsf{PN}}}\xspace}
\newcommand{\nopt}{\ensuremath{n_{\mathsf{opt}}}}

\newcommand{\nxt}{\ensuremath{\mathsf{next}}}
\newcommand{\prev}{\ensuremath{\mathsf{prev}}}
\newcommand{\dbrack}[1]{\ensuremath{\llbracket{#1}\rrbracket}}

\newcommand{\itemset}{\ensuremath{\mathsf{Bkt}}}
\newcommand{\pset}{Q}

\newcommand{\nbuck}{\ensuremath{{n_{\mathsf{bkt}}}}}
\newcommand{\rmax}{\ensuremath{\mathsf{r_{max}}}\xspace}
\newcommand{\thresh}{\tau}
\newcommand{\cpknap}[1][{B}]{\textnormal{KCP}\ensuremath{({#1})}}

\newcommand{\rk}{\ensuremath{r}}
\newcommand{\gr}{\ensuremath{\mathsf{gr}}\xspace}

\newcommand{\lvec}[1]{\svec[{\load}]{{#1}}}
\newcommand{\lvecopt}{{\vec{\load^*}}}
\newcommand{\lvo}{{\lvecopt}}
\renewcommand{\exp}{\ensuremath{\mathsf{exp}}}
\newcommand{\vt}{\ensuremath{\stdvec{t}}}

\newcommand{\tsg}{\ensuremath{\widetilde{\sigma}}}

\newcommand{\hsg}{\ensuremath{\widehat{\sg}}}
\newcommand{\jopt}{\optset}

\newcommand{\lpopt}[1][{\schrewdLP}]{\ensuremath{\mathsf{LP}^*_{{#1}}}\xspace}
\newcommand{\gnorm}{h}

\newcommand{\fload}[1]{\svec[{\mathsf{f}\load}]{{#1}}}
\newcommand{\onefrp}{\ensuremath{\mathsf{1FRP}}\xspace}
\newcommand{\constr}{budgeted\xspace}
\newcommand{\fvo}{{\vec{\mathsf{f}\load^*}}}
\newcommand{\copt}{C^*}
\newcommand{\fopt}{F^*}
\newcommand{\sepfllp}{\ensuremath{\mathsf{NBSFL}\text{-}\mathsf{LP}}\xspace}
\newcommand{\sepfllpp}{\sepfllp'}
\newcommand{\lpsepflopt}{\lpopt[\sepfllp]}
\renewcommand{\val}{\ensuremath{\mathsf{LPval}}\xspace}
\newcommand{\targ}{\ensuremath{\mathsf{Val}}\xspace}
\newcommand{\good}{\G}
\newcommand{\bad}{\B}
\newcommand{\fvec}{{\vec{\mathsf{f}\load}}}
\newcommand{\vopt}{V^*}
\newcommand{\guessind}{I}
\newcommand{\wt}{\mathsf{wt}}

\newcommand{\mksp}{\ensuremath{\mathsf{mksp}}\xspace}
\newcommand{\energy}{\ensuremath{\mathsf{energy}}\xspace}
\newcommand{\into}{\ensuremath{\mathsf{in}}\xspace}
\newcommand{\out}{\ensuremath{\mathsf{out}}\xspace}
\newcommand{\ssum}{\ensuremath{\mathsf{sum}}\xspace}
\newcommand{\smax}{\ensuremath{\mathsf{max}}\xspace}
\newcommand{\vo}{\stdvec{o}}

\newcommand{\mkp}{\ensuremath{\mathsf{MKP}}\xspace}
\newcommand{\tlvec}{{\widetilde{\load}}}

\newcommand{\sparse}[1][s]{{#1}parse\xspace}
\newcommand{\dense}[1][d]{{#1}ense\xspace}

\newcommand{\rem}{\ensuremath{\mathsf{rem}}\xspace}

\newcommand{\mclass}{\ensuremath{\Nc}}
\newcommand{\lmc}{\ensuremath{\D}}
\newcommand{\sml}{\ensuremath{\mathsf{sml}}\xspace}
\newcommand{\lrg}{\ensuremath{\mathsf{lrg}}\xspace}

\newcommand{\subnbknap}{\ensuremath{\mathsf{SubmodNBKnap}}\xspace}
\newcommand{\subnbsched}{\ensuremath{\mathsf{SubmodNBMaxGAP}}\xspace}
\newcommand{\noptset}{\ensuremath{\widetilde O}}
\newcommand{\sbktopt}{\ensuremath{\mathsf{SBkt}^*}}
\newcommand{\ellfin}{\ensuremath{\ell^{\mathsf{last}}}}
\newcommand{\rnopt}{\rewd^*}
\newcommand{\gensub}{g}
\newcommand{\genmulti}{G}
\newcommand{\brewd}{\ensuremath{\overline\rewd}}
\newcommand{\rseqset}{\Rc}
\newcommand{\POSg}{\POS_{>}}
\newcommand{\optpp}{O''}

\newcommand{\tlvmin}{\tlvec_{\mathsf{min}}}
\newcommand{\tlvmax}{\tlvec_{\mathsf{max}}}
\newcommand{\tlvavg}{\tlvec_{\mathsf{avg}}}
\newcommand{\mcset}{I}
\newcommand{\mcnum}{|\mcset|}
\newcommand{\cfg}{\zeta}
\newcommand{\cfgset}{\J}
\newcommand{\Totrewd}{\ensuremath{\mathsf{TotRewd}}\xspace}
\newcommand{\preload}{\Ld}

\newcommand{\work}{\ensuremath{\mathsf{work}}}
\newcommand{\wvec}[1]{\svec[{\work}]{{#1}}} %% work vector
\newcommand{\wvecopt}{{\vec{\work^*}}} %% optimal work vector
\newcommand{\twvec}{{\widetilde{\work}}}
\newcommand{\fast}{\ensuremath{\mathsf{fast}}} %% class of fast machines
\newcommand{\giant}{\ensuremath{\mathsf{giant}}}%%% Giant jobs
\newtheorem*{claim*}{Claim}
\newcommand{\wkestim}{\ensuremath{\widetilde\work}}
\newcommand{\wkest}{\wkestim}
\newcommand{\wksmooth}{\ensuremath{W^{\mathsf{smth}}}}
\newcommand{\anbuck}{\ensuremath{{n_{\mathsf{Abkt}}}}}
\newcommand{\abkt}{\ensuremath{\mathsf{ABkt}}}
\newcommand{\Mc}{M}
\newcommand{\spd}{s}
\newcommand{\ngrp}{K}
\newcommand{\est}{\ensuremath{\mathsf{est}}}
\newcommand{\seqset}{\Rc}

\newcommand{\sumnorm}{\ensuremath{\mathsf{sum}}}

\newcommand{\dptab}{D}

\title{Approximation Algorithms for Norm-Budgeted Packing Problems}

\author{
    David Alem\'{a}n Espinosa\thanks{{\tt\{dalemanespinosa,\,cswamy\}@uwaterloo.ca}.
    Dept. of Combinatorics and Optimization, Univ. Waterloo, Waterloo, ON N2L 3G1, Canada.
    Supported in part by C. Swamy's NSERC Discovery grant 2024-04532.}
\and
    Sharat Ibrahimpur\thanks{{\tt sibrahimpur@ethz.ch}.
    ETH Zurich, Zurich, Switzerland.}
\and
\addtocounter{footnote}{-2}
    Chaitanya Swamy\footnotemark
}
\date{}

\begin{document}

\maketitle

\def\thepage{}
\thispagestyle{empty}

\begin{abstract}
In recent years, much attention has been devoted to the study of 
optimization problems under norm-based objectives coming from the rich
class of monotone, symmetric norms (and their generalizations). This work has however
almost exclusively focused on covering problems, wherein one seeks to minimize the norm of
the cost vector induced by a solution.

We introduce and study the class of {\em norm-budgeted packing problems}, which are
packing problems where the resource constraints underlying the packing problem are modeled
via a {\em norm budget constraint} involving a {\em monotone, symmetric norm}.
%We introduce and study a broad class of packing problems in which the feasible region is
%constrained by a monotone, symmetric norm $f:\R^m\to\Rp$. 
Formally, we have some elements with associated rewards and sizes, a downwards-closed
collection $\sols$ of feasible solutions, and a budget $B$.
%Given a budget $B$ and a ground set of elements with associated rewards and sizes.
%Feasible solutions form a downwards-closed family of sets, and 
Each solution induces a size vector, and 
the goal is to maximize the total reward subject to the norm-budget constraint 
$f(\text{size vector}) \le B$.
As with minimum-norm covering problems, the versatility of monotone, symmetric norms
implies that a variety of classical packing problems, involving sum- or max- budget
constraints %sum-$\ell_1$-norm) of max-based 
can be captured under the umbrella of norm-budgeted packing problems.
%Taking $f$ to be the $\ell_1$ norm, or the $\ell_\infty$ norm, yields various standard
%packing problems.
%This formulation allows us to extend several \textcolor{blue}{well
%  studied/classical/standard} packing models under a common norm-based framework. 
Moreover, the closure properties of monotone, symmetric norms, also enable one to encode
multiple different norm-budget constraints via a single monotone, symmetric norm.

We consider the norm-budgeted versions of a variety of canonical packing problems,
including knapsack, matching, maximum-weight independent set in a matroid (and more
generally $k$-set system), maximum generalized assignment problem (\maxgap), and
$k$-facility location, and develop a framework that allows us to obtain 
{\em constant-factor approximation guarantees} for these problems, 
and {\em PTASes for knapsack, and \maxgap on identical and related machines}.

In order to do so, one fundamental and significant impediment that we need to overcome is
that the techniques that have been developed in the study of minimum-norm covering
optimization problems are not strong enough to deal with a {\em hard} norm-budget
constraint. All of this machinery leads to an {\em inherent} violation of the norm, and
moreover, in stark contrast with $\ell_1$ norms, 
%violating the norm by even a $(1+\ve)$-factor can lead to a 
there can be a huge gap between bicriteria and unicriteria solutions (even when every
single item consumes only a small portion of the budget).
We address this challenge by developing novel tools to handle the norm-budget constraint.  

We also develop constant-factor approximation algorithms for the {\em submodular} versions
of norm-budgeted knapsack and norm-budgeted \maxgap on related machines, wherein the
reward function is now specified by a monotone, submodular function.

\end{abstract}

\newpage \pagenumbering{arabic} \normalsize

\section{Introduction} \label{intro} 
Packing and covering problems are two broad and fundamental classes of
combinatorial-optimization problems that have been extensively investigated in the
Operations Research and theoretical Computer Science literature, and find applications in 
a variety of domains such as logistics, scheduling, clustering, network design. 
Covering problems 
%(e.g., set-cover, facility-location
tend to capture settings where there are some entities that need to be ``served'' (e.g., 
clients in a logistics problem, jobs in a scheduling problem), and one
needs to determine a minimum-cost way of providing service.
%must be provided service while
Packing problems on the other hand model settings where resource constraints preclude one
from serving all entities, and the goal is to therefore select a most ``profitable'' set
of entities to serve. %and we would like to do so while maximizing 
Well-known examples of covering problems include set cover, facility location, load
balancing, and some prominent examples of packing problems are bin packing, knapsack,
maximum coverage problem. %maximum multicommodity flow problems. Indeed,  
\begin{comment}
indeed, texts in approximation algorithms (see~\cite{Vazirani},
\cite{WilliamsonS11}) often devote a significant portion to the understanding of such
canonical covering and packing problems.  
%(e.g. set cover, facility location, knapsack, bin packing) and there
%are also books written on some classes of problems (see, e.g., \cite{KellererPP04}).
\end{comment}

Traditionally, covering problems were studied under the
%the commonly considered objectives work in both problem domains on 
min-sum and min-max objectives (and in some cases with $\ell_p$-norm objectives), but 
in recent years, a great deal of attention has been devoted to 
{\em minimum-norm
  optimization}~\cite{ChakrabartyS19a},
%ChakrabartyS19b,IbrahimpurS21,IbrahimpurS22,DengLR23,KesselheimMS24,HeroldKS25,ChenLRZ25,HeroldKS26},  
wherein one considers a much-richer class of objectives defined by
arbitrary {\em monotone, symmetric norms} (and their
generalizations~\cite{KesselheimMS24}). A norm $f:\R^n\mapsto\Rp$ is symmetric if it is
invariant under permutation of coordinates; a monotone norm satisfies $f(x)\geq f(y)$
whenever $x\geq y\geq 0$. In a minimum-norm optimization problem, %the objective to be
we are given a monotone, symmetric norm $f$, the goal is to minimize the $f$-norm of the
cost-vector induced by a solution (e.g., machine-load vector in load balancing, or
client-assignment cost vector in facility location and clustering).  
%The study of minimum-norm optimization problems was initiated by
Chakrabarty and Swamy~\cite{ChakrabartyS19a} initiated this research direction,
%as a far-reaching generalization of some earlier work on $k$-clustering that considered
%$\topl$ norms and ordered norms~\cite{ByrkaSS18,ChakrabartyS18},% 
%and devised constant-factor approximation algorithms for the minimum-norm generalizations
%of load balancing and $k$-clustering. %minimum-norm $k$-clustering problems,
%a fundamental class of monotone,
%symmetric norms, namely $\topl$ norms and ordered norms.
%As noted in~\cite{ChakrabartyS19a}, 
%providing two distinct reasons for why this an appealing topic of study. 
motivating it %topic of study 
from two distinct perspectives.
%mentioning two distinct reasons as motivation
%the modeling power afforded by monotone, symmetric norms, and the broad
%applicability of techniques 
First, monotone, symmetric norms constitute a very versatile class of objectives that
afford one a great deal of modeling power. %Not only do they encompass 
They include $\ell_p$ norms,
the natural class of $\topl$ norms%
\footnote{The $\topl$-norm of a vector $v\geq 0$ is defined as the sum of its $\ell$
  largest coordinates. An ordered norm is a nonnegative linear combination of $\topl$
  norms.} 
\addtocounter{footnote}{-1}
(which yield an alternate means of 
interpolating between the min-max and min-sum objectives), and ordered norms;\footnotemark\,
moreover, their closure properties %of monotone, symmetric norms 
imply that one can capture {\em multiple different norm constraints} by
suitably defining a {\em single} monotone, symmetric norm.%
\footnote{This feature is strikingly exploited by~\cite{IbrahimpurS21}.}
Second, min-sum and min-max problems are often handled via different approaches, and
the development of techniques for minimum-norm optimization yields a unified way of
approaching these, and various other, objectives.
Consequently, there has been much subsequent work %Following~\cite{ChakrabartyS19a}, 
investigating minimum-norm optimization for other
problems~\cite{DengLR23,AbbasiBBCGKMSS23,ChenLRZ25,HeroldKS25,HeroldKS26}, and in
stochastic~\cite{IbrahimpurS20,IbrahimpurS21} and online
settings~\cite{KesselheimMS23,KesselheimMS24,KesselheimMPS25}, some of which has 
also considered objectives that are more general than monotone, symmetric norms.

However, to our knowledge, {\em essentially all} prior work on minimum-norm optimization
problems has focused on covering problems, where the norm is in the objective,
%In particular, while packing and covering problems are often seen in classical
%combinatorial optimization 
and there is no prior work that considers packing problems with resource constraints that
are modeled by monotone, symmetric norms.%
\footnote{An exception is Kesselheim et al.~\cite{KesselheimMPS25}, who do formulate a
packing problem with a norm budget constraint as a means of solving an
%primary problem, which is an 
online covering problem with a norm-based objective, which is their primary problem of
concern. They violate the norm budget constraint, and as we discuss under ``Related work,'' 
%when solving the packing problem. This violation is unavoidable in the
%online setting, 
%for the packing problem, which is OK from the perspective of the covering problem.
%As we discuss later, 
our work can be seen as complementary to their work. 
%
%Somewhat peripherally 
Also somewhat related is Neogi et al.~\cite{NeogiPS24}, who consider a
mechanism-design setting with multiple $\topl$-budget constraints; we discuss this also
under ``Related work.''} 

\begin{comment}
Very recently~\cite{KesselheimMPS25} consider, as a means of solving an online covering
problem with a norm-based objective, the online packing problem obtained by moving the
norm objective to a norm budget constraint, where the objective is, roughly speaking to,
maximize the number of covering constraints that are satisfied. Their results violate the
norm budget constraint, which is unavoidable in the online setting. Our results are
orthogonal to, and can be seen as complementary to, their results: we consider the offline
setting, and are able to obtain stronger constant-factor approximation guarantees and 
{\em without violating the norm budget constraint} in the offline setting.
\end{comment} 

\subsection{Our contributions} \label{contrib} 
We initiate the study of {\em norm-budgeted packing problems},  
which are packing problems with the common distinguishing feature 
%packing constraint is modeled by a %monotone, symmetric 
%We consider a variety of packing problems with the common distinguishing feature 
that the resource constraints underlying the packing problem are modeled by a
%monotone, symmetric norm. 
{\em norm budget constraint}. %that restricts the space of allowed solutions.
We consider problems that fall into the following setup. As is standard in packing
problems, there is an underlying ground set of $n$ elements, where each element $e$
has a reward $\rewd_e\geq 0$ and size $\sz_e\geq 0$, and a non-empty collection 
$\sols\sse 2^{[n]}$ of candidate solutions that is {\em downwards-closed}, i.e., closed
under taking subsets. The tuple $([n],\sols)$ is often called an 
{\em independence system}. %~\cite{KorteV}. 
\begin{comment}
The space of candidate solutions will sometimes be a collection of subsets of $[n]$, but
could also involve other combinatorial objects; for instance, in maximum-throughput
scheduling, %items correspond to jobs, and a solution involves selecting a set of jobs, and
%specifying a job-to-machine assignment for the selected jobs.
\end{comment}
%with the common feature that one of the
%constraints defining the space of feasible solutions is given by a 
%{\em norm budget constraint}  under 
We are also given a {\em monotone, symmetric norm} $f$, and a budget $B$. 
Each solution $T\in\sols$ induces a {\em size-vector}, denoted $\svec{T}$. 
%If a solution corresponds to a set $T\sse[n]$,
%For a solution $T\in\sols$ to be feasible, we require that the $f$-norm of the size-vector
%induced by $T$ should be at most the given budget $B$.
%More precisely and formally, given a set $T\sse[n]$ of items, the {\em size-vector}
%then it induces the size-vector 
%associated with $T$ is defined as the vector 
Sometimes the size-vector $\svec{T}$ will simply be the 
{\em size-weighted characteristic vector of $T$}, which is the vector in $\R^n$ with
coordinates $\sz_e$  if $e\in T$, and $0$ otherwise. 
In other settings, this size-vector is formed by aggregating some of the $\sz_e$ values
for elements $e\in T$. 
For instance, in a scheduling problem, %(see \normsched below), 
where we have jobs and machines and have to select a suitable subset of jobs to assign to
machines, a natural choice for the size-vector is the vector of machine-loads. 
The goal is to maximize $\rewd(T):=\sum_{e\in T}\rewd_e$
subject to the constraint $T\in\sols$, and the norm budget constraint
$f\bigl(\svec{T}\bigr)\leq B$. 
%and so we call this problem, the 
%{\em norm-budgeted maximum-weight independent set} (\normmis) problem.

We use the term {\em norm-budgeted packing problem}, to describe a generic problem in the
above setup. %Unless otherwise stated, 
We assume that the norm $f$ is specified via a
feasibility oracle that given a vector $v$ and scalar $\ld$ determines if $f(v)\leq\ld$.
We state approximation guarantees as values that are at least $1$: 
an $\al$-approximation, for $\al\geq 1$, denotes that we obtain 
%an $\al$-approximation algorithm for a norm-budgeted packing problem, where
%$\al\geq 1$, returns a feasible solution of 
objective value at least (optimum)/$\al$.
%
\begin{comment}
We consider the norm-budgeted versions of a variety of fundamental packing problems,  
%including knapsack, matching, maximum-weight independent set in matroids (and more
%generally $k$-set systems), load balancing, and $k$-facility location, 
and develop constant-factor approximation algorithms or approximation schemes for these
problems.  
\end{comment}
Here are two canonical problems captured by this setup. 
%Section~\ref{probdefs}, 
%describes the norm-budgeted versions of the other canonical packing problems that we
%consider.

\begin{enumerate}[label=$\bullet$, leftmargin=*]
\item {\bf Norm-budgeted knapsack (\normknap), Section~\ref{knapsack}.} 
Here the ground set consists of
items that we seek to pack in a knapsack subject to the norm budget constraint. 
%We use the terms elements and items interchangeably for this problem.
%We are given $n$ items, with
%non-negative rewards $\{\rewd_i\}_{i\in[n]}$ and sizes $\{\sz_i\}_{i\in[n]}$. 
So we have $\sols=2^{[n]}$, and 
%have a {\em monotone, symmetric norm} $f:\R^n\mapsto\R_+$ and a knapsack budget $B$. The goal is
%to find a maximum-reward set of items that fits in the knapsack, where fitting in the
%knapsack is determined by requiring that the norm applied to the size-vector obtained from
%these items is at most the budget $B$. 
%More precisely and formally, given a set $T\sse[n]$ of items, 
the size-vector $\svec{T}$ corresponding to an item-set $T\sse[n]$ is simply the
size-weighted characteristic vector $\svec{T}\in\R^n$;  
%defined as the vector $\svec{T}\in\R^n$ with coordinates 
that is, $\svec{T}_i$ is $\sz_i$ if $i\in T$, and is $0$ otherwise. 
(So the monotone, symmetric norm $f$ is over $\R^n$.)
%The goal therefore is to maximize $\rewd(T):=\sum_{i\in T}\rewd_i$ 
%subject to the norm budget constraint $f\bigl(\svec{T}\bigr)\leq B$. 
%
Observe that taking $f$ to be the $\ell_1$ norm, we obtain the standard knapsack problem, so
\normknap is clearly \nphard. 

We develop a polynomial time approximation scheme (PTAS) for
this problem (Theorem~\ref{normknapthm}.)

\item {\bf Norm-budgeted maximum generalized-assignment problem (\normsched).}  
Here, we are given a set $J$ of jobs, and a set of $m$ machines. Processing a job $j$ on
machine $i$ incurs $p_{ij}$ time, and yields reward $\rewd_{ij}$. 
%In the context of this problem, we always use $i$ to index machines, and $j$ to index
%jobs. 
%Jobs cannot be preempted.
We have a monotone, symmetric norm $f:\R^m\mapsto\Rp$ and a budget $B$. A %candidate
solution %involves selecting a subset $S$ of jobs, and 
specifies an assignment $\sg:S\mapsto[m]$ of some subset $S\sse J$ of jobs to
machines, which induces a {\em load-vector} in $\R^m$, where the load on machine $i$ is
$\sum_{j\in S:\sg(j)=i}p_{ij}$, and earns reward $\sum_{j\in S}\rewd_{\sg(j)j}$. 
The goal is to find a maximum-reward solution 
satisfying the constraint $f(\text{load vector})\leq B$.
\begin{comment}
This can also be viewed as a generalization of \normknap involving multiple knapsacks, so
we sometimes refer to this as the {\em norm-budgeted multi-knapsack} problem.  
\end{comment}

To cast this in our general setup, we take the ground set $E$ to be the edges of the 
complete bipartite graph with node-set $J\cup[m]$, and map the $p_{ij}$s and $\rewd_{ij}$s
to the sizes and rewards respectively of edges in $E$. 
We take $\sols=\{T\sse E: |T\cap\dt(j)|\leq 1\text{ for all $j\in J$}\}$, and 
%where $\dt(j)$ is the set of edges incident to node $j$. 
the size-vector $\svec{T}\in\R^m$ associated with $T\in\sols$ is given by
$\bigl(\sum_{e\in T\cap\dt(i)}p_e\bigr)_{i\in[m]}$. 
%(Recall that $\sz_{ij}=p_{ij}$ for every edge $ij\in E$.)

The above setup corresponds to scheduling on {\em unrelated machines}. The setting of 
{\em identical machines} is the special case where we have $p_{ij}=p_j$ and
$\rewd_{ij}=\rewd_j$ for every machine $i\in[m]$, job $j\in J$. An intermediate setting is
{\em related machines}, wherein each machine $i\in[m]$ has a speed $s_i>0$, and we have
$p_{ij}=\frac{p_j}{s_i}$, $\rewd_{ij}=\rewd_j$ for every $i\in[m]$, $j\in J$.

The special case of \normsched where $f$ is the $\ell_\infty$ norm---i.e., find a
maximum-reward set of jobs that can be scheduled within a certain makespan---was 
considered by Fleischer et al.~\cite{FleischerGMS11}, who called this problem maximum-GAP. 
They devised a $\frac{e}{e-1}$-approximation for this problem, and this
problem is known to be \apx-hard~\cite{ChekuriK05}.
The further special cases involving identical machines and related machines 
%have been studied 
correspond respectively to the 
uniform {\em multiple knapsack problem} (\mkp) and {\em non-uniform \mkp},%
\footnote{In non-uniform \mkp, the capacities $\{u_i\}_{i\in[m]}$ of the knapsacks may be
different; by viewing $u_i$ as the speed of a machine $i$ and setting the norm budget to
$1$ (with the $\ell_\infty$ norm), we can cast this as \normsched on related machines.}
both of which are strongly \nphard~\cite{ChekuriK05} and %for which there is a
admit a PTAS~\cite{Kellerer99,ChekuriK05,Jansen09}. 

%We devise algorithms achieving approximation guarantees of $(5+\ve)$ for
%\normsched (Theorem~\ref{normschedthm}), and $(2+\ve)$ for \normsched on identical
%machines (Theorem~\ref{ident-biptas}).  

We obtain a $(4.582+\ve)$-approximation for \normsched (Theorem~\ref{normschedthm}). For
\normsched on identical and related machines, we obtain a PTAS. This yields a
tight complexity result for the latter two problems, as these problems are strongly
\nphard. Moreover, our PTAS substantially generalizes the PTAS for non-uniform
\mkp (which is the special case of the related-machines problem where $f$ is the
$\ell_{\infty}$-norm).  
%a $(2+\ve)$-approximation
%(Theorem~\ref{ident-biptas}), assuming a feasibility oracle for the norm, and devise a
%PTAS assuming stronger oracle access (Theorem~\ref{ident-ptasthm}).
\end{enumerate}

We consider the norm-budgeted versions of various other fundamental packing problems,
including matching, maximum-weight independent set, $k$-facility location;
we define these problems in Section~\ref{probdefs}. 
%including knapsack, matching, maximum-weight independent set in matroids (and more
%generally $k$-set systems), load balancing, and $k$-facility location, 
We develop constant-factor approximation algorithms %or approximation schemes 
for all these problems. Table~\ref{restable} summarizes the approximation factors we
obtain for the various norm-budgeted packing problems we consider.
%describes the norm-budgeted versions of the other canonical packing problems that we
%consider.

\begin{table}[t!]
\small
\begin{tabular}{l|l|l} \hline
\makecell*[l]{\bf Problem} & \makecell*[l]{\bf Approximation factor} & 
\makecell*[l]{\bf Comments} \\ \hline
\makecell*[l]{\normknap} & \makecell*[l]{PTAS (Theorem~\ref{normknapthm})} &
\multirow{3}{*}[-10pt]{\makecell*[lc]{Guarantees hold even when the budget constraint \\ is given 
by a monotone, symmetric function \\ (which need not be convex or homogeneous); \\ see
Remark~\ref{genfun}}} \\ \cline{1-2}
\makecell*[l]{\normmwis on \\ $k$-set system} & 
%\multirow{2}{*}[-2pt]
\makecell*[lc]{$(k+1)(1+\ve)$ (Theorem~\ref{normmwisthm})} \\ \cline{1-2}
\makecell*[l]{\normmatch and \\ norm-budgeted $b$-matching} & 
%\multirow{2}{*}[-2pt]
\makecell*[lc]{$3+\ve$ (Theorem~\ref{normmatchthm})} \\ \hline
\makecell*[l]{\normsched and \\ \normfl} & \makecell*[l]{$4.582+\ve$ (Theorem~\ref{normschedthm})} \\ \hline
\makecell*[l]{\normsched on \\ related machines} &
%\multirow{2}{*}[-2pt]
%\makecell*[lc]{$2+\ve$ (Theorem~\ref{ident-biptas}) \\ PTAS
%(Theorem~\ref{ident-ptasthm})} & 
\makecell*[lc]{PTAS (Theorem~\ref{ident-ptasthm}: identical machines;
  \\ Theorem~\ref{rel-ptasthm}: related machines)} & 
%\makecell*[l]{PTAS requires stronger oracle access to the norm} 
\\ \hline
\makecell*[l]{\normsap and \\ \normsepfl} & 
\makecell*[l]{$\gm(1+\ve)\bigl(\frac{e}{e-1}\cdot\rho+3\bigr)$ and \\
$O(\beta)$ (Theorems~\ref{normsepflthm} and~\ref{normsapthm})} &
\makecell*[l]{Given: $(\rho,\gm)$-approximation for \constr \mwis \\
on $\M_i$, or $\beta$-approximation for \mwis on $\M_i$, \\ for every machine/facility $i$} \\ \hline
%\makecell*[l]{\normfl} & \makecell*[l]{$O(1)$ (Theorem~\ref{normflthm})} \\ \hline
%\makecell*[l]{\normsepfl} & \makecell*[l]{$O(\beta)$ (Theorem~\ref{normsepflthm})} &
%\makecell*[l]{Given $\beta$-approximation for \mwis problem \\ on $\M_i$ for every facility $i$} \\ \hline
\end{tabular}
\captionsetup{font=small, labelfont=small}
\caption{A summary of (most of) the problems considered and the approximation factors obtained. 
%The PTAS for \normsched on identical machines is the only result that requires stronger
%oracle access (than a feasibility oracle) to the norm.
(We have not listed here the results we obtain for submodular norm-budgeted packing problems in
Section~\ref{subnb}.) 
All our results require only a feasibility oracle for the norm.
%\begin{comment}
%The guarantees for \normknap, \normmwis, and \normmatch actually apply even 
%%if we consider a 
%the budget constraint is given by a monotone, symmetric function (which need not be convex
%or homogeneous); see Remark~\ref{genfun}. 
%\end{comment}
\label{restable}}
\end{table}

\vspace*{-1ex}
\paragraph{Modeling power of monotone-symmetric-norm packing constraints.}
Besides the fact that monotone, symmetric norms include various
prominent norms of interest, %($\ell_p$ norms, $\topl$ norms etc.), 
a less-evident source of modeling power of monotone, symmetric norms, which was alluded to 
earlier, stems from the fact that a single
monotone-symmetric-norm packing constraint can be used to aggregate multiple norm-budget
constraints.
 %\begin{comment}

%For instance, %in a scheduling problem, where jobs have 
As an illustrative example, consider a setting involving selecting jobs for processing and
assigning them to servers (as in \normsched).  
A natural constraint that one would like to impose is a makespan bound
$\bigl\|\vec{\load}\bigr\|_\infty\leq B_{\mksp}$, where $\vec{\load}$ is the load vector
from the assignment, ensuring that no server is overloaded. %i.e., impose a makespan bound   
We may seek more fine-grained load balancing, where we also look to avoid (larger)
congestion hot-spots by imposing a bound %say $B_{\ell}$,
on the total load handled by any set of $\ell$ machines; this yields a norm constraint  
$\topl(\vec{\load})\leq B_{\topl}$. Additionally, energy considerations often arise and 
may dictate a bound on the total energy consumed; the energy consumed
%Energy expended by a server is often modeled as being proportional to the square of the
%load assigned to it, so the total energy consumed corresponds to the
is typically modeled by the $\ell_2^2$ objective of the load vector, so this leads to the 
norm constraint $\bigl\|\vec{\load}\bigr\|_2\leq B_{\energy}$.
%Capitalizing on the versatility of %budget constraint coming from a 
%monotone, symmetric norms, 
Now observe that these multiple constraints can be easily captured by defining the
aggregate monotone, symmetric norm
$f(x):=\max\left\{\frac{\|x\|_\infty}{B_{\mksp}},\frac{\|x\|_2}{B_{\energy}},\frac{\topl(x)}{B_{\topl}}\right\}$, 
and imposing the norm constraint $f(\vec{\load})\leq 1$.
%setting an $f$-norm budget of $1$. 
%illustrating the versatility of the norm-budgeted packing setup.
%here, $B_{\mksp}$, $B_{\energy}$ and $B_{\topl}$ are the desired bounds on the makespan,
%energy consumed, and $\topl$ load respectively.

%But in addition, we may want the energy consumed to be bounded ($\ell_2$-norm bound); or
%or achieved more fine-grained load balancing to avoid larger 
%where we bound the average load handled by any set of $\ell$ machines

It is not hard to imagine packing problems that feature such multiple resource
constraints, and {\em the framework of norm-budgeted packing problems gives a convenient
way of incorporating such constraints} when these constraints arise from monotone,
symmetric norms. In fact, more generally, one can incorporate budget constraints arising
from increasing, symmetric, convex functions.% 
\footnote{Given an increasing symmetric, convex function $h:\Rp^m\mapsto\Rp$ with
$h(0)=0$, and a budget constraint $h(x)\leq B$, 
consider the function $f_h(x):=\inf\bigl\{\al>0: h(|x|/\al)\leq B\}$. It is not to
hard to see that $f_h$ is a monotone, symmetric norm, and the constraint $h(x)\leq B$ is
equivalent to the norm-budget constraint $f_h(x)\leq 1$.}
Note that we crucially rely here on the {\em generality} (specifically, the closure
properties) of monotone, symmetric norms (which allows us to work with the maximum of the
scaled individual resource constraints%
\footnote{Observe that even if the individual resource constraints come from a
special class of norms, say $\ell_p$ norms, the resulting aggregated norm will 
not in general belong to the same class.}) 
and the versatility afforded by the setting of
norm-budgeted packing problems where we allow a budget constraint specified by a norm
belonging to this rich class of norms.
%an arbitrary monotone, symmetric norm. 
%Thus, allowing for general monotone, symmetric norms, yields

\medskip
Covering and packing problems can be viewed as ``flipped'' versions of each other; in the
covering counterpart of a norm-budgeted packing problem, we need to serve {\em all}
entities and minimize the norm of the size-vector induced by a solution. 
%(Depending on the underlying packing problem, the resulting minimum-norm covering problem
%may become trivial.) 
%A minimum-norm covering problem can be viewed as the ``flipped'' version of a
%norm-budgeted packing problem, where we must serve
%We feel that both of these views are
As discussed under ``Technical challenges and overview'' below,
the techniques developed for minimum-norm
(covering) optimization problems 
%(which, recall, has solely considered covering problems)
are not strong enough to yield ``true'' approximation guarantees for norm-budgeted packing 
problems, where we do not violate the norm budget. 
However, interestingly, and somewhat surprisingly, 
for certain norm-budgeted packing problems, we obtain {\em stronger guarantees} than what
is known for the covering counterpart of the problem. 
For instance, for {\em norm-budgeted matching}, %(see Section~\ref{probdefs}), 
wherein the norm-budget constraint applies to the size-weighted characteristic vector,
we obtain a $(3+\ve)$-approximation (Theorem~\ref{normmatchthm}), 
but for the covering version, where all nodes have to be matched, 
%i.e., one seeks a perfect matching, 
even on bipartite graphs, nothing better than an $O(\log n)$-approximation is known and the
natural LP-relaxation has large integrality gap~\cite{ChenLRZ25}. 
A similar situation arises %A similar , 
with {\em norm-budgeted $k$-facility location} (\normfl), %see Section~\ref{probdefs}),
wherein we need to open $k$ facilities and select clients to be assigned to facilities,
%assigned to facilities 
and there is a budget constraint on the norm of the resulting facility-cost vector. We
obtain an $O(1)$-approximation for \normfl (Theorem~\ref{normflthm}). In stark contrast,
for the corresponding covering problem, where all clients have to be assigned and we seek to
minimize the norm of the facility-cost vector,%
\footnote{This is an instance of the {\em generalized $(f^{\into},f^{\out})$-clustering}
problem introduced by~\cite{HeroldKS25}, where the inner norm $f^{\into}$ used to
aggregate the assignment costs of clients assigned to a facility and calculate the
cost of the facility is $\ell_1$, and the outer norm $f^{\out}$ applied to the
facility-cost vector is $f$.} 
no $O(1)$-approximation algorithm is known even for the special case where $f$ is the
$\ell_\infty$ norm, which is known as the {\em minimum-load $k$-facility location}
problem~\cite{AhmadianBFJSS18}.  

\medskip
Our work %addresses a gap in the literature in terms of bringing packing problems
opens up the area of norm-budgeted packing problems as a promising avenue for further
research. Our array of results shows that, notwithstanding the difficulties posed by a
hard norm budget constraint, one can develop strong approximation
guarantees for %norm-budgeted packing 
these problems. Various interesting research directions arise from our work. 
We mention a few of these below.
%First, %it would be interesting to consider 
%one would like 
One immediate direction is to consider the norm-budgeted packing versions of 
other combinatorial-optimization problems, including 
the generalized load balancing~\cite{DengLR23} and generalizing
clustering~\cite{HeroldKS25} problems that were proposed in the covering setting. 
%These problems use multiple norms, and two levels of aggregation of costs
In the covering setting, no $O(1)$-approximation is possible or known: generalized load
balancing is set-cover hard, and generalized clustering contains minimum-load $k$-facility
location as a special case, which has proved to be a bottleneck. However, as noted
earlier, these difficulties for the covering problem need not necessarily translate to the
norm-budgeted packing problem. It would be quite interesting if one could obtain
$O(1)$-approximation guarantees for the norm-budgeted packing counterparts of generalized
load balancing and generalized clustering.

Second, it would be very interesting to consider the {\em submodular} generalizations of
these norm-budgeted packing problems, wherein there is a monotone, submodular function that
specifies the reward of a set of items. This is a vast generalization of the current
setup, and it would be quite noteworthy if one could %show that one 
obtain guarantees for these submodular-reward problems that qualitatively match
the guarantees obtained for their regular (i.e., additive-reward) counterparts. 
%(especially if this can be obtained in some kind of black-box fashion).
In Section~\ref{subnb}, we obtain some partial results in this direction: 
{\em we obtain constant-factor approximation guarantees for submodular norm-budgeted
knapsack} 
(via a very different technique from what is used for \normknap) 
{\em and submodular norm-budgeted \maxgap on related machines}. We leave further
exploration of this class of problems for future work.

Finally, another exciting direction is to consider norm-budgeted packing problems in more
general environments, such as stochastic settings (as was done in~\cite{IbrahimpurS20} for
minimum-norm optimization), and online settings (analogous to the work
of~\cite{KesselheimMS23,KesselheimMS24,KesselheimMPS25}).

\subsubsection*{Technical challenges and overview}
One of the main challenges that arises when working with the generality of an arbitrary
monotone, symmetric norm $f$ is that the norm $f$ may couple the coordinates of the size
vector in complex ways. 
This makes it difficult to infer bounds on the norm $f$ from bounds on its individual
coordinates. 
(In contrast, with $\ell_\infty$ there is no coupling; with $\ell_1$,
the coupling manifests as a simple sum over all coordinates; with $\ell_p$ norms, one can
often move to the $\ell_p^p$ objective, which is separable over the coordinates; see 
``$\ell_p$ norms'' in Section~\ref{msn-prelim}.)

Two main insights have emerged from the work on minimum-norm covering optimization
%problems has developed various techniques 
to overcome this difficulty: 
%In particular, yielded two main insights to overcome this difficulty: 
%that have emerged from this work are that: 
(1) it suffices to reason about $\topl$ norms, for
logarithmically many $\ell$-values (say all powers of $(1+\ve)$); and (2) For a $\topl$ norm,
one can move to a suitable {\em coordinate-wise separable} proxy function by guessing certain
coordinates of the vector to which the norm is applied. Together, these two insights
enable one to treat the minimum-norm covering problem as a collection of logarithmically
many min-sum constraints, and this is the perspective that has largely led to the various
positive results for minimum-norm covering
problems~\cite{ChakrabartyS19a,IbrahimpurS21,DengLR23}. 

When looking to apply these same insights to norm-budgeted packing problems, where the
norm appears in the constraint, we run into an immediate stumbling block.
%To our knowledge, {\em all} prior work on minimum-norm optimization problems has focused
%on settings where the norm is in the objective, and various techniques have been devised
%for approximately minimizing the norm. 
%In our case, we have a norm-based budget constraint, and one of
%the main challenges lies in ensuring that this hard budget constraint is respected, as
%opposed to violated by a $(1+\ve)$-factor. 
%However, 
%Except for the simplest setting of $\topl$ norms
Essentially, {\em all} of this machinery %and 
%the tools developed for minimum-norm covering problems 
leads to an {\em inherent} $(1+\ve)$-factor violation in the norm. %For instance, %---e.g.,
In particular, the only known way of reasoning about a general monotone, symmetric norm is via
controlling the $\topl$ norms, but (D1) controlling $\topl$-norms only for $\ell$-values
that are powers of $(1+\ve)$ is too coarse to yield a tight bound on the norm. 
Furthermore, %the move to separable proxy costs 
%for a $\topl$ norm is based on estimating the $\ell$-th largest coordinate of the vector
%to which the norm is applied (i.e., the size-vector $\svec{T}$ in our case). 
(D2) for problems like \normsched, where coordinates of the size vector $\svec{T}$ are
obtained by aggregating the $\sz_e$'s for elements $e\in T$, we can only guess 
%this $\ell$-th largest entry 
the (logarithmically-many) coordinates of $\svec{T}$ within a $(1+\ve)$-factor, and this
is again too coarse to yield a tight bound on the norm (as seen from Lemma~\ref{toplestim}
(a)).  
(We remark that for the very special case when $f$ is a $\topl$-norm,
difficulty (D1) does not arise, and one can overcome difficulty (D2) in a simple way; 
%to obtain approximation guarantees; 
see Section~\ref{refine}. 
%and one can obtain 
For example, when $\svec{T}$ is the size-weighted characteristic vector, we can guess 
%$\svec{T}^{\down}_\ell$ exactly.)
the $\ell$-th largest entry of $\svec{T}$ exactly.) 
%and thereby reduce $\topl$-budget constraint to a budget constraint on the sum.)

The upshot is that for a general monotone, symmetric norm, the techniques
developed for minimum-norm covering problems 
%do not seem strong enough to yield guarantees where we do not violate the norm budget.
%These techniques can be utilized to 
only seem suitable for obtaining {\em bicriteria} solutions, where we violate the norm
budget by (at least) a $(1+\ve)$-factor.
(In Appendix~\ref{append-budgviol}, we sketch how such bicriteria solutions can be
obtained.)  
Moreover, unlike the setting with $\ell_1$ norms, the gap between bicriteria and
unicriteria solutions may be quite large, even when every single element ``fits within the
budget''; Theorem~\ref{knapbad} demonstrates this for \normknap. This also implies that
the natural convex-programming relaxation that incorporates the norm-budget constraint
has a large integrality gap (see Theorem~\ref{knapbad}). 
%(since the optimum of the convex program is well behaved under budget violations)
%violating the budget by even a
%$(1+\ve)$-factor can lead to a drastic increase in the optimum value, even when 
%$f\bigl(\svec{\{e\}}\bigr)$ is small compared to the budget $B$
%every individual element 
Thus, the chief challenge in obtaining a true 
approximation algorithm for a norm-budgeted packing problem lies in ensuring that the
norm-budget constraint is not violated, %Therefore, 
%Indeed, the lack of a ``target'' vector to aim for (compiled from the guessed coordinates
%of an optimal solution) considerably ties up one's handds, 
and we need to come up with novel ideas to address this challenge. 

\medskip
We come up with two main ideas to handle the norm budget constraint.
For a vector $v\geq 0$, we use $v^\down$ to denote $v$ with its coordinates sorted in
non-increasing order. 
Let $\optset$ denote some fixed optimal solution, and $\OPT:=\rewd(\optset)$ be the
optimal value. 
For the norm-budgeted versions of knapsack, matching, and maximum-weight independent set
(\mwis), where $\svec{T}$ is the size-weighted characteristic vector of a solution
$T\in\sols$ (recall that this means that $\svec{T}_e=\sz_e$ if $e\in T$, and $0$
otherwise), we proceed from first principles to find a solution $T$ such that
$\svec{T}^{\down}\leq\svec{\optset}^{\down}$. This is based on an enumeration step, where
we group elements of similar reward, and guess, up to a certain bounded error, the number
of elements $\num_j$ that $\optset$ includes from each reward bucket $\itemset_j$. We
argue that we can obtain these $\num_j$ estimates in polynomial time, and given these
guesses, for \normknap, we proceed by simply choosing the $\num_j$ smallest-size items 
from each reward bucket $\itemset_j$. This yields the PTAS for \normknap (Section~\ref{knapsack}).
For \normmwis (and \normmatch), we need to employ a more sophisticated
greedy strategy, where we choose a suitable number of small-size elements from each prefix set
$(\itemset_0\cup\ldots\cup\itemset_j)$ (Section~\ref{mwis}).
Interestingly, %we only utilize that the norm $f$ is monotone and symmetric, but we do not 
our guarantees for \normknap, \normmwis, and \normmatch utilize {\em only} that $f$ is
monotone and symmetric, and so these guarantees hold when the budget constraint is
prescribed by any monotone, symmetric function (which need not be convex or homogeneous); see
Remark~\ref{genfun}. 
%symmetric; we do not need convexity or homogeneity, 

For the norm-budgeted versions of \gap, separable assignment problem (\sap), and \kfl,
where the coordinates of the size-vector $\svec{T}$ are obtained by aggregating individual 
$\sz_e$ values, we come up with a general reduction (Section~\ref{reduction}) that 
allows us to reduce our task to that of obtaining a {\em bicriteria} approximation for the
problem, provided that we can also solve a ``one-job-per-machine'' variant of the problem
where we are additionally constrained to assign at most one job per machine 
(see Theorems~\ref{lbredn} and~\ref{sepflredn}). This reduction is extremely useful because
the flexibility of working with bicriteria solutions allows one to effectively utilize the
machinery developed for tackling minimum-norm covering problems. In Section~\ref{sepfl},
we develop a configuration-LP rounding approach to obtain a
bicriteria approximation for \normsepfl, which contains \normsap, \normfl (all defined in 
Section~\ref{probdefs}), and \normsched, as special cases. 
%Also, 
The one-job-per-machine problem becomes a special case of \normmwis on a suitably-defined
independence system. %we use this in conjunction with 
Combining these ingredients via our reduction 
yields our $O(1)$-approximation for \normsched and \normkfl. For
\normsap and \normsepfl, %(defined in Section~\ref{probdefs}), 
where the input specifies an independence system $\M_i$ for every machine/facility $i$,  
given a $\beta$-approximation algorithm for solving (standard) \mwis on $\M_i$, or a
bicriteria $(\rho,\gm)$-approximation algorithm for \constr \mwis on $\M_i$, we obtain
approximation factors of $O(\beta)$ or $O(\gm\rho)$
%an $O(\beta)$-approximation algorithm 
for the norm-budgeted problem.

We point out that the above reduction does not by itself alleviate the difficulty of
obtaining a true approximation guarantee for a norm-budgeted packing problem. 
%(due to the
%limitations of techniques used for minimum-norm covering problems for achieving this end). 
%do not seem amenable to producing true approximation guarantees. 
Rather, the reduction illuminates the insight that, if a bicriteria approximation can be 
obtained without much difficulty (for instance, by utilizing the machinery for
minimum-norm covering problems),  
then the {\em crux} of obtaining a true approximation guarantee boils down
to obtaining a true approximation guarantee for the one-job-per-machine variant of the
problem.  
%and consequently, one can focus attention on this latter problem. 
In a certain sense, this alludes to the more fundamental nature of \normknap, \normmatch,
and \normmwis problems. 

\vspace*{-1ex}
\paragraph{Submodular norm-budgeted \{knapsack,\,\maxgap on related machines\}
  (Section~\ref{subnb}).} 
For submodular norm-budgeted knapsack (\subnbknap), we proceed in a
fundamentally-different manner from (regular) \normknap. For \normknap, as outlined above, we
use a reward-bucketing approach. In essence, with reward-bucketing, we set things up so
that any solution choosing the correct number of items from each reward bucket 
yields good reward, and a particular choice---picking the smallest-size items---ensures
feasibility. With a 
submodular reward function, it is unclear how to define reward buckets, since 
%the contribution of an item to the reward of a set depends on other items in the set. 
the reward function is not separable across items;
moreover, picking the smallest-size items is not a strategy that yields good objective value for
a cardinality-constrained submodular-maximization problem.

Our approach is based, loosely speaking, on an alternate {\em size-bucketing} approach,
where the size buckets are defined using an optimal solution. %we set things up 
%so that any solution picking a certain number of items from each size bucket is guaranteed
%to be feasible. 
However, unlike
reward-bucketing, we cannot quite identify these size buckets, and items in a size-bucket
need not have similar size! Despite these challenges, we argue that, given a target
reward-sequence of incremental rewards to obtain from these size buckets, one can
build up a solution by including a suitable set of items from each size bucket. 
In contrast with reward-bucketing, the size-bucketing approach is tailored so that any choice
of picking a certain number of items from each size bucket yields feasibility, and in order
to obtain good reward, one solves a cardinality-constrained submodular-maximization
problem. The choice of the target incremental-reward sequence is rather tricky, and is
somewhat correlated with the execution of the algorithm. 

\smallskip
For submodular norm-budgeted \maxgap (\subnbsched) on related machines
(Section~\ref{submodlb}), we observe that 
the reduction used for \normsched still works with submodular rewards. Consequently, as
before, we need to: (a) obtain a bicriteria guarantee for the problem; and (b) solve the
one-job-per-machine variant of the problem. The latter essentially boils down to solving
\subnbknap: for identical machines, the one-job-per-machine problem is precisely \subnbknap; 
%For the special case of identical machines,
for related machines, we do not have such a crisp correspondence but one can nevertheless
argue that the guarantee of our algorithm for \subnbknap carries over to this problem.
So we focus on task (a), for which we devise an algorithm that
repeatedly calls an algorithm for {\em knapsack-constrained submodular maximization}.

\vspace*{-1ex}
\paragraph{PTAS for \boldmath\normsched on identical and related machines
  (Section~\ref{relatedmc}).} 
The reduction-based approach used to tackle \normsched (on unrelated machines) is too
coarse to yield a PTAS for related machines, and  
%the PTAS for identical and related machines is obtained via 
we utilize a very different approach to obtain the PTAS. This is quite problem-specific
and is the most
technically-involved portion of the paper. This section can be read independently of
Sections~\ref{mwis}--\ref{refine}. Similar to \normknap, we first
identify a set of jobs, $A$, assigned by a near-optimal solution $\tsg:A\mapsto[m]$. 

For identical machines (Section~\ref{ident-ptas}), we start by guessing the ``lonely''
jobs in $A$: these are jobs $j$ for which no other job is assigned by $\tsg$ to the
machine $\tsg(j)$; one can argue that they form a prefix of $A$ when we consider jobs
in $A$ in non-increasing order of size. %their $p_j$'s. 
%assigned all by themselves on their respective machines by $\tsg$
One useful insight is that the loads on the remaining machines $I$ (i.e., machines not
holding lonely jobs) are roughly balanced under $\tsg$. Given this, we can categorize jobs
as ``large'' or ``small'', %define some large jobs and some small jobs, 
where large jobs have the property that only $O(\frac{1}{\ve}\bigr)$ such jobs
are assigned to any given machine by $\tsg$. Hence, using a job-configuration enumeration
approach, we can find an assignment of these large jobs to the remaining machines that is
consistent with $\tsg$. Finally, to assign the small jobs, conceptually, we solve a convex
program to find a fractional assignment of these jobs. This convex program is structured
enough that we can give a closed-form expression for the optimal solution (see
Lemma~\ref{extnlem}). We round this fractional assignment using \gap rounding. The
resulting assignment need not be feasible, but removing at most one job per machine would
make it feasible. %The issue however is that 
These removed jobs may however carry large reward, so instead of dropping them, we
temporarily create some extra machines to hold these removed jobs. One can argue that the
number of extra machines is $O(\ve)|I|$; hence, 
%constitute only an $O(\ve)$-fraction of the total number of machines, 
retaining the $|I|$ largest-reward machines yields a feasible solution without
sacrificing the reward by much.

The PTAS for related machines (Section~\ref{rel-ptas}) is considerably more
complicated. At a high level, we are able to eventually group machines into 
$O\bigl(\frac{\log m}{\ve}\bigr)$ groups so that,
roughly speaking, we can treat each group as an identical-machines instance, 
%to eventually obtain a collection of $O(\log m)$
%identical-machine instances, 
and extend, to an extent, the approach used in the PTAS for identical machines to solve
these instances. One of the novel ideas here is that 
this grouping is {\em not based on machine speeds}, as is often the case 
%in the scheduling literature 
when working with related machines (e.g., speed smoothing~\cite{ImKPS18}), and 
also what is used for the special case of non-uniform \mkp~\cite{ChekuriK05}. 
Instead, the grouping is based on the
{\em total work assigned to a machine}, where work (as opposed to load) of a machine
denotes the total processing time of jobs assigned to the machine. Indeed, the speeds of
machines in the same group can vary considerably. One benefit of considering the
work-vector is that this is much more structured compared to the load vector: in
particular, one can assume that the work assigned to a machine is non-decreasing
in its speed. %it is sorted in non-decreasing order of speed. 
So given a
work-vector for a set of machines, we know how to assign (the jobs corresponding to) its
coordinates to machines. Given this, one can now ignore speeds, and treat each group as an
collection of identical machines, which %as noted above, 
yields a way forward by
leveraging and building upon suitable ideas from the PTAS for identical machines.

\subsection*{Related work}
We limit ourselves to a discussion of work that is more closely related to monotone,
symmetric norms, or their generalizations.
%and do not discuss $\ell_p$-norms. 
As noted earlier, Chakrabarty and Swamy~\cite{ChakrabartyS19a} initiated the study of
minimum-norm optimization problems, as a far-reaching generalization of some earlier work
on $k$-clustering that considered $\topl$ norms and ordered
norms~\cite{ByrkaSS18,ChakrabartyS18}. They devised constant-factor approximation
algorithms for the minimum-norm generalizations of load balancing and $k$-clustering. This
led to much follow-up work on minimum-norm optimization problems that considered other
combinatorial-optimization problems~\cite{ChenLRZ25}, more-general ways of aggregating
costs using monotone, symmetric
norms~\cite{DengLR23,AbbasiBBCGKMSS23,HeroldKS25,HeroldKS26}, 
and norm-based objectives in stochastic settings~\cite{IbrahimpurS20,IbrahimpurS21} and
online settings~\cite{KesselheimMS23,KesselheimMS24,KesselheimMPS25}. The latter work on
online problems also applies to objectives that are more general than monotone, symmetric
norms. 

All of this work focuses on minimum-norm covering problems, where the norm appears in the
objective. 
However, Kesselheim et al.~\cite{KesselheimMPS25} consider, as a means of solving an
online covering problem with a norm-based objective, the ``flipped'' online packing
problem where %obtained by 
the norm in the objective of the covering problem yields a norm budget constraint in
the packing problem, and the objective in the packing problem is,
%considers where the objective is, 
roughly speaking, to maximize the number of covering constraints that are satisfied. 
Their results violate the norm budget constraint, which is unavoidable in the online
setting. Our results can be seen as complementary to and orthogonal to, their
results: we consider the offline setting, and are able to obtain stronger constant-factor
approximation guarantees and {\em without violating the norm budget constraint} in the
offline setting. 
Neogi et al.~\cite{NeogiPS24} consider a mechanism-design problem with multiple $\topl$
budget constraints. The underlying algorithmic problem is to maximize a set-based reward
function $v(T)$ subject to the size-weighted characteristic vector of $T$ satisfying
$k$ given $\topl$ budget constraints, and they obtain an $O(k)$-approximation for this
when $v$ is subadditive.%
\footnote{Their focus is on the mechanism-design setting with private item
sizes, but their guarantee does not improve in the algorithmic setting where all
information is public.}
Since multiple $\topl$ budget constraints can be folded into one monotone, symmetric norm,
our work yields the following much-improved guarantees for their problem: a PTAS for their
$v$ is additive, and an $O(1)$-approximation when $v$ is submodular.

\section{Preliminaries and notation} \label{prelim}
For an integer $k$, we use $[k]$ to denote $\{1,\ldots,k\}$, and $\dbrack{k}$ to denote
$\{0,1,\ldots,k\}$. 
Recall that for a vector \mbox{$v\geq 0$}, we use $v^\down$ to denote $v$ with its coordinates
sorted in non-increasing order. For a subset $S$ of coordinates, we use $v(S)$ to denote
$\sum_{j\in S}v_j$. 
For a singleton set $\{e\}$, slightly abusing notation, we use the more-compact
expression, $f(\sz_e)$ to denote $f\bigl(\svec{\{e\}}\bigr)$. 

%With the exception of the PTAS for \normsched on identical machines, 
{\em All} our algorithms 
%use value-oracle access to the underlying monotone, symmetric norm 
%$f$. In fact, we 
only need a {\em feasibility oracle} for the norm $f$, which, given a candidate vector $v$ 
(e.g., $\svec{T}$) and scalar $\ld$, determines if $f(v)\leq\ld$ or not.% 
\footnote{This is a very basic primitive that is even weaker than assuming value-oracle
access to the norm $f$, and can often be obtained for a rational vector
$v$ even when the norm value itself may be irrational. For instance for $\ell_p$ norms,
where $p$ is an integer, we can efficiently compute $\ell_p(v)^p$ and compare this with
$\ld^p$.}

\subsection{Problem definitions} \label{probdefs}
%Besides norm-budgeted knapsack, and norm-budgeted maximum GAP (which were described
In Section~\ref{contrib}, we defined the \normknap and \normsched problems.
%the norm-budgeted knapsack and norm-budgeted maximum GAP problems. 
In addition to these problems, we consider various other %the following 
norm-budgeted packing problems, which we now define.
%canonical problems that are captured by the above setup.  
%that are captured by the \normmis problem. 
Specifying a norm-budgeted packing problem %in the above setup 
entails specifying the
candidate-solution set $\sols$ (i.e., the independence system)
and the size-vector $\svec{T}$ induced by a candidate solution $T\in\sols$ (which also
determines the dimension of the norm $f$). But it will sometimes be more natural and
intuitive to describe the problem first in its native context and then show how it
%can be cast as a special case of \normmis.
can be cast in the above framework. 
%The first three problems, and the last three problems listed below are in
%increasing order of generality; 
%and the last three problems 
Fig.~\ref{probcomp} depicts some relationships between the various problems considered in
this paper. %listed below. 
%some of which are also discussed below.

\begin{figure}[t!]
\center
\begin{picture}(300,240) %(0,-10)
%\graphpaper(0,0)(300,240)
\thicklines
\put(15,30){\small \nb{\knap}}
\put(55,32){\vector(1,0){45}}
\put(30,45){\vector(0,1){40}}
\put(105,30){\small \nb{\match}}
\put(105,18){\small on bipartite graphs}
\put(150,32){\vector(1,0){55}}
\put(210,30){\small \nb{\match}}
\put(205,45){\vector(-3,2){60}}
\put(10,102){\small \nb{\mwis}}
\put(10,90){\small on matroids}
\put(55,102){\vector(1,0){45}}
\put(110,102){\small \nb{\mwis}}
\put(110,90){\small on $2$-set systems}
\put(155,102){\vector(1,0){50}}
\put(210,102){\small \nb{\mwis}}
\put(210,90){\small on $k$-set systems} %${}^*$}
\put(10,225){\small \nb{\maxgap}}
\put(10,213){\small identical machines}
\put(30,205){\vector(0,-1){28}}
\put(10,165){\small \nb{\maxgap}}
\put(10,153){\small related machines}
\put(68,170){\vector(1,0){40}}
\put(110,165){\small \nb{\maxgap}} %${}^{\dag}$}
\put(168,170){\vector(1,0){35}}
\put(130,178){\vector(0,1){42}}
\put(210,165){\small \nb{\sap}}
\put(225,178){\vector(0,1){42}}
\put(115,225){\small \nb{\kfl}}
\put(150,228){\vector(1,0){55}}
\put(210,225){\small \nb{\sepfl}}
\qbezier(140,10)(420,-30)(280,130)
\put(280,130){\vector(-1,1){35}}
\end{picture}
\captionsetup{font=small, labelfont=small}
\caption{The norm-budgeted packing problems considered in the paper, and some
relationships between them. We abbreviate \normbudg{\ldots} to \nb{\ldots}. 
An arrow from problem $X$ to problem $Y$ denotes that problem $X$ can be cast as a special
case of problem $Y$. All problems considered are at least \nphard. %\\
%\addtocounter{footnote}{-2}
Some other hardness results follow from hardness results known for special cases:
%are as follows. 
(a) \normmwis on matroids %generalizes budgeted matroid independent set, which 
does not admit an FPTAS~\cite{AradKS24}; 
(b) 
%\normmwis on $k$-set systems generalizes $k$-matroid intersection, for which recently
There is an $\Omega(k)$-factor hardness of approximation for \normmwis on $k$-set
systems~\cite{LeeST25}; 
(c) %(\dag) 
\normsched %generalizes maximum-\gap, which 
is \apxhard~\cite{ChekuriK05}, and \normsched on identical machines 
%generalizes the uniform multiple-knapsack problem, which 
is strongly \nphard~\cite{ChekuriK05}. 
\label{probcomp}}
\end{figure}
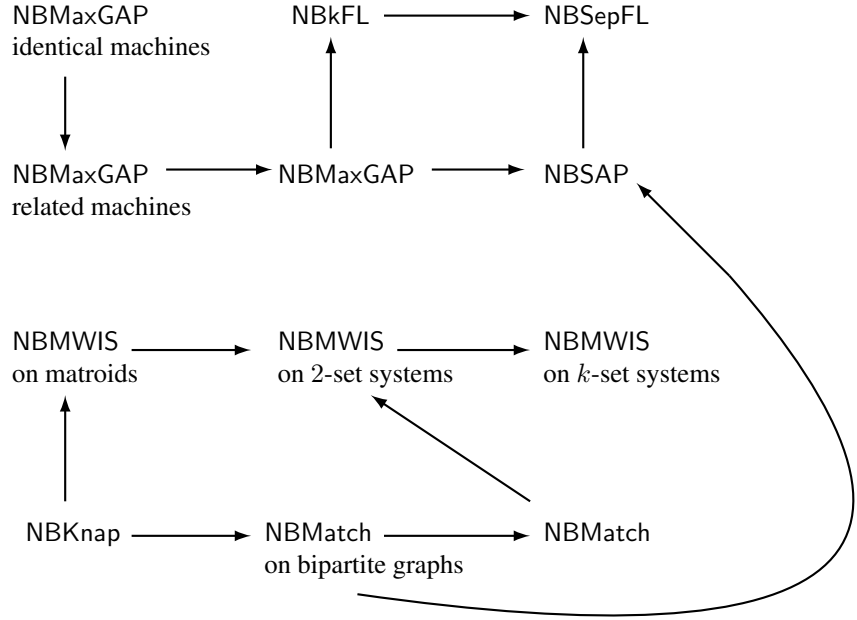

\begin{enumerate}[label=$\bullet$, leftmargin=*]
\item {\bf Norm-budgeted matching (\normmatch).} The ground set here is the edge-set
of an  undirected graph $G=(V,E)$, and $\sols$ is the collection of all matchings in
$G$. For a matching $T$, $\svec{T}\in\R^E$ is the size-weighted characteristic vector of
$T$. 
%given by $\svec{T}_e=\sz_e$ if $e\in T$, and
%$\svec{T}=0$ otherwise.

Observe that when the graph $G$ itself is a matching, \normmatch reduces to the
norm-budgeted knapsack problem. 
When $f$ is the $\ell_1$ norm, \normmatch has been studied as the budgeted matching
problem, which admits a PTAS~\cite{BergerBGS11,GrandoniRSZ14}. 

\begin{comment}
Observe that %norm-budgeted matching generalizes the norm-budgeted knapsack problem: given
\normknap can be cast as special case of norm-budgeted matching. Given a \normknap
instance with $n$ items, one can consider the complete bipartite graph whose node-set is
two copies of $[n]$, representing item-nodes (denoting items packed in the knapsack) and
``slot-nodes''; all edges incident to item-node $i$ have reward $\rewd_i$ and size
$\sz_i$. It is easy to see that feasible solutions to this \normmatch instance correspond
to \normknap solutions with the same reward, and vice-versa.

Observe also that with identical machines, this one-job-per-machine variant is the same as  
the norm-budgeted knapsack problem (where we also impose that at most $m$ items
can be packed in the knapsack). 
That is, \normknap with item-set $I$ can be cast as
the special case of \normmatch, where the bipartite graph $G$ on a
vertex-set consisting of two disjoint copies of $I$
bipartite graph $G=(I\cup I',E)$, with
$I'$ being a disjoint copy of $I$, and, for every $v\in I$, all edges in $\dt_G(v)$ have
the same reward and size. 
\end{comment}

\item {\bf Norm-budgeted maximum-weight independent set (\normmwis).}
%We now consider a quite general problem that abstracts all of the above problems. 
Here $\svec{T}$, for $T\in\sols$, is the size-weighted
characteristic vector of $T$. Clearly, \normmatch and \normknap are special cases of
\normmwis. 
When $f$ is the $\ell_1$ norm and the independence system is a matroid, we obtain the
budgeted matroid independent set problem, which admits a
PTAS~\cite{BergerBGS11,GrandoniRSZ14} but not an FPTAS~\cite{AradKS24}.
\normmwis on a $k$-set system generalizes $k$-matroid intersection, for which an
$\Omega(k)$-factor hardness of approximation was recently shown~\cite{LeeST25}.

\item {\bf Norm-budgeted separable assignment problem (\normsap).}
This is a generalization of norm-budgeted \gap, wherein the input also specifies an
independence system $\M_i=(J,\sols_i)$ for each machine $i$, and the output assignment 
must satisfy the additional constraint that the set of jobs scheduled on each machine
$i$ must be an independent set of $\M_i$; as before, the size-vector of an 
assignment is the resulting machine-load vector.
(\normsched is the special case where
each $\M_i$ is the free matroid, i.e., every subset of jobs is independent.)

As with \normsched, one can cast this in our general setup by considering the complete
bipartite graph $G=(J\cup[m],E)$ (where the rewards and sizes of edges are the rewards and
sizes of the corresponding (job, machine) pairs). The collection of candidate solutions is
now 
\[
\sols=\Bigl\{T\sse E:\ |T\cap\dt_G(j)|\leq 1\quad \forall j\in J,
\qquad \{j\in J: ij\in T\}\in\sols_i\quad \forall i\in[m]\Bigr\},
\]
and (as with \normsched) we have 
$\svec{T}\in\R^m=\bigl(\sum_{e\in T\cap\dt_G(i)}p_e\bigr)_{i\in[m]}$ for $T\in\sols$. 

\normsap is a substantial generalization of the {\em separable-assignment-problem} (\sap)
considered by~\cite{FleischerGMS11}. The problem studied
by~\cite{FleischerGMS11} does not involve job sizes or a norm budget constraint. 
Note that while job sizes and a bound on the load of a machine $i$ can be incorporated via
the independence system $\M_i$, 
{\em the norm-budget constraint couples the various machines} making the problem highly
{\em non-separable}. Hence, \normsap does not
fall into the framework of~\cite{FleischerGMS11}, and is a strict generalization 
of \sap.

The special case of \normsap where each $\M_i$ %is a matroid specifying 
%we are further constrained to 
encodes that at most one job can be scheduled on machine $i$ corresponds to 
%a special case of 
norm-budgeted matching on %Observe that when 
a bipartite graph whose vertex bipartition consists of job nodes and machine nodes. 
%norm-budgeted matching can be viewed as a
%variant of \normsched where we are further constrained to schedule at most one job
%per machine. 
As we will show in Section~\ref{reduction}, the approximability of the
general \normsap and \normsched problems is closely related to the approximability of this
one-job-per-machine variant. 

\item {\bf Norm-budgeted $k$-facility location (\normfl).}
In this problem, we are given a set of facilities $\F$, and a set of clients
$\C$. Assigning a client $j$ to facility $i$, incurs an assignment cost $c_{ij}$, and
earns reward $\rewd_{ij}$. For this problem, we always use $i$ to index facilities, and
$j$ to index clients. The input also specifies an integer $k\geq 0$, a monotone, symmetric
norm $f:\R^k\mapsto\Rp$, and budget $B$.
A solution specifies a set $F\sse\F$ of (at most) $k$ facilities to open, and an
assignment $\sg:S\mapsto F$ of some subset $S\sse\C$ of clients to these facilities. Such
an assignment induces a {\em facility-load vector} in $\R^k$ indexed by the facilities in
$F$, where the load on facility $i$ is the total assignment cost 
$\sum_{j\in S:\sg(j)=i}c_{ij}$ of clients assigned to $i$, and earns reward 
$\sum_{j\in S}\rewd_{\sg(j)j}$. The goal is to find a maximum-reward solution subject to
the norm-budget constraint $f(\text{facility-load vector})\leq B$.
Note that while facility-location problems 
%(that typically involve assigning all clients to open facilities) 
often assume metric assignment costs, we do not make this assumption here. 

As with \normsap, we also consider a more-general problem,
{\em norm-budgeted separable $k$-facility location} (\normsepfl), wherein we are
additionally given an independence system $\M_i=(\C,\sols_i)$ for each facility $i\in\F$,
and the output assignment must also satisfy that the set of clients assigned to each open
facility $i$ is independent in $\M_i$. For example, $\M_i$ can encode a capacity
constraint that at most $u_i$ clients may be assigned to facility $i$, and so \normsepfl
can be used to model, among other things, {\em capacitated \normfl}.

Observe that if we take $k=|\F|$ in \normsepfl, then we recover \normsap as a special case.
%\normkfl generalizes \normsched, which is the special case where
Similar to \normsap, we can cast \normsepfl in our general setup by considering the
ground set to be the edge-set of the complete bipartite graph $G=(\C\cup\F, E)$. The
reward and size of an edge are the reward and assignment cost of the corresponding $i,j$
pair respectively (where $i\in\F, j\in\C$). The solution-set is 
\[
\sols=\Bigl\{T\sse E:\, |T\cap\dt_G(j)|\leq 1 \ \,\forall j\in\C, \quad 
\{j\in\C: ij\in T\}\in\sols_i\ \,\forall i\in\F, \quad 
\bigl|\{i\in\F: T\cap\dt_G(i)\neq\es\}\bigr|\leq k\Bigr\}.
\]
One can easily verify that $(E,\sols)$ is an independence system. For a solution
$T\in\sols$, letting 
$F=\{i\in\F: T\cap\dt_G(i)\neq\es\}$, we have $|F|\leq k$, and we define
$\svec{T}=\bigl(\sum_{ij\in T\cap\dt_G(i)}c_{ij}\bigr)_{i\in F}$, with the understanding
that if $|F|<k$, we pad zeros to obtain a vector in $\R^k$.
\end{enumerate}

\subsection{Basic results for monotone, symmetric norms} \label{msn-prelim}
As noted earlier, our algorithms for the norm-budgeted versions of \gap, \sap, and \kfl
(where the size-vector aggregates individual $\sz_e$ values) are based on a reduction,
that, in part, requires one to obtain bicriteria approximation for the problem. In
doing so, we utilize some machinery developed for handling minimum-norm covering
problems. We will use the following notation and concepts from~\cite{IbrahimpurS21} 
(stated also in~\cite{ChakrabartyS19a}, albeit slightly differently).

\begin{definition}[\cite{IbrahimpurS21}] \label{posdef}
Let $\dt>0$. Let $N\geq 1$ be an integer.
%As in~\cite{IbrahimpurS21}, 
Define $\POS_{N,\dt}\sse[N]$ iteratively as follows: 
include the index $1$ in $\POS_{N,\dt}$;
as long as the largest index $\ell\in\POS_{N,\dt}$ is such that
$\ceil{(1+\dt)\ell}\leq N$, include $\ceil{(1+\dt)\ell}$ (which is larger than $\ell$)
in $\POS_{N,\dt}$ (and repeat). We have 
$|\POS_{N,\dt}|\leq O\bigl(\frac{\log N}{\dt}\bigr)$.

%We view $\POS$ as a set whose entries are sorted in increasing order.
For $i\in[N]$, let $\nxt_{N,\dt}(i)$ be the smallest index in $\POS_{N,\dt}$ strictly
larger than $i$; if no such index exists, define $\nxt_{N,\dt}(i):=N+1$ for notational
convenience. Similarly, let $\prev_{N,\dt}(i)$ be the largest index in $\POS_{N,\dt}$
strictly smaller than $i$; set $\prev_{N,\dt}(1):=0$. It is immediate from the definition
of $\POS_{N,\dt}$ that $\nxt_{N,\dt}(\ell)-1\leq(1+\dt)\ell$ for all
$\ell\in\POS_{N,\dt}$; it follows also that $\nxt_{N,\dt}(i)-1\leq(1+\dt)i$ for all
$i\in[N]$. 
\end{definition}

For a non-increasing vector $v\in\Rp^{\POS_{N,\dt}}$, define its expansion $v^\exp\in\Rp^N$  
as follows: $v^\exp_i=v_i$ for $i\in\POS_{N,\dt}$ and $v^\exp_i=v_{\prev_{N,\dt}(i)}$ for
$i\in[N]\sm\POS_{N,\dt}$.  
For a vector $u\in\R^N$ and $\tht\in\R$, define 
$Q^{>\tht}(u):=\bigl|\{i\in[N]: u_i>\tht\}\bigr|$. 
The following lemma (proved in Appendix~\ref{append-prelim}) shows that two vectors
sharing certain similar statistics have similar norm values. 
%We defer the proof to Appendix~\ref{append-prelim}.
Part (a) appears as Lemma 2.8 (b)
in~\cite{IbrahimpurS21}, and part (b) follows by mimicking the proof
in~\cite{IbrahimpurS21} for Lemma 2.8 (c). 

\begin{lemma} %[Lemma 2.8 in~\cite{IbrahimpurS21} paraphrased] 
\label{toplestim}
Let $\ve,\gm>0$ and $0<\dt\leq 1$.
Let $u\in\Rp^N$, and $v\in\Rp^{\POS_{N,\dt}}$ be a non-increasing vector. 
Let $\gnorm:\R^N\mapsto\Rp$ be a monotone, symmetric norm. %Let $\ve,\gm>0$.
\begin{enumerate}[(a), topsep=0.2ex, itemsep=0.1ex, leftmargin=*]
%\item If $u^{\down}_\ell\leq v_\ell$ for all $\ell\in\POS$, then
%$\topl[i](u)\leq\topl[i](v^\exp)$ for all $i\in[m]$, and hence,
%$\gnorm(u)\leq\gnorm(v^\exp)$. 
%
\item If $v_\ell\leq(1+\ve)u^{\down}_\ell+\gm$ for all $\ell\in\POS_{N,\dt}$, 
%then $\topl[i](v^\exp)\leq(1+\dt)(1+\ve)\topl[i](u)+i\gm$ for all $i\in[m]$, and hence,
then $\gnorm(v^\exp)\leq(1+\dt)(1+\ve)\gnorm(u)+N\gm\cdot\gnorm(1,0,\ldots,0)$. 

\item Let $\al\in\R^N$ be such that $\al^{\down}_1\leq v_1$ and
%$N^{>v_\ell}(\al):=\bigl|\{i\in[M]:\al_i>v_\ell\}\bigr|
$Q^{>v_\ell}(\al)\leq(1+\dt)(\ell-1)$ for all $\ell\in\POS_{N,\dt}$.
%Then, $\topl[i](\al)\leq(1+4\dt)(1+\ve)\topl[i](u)+5i\kp$ for all
%$i\in[m]$. Hence, 
Then $\gnorm(\al)\leq(1+3\dt)\gnorm(v^{\exp})$.

\begin{comment}
\item Let $v$ be as in part (a). Let $\dt\leq 1$. %(in $\POS=\POS_{m,\dt}$).
Let $\al\in\R^M$ be such that $\al^{\down}_1\leq v_1$ and
%$N^{>v_\ell}(\al):=\bigl|\{i\in[M]:\al_i>v_\ell\}\bigr|
$Q^{>v_\ell}(\al)\leq(1+\dt)\ell-1$ for all $\ell\in\POS_{N,\dt}$.
%Then, $\topl[i](\al)\leq(1+4\dt)(1+\ve)\topl[i](u)+5i\kp$ for all
%$i\in[m]$. Hence, 
Then
$\gnorm(\al)\leq(1+4\dt)(1+\ve)\gnorm(u)+5M\gm\cdot\gnorm(1,0,\ldots,0)$.
\end{comment}
\end{enumerate}
\end{lemma}

\paragraph{\boldmath $\ell_p$ norms.}
Although our focus is on handling general monotone, symmetric norms, we briefly discuss
here the setting of $\ell_p$ norms, which are perhaps the most prominent examples of such
norms. For problems such as norm-budgeted \{knapsack, matching, max-weight independent
set\}, where the size-vector $\svec{T}$ induced by a solution $T$ is the size-weighted
characteristic vector $(\sz_e\bon_{e\in T})_{e\in[n]}$ (where $\bon_{e\in T}$ is $1$ if
$e\in T$ and $0$ otherwise), one can simply move to the $\ell_p^p$ objective and cast the
budget constraint as
$\bigl\|\svec{T}\bigr\|_p^p:=\sum_{e\in T}\sz_e^p\leq B^p$. Thus, this reduces to an
$\ell_1$-budget constraint (with item sizes $\sz_e^p$), and guarantees for the latter
problem immediately translate to the norm-budgeted problem. So with an $\ell_p$ norm, we
obtain an FPTAS for \normknap, and PTASes for \normmatch, \normmwis with a matroid, and
\normmwis when the independence system is the intersection of two matroids.
However, for %the assignment and facility-location problems, 
norm-budgeted \{\maxgap, \sap, \kfl, \sepfl\!\}, where the coordinates of $\svec{T}$ are 
obtained by aggregating individual $\sz_e$ 
values, %we are not aware of any drastic simplification 
the $\ell_p$-norm setting does not seem to make the problem substantially
simpler. 

\begin{comment}
We can still utilize our reduction from Section~\ref{reduction}, and reduce to
the problem of obtaining a bicriteria solution. Here, one potential simplification is that
due to the separability of the $\ell_p^p$ objective, one can argue that if 
$\bigl\|\sz(T)\bigr\|_p\leq B$, then there can at most $\frac{1}{\ve^p}$ items $e\in T$
with $\sz_e\geq\ve\cdot B$. Therefore, such ``large'' items used by the optimal solution
$\optset$ can be enumerated. For \normsched, then a simpler approach to obtain a
bicriteria solution is to consider the remaining ``small'' items and write a convex
program incorporates the norm-budget constraint. One can round an optimal solution to this
using \gap rounding~\cite{ShmoysT93}. The load on each machine in the resulting integral
solution exceeds the load in the fractional solution by at most the size of one small item
$e$; by the triangle inequality
\end{comment}

%\paragraph
\subsection*{Scaling and rounding}
We may assume that $\{e\}\in\sols$ and $f(\sz_e)\leq B$ for all elements $e$,
as otherwise $e$ cannot belong to any feasible solution, and we can simply delete $e$ and
consider the independence system $\bigl([n]-\{e\},\sols':=\{T\in\sols: e\notin T\}\bigr)$.
Throughout, we use $\optset$ to denote some fixed optimal solution, and
$\OPT:=\rewd(\optset)$ to denote the optimal value.
Let $\rmax$ be the maximum reward of an element included in $\optset$. 
We may assume that $\rmax=\max_{e\in[n]}\rewd_e$, since we can simply ``guess'' a
maximum-reward element in $\optset$, and delete all higher-reward items from the instance.
By standard scaling and rounding ideas, incurring a $(1+\ve)$-factor loss in
approximation, we can move to an instance where all rewards are integers bounded by
$\frac{n}{\ve}$, so, unless otherwise stated, we will
assume that this holds in the sequel for all the problems considered.

\begin{theorem} \label{scaling}
Consider a norm-budgeted packing problem involving a ground set $[n]$ and element-rewards 
$\{\rewd_e\geq 0\}_{e\in [n]}$, satisfying the above assumptions. 
Let $\rmax:=\max_{e\in[n]}\rewd_e$. Let $\ve>0$.
Consider the same instance with element-rewards given by 
$\rewd'_e:=\ceil{\frac{n\cdot\rewd_e}{\ve\cdot\rmax}}$ for all $e\in[n]$.
Let $\OPT'$ denote the optimal value for the new instance. We have:
\begin{enumerate}[label=(\alph*), topsep=0.1ex, noitemsep, leftmargin=*]
%(a) 
\item $\OPT\leq\frac{\ve\cdot\rmax}{n}\cdot\OPT'\leq(1+\ve)\OPT$; \quad
%(b) 
\item If $T\sse[n]$ satisfies $\rewd'(T)\geq\frac{\OPT'}{\rho}$, for some $\rho\geq 1$,
%then 
we also have $\rewd(T)\geq\bigl(\frac{1}{\rho}-\ve\bigr)\OPT$. 
%\rewd(T)\geq\frac{\ve\rmax}{n}\cdot\rewd'(T)-\ve\OPT
%\geq\frac{\ve\rmax}{n}\cdot\frac{\OPT'}{\rho}-\ve\OPT\geq\frac{\OPT}{\rho}-\ve\OPT   
\end{enumerate}
\end{theorem}

\section{Norm-budgeted knapsack} \label{knapsack}
%\begin{defn} \label{normknap}
Recall that an instance of {\em norm-budgeted knapsack} (\normknap) 
is specified by an item-set $[n]$, 
%we are given $n$ items, with
non-negative rewards $\{\rewd_i\}_{i\in[n]}$ and sizes $\{\sz_i\}_{i\in[n]}$, 
%We also have 
a {monotone, symmetric norm} $f:\R^n\mapsto\R_+$, and a knapsack budget $B$. 
\begin{comment}
The goal is
to find a maximum-reward set of items that fits in the knapsack, where fitting in the
knapsack is determined by requiring that the norm applied to the size-vector obtained from
these items is at most the budget $B$. 
More precisely and formally, given a set $T\sse[n]$ of items, the {\em size-vector}
associated with $T$ is defined as the vector $\svec{T}\in\R^n$ with coordinates $\sz_i$ if
$i\in T$, and $0$ otherwise. 
\end{comment}
The goal is find $T\sse[n]$ that maximizes $\rewd(T):=\sum_{i\in T}\rewd_i$ subject to the
norm budget constraint $f\bigl(\svec{T}\bigr)\leq B$, where $\svec{T}_i$ is $\sz_i$ if 
$i\in T$, and is $0$ otherwise.  
%\end{defn}
%
We use items and elements interchangeably in this section.
We devise a polynomial-time approximation scheme (PTAS) for \normknap. 
%this problem.
%an $O(1)$-approximation algorithm for this problem. 
%(running in polynomial time).
%$n^{O(\log\log n)}$. 
%Our algorithm only requires value-oracle access to the norm $f$.

\begin{theorem} \label{normknapthm}
There is a %(polytime) $(2+\ve)$-approximation algorithm 
PTAS for \normknap, which obtains reward at least $(1-\ve)^2\OPT$ in 
$\bigl(\frac{n}{\ve}\bigr)^{O(1/\ve^2)}$ time, for any $\ve>0$. 
%with running time $n^{O(\log\log n)}$. 
%The algorithm uses only value-oracle access to the underlying
%monotone, symmetric norm. 
\end{theorem}

Before delving into the proof of Theorem~\ref{normknapthm}, we give some examples
illustrating the challenges that arise in working with norm-budgeted knapsack, 
and why \normknap is significantly harder to tackle than the standard knapsack
problem  (where $f$ is the $\ell_1$ norm).
We demonstrate that certain properties or approaches that apply to standard knapsack 
%problem (i.e., where $f$ is the $\ell_1$ norm) 
fail badly for \normknap, showing a sharp contrast between the two problems.
%In particular,
Theorem~\ref{knapbad} shows that increasing the budget even very slightly 
%i.e., by a $(1+\ve)$-factor, 
can cause a drastic increase in the optimum value, even when every
individual item has a small size relative to the budget. %(i.e., $f(\sz_e)\ll B$). 
This also implies that the natural convex-programming relaxation for \normknap has a
large integrality gap. 
%These facts are in stark contrast with the situation for standard knapsack.
%Finally, we show that the best of the maximum-reward item, and the a greedy approach based
%on picking items in non-increasing $\rewd_i/\sz_i$ order, can lead to arbitrarily poor
%solutions.
(We remark that for the quite special case of ordered norms, the dynamic-programming based
FPTAS for standard knapsack can be extended to yield an FPTAS for \normknap; see
Appendix~\ref{dpknap}.)  
%but this approach is not amenable to extension to anything beyond ordered norms.

\begin{theorem} \label{knapbad}
\
\begin{enumerate}[label=(\alph*), topsep=0ex, itemsep=0.1ex, leftmargin=*]
\item For any constants $M,\ve>0$, there is a \normknap instance %$\I$
%$\I=\bigl([n],\{\rewd_i,\sz_i\}_{i\in[n]},f,B\bigr)$ 
with $f(\sz_i)\leq\ve B$ for all items $i$, such that increasing the budget to $(1+\ve)B$ 
increases the optimum by a factor of at least $M$. This holds even when $f$ is a
$\topl$ norm.
%$\OPT_{\I}((1+\ve)B)\geq K\cdot\OPT_{\I}(B)$, where $\OPT_\I(W)$ denotes the optimum
%value with budget $W$.
%the optimum value with budget $(1+\ve)B$ is at least $K$

\item The following convex-programming relaxation for \normknap has unbounded integrality
gap. 
\begin{equation}
\max \ \sum_{i\in[n]}\rewd_ix_i \quad\ \ \text{\textnormal s.t.} \quad\ \  
f(\sz_1 x_1,\sz_2x_2,\ldots,\sz_nx_n)\leq B, \quad\ \ 
0\leq x_i\leq 1 \quad \forall i\in[n].
\tag{\cpknap} \label{knapcp}
\end{equation}
%is unbounded.
\end{enumerate}
\end{theorem}

\begin{proof}
Note that part (b) follows from part (a), because letting $\OPT_{\text{\ref{knapcp}}}$ denote
the optimal value of \eqref{knapcp}, we can see that 
$\OPT_{\cpknap[{(1+\ve)B}]}\leq (1+\ve)\OPT_{\text{\ref{knapcp}}}$, as scaling
down an optimal solution to ($\cpknap[{(1+\ve)B}]$) by a $(1+\ve)$-factor yields a
feasible solution to \eqref{knapcp}. 
%So $\lpopt[{\ref{knapcp}((1+\ve)B)}]\leq (1+\ve)\lpopt[{\ref{knapcp}(B)}]$, 
But part (a) shows that the optimal value of the problem with budget $(1+\ve)B$ can be
arbitrarily large compared to the optimal value of the problem with budget $B$. 

For part (a), let $\ell$ be such that $\ell\geq 1+\frac{1}{\ve}$, so
$(1+\ve)(\ell-1)\geq\ell$, and $n$ be such that $n\geq M(\ell-1)$. The simplest instance
demonstrating this is to take $n$ items, each with unit reward and size, $f$ as the
$\topl$ norm. Let $\OPT(B)$ denote the optimum value with budget $B$.
Taking $B=\ell-1$, we clearly have $\OPT(B)=\ell-1$; but since $(1+\ve)B\geq\ell$, and any
set of at most $\ell$ items have total size at most $\ell$, we have that
$\OPT\bigl((1+\ve)B\bigr)=n$. 

When all items have the same size, the problem can be solved in polytime (as the norm
budget constraint reduces to a cardinality constraint), but perturbing the above instance
avoids this over-simplification. For instance, for some $1\leq k<n$, for $i\in[k]$, we can
set the size and reward item $i$ to be some number $x_i\in[1,2]$, and take $\ell$ such that
$(1+\ve)(\ell-1)\geq\ell+k$; this yields the same outcome. 
\end{proof}

\subsection*{Proof of Theorem~\ref{normknapthm}}
%\subsection{\boldmath $O(1)$-approximation in polytime: proof of Theorem~\ref{normknapthm}}
%\label{normknapalg}
We begin by establishing some notation and terminology. 
%We remark that we have not sought to optimize the constant in the approximation guarantee,
%and instead chosen to keep exposition simple. 
%
%By padding the instance with items of zero reward and zero size, we may assume, for
%simplicity, that $n$ is a power of $2$. 
%Fix an optimal solution $\optset\sse[n]$, and let $\OPT=\rewd(\optset)$. 
We use items and elements interchangeably.
Recall that $\optset\sse[n]$ denotes some fixed optimal solution and
$\OPT=\rewd(\optset)$.  
\begin{comment}
By padding the instance with items of zero reward and zero size, 
%Again, by padding, 
we may assume for notational convenience that
$\nopt:=|\optset|$ is a power of $2$.\footnote{Note that we can achieve this by
padding the instance by at most the original number of items.}
\end{comment}
%Recall also that we may assume that all rewards are integers bounded by $\frac{n}{\ve}$.
Recall also that $\rmax$ denotes the maximum-reward of an item in $\optset$, and we may 
assume that $\rmax=\max_{i\in[n]}\rewd_i$, and 
%$\rmax$ is an integer that is at most $\frac{n}{\ve}$. 
all rewards are integers bounded by $\ceil{\frac{n}{\ve}}$ (due to Theorem~\ref{scaling}).
%Let $\nopt:=|\optset|$.

\begin{comment}
Let $\rmax:=\max_{i\in[n]}\rewd_i$. We may assume that an item with reward $\rmax$ is
included in $\optset$, since we can simply ``guess'' the maximum-reward item in $\optset$,
and discard all higher-reward items. 
\end{comment}

The idea is to bucket items having similar reward, %(up to a factor of $2$), 
and guess the number of items that $\optset$ includes from each bucket. 
%again within a factor of $2$. 
Some notation will be handy here. 
For an integer $j\geq 0$, define $\thresh_j:=\frac{\rmax}{(1+\ve)^j}$, and let
$\itemset_j:=\bigl\{i\in[n]:\frac{\thresh_j}{1+\ve}<\rewd_i\leq\thresh_j\bigr\}$
denote all items with reward roughly %(i.e., at most, and within a factor of $2$ of) 
$\thresh_j$; we call this a reward bucket. %$\frac{\rmax}{2^j}$. 
Note that there are $\nbuck\leq O\bigl(\frac{1}{\ve}\log\frac{n}{\ve}\bigr)$ such reward
buckets that  
together cover all items with non-zero reward.

Now, if we know that $|\optset\cap\itemset_j|\in[\ell,(1+\ve)\ell)$, 
%items from the %bucket corresponding to items with rewards in $(r/2,r]$, 
%bucket $\itemset_j$, 
then if we pick the $\ell$ {\em smallest-size items} from this bucket, we are assured that
the size-vector of our item-set is coordinate-wise at most $\svec{\optset\cap\itemset_j}$.
%the size-vector of the items from $\optset$ in this bucket. 
So if we do this for all buckets, then we ensure that we satisfy the budget constraint. 
Moreover, since all items in $\itemset_j$ have roughly the same reward, and we pick at
least $|\optset\cap\itemset_j|/(1+\ve)$ items from this bucket, it is not hard to
see that the reward that we obtain from the items we pick from $\itemset_j$ is at least 
$\rewd(\optset\cap\itemset_j)/(1+\ve)^2$, and this holds for all $j$.
%$\Omega(\text{reward that $\optset$ obtains from this bucket})$.
So we obtain a feasible solution of reward at least $\OPT/(1+\ve)^2$. 
Now, as stated, implementing the above plan would need 
$O\bigl(\nbuck^{(\log n)/\ve}\bigr)=n^{O(\log\log(n/\ve)/\ve)}$ time, since there are
%$O\bigl(\log\frac{n}{\ve}\bigr)$ buckets, and there are 
$O\bigl(\frac{\log n}{\ve}\bigr)$ choices of powers of $(1+\ve)$ for
$|\optset\cap\itemset_j|$ for each reward bucket. 
%We note that the approximation guarantee can be improved to $(1+\ve)$
%by considering powers of $(1+\ve)$ instead of $2$, so this yields a 
%``mildly quasi-polytime'' approximation scheme for the problem; see Section~\ref{qptas}. 

To refine the above approach and obtain polynomial running time, we obtain the
$|\optset\cap\itemset_j|$ estimates indirectly, by guessing the total
reward collected by $\optset$ from a reward-bucket, and then translating this to an
(more noisy) estimate of $|\optset\cap\itemset_j|$. This is similar to the
approach used by~\cite{ChekuriK05} for the multiple knapsack problem.
%We first guess the reward
%$\rewd(\optset\cap\itemset_j)$ for each reward bucket within an additive error of
%$\frac{\ve}{\rmax}{\nbuck}$, and then use this to obtain an estimate of
%$|\optset\cap\itemset_j|$. 

To elaborate, %we move to a version of the problem, where we are given a target reward,
%and we would like to return a \normknap solution that collects reward close to this
%target, if such a solution exists. 
by considering all values of the form $\rmax(1+\ve)^\ell$, where $\ell\in\Z_+$, in the
range $[\rmax,n\rmax]$, we may assume that we have an estimate $\optval$ such that
$\optval\leq\OPT\leq(1+\ve)\optval$.  
%and suppose that there is a feasible \normknap solution
%$\targopt\sse[n]$ such that $\rewd(\targopt) 
Let $\Dt=\frac{\ve\cdot\optval}{\nbuck}$, and 
%and let $\rnum_j\in\Z_+$ be the closest multiple of $\Dt$
let $\rnum_j:=\floor{\frac{\rewd(\optset\cap\itemset_j)}{\Dt}}$ for all
$j\in\dbrack{\nbuck}$. 
%So we have
%$\rnum_j\cdot\Dt\leq\rewd(\optset\cap\itemset_j)<(\rnum_j+1)\cdot\Dt$  
%$\rnum_j\cdot\frac{\ve\cdot\optval}{\nbuck}\leq\rewd(\optset\cap\itemset_j)<(\rnum_j+1)\cdot\frac{\ve\cdot\optval}{%\nbuck}$  
%(that is,
%$\rnum_j:=\floor{\frac{\nbuck\cdot\rewd(\optset\cap\itemset_j)}{\ve\cdot\optval}}$) 
%for all $j\in\dbrack{\nbuck}$. 
(Note that $\rnum_j=0$ if $\rewd(\optset\cap\itemset_j)<\Dt$.) 
%\frac{\ve\optval}{\nbuck}$.)

A key observation is that 
$\sum_{j=0}^\nbuck\rnum_j\leq\frac{\rewd(\optset)}{\Dt}\leq\bigl(1+\frac{1}{\ve}\bigr)\nbuck$, since
$\rewd(\optset)\leq (1+\ve)\optval$.
%Thus, $\{\rnum_j\}_{j\in\dbrack{\nbuck}}$ is a 
Since the number of sequences of $\nbuck+1$ nonnegative integers that sum up to at most
$\bigl(1+\frac{1}{\ve}\bigr)\nbuck$ is at most
$2^{O(\frac{\nbuck}{\ve})}=\bigl(\frac{n}{\ve}\bigr)^{O(\frac{1}{\ve^2})}$
(Claim~\ref{polybnd}), that is, 
polynomially bounded, by enumerating over all such sequences, we may assume that we know
$\rnum_j$ for all $j\in\dbrack{\nbuck}$.

%Let $\Dt:=\frac{\ve\optval}{\nbuck}$.
Since the rewards of all items in $\itemset_j$ is roughly $\thresh_j$, we can set
$\num_j:=\rnum_j\cdot\frac{\Dt}{\thresh_j}$
%\frac{\ve\cdot\optval}{\nbuck\cdot\thresh_j}$ 
(which need not be an integer),
for every $j\in\dbrack{\nbuck}$. The bounds on $\rewd(\optset\cap\itemset_j)$ translate to
the bounds 
$\num_j\leq|\optset\cap\itemset_j|<(1+\ve)\bigl(\num_j+\frac{\Dt}{\thresh_j}\bigr)$.
%note that $\num_j$ need not be an integer.

\vspace*{-1ex}
\paragraph{The algorithm.}
Given the $\{\num_j\}_{j\in\dbrack{\nbuck}}$ estimates, the algorithm is now fairly
immediate.  
For each $j\in\dbrack{\nbuck}$, let $S_j$ be the set of $\ceil{\num_j}$ smallest-size 
items from $\itemset_j$. We return $A:=\bigcup_{j\in\dbrack{\nbuck}}S_j$.

\paragraph{Analysis.} 
The following standard claim justifies our earlier statement about the
polynomial number of candidate $\{\rnum_j\}_{j\in\dbrack{\nbuck}}$ sequences. 
We include the proof in Appendix~\ref{append-knapsack}, for completeness.

\begin{claim} \label{polybnd}
There are at most $(2e)^{\max\{M,k\}}$ sequences of $k$ nonnegative integers that sum to
at most $M$. The same bound applies to the number of non-increasing sequences of $k$
integers chosen from $\dbrack{M}$.
\end{claim}

\begin{proof}[Finishing up the proof of Theorem~\ref{normknapthm}]
%\begin{proofof}{Theorem~\ref{normknapthm}}
By Claim~\ref{polybnd}, we may assume that %the $\rnum_j$ integers satisfy
$\rnum_j\cdot\Dt\leq\rewd(\optset\cap\itemset_j)<(\rnum_j+1)\Dt$,
%$\rnum_j\cdot\frac{\ve\cdot\optval}{\nbuck}\leq\rewd(\optset\cap\itemset_j)<(\rnum_j+1)\cdot\frac{\ve\cdot\optval}{\nbuck}$,  
and hence $|\optset\cap\itemset_j|\geq\num_j$,
for all $j\in\dbrack{\nbuck}$.
Given this, we argue that, for every $j\in\dbrack{\nbuck}$, we have
(a) $\svec{S_j}^{\down}\leq\svec{\optset\cap\itemset_j}^{\down}$; and
(b) $\rewd(\optset\cap\itemset_j)<(1+\ve)\rewd(S_j)+\Dt$. %\frac{\ve\cdot\optval}{\nbuck}$.

From (a), we obtain that $\svec{A}^{\down}\leq\svec{\optset}^{\down}$ so $A$ is a feasible
\normknap solution. From (b), we obtain that 
$\rewd(A)=\sum_{j\in\dbrack{\nbuck}}\rewd(S_j)
>\frac{\OPT-\ve\cdot\optval}{1+\ve}\geq\frac{1-\ve}{1+\ve}\cdot\OPT\geq(1-\ve)^2\OPT$.

To prove claims (a) and (b), consider an index $j\in\dbrack{\nbuck}$.
Both claims are fairly immediate, given the bounds on $|\optset\cap\itemset_j|$ and
$\rewd(\optset\cap\itemset_j)$. 
Since $S_j$ consists of the $\ceil{\num_j}$ smallest-size items from $\itemset_j$ and 
$|\optset\cap\itemset_j|$ is an integer that is at least $\num_j$, claim (a) follows.
We have 
$\rewd(S_j)\geq\frac{\num_j\thresh_j}{1+\ve}=\rnum_j\cdot\frac{\Dt}{1+\ve}$
since all items in $\itemset_j$ have reward at least $\frac{\thresh_j}{1+\ve}$.
So $\rewd(\optset\cap\itemset_j)<(1+\ve)\rewd(S_j)+\Dt$, and (b)
follows. 

We need to enumerate all possible $\{\rnum_j\}_{j\in\dbrack{\nbuck}}$ sequences and 
$\optval$ values, which takes 
$\bigl(\frac{n}{\ve}\bigr)^{O(1/\ve)}\cdot O\bigl(\frac{\ln n}{\ve}\bigr)$ time.
The analysis above shows that for the right guess, we obtain reward at least
$(1-\ve)^2\OPT$, so we simply return the maximum-reward feasible solution found across all
the guesses; detecting feasibility of a candidate solution can be done using a feasibility
oracle for the norm. 
\end{proof}

\begin{remark} \label{genfun}
Note that we only utilize that $f$ is monotone and symmetric---in concluding that
$\svec{A}^{\down}\leq\svec{\optset}^{\down}$ implies 
$f(\svec{A})\leq f(\svec{\optset})$---and nowhere use the fact that $f$ is convex and
homogeneous. Thus, our guarantee continues to hold more generally, when the budget
constrained is prescribed by a monotone, symmetric function (which need not be convex or
homogeneous).
\end{remark}

\section{Norm-budgeted maximum-weight independent set} \label{mwis}
Recall that in the {\em norm-budgeted maximum-weight independent set} (\normmwis) problem,
we are given an {\em independence system} $\M=([n],\sols)$, that is, $\sols\sse 2^{[n]}$ is
non-empty and is closed under taking subsets: $\es\in\sols$ and 
$A\in\sols, B\sse A\implies B\in\sols$. Elements in $[n]$ have non-nonnegative rewards
$\{\rewd_e\}_{e\in[n]}$ and sizes $\{\sz_e\}_{e\in[n]}$, and we have a monotone, symmetric
norm $f:\R^n\mapsto\R_+$ and budget $B$. The goal is to find a maximum-reward set
$T\in\sols$ satisfying $f(\svec{T})\leq B$, where $\svec{T}_e$ is $\sz_e$ if $e\in T$ and
$0$ otherwise. 

We need some terminology to state our result for a general independence system. 
%we state some terminology and notation 
%Given an independence system $\M$, 
Sets in $\sols$ are often called independent sets.
As is standard, we assume that $\M$ is specified via an independence oracle that
determines whether a given input set is independent.
The rank function $\rk:2^n\mapsto\Z_+$ of $\M$, is defined as 
$\rk(A):=\max\{|B|: B\sse A, B\in\sols\}$, i.e., $\rk(A)$ is the maximum size
of an independent set contained in $A$. 
%Independence systems can be quite general, and 
%
%\begin{defn} \label{indepclass}
%Let $\M=([n],\sols)$ be an independence system.
%\begin{enumerate}[label=(\alph*), topsep=0.1ex, noitemsep, leftmargin=*]
%\item 
We say that $\M$ is a {\em $k$-set system}, %~\cite{KorteV}, 
where $k\geq 1$, if for every set $A\sse[n]$, every {\em maximal} independent set
contained in $A$ has size at least $\rk(A)/k$.
%
%\item We say that $\M$ is a {\em $k$-extendible system}~\cite{Mestre06}, where $k\geq 1$,
%if for every $B\in\sols$, $A\sse B$, and $e\notin A$ such that $A\cup\{e\}\in\sols$, there
%exists $X\sse B-A$ with $|X|\leq k$ such that $B-X\cup\{e\}\in\sols$.
%\end{enumerate}
%\end{defn}
%
%The above clases of independence systems %based on certain structural properties 
%are often considered in algorithmic work on independence systems. For instance, it is
%known that the greedy algorithm yields a $k$-approximation to the maximum-weight
%independent set problem on any $k$-set system.
%It is well known that a $k$-extendible system is also a $k$-set system, and both these
%classes coincide with the class of matroids when $k=1$.
When $k=1$, we obtain the class of matroids. 
The following are two well-known examples of $k$-set systems.%
\footnote{These examples are actually {\em $k$-extendible systems}~\cite{Mestre06},
which form a subclass of $k$-set systems.} %gives various other examples.
%The following facts are well-known $k$-set systems. 
%(1) $k=1$ iff $\M$ is a matroid; 
%It is well-known and not hard to show that: 
\begin{enumerate}[label=$\bullet$, topsep=0.5ex, itemsep=0.1ex, leftmargin=*]
%(1) The collection of 
\item {\bf \boldmath Matchings and $b$-matchings.}. The collection of matchings of a graph
$G=(V,E)$ forms a $2$-set system. More generally, given node-degree bounds
$b:V\mapsto\Z_+$, a $b$-matching is a set $F\sse E$ such that $|\dt(v)\cap F|\leq b_v$ for
all $v\in V$. 
The collection of $b$-matchings of a graph also forms a $2$-set system.
%$2$-extendible system.
%(2) the collection of 

\item {\bf Intersection of matroids.} The collection of common independent sets of $k$
matroids, forms a $k$-set system. %$k$-extendible system.
\end{enumerate}

We design an $O(k)$-approximation algorithm for \normmwis when $\M$ is a $k$-set system,
by suitably extending the insights leading to the PTAS for \normknap.
%our $O(1)$-approximation algorithm for \normknap described in Section~\ref{knapsack}.  
Our approximation guarantee is actually more refined, and depends on a certain structural 
parameter of an independence system. %that we define below.
To motivate and define this parameter, we first discuss why the approach used for
\normknap does not quite work as is, and sketch the changes needed to handle \normmwis.

Recall that the idea underlying our algorithm for \normknap was to (eventually) obtain
estimates $\{\num_j\}_{j\in\dbrack{\nbuck}}$ for $|\optset\cap\itemset_j|$, where
$\itemset_j$ is a reward bucket consisting of all items with rewards roughly $\thresh_j$,
%the $\thresh_j$ values are geometrically decreasing, and
and $\nbuck=O\bigl(\log\frac{n}{\ve}\bigr)$ is the number of reward buckets. These
estimates satisfied the bounds
$\num_j\leq|\optset\cap\itemset_j|\leq(1+\ve)\bigl(\num_j+\frac{\Dt}{\thresh_j}\bigr)$
for all $j\in\dbrack{\nbuck}$, for a suitable value $\Dt$.
(Recall that $\optset\in\sols$ is an optimal solution.)
\begin{comment}
Observe that the $\num_j$ estimates are in a sense quite noisy, in that the gap between
the lower and upper bounds on $|\optset\cap\itemset_j|$ can be large relative to $\num_j$,
and moreover this gap increases (geometrically) with $j$. For norm-budgeted knapsack, this
did not present a problem, because $(\text{gap for bucket $j$})\times\thresh_j$ is
bounded, and this translates to a bounded gap in the reward obtained from this bucket
\end{comment}
We then chose a set $S_j$ of $\num_j$ items from each reward-bucket $\itemset_j$. 
The above bounds on $\num_j$ ensure that $\rewd(S_j)$ is close to
$\rewd(\optset\cap\itemset_j)$; 
also, importantly, we could choose $S_j$ so that
$\svec{S_j}^{\down}\leq\svec{\optset\cap\itemset_j}^{\down}$. 
The latter ensures that the union of the $S_j$-sets is feasible, and the former ensures
that this obtains reward close to $\OPT$.

Once we move to an independence system, even a matroid, %it becomes problematic to 
one cannot necessarily pick $\num_j$ items from $\itemset_j$ (even though $\optset$
does so) because the elements picked from earlier buckets may block us from doing so. 
That is, if $A_{j-1}\sse\itemset_0\cup\ldots\cup\itemset_{j-1}$ denotes the
previously-picked elements (i.e., elements picked from earlier buckets), then it need
not be that one can extend $A_{j-1}$ to an 
independent set by picking $\num_j$ elements from $\itemset_j$. A better option is to
consider the prefix-set $\pset_j:=\itemset_0\cup\ldots\itemset_j$; now,
if $|A_{j-1}|\leq\sum_{q=1}^{j-1}\num_q$ since
$|\optset\cap\pset_j|\geq\sum_{q=1}^j\num_q$, by the exchange property of matroids, one is
assured that one can extend $A_{j-1}$ to an independent set $A_j\sse\pset_j$ by adding
s set $S_j$ of $\num_j$ elements from $\pset_j$. In terms of reward, this is good enough,
as one still obtains that $\rewd(S_j)$ is roughly $\num_j\thresh_j$.  
But in order to argue feasibility of the final solution $\bigcup_{j=0}^{\nbuck}S_j$,
we would like to ``charge'' the sizes of elements in $S_j$ to those of elements in
some set $L_j\sse\optset\cap\pset_j$, where the $L_j$'s are disjoint for different 
indices $j$. 
If we pick $S_j$ of size $\num_j$ by running the greedy algorithm on $\pset_j-A_{j-1}$, we
do obtain that there is some set $L_j\sse\optset\cap\pset_j$ such that 
$\svec{S_j}^{\down}\leq\svec{L_j}^{\down}$, but we cannot ensure that the $L_j$s are
disjoint. %(see Appendix~\ref{tightex}). To ensure disjointness, we 
%need to exclude the $L_q$-sets for $q\leq j-1$, that is, pick 

Disjointness requires that we can find a suitable charging set in
$\optset\cap\pset_j-(L_0\cup\ldots L_{j-1})$. But this is problematic
since $\optset-(L_0\cup\ldots L_{j-1})$ need not be a larger independent set than
$A_{j-1}=S_0\cup\ldots\cup S_{j-1}$, so there need not be $\num_j$ items from
$\optset\cap\pset_j-(L_0\cup\ldots L_{j-1})$ that can be used to extend $A_{j-1}$ to an
independent set. These considerations indicate that we need a compromise: 
for an index $j$, instead of choosing $\num_j$ items from $\pset_j-A_{j_1}$, we choose
fewer items, say $\num_j/2$ items (for a matroid), to add to our solution. Now, if we have
picked charging sets $L_0,\ldots,L_{j-1}\sse\optset\cap\pset_j$ with
$|L_q|=|S_q|=\num_q/2$ for all $q=0,\ldots,j-1$, we still have 
$|\optset\cap\pset_j-(L_0\cup\ldots L_{j-1})|-|A_{j-1}|\geq\num_j$. Thus, we can extend
$A_{j-1}$ by picking $\num_j/2$ items from $\optset\cap\pset_j-(L_0\cup\ldots L_{j-1})$,
and charge the sizes of these items to a suitable charging set 
in $\optset\cap\pset_j-(L_0\cup\ldots L_{j-1})$ of size $\num_j/2$.

For a general independence system $\M$, we introduce a parameter (defined below) 
%that we call the greedy parameter of $\M$ 
that allows us to quantify the largest fraction of $\num_j$
items that one can pick for index $j$ so that the above extension argument works out. 
%It's definition is tailored to ensure that 

\begin{defn} \label{grparam}
Let $\M=([n],\sols)$ be an independence system with rank function $\rk:2^n\mapsto\Z_+$. 
The {\em greedy parameter of $\M$}, denoted $\gr(M)$ is the smallest $\beta\geq 1$ such
that the following holds. 
For any sets $Z\sse Y\sse[n]$ with $Z\in\sols$,
%any independent set $Z\sse Y$, 
any $\ell\in[0,\rk(Y)]$, and any cost vector $c\in\Rp^n$, 
%(with possibly negative entries), 
there exists a set $S\sse Y-Z$ such that: \ 
\begin{enumerate*}[label=(G\arabic*), topsep=0.1ex, noitemsep, leftmargin=*]
\item $Z\cup S\in\sols$; \label{gind} {\quad}
\item $|Z\cup S|\geq\frac{\ell}{\beta}$; \label{gsize} \, and {\quad}
\item $\svec[c]{S}^{\down}\leq\svec[c]{B}^{\down}$ for {\em every} independent set 
$B\sse Y-Z$ with $|B|\geq \ell-|Z|$. \label{gcost}
\end{enumerate*}

We say that an upper bound $\gr(\M)\leq\beta$ is {\em efficiently certifiable} if there is
a polytime algorithm that, for any input $(Y,Z,\ell,c)$, produces an
independent set $S\sse Y-Z$ with the above properties.  
\end{defn}

\begin{comment}
While the definition of greedy parameter may seem a bit technical, as we discuss 
when describing our algorithm for \normmwis, 
%%in the proof of Theorem~\ref{normmwisthm}, the definition of greedy parameter 
this definition nicely distills the key property that we need in order to build a
\normmwis-solution in a manner akin to how this is done for the norm-budgeted knapsack
problem; with this definition in hand, it becomes relatively effortless to extend our
algorithm for \normknap to the \normmwis problem.
\end{comment}

To tie this in with the earlier discussion, suppose we take
$Y=\pset_j$, $Z=A_{j-1}$, $\ell=\sum_{q=0}^j\num_q$, and $c=\sz$. 
%and suppose that $|A_{j-1}|$
Then $\gr(\M)\leq\beta$ ensures that we can extend $A_{j-1}$ to $A_j\sse\pset_j$ with
$|A_j|\geq\ell/\beta$ (due to \ref{gsize}), and \ref{gcost} ensures that even after we
have charged the sizes 
of items in $A_{j-1}$ to some subset $C\sse\optset\cap\pset_j$ of size $|A_{j-1}|$, we
we can still charge the sizes of the items added to a suitable charging set from 
$\optset\cap\pset_j-C$.

We show that if $\gr(\M)\leq\beta$ is an efficiently certifiable bound, then one
can devise an $O(\beta)$-approximation algorithm for \normmwis on $\M$
(Theorem~\ref{normmwisthm}). 
%and show is bounded by $k+1$ for any $k$-set system, and by
%$k$ for any $k$-extendible system.
%We will bound $\gr(\M)$ by devising an efficient algorithm for producing an independent 
%set $S\sse A$ with the desired properties (for any $\bigl(A,\ell,\{c_e\}\bigr)$. 
Complementing this, we show that $\gr(\M)\leq k+1$ is an efficiently-certifiable bound for
$k$-set systems. %and $k$-extendible systems.
%$\gr(\M)\leq k+1$ for a $k$-set system, and $\gr(\M)\leq k$ for a $k$-extendible system 
%(Theorem~\ref{grparambounds}). 

%We prove the following result for a $k$-set system that will allow us to build a solution 
%to \normmwis in a manner akin to how we build a solution in the algorithm for \normknap. 
%that can be seen as a refined statement %a stronger statement 
%on the performance guarantee of the greedy algorithm.

\begin{theorem} \label{greedykset} \label{grparambounds}
Let $\M=([n],\sols)$ be a $k$-set system. Then $\gr(M)\leq k+1$ is an
efficiently-certifiable bound.
%We have the following
%efficiently-certifiably bounds: 
%(a) $\gr(\M)\leq k+1$; and (b) $\gr(\M)\leq k$ if $\M$ is a $k$-extendible system.
\end{theorem}

%We can now state for approximation guarantee for \normmwis precisely.
%This yields the following results.

\begin{theorem} \label{normmwisthm}
Let $\gr(\M)\leq\beta$ be an efficiently-certifiable bound for an independence system
$\M$. One can devise a $(1+\ve)\beta$-approximation algorithm for \normmwis on $\M$
with $\bigl(\frac{n}{\ve}\bigr)^{O(1/\ve^2)}$ running time.
Thus, we obtain
%\begin{enumerate}[label=(\alph*), topsep=0.1ex, noitemsep, leftmargin=*]
%\item 
a $(1+\ve)(k+1)$-approximation algorithm for \normmwis when $\M$ is a $k$-set system.
%\item a $(2+\ve)k$-approximation algorithm when $\M$ is a $k$-extendible system; and
%\end{enumerate}   
\end{theorem}

Observe that the above guarantee for \normmwis strictly generalizes the PTAS for
\normknap, since \normknap is \normmwis on the {\em free matroid} (where every subset is
independent), and it is not hard to see that $\gr(\M)\leq 1$ is an efficiently-certifiable
bound for a free matroid $\M$.
%the greedy parameter for a free matroid is
%$1$ (and this is efficiently certifiable)

Since $b$-matchings form a $2$-set system, we obtain the following corollary.
%$(4+\e)$-approximation algorithm for norm-budgeted matching. 

\begin{corollary} \label{normmatchthm}
There is a $(3+\ve)$-approximation algorithm for %norm-budgeted matching and 
norm-budgeted $b$-matching.
\end{corollary} 

\begin{comment}
For matroids, we also design an alternate algorithm that provides a
$(1+\ve)$-approximation in $n^{O(\log\log(n/\ve))}$ time. We discuss this in
Section~\ref{matroid-qptas}.
\end{comment} 

The efficiently-certifiable bound in Theorem~\ref{grparambounds}
%on $\gr(\M)$ when $\M$ is a $k$-set or $k$-extendible system 
is obtained via a greedy algorithm, 
%and %This algorithm will be based on the greedy algorithm, 
%the greedy algorithm for finding an independent set $S\sse A$ with the desired
%properties, 
%this indicates a close connection between the parameter $\gr(\M)$ %is closely-related to 
%and the performance guarantee of the greedy algorithm 
%and how well the greedy algorithm performs for the maximum-weight independent-set (\mwis)
%problem on $\M$. (This 
which is why we call $\gr(\M)$ the greedy-parameter of $\M$.
To avoid detracting the reader, we defer the proof of Theorem~\ref{grparambounds} to the
end of the section, and delve now into the proof of Theorem~\ref{normmwisthm}.

\begin{comment}
It is not hard to infer that the greedy algorithm %for \mwis% 
%\footnote{This algorithm considers elements in decreasing order of weight, and picks a
%maximal prefix of nonnegative-reward elements that is independent.}
yields a $\gr(\M)$-approximation (see Theorem~\ref{greedyapx}), but
Definition~\ref{grparam} requires a stronger property
\end{comment}

\begin{comment}
It is well-known that for a $k$-set system $\M$, the greedy algorithm 
that considers elements in decreasing order of weight, and picks a maximal prefix of
nonnegative-weight elements %(under this sorted order) 
that is independent yields a $k$-approximation to the maximum-weight independent set
(\mwis) problem on $\M$. We capitalize on this, and 
\end{comment}

\subsection*{Proof of Theorem~\ref{normmwisthm}}
%In the sequel, we assume that $\M$ is a $k$-set system.
We focus on the main statement that if $\gr(\M)\leq\beta$ is efficiently certifiable, then
we obtain a $(1+\ve)\beta$-approximation for \normmwis on $\M$. The guarantee for a
$k$-set system %and $k$-extendible systems 
then follows from Theorem~\ref{grparambounds}.
\begin{comment}
To illustrate the main ideas with minimal notational overhead, we describe here a
$4\beta/(1-\ve)$-approximation algorithm. 
Improving the approximation factor to $\bigl(1+O(\ve)\bigr)\beta$ requires a little extra
work, and we defer this to Appendix~\ref{append-improv}.
\end{comment}

%We may assume that $\{e\}\in\sols$ and $f(\svec{e})\leq B$ for every element $e\in[n]$.
%The algorithm and analysis are along the same lines as that for \normknap.
As always, let $\optset\in\sols$ be an optimal solution, $\OPT=\rewd(\optset)$, 
and $\rmax=\max_{e\in[n]}\rewd_e\leq\frac{n}{\ve}$ (due to Theorem~\ref{scaling}) be the
maximum-reward of an element in $\optset$. 
%If $\rmax\geq\frac{\OPT}{4\beta}$, then we can simply return an item with reward
%$\rmax$, and we are done. So suppose otherwise.

%Let $\nopt=|\optset|$. 
As in the proof of Theorem~\ref{normknapthm},   
%
%\begin{enumerate}[label=$\bullet$, topsep=0.1ex, noitemsep, leftmargin=*]
%\item 
for $j\geq 0$, define $\thresh_j:=\frac{\rmax}{(1+\ve)^j}$, %and let
and $\itemset_j:=\bigl\{e\in[n]:\frac{\thresh_j}{1+\ve}<\rewd_e\leq\thresh_j\bigr\}$. 
Let $\nbuck=O\bigl(\frac{\log(n/\ve)}{\ve}\bigr)$ be the number of reward buckets that
cover all items with non-zero reward.
Let $\pset_j:=\itemset_0\cup\ldots\cup\itemset_j$ for all $j\in\dbrack{\nbuck}$.
%Note that $\pset_{\nbuck}$ contains all items with non-zero reward (as rewards are
%integers). 

As in norm-budgeted knapsack, we may assume that we have an estimate
$\optval\leq\OPT\leq(1+\ve)\optval$, and know
$\rnum_j:=\floor{\frac{\rewd(\optset\cap\itemset_j)}{\Dt}}$ for all
$j\in\dbrack{\nbuck}$, where $\Dt=\frac{\ve\cdot\optval}{\nbuck}$.  
So taking $\num_j:=\rnum_j\cdot\frac{\Dt}{\thresh_j}$ for all
$j\in\dbrack{\nbuck}$, we obtain that 
$|\optset\cap\itemset_j|\geq\num_j$ for all $j\in\dbrack{\nbuck}$.
%<(1+\ve)\bigl(\num_j+\frac{\ve\optval}{\nbuck\thresh_j}\bigr)$.
%\end{enumerate}
Define $\pnum_j:=\sum_{q=0}^j\num_q$, which is a lower bound on $|\optset\cap\pset_j|$,
for all $j\in\dbrack{\nbuck}$. 
%clearly  $|\optset\cap\pset_j|\geq\pnum_j$. 

\vspace*{-1ex}
\paragraph{The algorithm.}
Let $\alg^{\gr}$ be the algorithm underlying the efficiently-certifiable bound
$\gr(\M)\leq\beta$. 
Define $A_{-1}:=\emptyset$. 
%We again build a solution by considering each breakpoint $j\in\brset$ (in non-decreasing
%order). 
%We maintain the invariant $|A_{j}|\leq\ceil{\frac{\pnum_{j}}{\beta}}$ for all 
%$j\in\brset$. 
%We iterate over all the breakpoints in $\brset$ (in sorted order), and 
%For a breakpoint $j$, %(considered in increasing order), 
For each $j\in\dbrack{\nbuck}$ (considered in increasing order), 
we do the following. We invoke $\alg^{\gr}$ on the input $(Y,Z,\ell,c)$, where
$Y=\pset_j$, $Z=A_{j-1}$, $\ell=\pnum_j$, and $c_e=\sz_e$ for all $e\in[n]$.
Since $\optset\cap\pset_j\in\sols$ and $|\optset\cap\pset_j|\geq\ell$, we have
$\rk(Y)\geq\ell$. 
%Also, we have $|A_{\prev(j)}|\leq\pnum_{\prev(j)}$ by the invariant, so
%$|A_{\prev(j)}|\leq\ell$.  
So the tuple $(Y,Z,\ell,c)$ is a valid input to $\alg^{\gr}$, and 
$\alg^{\gr}$ will return some set $S_j\sse\pset_j-A_{j-1}$ such that $A_{j-1}\cup S_j\in\sols$ 
and $|A_{j-1}\cup S_j|\geq\pnum_j/\beta$. 
%and %$\svec[c]{S}^{\down}\leq\svec[c]{B-A_{\prev(j)}}^{\down}$ for every independent set 
%$B\sse\pset_j$ with $|B|\geq\ell$.
%Let $S_j\sse S'_j$ be any subset of $S'_j$ such that
%$|A_{j-1}\cap S_j|=\ceil{\pnum_j/\beta}$. We argue in the analyis that this is always
%possible. 
We set $A_j\assign A_{j-1}\cup S_j$.
We return the set $A_{\nbuck}$. 

\paragraph{Analysis.} 
The theorem follows from Lemma~\ref{mwisfeas}, which shows feasibility, and
Lemmas~\ref{mwisrewd} and \ref{mwisasgnmt}, which lower bound the reward obtained.
The following simple claim will be useful.

\begin{claim}  \label{sortmap}
Let $R,T\sse[n]$ be such that $\svec{R}^{\down}\leq\svec{T}^{\down}$. 
There is a one-to-one mapping $\sg:R-T\mapsto T-R$ such that $\sz_i\leq\sz_{\sg(i)}$ for
all $i\in R-T$, and so $\svec{R-T}^{\down}\leq\svec{T-R}^{\down}$. 
\end{claim}

\begin{proof}
Since $\svec{R}^{\down}\leq\svec{T}^{\down}$, there is a one-to-one function
$\pi:R\mapsto T$ be a permutation such that $\sz_i\leq\sz_{\pi(i)}$ for all $i\in R$.
Consider the following directed bipartite graph $G$ with vertex bipartition 
$\{r_i:i\in R\}\cup\{t_i:i\in T\}$.
We have edges $(r_i,t_{\pi(i)})$ for all $i\in R$, and edges $(t_i,r_i)$ for all 
$i\in R\cap T$. Note that by construction for every edge $(u,v)$ corresponding to some
items $i,j\in\R\cup T$, we have $\sz_i\leq\sz_j$.
Every node in $G$ has in-degree and out-degree at most $1$, so the edges
of $G$ can be partitioned into vertex-disjoint paths and cycles. Due to
vertex-disjointness, each path in this decomposition must be maximal; in particular, it's
start node must have zero in-degree, and its end-node must have zero out-degree and hence
must be a $t_j$-node.
For every $i\in R-T$, $r_i$ has out-degree $1$ and zero in-degree, so it is the start-node of
one of these paths; if $t_j$ is the end-node of this path, then since $t_j$ has zero
out-degree, we have $j\in T-R$. Also, by construction, we have $\sz_i\leq\sz_j$. So we can
define $\sg(i)=j$; doing this for all $i\in R-T$ yields the desired one-to-one mapping.
\end{proof}

\begin{lemma} \label{mwisfeas}
$A_{\nbuck}$ is a feasible \normmwis solution.
\end{lemma}

\begin{proof}
Let $\optset_j:=\optset\cap\pset_j$ for $j\in\dbrack{\nbuck}$.
We prove that $\svec{A_j}^{\down}\leq\svec{\optset_j}^{\down}$ for all
$j\in\dbrack{\nbuck}$, by induction on $j$. Taking $j=\nbuck$, this shows that
$A_{\nbuck}$ is a feasible \normmwis solution.

Recall that in every iteration $j=0,1,\ldots,\nbuck$, we invoke $\alg^{\gr}$ on the input 
$(Y,Z,\ell,c)=(\pset_j,A_{j-1},\pnum_j,\sz)$ %for all $e\in[n]$ 
to obtain set $S_j\sse Y-Z$, and by \ref{gcost}, we have
$\svec{S_j}^{\down}\leq\svec{B}^{\down}$ for all $B\sse Y-Z$, $B\in\sols$ with
$|B|\geq\ell-|Z|$. 

The base case $j=0$ holds since $|\optset_0|\geq\pnum_0$,
and $\svec{S_0}^{\down}\leq\svec{\optset_0}^{\down}$. 
%Since $A_0\sse S'_0$, we also have $\svec{A_0}^{\down}\leq\svec{\optset_0}^{\down}$.
Now let $j\in\dbrack{\nbuck}$, $j\neq 0$. %First, consider \ref{prefsize}.
Let $(Y,Z,\ell,c)=(\pset_j,A_{j-1},\pnum_j,\sz)$, and $T=\optset_j$.
By the induction hypothesis, we have
$\svec{Z}^{\down}\leq\svec{\optset_{j-1}}^{\down}$, and so
$\svec{Z}^{\down}\leq\svec{T}^{\down}$. By Claim~\ref{sortmap}, 
%we have $\svec{R-T}^{\down}\leq\svec{T-R}^{\down}$. 
there is a one-to-one mapping $\sg:Z-T\mapsto T-Z$ such that $\sz_i\leq\sz_{\sg(i)}$ for
all $i\in Z-T$. Now let $L=(Z\cap T)\cup\{\sg(i): i\in Z-T\}$. 
Clearly, we have $L\sse T$, $|L|=|Z|$, and $\svec{Z}^{\down}\leq\svec{L}^{\down}$. 
Also, (i) $T-L\sse T-Z\sse Y-Z$, (ii) $T-L\in\sols$, 
and (iii) $|T-L|=|T|-|Z|\geq\pnum_j-|Z|$. So by \ref{gcost}, the set $S_j$ obtained in
iteration $j$ satisfies $\svec{S_j}^{\down}\leq\svec{T-L}^{\down}$. 
%and hence $\svec{S_j}^{\down}\leq\svec{T-L}^{\down}$. 
Coupled with $\svec{Z}^{\down}\leq\svec{L}^{\down}$, this shows that 
$\svec{Z\cup S_j}^{\down}\leq\svec{T}^{\down}$. 
%that is, $\svec{A_j}^{\down}\leq\svec{\optset\cap\pset_j}^{\down}$.   
This proves the induction step, and hence the lemma. %for invariant \ref{prefsize}.
\end{proof}

To bound $\rewd(A_{\nbuck})$, we argue that there is a fractional assignment 
$x\in\Rp^{A_{\nbuck}\times\dbrack{\nbuck}}$ of elements in $A_{\nbuck}$ to buckets
$j\in\dbrack{\nbuck}$ such that: 
(i) each element is assigned to an extent of at most $1$, i.e., 
$\sum_{j=0}^\nbuck x_{e,j}\leq 1$ for every $e\in A_{\nbuck}$;
and
%(ii) for an element $e$, if  is only assigned to buckets $j$ or higher; 
(ii) each $j\in\dbrack{\nbuck}$ is assigned at least $\num_j/\beta$ elements, and this
assignment is supported on a subset of $A_j$, i.e., 
\begin{equation*}
\sum_{e\in A_{\nbuck}}x_{e,j}\geq\frac{\num_j}{\beta}
\quad \forall j\in\dbrack{\nbuck}, \qquad
x_{e,j}=0\ \ \text{if $e\notin A_j$} \quad \forall e\in A_{\nbuck}, j\in\dbrack{\nbuck}.
\end{equation*} 
%Property (ii) implies that an element $e\in A_j-A_{j-1}$ is only assigned to buckets $j$
%or higher.
We first show that such a fractional assignment implies the desired guarantee on
$\rewd(A_{\nbuck})$.

\begin{lemma} \label{mwisrewd}
If we have a fractional assignment $x$ satisfying (i) and (ii), then
$\rewd(A_{\nbuck})\geq\frac{(1-\ve)^2}{\beta}\cdot\OPT$. 
\end{lemma}

\begin{proof}
Define $\rewd(x):=\sum_{e\in A_{\nbuck}}\sum_{j\in\dbrack{\nbuck}}\rewd_e x_{e,j}$.
Due to (i), we have $\rewd(x)\leq \rewd(A_{\nbuck})$. 
Due to (ii), we have that $x_{e,j}>0$ implies that $e\in A_j$, and so 
$\rewd_e\geq\frac{\thresh_j}{1+\ve}$.
This yields 
\begin{equation}
\rewd(x)\geq
\frac{1}{1+\ve}\cdot\sum_{j=0}^{\nbuck}\thresh_j\cdot\Bigl(\sum_{e\in A_{\nbuck}}x_{e,j}\Bigr)
\geq\frac{1}{\beta(1+\ve)}\cdot\sum_{j=0}^{\nbuck}\num_j\thresh_j.
\label{fracrewd}
\end{equation}
%The total reward of $x$ can be lower bounded by
%$\sum_{j=0}^\nbuck\frac{\num_j}{\beta}\cdot\frac{\thresh_j}{1+\ve}$. 
%Together, these imply 
Also, for all $j\in\dbrack{\nbuck}$, we have
$\rewd(\optset\cap\itemset_j)\leq(\rnum_j+1)\Dt=\num_j\thresh_j+\Dt$.
%\[
%\rewd(\optset\cap\itemset_j)\leq\rnum_j\cdot\frac{\ve\cdot\optval}{\nbuck}+\frac{\ve\cdot\optval}{\nbuck}
%=\num_j\thresh_j+\frac{\ve\cdot\optval}{\nbuck}.
%\]
Adding these inequalities for all $j\in\dbrack{\nbuck}$, and combining this with
$\rewd(A_{\nbuck})\geq\rewd(x)$ and \eqref{fracrewd},
%\geq\frac{1}{\beta(1+\ve)}\cdot\sum_{j\in\dbrack{\nbuck}}\num_j\thresh_j$,
we obtain that $\OPT\leq\beta(1+\ve)\rewd(A_{\nbuck})+\ve\OPT$.
Rearranging, gives 
$\rewd(A_{\nbuck})\geq\frac{1-\ve}{\beta(1+\ve)}\cdot\OPT\geq\frac{(1-\ve)^2}{\beta}\cdot\OPT$. 
\end{proof}

\begin{lemma} \label{mwisasgnmt}
A fractional assignment $x\in\Rp^{A_{\nbuck}\times\dbrack{\nbuck}}$ satisfying (i) and (ii) exists.
\end{lemma}

\begin{proof}
We argue by induction on $j$ that there is a partial fractional assignment 
$x\in\Rp^{A_{j}\times\dbrack{j}}$ of elements in $A_j$ to buckets with indices in $\dbrack{j}$ 
%satisfying (i) and (ii). 
%More precisely, there is an assignment 
such that: 
(i') $\sum_{q=0}^j x_{e,q}\leq 1$ for every $e\in A_j$;
and
(ii') %each $j\in\dbrack{\nbuck}$ is assigned at least $\num_j/\beta$ elements, and this
%assignment is supported on a subset of $A_j$, i.e., 
\begin{equation*}
\sum_{e\in A_j}x_{e,q}=\frac{\num_q}{\beta}\quad \forall q\in\dbrack{j}, \qquad
x_{e,q}=0\ \ \text{if $e\notin A_q$} \quad \forall e\in A_j, q\in\dbrack{j}.
\end{equation*}

For the base case, $j=0$, we can easily ensure these properties since
$|A_0|\geq\num_0/\beta$. For the induction step, consider $j>0$, and suppose we
have an assignment $x\in\Rp^{A_{j-1}\times\dbrack{j-1}}$ for index $j-1$ satisfying (i')
and (ii'). Initialize
$x_{e,j}=0$ for all $e\in A_{j-1}$, and $x_{e,q}=0$ for all $e\in A_j-A_{j-1}$,
$q\in\dbrack{j}$. 
Since $|A_j|\geq\pnum_j/\beta$, we have 
$\sum_{e\in A_j}\bigl(1-\sum_{q=0}^{j-1}x_{e,q}\bigr)=|A_j|-\bigl(\sum_{q=0}^{j-1}\num_q\bigr)/\beta\geq\num_j/\beta$,
where the first equality is due to (i').
So we can find some $y\in\Rp^{A_j}$ with $y_{e}\leq 1-\sum_{q=0}^{j-1}x_{e,q}$ for all
$e\in A_j$ such that 
$\sum_{e\in A_j}y_e=\num_j/\beta$. Setting $x_{e,j}=y_e$ for all $e\in A_j$ yields the
desired partial fractional assignment for index $j$.
\end{proof}

\begin{proofof}{Theorem~\ref{grparambounds}}
Consider any $Y\sse[n]$, independent set $Z\sse Y$,
$\ell\in[0,\rk(Y)]$, and cost vector $c\in\R^n$.
%$\{c_e\}_{e\in[n]}$. 
To obtain the efficiently-certifiable bound $\gr(\M)\leq k+1$, recall from
Definition~\ref{grparam} that we need to efficiently compute $S\sse Y-Z$
such that 
\ref{gind} $Z\cup S\in\sols$,\ 
\ref{gsize} $|Z\cup S|\geq\frac{\ell}{k+1}$, and 
\ref{gcost} $\svec[c]{S}^{\down}\leq\svec[c]{B}^{\down}$ for every independent set 
$B\sse Y-Z$ with $|B|\geq\ell-|Z|$.

%If $\ell=0$, we can simply take $S=\es$, so suppose $\ell\geq 1$.
If $|Z|\geq\ell/(k+1)$, we can simply take $S=\es$, so suppose this is not the case.
We obtain $S$ by running the greedy algorithm on $Y-Z$, where we consider elements in
non-decreasing order of cost. We start with $S\assign\es$. We
consider elements in $Y$ in non-decreasing order of cost (breaking ties arbitrarily) and
keep adding elements to $S$ as long as this maintains $Z\cup S\in\sols$. 
%If $\M$ is $k$-extendible, we continue until $|S|\geq\ell/k$; otherwise, 
We continue until $|Z\cup S|\geq\ell/(k+1)$. 
%\beta, where $\beta=k$ for a $k$-extendible system,
%and $\beta=k+1$ otherwise. 
We argue that this termination condition is reached, so the algorithm is well-defined
and hence, $S$ satisfies \ref{gind} and \ref{gsize}.
%hence, $S\sse A$, $S\in\sols$ and part (b) holds. %We also argue that $S$ satisfies (i).

Fix some $B\sse Y-Z$, $B\in\sols$ with $|B|\geq\ell-|Z|$. 
\begin{comment}
We argue that as long as
$|S|<\frac{\ell}{k+1}$, in every iteration of the above algorithm, 
there is a distinct element $e\in B-S$ that $S\cup\{e\}\in\sols$. For a $k$-extendible
system, we argue that this holds as long as $|S|<\ell/k$. This will prove (b), and since
the greedy algorithm always chooses the least-cost element to add to $S$ among all
candidate elements that maintain independence, this will also prove (a).
\end{comment}
%
%First, suppose that we only know that $\M$ is a $k$-set system. 
We argue inductively that we can maintain a set $B'\sse B\cup S$ such that: 
%(i) $|B-B'|$ equal to the size of the current-set $S$; 
(i) $\svec[c]{S}^{\down}\leq\svec[c]{B-B'}^{\down}$; and 
(ii) $Z\cup S\cup\{e\}\in\sols$ for some $e\in B'-S$ as long as 
$|Z\cup S|<\frac{\ell}{k+1}$. 
%and as long as $|S|<\frac{\ell}{k}$ if $\M$ is a $k$-extendible system. 
Property (ii) shows that the algorithm is well-defined and so \ref{gind}, \ref{gsize} hold;
property (i) shows that \ref{gcost} holds.

%The construction of $B'$ will differ based on whether $\M$ is a $k$-set system, or a
%$k$-extendible system. Suppose first that we only know that $\M$ is a $k$-set system.
We will also ensure: (*) $|B-B'|=|S|$.
At the start of the first iteration, when $S=\es$, we take $B'=B$. Clearly, this satisfies
(i) and (*). Property (ii) holds because $|Z\cup S|=|Z|<\frac{\ell}{k+1}$ implies that
$|Z|<k(\ell-|Z|)\leq k|B'|\leq k\rk(B'\cup Z)$, and so since $\M$ is a $k$-set system, $Z$
cannot be a maximal independent set contained in $B'\cup Z$; therefore, there is some
$e\in (B'\cup Z)-Z=B'-Z$ such that $Z\cup\{e\}\in\sols$.
 
%clearly, there is some $e\in B'$ such that $\{e\}\in\sols$. 
Suppose inductively that (i), (ii), (*) hold at the beginning of some iteration for which
$|Z\cup S|<\frac{\ell}{k+1}$. Let $S_1$ denote $S$ at the start of the iteration.
So there exists $e\in B'-S_1$ such that
$Z\cup S_1\cup\{e\}\in\sols$. (Note that $e\notin Z$.)
So in this iteration, greedy adds some $f\in Y-(S_1\cup Z)$ to $S_1$ with
$c_f\leq c_e$ to obtain the $S$-set at the end of the iteration. 
Take $B''=B'-\{e\}$. Since $\svec[c]{S_1}^{\down}\leq\svec[c]{B-B'}^{\down}$ and 
$c_f\leq c_e$, we obtain that $\svec[c]{S}^{\down}\leq\svec[c]{B-B''}^{\down}$. 
%by construction. 
Hence, (i) and (*) hold (with $B'=B''$) at the end of the iteration.
%If $f\in B$, take $B''=B'\cup\{f\}$, otherwise, take $B''=B'\cup\{e\}$. 
If $|Z\cup S|<\frac{\ell}{k+1}$, then %consider the set $B''\cup S$.
%We have 
$\rk(B''\cup Z\cup S)\geq |B''|=|B|-|S|\geq\ell-|Z|-|S|>k(|Z\cup S|)$. 
So since $\M$ is a $k$-set system, it cannot be
that $Z\cup S$ is a maximal independent set contained in $B''\cup Z\cup S$; that is, there
exists some $e\in (B''\cup Z\cup S)-(Z\cup S)=B''-S$ such that $Z\cup
S\cup\{e\}\in\sols$. This shows that (ii) holds at 
the end of the iteration. Thus, we obtain the efficiently-certifiable bound 
$\gr(\M)\leq k+1$. %proving part (a).
\end{proofof}

\section{\boldmath A reduction for norm-budgeted \sap and norm-budgeted \sepfl} 
\label{reduction}
In this section, we describe a very useful general reduction applicable to
\normsap and \normsepfl (and hence \normsched, \normkfl), that will allow us to reduce the
task of developing an 
approximation algorithm for these problems to that of obtaining a {\em bicriteria}
approximation guarantee for the problem, where we may violate the norm budget by a bounded
factor, and an approximation guarantee for a norm-budgeted bipartite
matching problem. Since we have an $O(1)$-approximation algorithm for \normmatch, this
will imply that to obtain an $O(1)$-factor for any of these problems 
(i.e., norm-budgeted \{\maxgap, \sap, \kfl, \sepfl\}),
%(i.e., \normsched, \normsap, \normsepfl), 
it suffices to obtain a bicriteria
$(\rho,\gm)$-approximation with $\rho,\gm=O(1)$, for the problem, i.e., return a solution
with reward at least $\OPT/\rho$, which may violate the budget by at most a $\gm$-factor.
(Complementing this, in Section~\ref{sepfl}, we show how to obtain such bicriteria
guarantees.)  

This reduction turns out to be extremely useful, because allowing for bicriteria
guarantees considerably frees our hand and, in particular, enables one to leverage 
the machinery developed for tackling minimum-norm optimization
problems~\cite{ChakrabartyS19a,IbrahimpurS21}, such as estimating suitable features of
the size-vector of an optimal solution (e.g., the $\ell$-th largest components), coming up 
with a solution with roughly similar features.  

Although \normsepfl contains \normsap %and \normsched 
as a special case, we state our
reduction separately for \normsepfl and \normsap because the norm-budgeted
matching problem that we need to solve is slightly different in these two settings 
%(though both admit a $(3+\ve)$-approximation algorithm). 
%due to Theorem~\ref{normmwisthm} and Corollary~\ref{normmatchthm} 
%and (ii) the reduction is slightly simpler to state for \normsap.

\begin{theorem} \label{lbredn} \label{sapredn}
Given a bicriteria $(\rho,\gm)$ approximation algorithm $\alg^{\sap}$ for \normsap, and
an $\al$-approximation algorithm $\alg^{\match}$ for \normmatch on bipartite graphs, one
can obtain a $\beta$-approximation algorithm for \normsap, where
$\beta=\min\bigl\{\gm(\rho+\al),\,\ceil{\gm}(\rho+\al)-\al\bigr\}$.
%approximation algorithm for \normsap, where $\gm'=\ceil{\gm}$. 
\end{theorem}

For \normsepfl, we need to solve a variant of norm-budgeted matching where we seek a
matching of size at most $k$ (satisfying the norm-budget constraint); we call this
\normkmatch. 
Note that \normkmatch on a bipartite graph can be cast as \normmwis on a $2$-set
system $\M$, because the independence system encoding the degree bounds for vertices on
one side of the bipartite graph {\em and} the cardinality constraint of $k$ is still a
matroid. So just as with \normmatch (Corollary~\ref{normmatchthm}), there is a
$(3+\ve)$-approximation algorithm for \normkmatch. 

\begin{theorem} \label{sepflredn}
Given a bicriteria $(\rho,\gm)$ approximation algorithm $\alg^{\sepfl}$ for \normsepfl,
and an $\al$-approximation algorithm $\alg^{\kmatch}$ for \normkmatch on bipartite graphs,
one can obtain a $\beta$-approximation algorithm for \normsepfl, where
$\beta=\min\bigl\{\gm(\rho+\al),\,\ceil{\gm}(\rho+\al)-\al\bigr\}$.
%a $\bigl(\gm'\rho+(\gm'-1)\al\bigr)$-approximation algorithm for \normsepfl, where
%$\gm'=\ceil{\gm}$. 
\end{theorem}

The proofs of Theorems~\ref{lbredn} and~\ref{sepflredn} utilize the following simple
claim. (Recall that for a vector $v$, and a subset $S$ of coordinates, $v(S)$ denotes 
$\sum_{j\in S}v_j$.) 

\begin{claim} \label{jobsplit}
Let $a\in\R_+^N$, $L=a([N])$, and $k\geq 1$ be an
integer. We can obtain $(2k-1)$ sets $T_1,\ldots,T_{2k-1}$ whose union is $[N]$
such that: (a) $a(T_{2\ell-1})\leq L/k$ for all $\ell\in[k]$, (b) $|T_{2\ell}|\leq 1$ for
all $\ell\in[k-1]$, and (c) $(T_1\cup\ldots\cup T_\ell)=[j]$, for some $j\in[N]$, for all
$\ell\in[2k-1]$. 
\end{claim}

\begin{proof}
Let $I=[N]$. %For a set $S\sse[n]$, let $a(S):=\sum_{j\in S}a_j$.
For $\ell=1,\ldots,k-1$, we repeat the following: let $T'_\ell$ be a minimal prefix of $I$
such that $a(T'_\ell)>L/k$, and let $j$ be the last index in $T'_\ell$. 
We set $T_{2\ell-1}=T'_\ell-\{j\}$, and $T_{2\ell}=\{j\}$, and update $I\assign I-T'_\ell$.
Finally, we set $T_{2k-1}$ to be the index-set $I$ at the end of the above loop.
Clearly, by minimality of $T'_\ell$, we obtain $a(T_{2\ell-1})\leq L/k$ for all
$\ell\in[k-1]$. We also have 
$a(T_{2k-1})=a([N]-(T'_1\cup\ldots\cup T'_{k-1}))<L-(k-1)\cdot L/k=L/k$. Finally, (c)
holds because, for all $\ell\in[2k-1]$, by construction, we have that  
$(T_1\cup\ldots\cup T_\ell)$ is some prefix of $I$.  
\end{proof}

\begin{proofof}{Theorem~\ref{lbredn}}
Let $\I=\bigl(J,m,\{p_{ij},\rewd_{ij}\}_{i\in[m],j\in J},f:\R^m\mapsto\R_+,B,\{\M_i=(J,\sols_i)\}_{i\in[m]}\bigr)$ be
a \normsap instance.
(Recall that we use $i$ to index machines, and $j$ to index jobs.)
The algorithm leading to the stated guarantee is simple.
Let $\I'$ be the \normsap instance specified by the same data as $\I$, except that the
norm budget is reduced to $B/\gm$, 
%(so $\I'=\bigl(J,m,\{p_{ij},\rewd_{ij}\}_{i,j},f,B/\gm,\{\M_i\}_i\bigr)$), 
and let $\sg_1:S_1\mapsto[m]$ be the solution returned by the $(\rho,\gm)$-approximation
algorithm $\alg^{\sap}$ on input $\I'$. Note that 
$f(\svec[{\load}]{\sg_1})\leq\gm\cdot B/\gm\leq B$, so $\sg_1$ is a feasible solution to
the original \normsap instance.
Let $\I''$ be the norm-budgeted matching instance specified by the
bipartite graph $G=\bigl(V=J\cup[m],\,E=\bigl\{ij:\{j\}\in\sols_i\bigr\}\bigr)$, where
the size and reward of an edge $ij$ are given by $p_{ij}$ and $\rewd_{ij}$ respectively, 
the same norm $f$,%
\footnote{More precisely, since \normmatch requires a norm over $\R^{|E|}$, where $E$ is the
edge-set of the graph, we lift $f$ to a norm $f':\R^{m|J|}\mapsto\R_+$ by setting
$f'(v)=f\bigl((v^{\down}_\ell)_{\ell=1,\ldots,m}\bigr)$.}
and the same budget $B$.
Consider the solution returned by the $\al$-approximation
algorithm $\alg^{\match}$ for $\I''$, %this \normmatch-instance, 
viewed as an assignment $\sg_2:S_2\mapsto[m]$ (so $\sg_2$ assigns at most one job per
machine). Note that $\sg_2$ is a feasible solution to $\I$.
We return the better of the two solutions, $\sg_1$, $\sg_2$.

\medskip
To analyze this, we lower bound $\OPT$ in terms of $\OPT_{\I'}$ and $\OPT_{\I''}$, which
are the optimal values for the \normsap-instance $\I'$ and the \normmatch-instance
$\I''$. 
Let $\sg^*:\jopt\mapsto[m]$ be an optimal solution to $\I$, and let
$\lvecopt=\svec[{\load}]{\sg^*}$ be the load-vector induced by $\sg^*$.
Let $\jopt_i:=\{j\in\jopt: \sg^*(j)=i\}$ be the jobs assigned to machine $i$ under $\sg^*$.
%Let $\gm'=\ceil{\gm}$.

First, suppose that $\gm$ is an integer.
Using Claim~\ref{jobsplit}, we can divide each $\jopt_i$-set into $2\gm-1$
sets, ${\jopt_i}^{(1)},\ldots,{\jopt_i}^{(2\gm-1)}$ (some of which could be possibly empty),
such that, for $\ell=1,3,\ldots,2\gm-1$, we have 
$\sum_{j\in{\jopt_i}^{(\ell)}}p_{ij}\leq\lvecopt_i/\gm$, and the sets ${\jopt_i}^{(\ell)}$ for
$\ell=2,4,\ldots,2\gm-2$ consist of at most one job. 
Note that since each $\M_i$ is an independence system, we have that
${\jopt_i}^{(\ell)}\in\sols_i$, for every $i\in[m]$ and $\ell=1,\ldots,2\gm-1$.

For $\ell=1,\ldots,2\gm-1$, consider the assignment $\sg^{(\ell)}$ that assigns jobs in
${\jopt_i}^{(\ell)}$ to each machine $i\in[m]$. 
For $\ell=1,3,\ldots,2\gm-1$, by construction, the
resulting load-vector is coordinate-wise at most $\lvecopt/\gm$. Therefore,
$f\bigl(\svec[{\load}]{\sg^{(\ell)}}\bigr)\leq B/\gm$,%
\footnote{This is the {\em only} place where we use the (sub-) homogeneity of the norm $f$.}
and $\sg^{(\ell)}$ is a feasible solution to the \normsap-instance $\I'$. 
%specified by all
%the same data as $\I$, except that the norm budget is reduced to $B/\gm'$.
%
\begin{comment}
Let $\OPT_{\sap}=\OPT_{\normsap}(B/\gm')$ denote the optimal value for the
\normsap-instance $\I'$. 
%given by the same $p_{ij}$'s, $\rewd_{ij}$'s, norm $f$, but with norm budget $B/\gm'$. 
Let $\sg_1:S_1\mapsto[m]$ be the solution returned by the $(\rho,\gm)$-approximation
algorithm $\alg^{\sap}$ on input $\I'$. Note that 
$f(\svec[{\load}]{\sg_1})\leq\gm\cdot B/\gm'\leq B$, so $\sg_1$ is a feasible solution to
the original \normsap instance.
\end{comment}
%
For $\ell=2,4,\ldots,2\gm-2$, $\sg^{(\ell)}$ is a feasible solution to the
\normmatch-instance $\I''$.
%
\begin{comment}
feasible solution to the norm-budgeted matching instance $\I''$, specified by the
bipartite graph $G=\bigl(V=J\cup[m],\, E=\bigl\{ij:\{j\}\in\sols_i\bigr\}\bigr)$, where
the size and reward of an edge $ij$ are given by $p_{ij}$ and $\rewd_{ij}$ respectively, 
the same norm $f$\,%
\footnote{More precisely, since \normmatch requires a norm over $R^{|E|}$, where $E$ is the
edge-set of the graph, we lift $f$ to a norm $f':\R^{m|J|}\mapsto\R_+$ by setting
$f'(v)=f\bigl((v^{\down}_\ell)_{\ell=1,\ldots,m}\bigr)$ for $v\geq 0$.}
and the same budget $B$. Let $\OPT_{\match}$ denote the optimal value for this
\normmatch-instance $\I''$. 
%this norm-budgeted matching problem. 
Consider the solution returned by the $\al$-approximation
algorithm $\alg^{\match}$ for $\I''$, %this \normmatch-instance, 
viewed as an assignment $\sg_2:S_2\mapsto[m]$ (where $\sg_2$ assigns at most one job per
machine). Note that $\sg_2$ is a feasible solution to $\I$.
\end{comment}
%
Since $\jopt=\bigcup_{i\in[m],\ell\in[2\gm'-1]}{\jopt_i}^{(\ell)}$, it follows that
$\OPT\leq\gm\OPT_{\I'}+(\gm-1)\OPT_{\I''}$. 
So returning the better of $\sg_1$ and $\sg_2$, yields reward 
\begin{equation*}
\begin{split}
\max\Bigl\{\rewd(S_1),\rewd(S_2)\Bigr\}
& \geq \frac{\gm\rho}{\gm\rho+(\gm-1)\al}\cdot\rewd(S_1)+
\frac{(\gm-1)\al}{\gm\rho+(\gm-1)\al}\cdot\rewd(S_2) \\
& \geq\frac{1}{\gm\rho+(\gm-1)\al}\cdot
\Bigl(\gm\cdot\OPT_{\I'}+(\gm-1)\OPT_{\I''}\Bigr) \\
& \geq\frac{\OPT}{\gm\rho+(\gm-1)\al}\geq\frac{\OPT}{\beta}.
\end{split}
\end{equation*}
The second inequality follows due to the approximation guarantees of algorithms
$\alg^{\sap}$ and $\alg^{\match}$.
%the algorithms used for \normsched and \normmatch.

\smallskip
Now suppose $\gm$ is not an integer. Since $\alg^{\sap}$ is also a
$(\rho,\ceil{\gm})$-approximation algorithm for \normsap, the above analysis shows that
we obtain a solution of reward at least $\OPT/\bigl(\ceil{\gm}\rho+(\ceil{\gm}-1)\al\bigr)$.

For the other guarantee, suppose $\gm$ is a rational number $a/b$, where $a>b\geq 1$, and
$a,b\in\Z$. 
For $i\in[m]$, let $A_i$ be the {\em ordered} multiset consisting of $b$ copies of
$\jopt_i$, in sequence. Clearly, $\sum_{j\in A_i}p_{ij}=b\cdot\lvo_i$ and
$\rewd(A_i)=b\cdot\rewd(\jopt_i)$; also, note that any subsequence $T\sse A_i$ 
%(i.e., $T$ is the difference of two prefixes of $A_i$) 
with $\sum_{j\in T_i}p_{ij}\leq\lvo_i$ consists of all distinct jobs, i.e., $T$ is a 
{\em set} of jobs.
We now apply Claim~\ref{jobsplit} to the multiset $A_i$ taking $k=a$, for all
$i\in[m]$. This yields sets $B_i^{(\ell)}$ for all $i\in[m]$, $\ell\in[2a-1]$, where
for every $i\in[m]$, we have:
(1) $\sum_{j\in B_i^{(\ell)}}p_{ij}\leq\frac{b}{a}\cdot\lvo_i=\frac{\lvo_i}{\gm}$ for
$\ell=1,3,\ldots,2a-1$, and (2) $|B_i^{(\ell)}|\leq 1$ for $\ell=2,4,\ldots,2a-2$.
By part (c) of Claim~\ref{jobsplit}, each $B_i^{(\ell)}$ is a subsequence of $A_i$, so
from (1) and our earlier observation, it follows that every $B_i^{(\ell)}$-set consists of
distinct jobs. 

As before, let $\sg^{(\ell)}$ be the assignment that assigns jobs in $B_i^{(\ell)}$ to
each machine $i\in[m]$. Then, due to (1), for all $\ell=1,3,\ldots,2a-1$, $\sg^{(\ell)}$
is a feasible solution to the \normsap-instance $\I'$. Also, for $\ell=2,4,\ldots,2a-2$,
$\sg^{(\ell)}$ is a feasible solution to the \normmatch-instance $\I''$.
Since $A_i=\bigcup_{\ell\in[2a-1]}B_i^{(\ell)}$ for all $i\in[m]$, we have 
$b\cdot\OPT=\sum_{i\in[m],\ell\in[2a-1]}\rewd(B_i^{(\ell)})\leq a\cdot\OPT_{\I'}+(a-1)\OPT_{\I''}$.
Recall that we return the better of the two assignments $\sg_1:S_1\mapsto[m]$ and
$\sg_2:S_2\mapsto[m]$. 
%obtained respectively by running $\alg^{\sap}$ on $\I'$ and $\alg^{\match}$ on $\I''$. 
This yields reward
\begin{equation*}
\begin{split}
\max\Bigl\{\rewd(S_1),\rewd(S_2)\Bigr\}
& \geq \frac{\rho}{\rho+\al}\cdot\rewd(S_1)+
\frac{\al}{\rho+\al}\cdot\rewd(S_2) \\
& \geq\frac{1}{\rho+\al}\cdot
\Bigl(\OPT_{\I'}+\OPT_{\I''}\Bigr)
\geq\frac{b}{a}\cdot\frac{\OPT}{\rho+\al}=\frac{\OPT}{\gm(\rho+\al)}.
\end{split}
\end{equation*}
The same guarantee holds for irrational $\gm$, by taking a limit of rationals approaching
$\gm$ from above, since $\gm(\rho+\al)$ is a continuous function of $\gm$. So for any
$\gm\geq 1$, we 
obtain reward least $\OPT/\bigl(\gm(\rho+\al)\bigr)$. Combining the two bounds, we obtain
reward at least $\OPT/\beta$.
\end{proofof}

\begin{proofof}{Theorem~\ref{sepflredn}}
The proof is essentially identical to that of Theorem~\ref{lbredn}. The only change is in
the definition of the norm-budgeted matching problem $\I''$, where now we need to
additionally enforce that a solution is a matching of size at most $k$, so that this maps
to a feasible solution to the original \normsepfl instance. Thus, $\I''$ is now a
\normkmatch-instance. 
\end{proofof}

\begin{comment}
As noted earlier, since \normmatch, and \normkmatch on bipartite graphs both admit a
$(3+\ve)$-approximation, Theorems~\ref{lbredn} and~\ref{sepflredn} imply
that a bicriteria $(\rho,\gm)$-approximation for norm-budgeted \{maximum-\gap, \sap,
\kfl\} can be used to obtain a (unicriterion) $(\rho\gm'+3\gm'-3+\ve)$-approximation for
the respective problem.
\end{comment}

\section{\boldmath \mbox{Approximation results for norm-budgeted \{\maxgap, \sap, \kfl, \sepfl\!\}}} 
\label{sap} \label{kfl} \label{sepfl}
We now consider the most general problem, 
{\em norm-budgeted separable $k$-facility location} (\normsepfl), which contains \normfl,
\normsap, and \normsched as special cases. Recall that the input to \normsepfl consists of a
facility-set $\F$ and client-set $\C$. Assigning client $j$ to facility $i$
incurs assignment cots $c_{ij}$, and yields reward $\rewd_{ij}$. Throughout this section,
we use $i$ to index facilities, and $j$ to index clients. Each facility $i$ comes
with an independence system $\M_i=(\C,\sols_i)$. We also have a monotone,
symmetric norm $f:\R^k\mapsto\Rp$ and a budget $B$. A solution opens a set $F\sse\F$
of facilities with $|F|\leq k$, and specifies an assignment $\sg:S\mapsto F$ of some
subset $S\sse\C$ of clients to facilities in $F$ such that $\{j:\sg(j)=i\}\in\sols_i$ for
every $i\in F$.
This induces the facility-load vector 
$\fload{\sg}:=\bigl(\sum_{j:\sg(j)=i}c_{ij}\bigr)_{i\in F}$. 
The goal is to find a maximum-reward solution $(F,\sg)$ 
%(where $|F|\leq k$ and $\{j:\sg(j)=i\}\in\sols_i$ for all $i\in F$), 
satisfying $f(\fload{\sg})\leq B$. 

The main result of this section is that given an algorithm for a subproblem involving a 
single facility, we can obtain an approximation algorithm for \normsepfl. 
%The following definition makes this precise.

\begin{definition} \label{singlefrp}
We say that $\alg$ is a $\beta$-approximation algorithm for the 
{\em single-facility max-reward problem} (\onefrp) if it is a
%The main result of this section is that if we have 
$\beta$-approximation algorithm for the \mwis problem on $\M_i$, for every facility
$i\in\F$; that is, for any $i\in\F$ and any rewards $v\in\Rp^\C$,
$\alg$ returns $A\in\sols_i$ such that 
$v(A)\geq\bigl(\max_{S\in\sols_i}v(S)\bigr)/\beta$.  

We say that $\alg$ is a $(\rho,\gm)$-approximation algorithm for {\em \constr \onefrp} 
%one can devise a polytime algorithm $\alg$ that 
if for any $i\in\F$, any $v\in\Rp^\C$, and any {\em weight-vector} $\wt\in\Rp^\C$ and
budget $t\in\Rp$, $\alg$ returns a set $A\in\sols_i$ such that 
$v(A)\geq\frac{1}{\rho}\cdot\bigl(\max\,\{v(S):\ S\in\sols_i,\ \wt(S)\leq t\}\bigr)$, and 
$\wt(A)\leq\gm t$. 
\end{definition}

\begin{comment}
We show that given a $\beta$-approximation algorithm for \onefrp, or a
$(\rho,\gm)$-approximation algorithm for \constr \onefrp, one can obtain a
corresponding approximation guarantee ($O(\beta)$ or $O(\rho\gm)$) for \normsepfl.
\end{comment}

\begin{lemma} \label{constfrp}
Given a $\beta$-approximation algorithm for \onefrp, one can obtain a
$(\beta+1,1+\ve)$-approximation algorithm for \constr \onefrp with running time
$\bigl(\frac{|\C|}{\ve}\bigr)^{O(1/\ve)}$, for any $\ve>0$. 
\end{lemma}

We show that given an approximation algorithm for \onefrp, or 
%a bicriteria approximation algorithm for 
\constr \onefrp, one can obtain a corresponding guarantee for \normsepfl.
The guarantee for \normsepfl actually depends on the approximability of \constr
\onefrp. Lemma~\ref{constfrp} (proved in Appendix~\ref{append-constfrp}) 
shows that an approximation algorithm for \onefrp can 
be used to obtain a bicriteria approximation for \constr \onefrp,
%$(\beta+1,1+\ve)$-approximation algorithm for \constr \onefrp,
%(Lemma~\ref{constfrp}), 
but in various settings, better guarantees are possible for \constr \onefrp by
tackling this problem directly.  
We therefore state the guarantee we obtain for \normsepfl in terms of the 
approximation guarantees of both \onefrp and \constr \onefrp.

%given such a $\beta$-approximation algorithm,
%one can obtain an $O(\beta)$-approximation algorithm for \normsepfl.
%that if we have a $\beta$-approximation algorithm for \onfrp, then one can obtain an
%$O(\beta)$-approximation algorithm for \normsepfl. 

\begin{theorem} \label{normsepflthm}
%Suppose we have an algorithm $\alg^{\mwis}$ that, Then,
%For any $\ve>0$, 
We can obtain the following guarantees for \normsepfl.
\begin{enumerate}[label=(\alph*), topsep=0.25ex, itemsep=0.1ex, leftmargin=*]
%\item $(1+\ve)\bigl(\frac{e}{e-1}\cdot(\beta+1)+3\bigr)$-approximation, 
%given a $\beta$-approximation algorithm for \onefrp; 
%one can obtain an $O(\beta)$-approximation algorithm for \normsepfl. \label{onefrptrans}
%
%\item $(1+\ve)\bigl(\frac{e}{e-1}\cdot\rho+3\bigr)$-approximation, 
%given a $(\rho,1+\ve)$-approximation algorithm for consrained \onefrp; \label{constfrptrans}
%
%\item 
%$\min\bigl\{\frac{e}{e-1}\cdot\gm\rho+(3+\ve)\gm,\frac{e}{e-1}\cdot\ceil{\gm}\rho+(\ceil{\gm}-1)(3+\ve)\bigr\}$-%
%approximation, 
%Given a $(\rho,\gm)$-approximation algorithm for \constr \onefrp, one can obtain a 
\item A $\gm(1+\ve)\bigl(\frac{e}{e-1}\cdot\rho+3\bigr)$-approximation algorithm 
%for \normsepfl 
%\label{gen-constfrptrans}
with running time $\bigl(\frac{|\C|+|\F|}{\ve}\bigr)^{O(1/\ve^2)}$, for any $\ve>0$, 
given a $(\rho,\gm)$-approximation algorithm for \constr \onefrp. 
For the running time, %includes the number of 
we treat each call to the algorithm for \constr \onefrp as an elementary operation. 

\item An $O(\beta)$-approximation algorithm, given a $\beta$-approximation algorithm for
\onefrp. 
%the \mwis problem on $\M_i$, for every machine $i$, 
%one can obtain an  algorithm for \normsap. 
\end{enumerate}
\end{theorem}

\begin{comment}
%We emphasize that we have not attempted to optimize the constant in the $O(\beta)$
%approximation factor, preferring simplicity of exposition instead.
Part~\ref{gen-constfrptrans} above states the most general guarantee, which yields
part~\ref{constfrptrans} as a special case, %which in turn
%yields part (a), 
but we state the guarantee in part~\ref{constfrptrans}
separately as it is easier to parse, and is the guarantee that we utilize.
\end{comment}

Part (b) above follows from part (a) due to Lemma~\ref{constfrp}. 
%because (as noted above) a $\beta$-approximation for \onefrp can be used to obtain a
%$(\beta+1,1+\ve)$-approximation for \constr \onefrp. 
(The approximation factor in part (b) is more precisely 
$(1+\ve)\bigl(\frac{e}{e-1}\cdot(\beta+1)+3\bigr)$, in time 
$\bigl(\frac{|\C|+|\F|}{\ve}\bigr)^{O(1/\ve^2)}$.)

As noted in Section~\ref{probdefs}, when $k=|\F|$, \normsepfl
reduces to \normsap, since we can simply open all facilities and we only need to
determine which clients to assign to which facilities.
An orthogonal special case of \normsepfl is \normfl, wherein there are no constraints on
the set of clients that may be assigned to an open facility, i.e., $\M_i$ is the free
matroid with $\sols_i=2^{\C}$ for all $i\in\F$. 
%which again corresponds to $\M_i$
%being the free matroid (i.e., $\sols_i=2^{\C}$) for all $i\in\F$.
Also, \normsched is a special case of both \normsap and \normfl ($k=|\F|$ and $\M_i$ is a 
free matroid for all $i\in\F$).
%wherein there are no constraints on the set of jobs that may be assigned
%to a machine, i.e., $\M_i$ is the free matroid ($\sols_i=2^J$) for all $i\in\F$.
Clearly, when $\M_i$ is the free matroid for all $i\in\F$, \constr \onefrp is simply
the (standard) knapsack problem, and we have an FPTAS for \constr \onefrp (i.e., a
$(1+\ve,1)$-approximation).
%$1$-approximation algorithm for \onefrp. %the \mwis problem on $\M_i$.
Thus, Theorem~\ref{normsepflthm} leads to the following results for \normsap, \normsched,
and \normfl. 

\begin{theorem}[{\bf Corollary of Theorem~\ref{normsepflthm}}] 
\ 
\begin{enumerate}[label=(\alph*), ref={\thetheorem(\alph*)},
topsep=0.25ex, itemsep=0.1ex, leftmargin=*]
\item \label{normsapthm}
For \normsap, we obtain the same guarantees as in parts (a) and (b) of
Theorem~\ref{normsepflthm}. 

\begin{comment}
\item Given a $(\rho,1+\ve)$-approximation for constrained \onrfrp, one can obtain an 
$(1+\ve)\bigl(\frac{e}{e-1}\cdot\rho+3\bigr)$-approximation algorithm for \normsap. 

\item %{\bf \normsap.} 
Given a $\beta$-approximation algorithm for \onefrp,
%the \mwis problem on $\M_i$, for every machine $i$, 
one can obtain an $O(\beta)$-approximation algorithm for \normsap. 
%The approximation factor is more precisely,
%$(1+\ve)\bigl(\frac{e}{e-1}\cdot(\beta+1)+3\bigr)$. 
%one can obtain an $O(\beta)$-approximation algorithm for \normsepfl.
\end{comment}

\item %{\bf \normfl.} 
\label{normschedthm} \label{normflthm} 
We can obtain a $(4.582+\ve)$-approximation algorithm for \normfl and \normsched.
\end{enumerate}
\end{theorem}

%Part (a) of 
Theorem~\ref{normsapthm} is a direct consequence of \normsap being a special
case of \normsepfl. 
\begin{comment}
immediately from part~\ref{constfrptrans} of
Theorem~\ref{normsepflthm}, and part (b) follows from part (a) due to
Lemma~\ref{constfrp}. 
%because (as noted above) a $\beta$-approximation for \onefrp can be used to obtain a
%$(\beta+1,1+\ve)$-approximation for \constr \onefrp. 
(So the approximation factor in part (b) is more precisely 
$(1+\ve)\bigl(\frac{e}{e-1}\cdot(\beta+1)+3\bigr)$.)
\end{comment}
%Part (b) of 
Theorem~\ref{normschedthm} follows from Theorem~\ref{normsepflthm}(a)
by taking $\rho=1+\ve$, $\gm=1$.

%The rest of this section is devoted to the proof of Theorem~\ref{normsepflthm}. %Again, 
Armed with the reduction given by Theorem~\ref{sepflredn}, 
to obtain Theorem~\ref{normsepflthm}, we can focus on developing a bicriteria
approximation algorithm for \normsepfl. We obtain the following bicriteria guarantee,
which is the main technical result of this section.  

\newcommand{\sepflbiapx}{\frac{e}{e-1}}
\newcommand{\badprob}{\dt}

\begin{theorem} \label{normsepfl-bi}
Given a $(\rho,\gm)$-approximation algorithm for \constr \onefrp, we can obtain a 
bicriteria \mbox{$\bigl(\sepflbiapx\cdot\rho(1+\ve),\gm(1+\ve)\bigr)$-approximation} algorithm for
\normsepfl with running time $\bigl(\frac{|\F|+|\C|}{\ve}\bigr)^{O(1/\ve^2)}$, for any
$\ve>0$.  
\end{theorem}

Theorem~\ref{normsepflthm}(a) immediately follows from Theorem~\ref{normsepfl-bi} and
Theorem~\ref{sepflredn}, by noting also that there is a $(3+\ve)$-approximation for
\normkmatch. 
%Specifically, this yields an approximation factor of 
%$(2+\ve)(\sepflbiapx(\beta+1)+3+\ve)=\sepflapx\beta+\sepflconst+O(\beta\ve)$.

\paragraph{\boldmath \normsched on identical machines.} 
%proof of Theorem~\ref{identicalbiptas}} 
%\subsubsection{A bicriteria PTAS: %and \boldmath $(2+\ve)$-approximation: 
%proof of Theorem~\ref{identicalbiptas}} 
%
Before delving into the proof of Theorem~\ref{normsepfl-bi}, it is worth noting that
we can obtain better guarantees for \normsched on identical machines 
%via the same reduction approach 
by observing
%by combining 
%the reduction given by Theorem~\ref{sapredn}. 
%and the approach used in our PTAS for \normknap.
%This is because 
that the \normsched and the \normmatch problems that we need to solve in the
reduction of Theorem~\ref{lbredn} involve the same set of jobs, machines, $p_{ij}$'s,
$\rewd_{ij}$'s, and norm $f$ as in the original \normsched instance. %This can 
%be useful if one is interested in solving a structured class of \normsched instances.
%$(K_{J,[m]},\{p_{ij}\},\{\rewd_{ij}\},f,B)$ as for
Thus, if we consider the setting of identical machines, we only need a bicriteria
approximation for \normsched on identical machines, and the \normmatch problem that we
need to solve is simply norm-budgeted knapsack. 
The latter admits a PTAS. Also, one can obtain a
$(1+\ve,1+\ve)$-approximation algorithm for \normsched on identical machines
by: (a) using the approach for \normknap to identify 
set $A$ of jobs with $\rewd(A)\geq(1-\ve)^2\OPT$ that admits an assignment satisfying the
norm-budget constraint; and (b) using the PTAS for minimum-norm load-balancing on
identical machines from~\cite{IbrahimpurS21} to find an assignment for $A$ that violates
the norm-budget constraint by a $(1+\ve)$-factor. 
Combining these ingredients yields %a $(2+\e)$-approximation for identical machines. 
the following result. 

%Combining this with the PTAS for \normknap using Theorem~\ref{lbredn},
%we therefore obtain an improved $(2+\ve)$-approximation for \normsched on identical
%machines. 

%For the setting of identical machines, we directly tackle the problem, and devise a PTAS
%by leveraging the ideas used in our PTAS for \normknap.

\begin{theorem} \label{identicalbiptas} \label{ident-biptasthm}
One can obtain a bicriteria $(1+\ve,1+\ve)$-approximation algorithm, and hence, a
$(2+\ve)$-approximation algorithm, for \normsched on identical machines. 
\end{theorem}

We defer the proof of Theorem~\ref{identicalbiptas} to Appendix~\ref{append-identbiptas}. 
In Table~\ref{restable}, we mention
%because in Section~\ref{relatedmc}, 
%we show that we can further improve the above guarantee and 
%we prove 
a stronger result, namely, %that one can obtain 
a {\em PTAS} for \normsched on identical machines (and also related machines). 
%see Section~\ref{relatedmc}.).  
%assuming stronger oracle access to the underlying norm.
%Notably, since even the special case with the $\ell_{\infty}$-norm, i.e., the
%uniform multiple-knapsack problem is already strongly \nphard, this settles the complexity
%states of \normsched on identical machines.
This utilizes a very different, more-sophisticated, problem-specific
enumeration-based approach and is much-more involved. Hence,  
%and is perhaps our most technically-involved result;
%and we require a stronger oracle-access to the underlying norm. 
we discuss this separately in Section~\ref{relatedmc}, and instead focus here on 
general techniques and results that apply more broadly to norm-budgeted packing problems.  
%~\ref{ident-ptas}.
The reader who is interested in this PTAS can skip directly to Section~\ref{relatedmc}.

%\bigskip
\subsection*{Proof of Theorem~\ref{normsepfl-bi}}
The rest of this section is devoted to the proof of Theorem~\ref{normsepfl-bi}. %Again, 
Let $(\fopt,\,\sg^*:\copt\mapsto\fopt)$ be an optimal 
solution to the \normsepfl instance, and let $\fvo=\fload{\sg^*}$ be the 
facility-load vector induced by $\sg^*$.
%
%Since we can violate the norm budget, we can utilize the machinery developed
We utilize an idea that has become standard in the study of minimum-norm optimization
problems, namely, working with guesses of certain coordinates of the load vector $\fvo$. 
Let $\dt=\min\{\ve, 1\}$, and $\POS=\POS_{k,\dt}$ (see Definition~\ref{posdef}). 
%
\begin{comment}
As in Section~\ref{bicriteria}, define $\POS=\POS_{k,\dt}\sse[k]$ iteratively as follows: 
include the index $1$ in $\POS_{k,\dt}$;
as long as the largest index $\ell\in\POS_{k,\dt}$ is such that
$\ceil{(1+\dt)\ell}\leq m$, include $\ceil{(1+\dt)\ell}$ (which is larger than $\ell$)
in $\POS_{m,\dt}$ (and repeat). We have 
$|\POS_{m,\dt}|\leq O\bigl(\frac{\log m}{\dt}\bigr)$.
\end{comment}
%
We drop the subscripts $k,\dt$ from $\nxt$ and $\prev$. 
Recall that $\nxt(i)$ is the smallest index in $\POS$ strictly larger than $i$, or $k+1$
if there is no such index; $\prev(i)$ is the largest index in $\POS$ strictly smaller than
$i$; and $\prev(1):=0$.
%Recall also that 
Recall also that for a non-increasing vector $v\in\Rp^\POS$, we define $v^\exp\in\Rp^k$ 
as follows: $v^\exp_i=v_i$ for $i\in\POS$ and $v^\exp_i=v_{\prev(i)}$ for $i\in[k]\sm\POS$. 

%As in Section~\ref{bicriteria}, 
We may assume that $f$ is normalized so that $f(1,0,\ldots,0)=1$. Then,
$\fvo^{\down}_1\leq f(\fvo)\leq B$. We can %then  
identify the following polynomial-size set 
\[
\T:=\Bigl\{v\in\R^{\POS}:\ v\text{ is non-increasing}, 
\qquad v_\ell=\frac{B}{(1+\ve)^r},\ r\in\Z_+,\ \ v_\ell\geq\frac{\kp}{1+\ve} 
\quad \forall\ell\in\POS\Bigr\},
\]
where $\kp=\ve B/k$, which contains a non-increasing vector $\vt\in\Rp^\POS$ 
%assume that we have a non-increasing vector $\vt$ 
such that 
$\fvo^{\down}_\ell\leq t_\ell\leq(1+\ve)\fvo^{\down}_\ell+\kp$ for all $\ell\in\POS$.
Note that $|\T|=O\bigl((\frac{k}{\dt})^{1/\dt}\bigr)$. (Recall that $\dt=\min\{\ve,1\}$.) 
We formulate an LP-relaxation keeping in mind that we have such a vector $\vt$.
Set $t_{k+1}:=0$ for notational convenience.
  
\paragraph{LP relaxation.}
%The above observations motivate 
We consider a configuration-style LP for \normsepfl, where a configuration corresponds
to the set of clients assigned to an open facility.
We %view the problem as one of opening $k$ facilities, and 
think of a solution as also assigning each open facility $i$ to an index $\ell\in\POS$, 
to denote that its total load lies in $(t_{\nxt(\ell)},t_\ell]$, 
and use $y_{i,\ell}$ variables to encode this. %for every facility $i$ and index $\ell\in\POS$,  
%where assigning a facility $i$ to index $\ell$  indicating that
%facility $i$ is opened and has total load at most $t_\ell$. 
Note that $\fvo$ satisfies that $\bigl|\{i:\fvo_i>t_\ell\}\bigr|\leq \ell-1$ 
%coordinates of $\fvo$ are larger than $t_\ell$, 
for all $\ell\in\POS$.
Given the semantics of the $y_{i,\ell}$ variables, the number of facilities with load
larger than $t_\ell$ is $\sum_i\sum_{\ell'\in\POS:\ell'\leq\prev(\ell)}y_{i,\ell'}$, and
constraints \eqref{numl} below enforce that this is at most $\ell-1$, for every
$\ell\in\POS$. 
Note that this constraint is vacuous for $\ell=1$.  

%In order to obtain certain concentration properties 
In our rounding algorithm, it will be convenient to essentially work with only large
indices in $\POS$. To achieve this, we guess the values of the $y_{i,\ell}$ variables for
small $\ell$ indices. Let $\ell_0$ be the smallest index in $\POS$ that is at least
$\frac{12}{\dt^2}\cdot\ln\frac{15}{\dt^3}$.
Clearly, $\ell_0=O\bigl(\frac{1}{\dt^2}\ln(\frac{1}{\dt})\bigr)$. 
%(Clearly $\ell_0$ is then also at most $\frac{24}{\dt^2}\cdot\ln\bigl(\frac{15}{\dt^3}\bigr)$.)
%Note that $\ell_0\leq\ceil{32(1+\dt)/\dt^3}$. 
%For every $\ell\in\POS$, $\ell\leq\ell_0$, 
We guess all the facilities that are assigned some index $\ell\leq\ell_0$ under an 
optimal solution. This involves guessing some $\ell_0$ facilities %in $\F$ 
%assigned to the first $\ell_0$ indices in $\POS$, 
and the indices in $\POS\cap[\ell_0]$ assigned to these facilities, which
overall takes $O\bigl(|\F|^{\ell_0}\bigr)$ time. 
(Note that if $\ell_0\geq|\F|$, then we are guessing all the open facilities and the
$\POS$-index assignments of these open facilities.)
Let $\guessind\sse\F\times(\POS\cap[\ell_0])$ denote these guessed facilities and their
$\POS$-index assignments. Note that this fixes $y_{i,\ell}\in\{0,1\}$ for every $i\in\F$
and every $\ell\in\POS\cap[\ell_0]$, as specified by constraints \eqref{lguess} below.

For $i\in\F$ and $t\in\Rp$, let
$\sols_{i,t}:=\{S\in\sols_i: \sum_{j\in S}c_{ij}\leq t\}$. (Note that
$(\C,\sols_{i,t})$ is also an independence system.) We use configuration variables
$x_{i,\ell,S}$ for all $i\in\F$, $\ell\in\POS$, and $S\in\sols_{i,t_\ell}$ to denote that
$S$ is the set of clients assigned to the open facility $i$, and $i$ is assigned to index 
$\ell$. This yields the following LP. (Recall that $i$ indexes facilities in $\F$, and $j$
indexes clients in $\C$.)

\begin{alignat}{3}
\max & \quad & \sum_i\sum_{\ell\in\POS}\sum_{S\in\sols_{i,t_\ell}}\sum_{j\in S}\rewd_{ij}&x_{i,\ell,S}
\tag{\sepfllp} \label{lpsepfl} \\
\text{s.t.} & \quad & \sum_{S\in\sols_{i,t_\ell}} x_{i,\ell,S} & \leq y_{i,\ell} \qquad &&
\forall i,\,\forall\ell\in\POS \label{config} \\
&& \sum_i\sum_{\ell\in\POS}\sum_{S\in\sols_{i,t_\ell}:j\in S}x_{i,\ell,S} & \leq 1 \qquad
&& \forall j \label{casgn} \\
&& \sum_{\ell\in\POS} y_{i,\ell} & \leq 1 \qquad && \forall i \label{fasgn} \\
&& \sum_i\sum_{\ell'\in\POS: \ell'\leq\prev(\ell)}y_{i,\ell} & \leq \ell-1 \qquad && \forall 
\ell\in\POS \label{numl} \\
&& \sum_i\sum_{\ell\in\POS} y_{i,\ell} & \leq k \label{numk} \\
&& x,y \geq 0, \qquad
y_{i,\ell} & =\begin{cases}1 & \text{if $(i,\ell)\in\guessind$} \\ 0 & \text{otherwise}\end{cases}
\qquad && \forall i,\ell\in\F\times(\POS\cap[\ell_0])
\label{lguess}
\end{alignat}
%
%We have already discussed constraints \eqref{numl}. 
Constraints \eqref{config} encode that a
configuration can be chosen for an $(i,\ell)$ pair only if $i$ is assigned to index
$\ell$, and constraints \eqref{casgn} ensure that a a client is assigned at most once.  
Constraints \eqref{fasgn} encode that every facility is assigned to at most one index, and
constraints \eqref{numk} enforce that at most $k$ facilities are opened.

Let $\lpsepflopt$ denote the optimal value of \eqref{lpsepfl}.
%Lemma~\ref{lpsepflval} shows that $\lpsepflopt\geq\OPT$, under some
%conditions on $\vt$, and hence, 
Lemma~\ref{lpsepflval} shows there is a suitable vector $\vt\in\T$ and suitable set
$\guessind$ such that $\lpsepflopt$ is an upper bound on $\OPT$ and $f(t^{\exp})$ is close to
$B$, and Theorem~\ref{lpsepflsolve} shows that \eqref{lpsepfl} can be approximately
solved.  

\begin{lemma} \label{lpsepflval}
%If $t_\ell\geq\fvo^{\down}_\ell$ for all $\ell\in\POS$, then $\lpsepflopt\geq\OPT$. 
%the optimal value of \eqref{lpsepfl}, denoted $\lpsepflopt$, is at least $\OPT$.
%Therefore, 
There exists a vector $\vt\in\T$ and a choice of $\guessind$ such that
$\lpsepflopt\geq\OPT$ and $f(t^\exp)\leq (1+4\ve)B$.
\end{lemma}

\begin{comment}
\begin{lemma} \label{solapprox} \label{constfrp}
Given a $\beta$-approximation algorithm for \onefrp, 
%the \mwis problem on $\M_i$, for any $i\in\F$, 
one can devise a polytime algorithm $\alg$ that for any $i\in\F$, any
$t\in\Rp$, and any $v\in\Rp^\C$, returns a set $A\in\sols_i$ such that 
$v(A)\geq\frac{1}{\beta+1}\cdot\bigl(\max_{S\in\sols_{i,t}}v(S)\bigr)$, and 
$\sum_{j\in A}p_{ij}\leq 2t$. 
%
We say that $\alg$ is a $(\beta+1,2)$-approximation algorithm for 
{\em \constr \onefrp}. 
%the \mwis problem on the independence system $(\C,\sols_{i,t})$ for any
%$i\in\F$, $t\in\Rp$.
\end{lemma}
\end{comment}

\begin{theorem} \label{lpsepflsolve}
Let $\alg$ be a $(\rho,\gm)$-approximation algorithm for \constr \onefrp.
%the \mwis problem on the independence system $(\C,\sols_{i,t})$, for any $i$, any
%$t\in\Rp$. 
One can efficiently compute $(\bx,\by)$ that is a feasible solution to \sepfllpp, which is 
\eqref{lpsepfl} where we replace $\sols_{i,t_\ell}$ by $\sols_{i,\gm t_\ell}$ everywhere,
and such that $(\bx,\by)$ has objective value at least $\lpsepflopt/\rho$. 
More precisely, we have: 
\begin{enumerate}[label=(\alph*), topsep=0.25ex, itemsep=0.1ex, leftmargin=*]
%(a) 
\item $\bx_{i,\ell,S}>0\implies S\in\sols_{i,\gm t_\ell}$ for all $i$, all $\ell\in\POS$;
\label{configval}
%(b) 
\item $\sum_{S\in\sols_{i,\gm t_\ell}} \bx_{i,\ell,S} \leq \by_{i,\ell}$ for all  $i,\,\ell\in\POS$;
%(c) 
\item $\sum_i\sum_{\ell\in\POS}\sum_{S\in\sols_{i,\gm t_\ell}:j\in S}\bx_{i,\ell,S} \leq 1$ 
for all $j$;
%(d) 
\item $\by$ satisfies \eqref{fasgn}--\eqref{lguess}; and
%(e) 
\item $\sum_i\sum_{\ell\in\POS}\sum_{S\in\sols_{i,\gm t_\ell}}\sum_{j\in S}\rewd_{ij}\bx_{i,\ell,S}\geq\lpsepflopt/\rho$. 
\label{obj}
\end{enumerate}
\end{theorem}

The proof of Lemma~\ref{lpsepflval} %and~\ref{solapprox} 
is fairly routine.
Theorem~\ref{lpsepflsolve} follows from an application of the ellipsoid method, using the
algorithm for \constr \onefrp as an approximate separation oracle for the dual, an
approach that has been used in other settings. We defer these proofs to the end of this
section, to avoid detracting the reader. 

\begin{comment}
Recall that we assuming that we have a $\beta$-approximation algorithm $\alg$ for
\onefrp. By Lemma~\ref{solapprox}, this yields a $(\beta+1,2)$-approximation algorithm for
\constr \onefrp. 
So using this in Theorem~\ref{lpsepflsolve}, we know that 
%if $\vt$ satisfies $t_\ell\geq\fvo^{\down}_\ell$ for all $\ell\in\POS$, 
for a suitable choice of the vector $\vt$, we will obtain a solution of objective value at
least $\OPT/(\beta+1)$. 
\end{comment}
%By Lemma~\ref{lpsepflval}, we know that there is a vector $\vt\in\T$ and set $F$ for which
%Theorem~\ref{lpsepflsolve} will return objective value at least $\OPT/\rho$. 
In the sequel, we will assume that we have a vector $\vt\in\T$ and subset $\guessind$ 
%the solution returned by Theorem~\ref{lpsepflsolve} has objective value at least
%$\OPT/(\beta+1)$, 
for which Theorem~\ref{lpsepflsolve} returns objective value at least $\OPT/\rho$,  
and such that $f(t^{\exp})\leq (1+4\ve)B$. This is justified because one can simply
consider the vector $\vt\in\T$ satisfying $f(t^{\exp})\leq(1+4\ve)B$ and subset
$\guessind$ for which the solution returned by Theorem~\ref{lpsepflsolve} has maximum
objective value. 

\vspace*{-1ex}
\paragraph{Rounding algorithm.}
Let $(\bx,\by)$ be the solution returned by Theorem~\ref{lpsepflsolve} (for the above 
choice of $\vt$ and $I$) using the $(\rho,\gm)$-approximation algorithm $\alg$. 
%for \constr \onefrp. %supplied by Lemma~\ref{solapprox}. 
Let $\val=\sum_i\sum_{\ell\in\POS}\sum_{S\in\sols_{i,\gm t_\ell}}\sum_{j\in S}\rewd_{ij}\bx_{i,\ell,S}$ 
be its objective value. 
We may assume that $\sum_{S\in\sols_{i,\gm t_\ell}}\bx_{i,\ell,S}=\by_{i,\ell}$ for all
$(i,\ell)\in\F\times\POS$, since we can always set $\bx_{i,\ell,\es}$ appropriately to
achieve this. 

We use randomized rounding, in two independent stages. 
Independently, for each $(i,\ell)$, we choose exactly one set  
$\tS_{i,\ell}\in\sols_{i,\gm t_\ell}$ by picking set $S$ with probability
$\frac{\bx_{i,\ell,S}}{\by_{i,\ell}}$. 
Next, for every facility $i$ independently, we assign $i$ to at
most one index $\ell\in\POS$ by picking index $\ell$ with probability $\by_{i,\ell}$.
%$\frac{\by_{i,\ell}}{5}$. 
Let $T$ be the (random) set of $(i,\ell)$ pairs so chosen. 
Let $F=\{i\in\F: (i,\ell)\in T\text{ for some }\ell\in\POS\}$. 
(Note that for $i\in F$ there is {\em exactly} one index $\ell$ such that
$(i,\ell)\in T$, and $|T|=|F|$.)
For $i\in F$, we define $\tS_i=\tS_{i,\ell}$, where $\ell$ is the (unique) index such that 
$(i,\ell)\in T$.     

Let $\good$ be the good event that $|T|\leq(1+\dt)k$, and for every $\ell\in\POS$, we have    
$\bigl|\{(i,\ell')\in T: \ell'\leq\prev(\ell)\}\bigr|\leq (1+\dt)(\ell-1)$. 
If $\good$ does not occur, we return the empty solution where we do not open any
facilities and do not assign any clients. Otherwise,
%corresponding to the pairs in $T$.
for each client $j$, we consider the facilities in $F$ in non-increasing order of
$\rewd_{ij}$, and assign $j$ to the first facility $i\in F$ for which $j\in\tS_i$.
Call this the client pre-assignment.
%Under this client assignment, 
We compute the reward of each facility in $F$, which is the total reward of the clients
pre-assigned to it, and open the $k$ largest-reward facilities from $F$. We retain the
client pre-assignments to the opened facilities; the other clients are not assigned.

\vspace*{-1ex}
\paragraph{Analysis.}
To keep notation simple, for every $i,j$, define
$\bx_{i,\ell,j}=\sum_{S\in\sols_{i,\gm t_\ell}:j\in S}\bx_{i,\ell,S}$ for all $\ell\in\POS$, and 
$\bx_{ij}=\sum_{\ell\in\POS}\bx_{i,\ell,j}$. So we have $\sum_i\bx_{ij}\leq 1$.
Also, let $\val_j:=\sum_i\rewd_{ij}\bx_{ij}$ for a client $j$.
So we can write the objective value $\val$ of $(\bx,\by)$ as
$\sum_{i,j}\rewd_{ij}\bx_{ij}=\sum_j\val_j$. 

For $\ell\in\POS':=\POS\cup\{k+1\}$, let $\bad_\ell$ be the bad event that 
$\bigl|\{(i,\ell')\in T: \ell'\leq\prev(\ell)\}\bigr|>(1+\dt)(\ell-1)$. 
Note that $\bad_{k+1}$ is the event that $|T|>(1+\dt)k$.
%and let $\badk$ be the event that $|T|>k$. 
Let $\bad=\good^c$, and %denote that $\good$ does not happen. 
note that $\bad=\bigvee_{\ell\in\POS'}\bad_\ell$.
Observe that $\bigl|\{(i,\ell')\in T: \ell'\leq\prev(\ell)\}\bigr|$ is the sum
of independent (but not identical) Bernoulli random variables, and we only need to
consider indices $\ell\in\POS$ with $\ell>\ell_0$. 
Therefore, Chernoff bounds imply that $\Pr[\bad_\ell]=e^{-O(\ell)}$, and we argue that
this holds even when we condition on some of the random choices leading to $T$.
So using the union bound, one can argue that $\bad$ happens with very low probability,
under the same conditioning (see Lemma~\ref{badbnd}). Given
this, we can argue that the expected reward obtained from each client $j$ 
%conditioned on the event $\good$ 
is $\Omega(\val_j)$ (Lemma~\ref{clrewd}), and so
the overall expected reward from the client pre-assignment is
$\Omega(\val)=\Omega(\OPT/\rho)\bigr)$. The final step where we open a subset of $F$ and
drop some clients can cause the total reward to decrease by at most a $\frac{1}{1+\dt}$-factor,
since $|F|\leq(1+\dt)k$ and we open the $k$ largest-reward facilities in $F$.
Also, due to the choice of the vector $\vt$ and part~\ref{configval} of
Theorem~\ref{lpsepflsolve}, one can argue that the facility-load vector of our solution 
has norm at most $\bigl(1+O(\ve)\bigr)\gm\cdot B$ (Lemma~\ref{fvecnorm}).

%Let $Y\in\{0,1\}^{\F\times\POS}$ be the characteristic vector of the random set $T$, and
%let $Y_i=\sum_{\ell\in\POS}Y_{i,\ell}$ for every $i\in\F$; so $Y$ is the
%characteristic vector of $F$. 

\begin{lemma} \label{badbnd}
%Let $S\sse\F$ and $i'\in \F-S$.
%We have $\Pr[\bad\,|\,i'\in F,\,S\cap F=\es]\leq\badprob$.
Consider any $(i',\ell')$ pair with $\by_{i',\ell'}>0$.
We have $\Pr[\bad\,|\,(i',\ell')\in T]\leq\badprob$.
\end{lemma} 

\begin{proof}
Let $\Om$ denote the event $(i',\ell')\in T$. %$\{i\in F,\, S\cap F=\es\}$.
Let $Y_{i,\ell}$ be an indicator random variable that is $1$ if $(i,\ell)\in T$, and $0$
otherwise. Note that for $\ell\leq\ell_0$, $Y_{i,\ell}$ is actually a deterministic
quantity, but we can still treat it as a random variable.
By the union bound, we have 
$\Pr[\bad\,|\,\Om]\leq\sum_{\ell\in\POS'}\Pr[\bad_\ell\,|\,\Om]$.
%Clearly, $\Pr[\bad_1]=0$, so $\Pr[\bad_1\,|\,\Om]=0$.
Clearly, $\Pr[\bad_\ell]=0$ for all $\ell\in\POS'$, $\ell\leq\ell_0$, and so
$\Pr[\bad_\ell\,|\,\Om]=0$ for all such indices $\ell$.

Consider $\ell\in\POS'$, $\ell>\ell_0$.
Recall that $\ell_0\geq\frac{12}{\dt^2}\cdot\ln\frac{15}{\dt^3}$.
Define $Z_i=\sum_{\ell''\in\POS:\ell''\leq\prev(\ell)}Y_{i,\ell''}$. Note that the $Z_i$
variables are independent. 
%and 
%$\mu_\ell=\E{Z(\F)}=\bigl(\sum_{i\in\F}\sum_{\ell'\in\POS:\ell'\leq\prev(\ell)}\by_{i,\ell'}\bigr)/5
%\leq\frac{\ell-1}{5}$. 
We have 
\begin{equation*}
\begin{split}
\Pr[\bad_\ell\,|\,\Om]& =\frac{\Pr[\{Z(\F)>(1+\dt)(\ell-1)\}\wedge\Om]}{\Pr[\Om]}
\leq\frac{\Pr[\{Z(\F-\{i'\})\geq(1+\dt/2)\ell\}\wedge\Om]}{\Pr[\Om]} \\
& =\frac{\Pr[Z(\F-\{i'\})\geq\ell-1]\cdot\Pr[\Om]}{\Pr[\Om]}=\Pr[Z(\F-\{i'\})\geq\ell-1].
\end{split}
\end{equation*}
The first inequality follows because 
%event $\Om$ implies that $Z_i=0$ for all $i\in S$; also, 
$Z_{i'}\leq 1$ and $(1+\dt)(\ell-1)-1\geq (1+\dt)\ell-3\geq(1+\dt/2)\ell$ since 
$\ell\geq\ell_0\geq\frac{6}{\dt}$.
The subsequent equality follows because $Z(\F-\{i'\})$ depends only
on the random choices made for the facilities in $\F-\{i'\}$.
%So if $Z(\F)>\ell-1$ and event $\Om$ happens, then it must be that
%$Z(\F-S-\{i'\})\geq\ell-1$
We have
$\E{Z(\F-\{i'\})}\leq\E{Z(\F)}=\sum_{i\in\F}\sum_{\ell''\in\POS:\ell''\leq\prev(\ell)}\by_{i,\ell''}\leq\ell-1$. 
So using Chernoff bounds,% 
\footnote{We are using the following Chernoff bound. Let $X_1,\ldots,X_n$ be
independent $[0,1]$ random variables with $\sum_{i\in[n]}\E{X_i}\leq\mu$. Then, for
$\tht\in[0,1]$, we have $\Pr[\sum_{i\in[n]}X_i\geq(1+\tht)\mu]\leq e^{-\tht^2\mu/3}$.}
we obtain that 
$\Pr[Z(\F-\{i'\})\geq(1+\dt/2)\ell]\leq e^{-\frac{1}{3}\cdot(\dt/2)^2\cdot\ell}$.

\begin{comment}
\begin{equation}
\biggl(\frac{e^{4}}{5^5}\biggr)^{(\ell-1)/5}=\biggl(\frac{e^{4/5}}{5}\biggr)^{\ell-1}
\leq(0.4452)^{\ell-1}. \label{badlineq}
\end{equation}
Also, by Markov's inequality, $\Pr[Z(\F-\{i'\})\geq\ell-1]\leq 0.2$. 
For the smallest two indices in $\POS'$ that are larger than $1$, we use the bound given
by Markov's inequality, and for the others, we use \eqref{badlineq}. Note that the fourth
smallest index in $\POS'$ is at least $4$.
\end{comment}

Note that $e^{-{\dt^2}/{12}}\leq 1-\bigl(1-e^{-1/12}\bigr)\dt^2\leq 1-0.07\dt^2$, since
the function $e^{-x/12}$ is convex and decreasing, and so 
$e^{-x/12}\leq 1-\bigl(1-e^{-1/12}\bigr)x$ for all $x\in[0,1]$. 
The above bound on $\Pr[\bad_\ell\,|\,\Om]$ gives
\begin{equation*}
\Pr[\bad\,|\,\Om]\leq\sum_{\ell\in\POS':\ell>\ell_0}\Pr[\bad_\ell\,|\,\Om]
\leq\sum_{\ell>\ell_0}e^{-\frac{\dt^2}{12}\cdot\ell}
\leq\frac{e^{-\frac{\dt^2}{12}\cdot\ell_0}}{1-e^{-\dt^2/12}}
\leq\frac{e^{-\ln(15/\dt^3)}}{0.07\dt^2}\leq\frac{\dt}{1.05}. \qedhere
%\leq\frac{0.07\cdot e^{-1/\dt}}{0.07\dt^2}\leq\dt.
\end{equation*}
%The last three inequalities follow because $\frac{\ell_0}{12}\geq\frac{-\ln(0.07)}{\dt^3}$,
%$\ln(0.07)/x\leq\ln(0.07)-1/x$ for all $x\leq 0.5$, and 
%$e^x\leq x^3$ for all $x\geq 5$.
%\leq 0.4+\sum_{\ell\geq 4}(0.4452)^{\ell-1}\leq 0.4+\frac{0.4452^3}{1-0.4452}\leq 0.56. \qedhere
%\end{equation*}
\end{proof}

\begin{lemma} \label{clfacbnd}
Consider any facility $i$ and client $j$.
Let $A\sse\F$ be the set of facilities that come before $i$ in $j$'s ordering of
facilities (in non-increasing order of $\rewd_{i'j}$'s).
We have 
\[
\Pr[\text{$j$ is pre-assigned to facility $i$}] %in the solution returned}]
\geq\bx_{ij}\cdot\biggl(\prod_{i'\in A}(1-\bx_{i'j})-\badprob\biggr)
\]
%$\Pr[j\in\tS_i\,|\,\good]\geq\bx_{ij}\cdot.$
\end{lemma}

\begin{proof}
For any $i'\in\F$ and $\ell\in\POS$, let $X_{i',\ell,j}$ be the random variable indicating
if $j\in\tS_{i',\ell}$, and let $X_{i'j}=\sum_{\ell\in\POS}X_{i',\ell,j}$. 
Let $\tX_{ij}$ be the random variable indicating if $j$ is assigned to facility $i$ in the
solution returned. 
%
%Recall that $Y$ is the characteristic vector of the random set $T$.
We have 
$\Pr[\tX_{ij}=1]=\sum_{\ell\in\POS}\Pr[(i,\ell)\in T]\cdot\Pr[\tX_{ij}=1\,|\,(i,\ell)\in T]=
\sum_{\ell\in\POS}\by_{i,\ell}\cdot\Pr[\tX_{ij}=1\,|\,(i,\ell)\in T]$.
We proceed to lower bound $\Pr[\tX_{ij}=1\,|\,(i,\ell)\in T]$ for $(i,\ell)$ such
that $\by_{i,\ell}>0$. 

Client $j$ is pre-assigned to $i$, %in the solution returned, 
conditioned on $(i,\ell)\in T$,
precisely if the following three things happen (conditioned on $(i,\ell)\in T$).
\begin{enumerate*}[label=(\arabic*), topsep=0.1ex, noitemsep, leftmargin=*]
\item We have $X_{i,\ell,j}=1$; {\ }
\item The good event $\good$ occurs; {\ }
\item For every $i'\in A$ and every $\ell'\in\POS$, we have 
$X_{i',\ell',j}=0$ or $(i',\ell')\notin T$.
\end{enumerate*}
Let $\Gm_{i'}$ denote the event that for every $\ell'\in\POS$, we have $X_{i',\ell',j}=0$
or $(i',\ell')\notin T$. Then 
$\Pr[\Gm_{i'}]=1-\sum_{\ell'\in\POS}\by_{i',\ell'}\cdot\frac{\bx_{i',\ell',j}}{\by_{i',\ell'}}=1-\bx_{i'j}$.
Also, note that all the $\Gm_{i'}$ events are independent, and they are independent of the
event $\{(i,\ell)\in T\}$, since the choices made for different facilities are completely
independent. 
Recall that $\bad=\good^c$ and $\bad=\bigvee_{\ell'\in\POS'}\bad_{\ell'}$.
So we have
\begin{alignat}{1}
\Pr[\tX_{ij}=1\,&|\,(i,\ell)\in T] 
= \Pr\Bigl[\{X_{i,\ell,j}=1\}\wedge(\bigwedge_{i'\in A}\Gm_{i'})\wedge\good\,|\,(i,\ell)\in T\Bigr]
\notag \\ &
= \Pr\Bigl[\{X_{i,\ell,j}=1\}\wedge(\bigwedge_{i'\in A}\Gm_{i'})\,|\,(i,\ell)\in T\Bigr]-
\Pr\Bigl[\{X_{i,\ell,j}=1\}\wedge(\bigwedge_{i'\in A}\Gm_{i'})\wedge\bad\,|\,(i,\ell)\in T\Bigr]
\notag \\ &
\geq \Pr\bigl[X_{i,\ell,j}=1\,|\,(i,\ell)\in T\bigr]\cdot\prod_{i'\in A}\Pr[\Gm_{i'}]-
\Pr\bigl[\{X_{i,\ell,j}=1\}\wedge\bad\,|\,(i,\ell)\in T\bigr] \notag \\
%& = \Pr\bigl[X_{i,\ell,j}=1\,|\,(i,\ell)\in T\bigr]\cdot
%\biggl(\prod_{i'\in A}\Pr[\Kp_{i'}]-\Pr\bigl[\bad\,|\,X_{i,\ell,j}=1,\,(i,\ell)\in T\bigr]\biggr).
& = \Pr[X_{i,\ell,j}=1]\cdot
\biggl(\prod_{i'\in A}\Pr[\Gm_{i'}]-\Pr\bigl[\bad\,|\,X_{i,\ell,j}=1,\,(i,\ell)\in T\bigr]\biggr).
\label{clfacineq1}
\end{alignat}
The last equality is because $X_{i,\ell,j}$ depends only on the set $\tS_{i,\ell}$ chosen
for $(i,\ell)$.
Note that $\bad$ only depends on the set $T$, and not on the random choice of the client-sets
chosen for different $(i',\ell')$ pairs. So we have
$\Pr\bigl[\bad\,|\,X_{i,\ell,j}=1,\,(i,\ell)\in T\bigr]=\Pr\bigl[\bad\,|\,(i,\ell)\in  T\bigr]\leq\badprob$ 
by Lemma~\ref{badbnd}. Also, $\Pr[X_{i,\ell,j}=1]=\frac{\bx_{i,\ell,j}}{\by_{i,\ell}}$
So using \eqref{clfacineq1}, we have
\[
\Pr[\tX_{ij}=1\,|\,(i,\ell)\in T]\geq\frac{\bx_{i,\ell,j}}{\by_{i,\ell}}\cdot
\biggl(\prod_{i'\in A}(1-\bx_{i'j})-\badprob\biggr).
\]
This yields that 
$\Pr[\tX_{ij}=1]\geq\bx_{ij}\cdot\bigl(\prod_{i'\in A}(1-\bx_{i'j})-\badprob\bigr)$.
\end{proof}

\begin{lemma} \label{clrewd}
The expected reward obtained in the pre-assignment from any client $j$ 
is at least $\bigl(1-e^{-1}-\badprob\bigr)\val_j$.
\end{lemma}

\begin{proof}
We will utilize the following well-known claim. 
%(The proof is included in Appendix~\ref{append-sepfl} for completeness.)

\begin{claim}[Claim 4.3 in~\cite{FriggstadS21}] \label{helper}
Let $a_1\geq a_2\geq\ldots\geq a_q\geq a_{q+1}:=0$. Let $z\in[0,1]^q$, and 
$t\geq\sum_{i=1}^qz_i$. 
We have 
$z_1a_1+(1-z_1)z_2a_2+\ldots+(1-z_1)(1-z_2)\cdots(1-z_{q-1})z_qa_q\geq (1-e^{-t})\cdot\frac{\sum_{i=1}^qa_iz_i}{t}$.
\end{claim}

Fix a client $j$. Let $i_1,i_2,\ldots,i_{|\F|}$ be $j$'s ordering of facilities in
non-increasing order of $\rewd_{ij}$ values.
Using Lemma~\ref{clfacbnd}, we obtain that 
\begin{equation*}
\begin{split}
\E{\text{expected reward from $j$}} & \geq
\sum_{q=1}^{|\F|}\rewd_{i_q j}\cdot\bx_{i_q j}\cdot
\biggl(\prod_{r=1}^{q-1}(1-\bx_{i_r j})-\badprob\biggr) \\
& \geq \bigl(1-e^{-1}\bigr)\cdot\sum_{q=1}^{|F|}\rewd_{i_q j}\cdot\bx_{i_q j}-
\badprob\cdot\sum_{q=1}^{|F|}\rewd_{i_q j}\cdot\bx_{i_q j} 
\geq\bigl(1-e^{-1}-\badprob\bigr)\val_j
\end{split}
\end{equation*}
where the first inequality follows from Claim~\ref{helper}, since
$\sum_{i\in\F}\bx_{ij}\leq 1$.
\end{proof}

\begin{lemma} \label{fvecnorm}
We have 
$f(\text{facility-load vector of solution returned})\leq\bigl(1+O(\ve)\bigr)\cdot\gm B$ 
with probability $1$.
\end{lemma}

\begin{proof}
Let $\fvec$ be the facility-load vector of the solution returned. If $\good$ does not
occur, then $\fvec=\stdvec{0}$, so suppose otherwise. 
Then, by design, we have $\bigl|\{(i,\ell')\in T: \ell'\leq\prev(\ell)\}\bigr|\leq(1+\dt)(\ell-1)$
for every $\ell\in\POS'=\POS\cup\{k+1\}$, and if $(i,\ell)\in T$, then the load
of facility $i$ is at most $\gm t_\ell$.
Thus, $\fvec$ has the property that there are at most $(1+\dt)(\ell-1)$ coordinates of
value larger than $\gm t_\ell$, for every $\ell\in\POS$. 
By Lemma~\ref{toplestim} (b), we therefore obtain that 
$f(\fvec)\leq (1+3\dt)f(\gm t^{\exp})$. Since $f(t^{\exp})\leq(1+4\ve)B$ and 
$\dt=\min\{\ve, 1\}$, this yields 
\begin{equation*}
f(\fvec)\leq\gm(1+3\dt)f(t^{\exp})
\leq\gm(1+3\dt)(1+4\ve)B\leq (1+19\ve)\cdot\gm B. \qedhere
\end{equation*}
%It follows that $\fvec^{\down}_\ell\leq \gm t_\ell$
%for every $\ell\in\POS$, and hence $\fvec^{\down}\leq 2t^{\exp}$.
%Therefore, $f(\fvec)\leq 2f(t^{\exp})\leq(2+8\ve)B$.
\end{proof}

%The following claim, which follows from a result in~\cite{IbrahimpurS21}, shows that
%then we must have $f(\fvec)\leq 2(1+4\dt)f(t^{\exp})\leq \bigl(2+O(\ve)\bigr)B$.

%shows that if the facility-load vector of our solution satisfies these bounds
%approximately,  %is ``close'' to $\vt$ in a certain sense, 
%then we approximately satisfy the norm budget.

\begin{comment}
\begin{lemma}[Follows from Lemma 2.8 in~\cite{IbrahimpurS21}]
%Let $\vt$ be as above, i.e., 
%$\fvo^{\down}_\ell\leq t_\ell\leq(1+\ve)\fvo^{\down}_\ell+\kp$ for all $\ell\in\POS$.
Let $\al\in\R^k$ be such that $\al^{\down}_1\leq\rho t_1$, and 
$\bigl|\{i\in[k]: \al_i>\rho t_\ell\}\bigr|\leq\ell-1$ for all $\ell\in\POS$. Then, 
$f(\al)\leq \rho(1+4\dt)f(t^{\exp})$.
\end{lemma}
\end{comment}

\begin{proofof}{Theorem~\ref{normsepfl-bi}}
Recall that the objective value of the LP solution $(\bx,\by)$ is at least
$\OPT/\rho$. %due to Theorem~\ref{lpsepflsolve}~\ref(obj)
So Lemma~\ref{clrewd} shows that the expected reward obtained from the pre-assignment at
least $\bigl(1-e^{-1}-\badprob\bigr)\frac{\OPT}{\rho}$. 
%If the good event $\good$ does not occur, then no clients are pre-assigned, so this
%also lower bounds the expected reward of the pre-assignment conditioned on
%$\good$ occurring. 
Since $|F|\leq k(1+\dt)$ when the good event $\good$ occurs, and we return the $k$
facilities in $F$ obtaining the largest reward under the pre-assignment, it follows that
the expected reward of the final solution is always at least $\frac{1}{1+\dt}$ times the
expected reward of the pre-assignment. (This holds even when $\good$ does not occur, as
then no clients are pre-assigned.) So the expected reward of the final solution is at
least $\frac{1}{1+\dt}\cdot\bigl(1-e^{-1}-\badprob\bigr)\cdot\frac{\OPT}{\rho}$.
Lemma~\ref{fvecnorm} bounds the norm of the facility-load vector returned.
\end{proofof}

%\subsection*{Proofs of Lemma~\ref{lpsepflval} and Theorem~\ref{lpsepflsolve}} 

\begin{proofof}{Lemma~\ref{lpsepflval}}
Suppose $t_\ell\geq\fvo^{\down}_\ell$ for all $\ell\in\POS$; equivalently
$t^{\exp}\geq\fvo^{\down}$. 
Consider the optimal solution $(\fopt,\,\sg^*:\copt\mapsto\fopt)$. 
For every $(i,\ell)$ pair, we set $y_{i,\ell}=1$ if $i\in\fopt$ and
$\fvo_i\in(t_{\nxt(\ell)},t_\ell]$, and $0$ otherwise. 
Also, let $\guessind$ be the $(i,\ell)$ pairs where $y_{i,\ell}=1$ and $\ell\leq\ell_0$.
These $y$-values satisfy \eqref{fasgn}, \eqref{numk}, and \eqref{lguess}, and since 
$t^{\exp}\geq\fvo^{\down}$, we also satisfy \eqref{numl}.
For every $(i,\ell)$ for which $y_{i,\ell}=1$, we set $x_{i,\ell,S}=1$ if $S$ is the set
of clients assigned to $i$; all other $x$-variables are set to $0$. It is easy to see that
by construction, this satisfies \eqref{config}, \eqref{casgn}. 
So $\lpsepflopt$ is at least the objective value of this feasible solution, which is
$\OPT$. 

There exists $\vt\in\T$ such that 
$\fvo^{\down}_\ell\leq t_\ell\leq(1+\ve)\fvo^{\down}_\ell+\kp$. For this vector $\vt$, we
have $\lpsepflopt\geq\OPT$, as argued above. By Lemma~\ref{toplestim} (a), we also have
$f(t^{\exp})\leq(1+\dt)(1+\ve)f(\fvo)+\ve B\leq (1+4\ve)B$, where we use the fact that
$\dt\leq\ve$, $\dt\leq 1$.
\end{proofof}

\newcommand{\dualconstr}{\ensuremath{\text{(*)}}\xspace}
\newcommand{\dlpsepfl}{(D)\xspace}

\begin{proofof}{Theorem~\ref{lpsepflsolve}}
We argue that $\alg$ can be used to obtain a kind of approximate separation oracle for the
dual, which, in conjunction with the ellipsoid method, yields the stated result.
This idea has been used in other settings (see, e.g.,~\cite{FriggstadS14},
\cite{ChakrabartyS16}). 
The dual \dlpsepfl has a polynomial number of variables, exponentially many
constraints corresponding to the $x_{i,\ell,S}$ variables of the primal LP
\eqref{lpsepfl}, and a polynomial number of other constraints corresponding to the
$y_{i,\ell}$ variables. Let $\mu_{i,\ell}\geq 0$, and $\tht_j$ be the dual variables
corresponding to constraints \eqref{config} and \eqref{casgn} respectively. 
The dual constraints corresponding to the $x_{i,\ell,S}$ variables
are:
\begin{equation}
\mu_{i,\ell}+\sum_{j\in S}\tht_j\geq\sum_{j\in S}\rewd_{ij} \qquad 
\forall i\in\F,\ \forall\ell\in\POS,\ \forall S\in\sols_{i,t_\ell}. \label{dlpconfig}
\end{equation}

\newcommand{\dlpc}[1][{}]{\ensuremath{\text{(\ref{dlpconfig}}_{#1}\text{)}}\xspace}

Defining $v_{ij}=\max\{\rewd_{ij}-\tht_j,0\}$ for all $j\in\C$, 
it is easy to see that constraints \eqref{dlpconfig} are equivalent to 
$\bigl(\max_{S\in\sols_{i,t_\ell}}\sum_{j\in S}v_{ij}\bigr)\leq\mu_{i,\ell}$.
So for any $\mu,\tht$, we can use $\alg$ to determine if constraints
\eqref{dlpconfig} hold, or find $i,\ell,A\in\sols_{i,\gm t_\ell}$ such that 
$\frac{\mu_{i,\ell}}{\rho}<\sum_{j\in A}\bigl(\rewd_{ij}-\tht_j\bigr)$. This is because,
suppose we run $\alg$ 
%on the independence system $(\C,\sols_{i,t_\ell})$
with the input tuple $(i, t_\ell, \{v_{ij}\}_{j\in\C})$ and obtain a set
$A\in\sols_{i,\gm t_\ell}$. If $\sum_{j\in A}v_{ij}>\frac{\mu_{i,\ell}}{\rho}$, then we also
have $\sum_{j\in S}\bigl(\rewd_{ij}-\tht_j\bigr)>\frac{\mu_{i,\ell}}{\rho}$ for 
%the subset of $A$ corresponding to clients $j$ with 
$S=\{j\in A: v_{ij}\geq 0\}$ (which also lies in $\sols_{i,\gm t_\ell}$). Otherwise, we
know that $\bigl(\max_{S\in\sols_{i,t_\ell}}\sum_{j\in S}v_{ij}\bigr)\leq\mu_{i,\ell}$, i.e.,
constraints \eqref{dlpconfig} hold for $i,\ell$, and all $S\in\sols_{i,t_\ell}$.

We utilize this as follows. Let $\phi$ denote the remaining dual variables, and
\dualconstr denote the dual constraints (including nonnegativity) constraints other than 
\eqref{dlpconfig}. The objective function of \dlpsepfl is of the form 
$\sum_j\tht_j+h^T\phi$, where $h$ is some fixed vector. 
%Let \eqref{dlpcp} denote 
Consider the dual LP with the following modified version of \eqref{dlpconfig}:
\begin{equation}
\frac{\mu_{i,\ell}}{\rho}+\sum_{j\in S}\tht_j\geq\sum_{j\in S}\rewd_{ij} \qquad 
\forall i\in\F,\,\forall\ell\in\POS,\,\forall S\in\sols_{i,\gm t_\ell}. 
\tag{\ref{dlpconfig}'} \label{dlpcp}
\end{equation}
%For $b\geq 0$, let \dlpc[b] denote \eqref{dlpconfig} with $\mu_{i,\ell}$ replaced by
%$\mu_{i,\ell}/b$. 
%In the primal LP, this changes constraint \eqref{config} to
%$\frac{1}{b}\cdot\sum_{S\in\sols_{i,t_\ell}}x_{i,\ell,S}\leq y_{i,\ell}$ for every $i$,
%$\ell\in\POS$. 
The effect of this in the primal LP is that we now have variables $x_{i,\ell,S}$ for every
$i$, $\ell\in\POS$, and $S\in\sols_{i,\gm t_\ell}$, and constraint \eqref{config} changes
to $\frac{1}{\rho}\cdot\sum_{S\in\sols_{i,\gm t_\ell}}x_{i,\ell,S}\leq y_{i,\ell}$ for every $i$,
$\ell\in\POS$. Let (\sepfllpp) denote this modified primal LP.

Let
$\K(\nu):=\bigl\{(\mu,\tht,\phi): \dualconstr,\ \eqref{dlpconfig},\ \sum_j\tht_j+h^T\phi\leq\nu\bigr\}$
denote the set of dual feasible solutions achieving objective value at most $\nu$.
Thus, the optimal value of the dual, and hence \eqref{lpsepfl}, is the smallest $\nu$
such that $\K(\nu)\neq\es$.
Also, let 
$\K'(\nu):=\bigl\{(\mu,\tht,\phi): \dualconstr,\ \eqref{dlpcp},\ \sum_j\tht_j+h^T\phi\leq\nu\bigr\}$.
Given $\nu$, $(\mu,\tht,\phi)$, by our earlier discussion, we can use $\alg$ to either
show that $(\mu,\tht,\phi)\in\K(\nu)$, or find a hyperplane separating
$(\mu,\tht,\phi)$ from $\K'(\nu)$. 
\begin{comment}
We first check if $\mu,\tht\geq 0$, $\dualconstr$
and $\sum_{i=0}^\num\tht^i+h^T\beta\leq\nu$ hold, and if not use the violated inequality as
the separating hyperplane. Next, for each $i=0,\ldots,\num$, we run $\A$ on the
\ptp-orienteering instance with start and end nodes $u_i,w_i$ respectively, length bound
$c_{u_iw_i}+\regb[i]$, and node-rewards $\{\mu^i_v\}_{v\in\nV}$. If for some $i$, we obtain a
path $P\in\Pc^i$ with reward greater than $\tht^i/\al$, then we return 
$\sum_{v\in P}\mu^i_v\leq\tht^i/\al$ as the separating hyperplane. Otherwise, for all
$i=0,\ldots,\num$ and all $P\in\Pc^i$, we know that $\sum_{v\in P}\mu^i_v\leq\tht^i$, and so 
$(\beta,\mu,\tht)\in\K(\nu;1)$.
\end{comment}
Thus, for a fixed $\nu$, by running the ellipsoid method, in polynomial time, we either
certify that $\K'(\nu)=\es$, or find a point $(\mu,\tht,\phi)\in\K(\nu)$. 

Now we can combine this with binary search to find $\nu^*$ that is an upper bound on
$\lpsepflopt$, and a near-feasible solution to (\sepfllpp) achieving this objective
value; scaling this solution will yield $(\bx,\by)$ satisfying the stated properties.

It is easy to find an upper bound $\ub$ such that $\K(\ub)\neq\es$.
For a given $\e>0$, we use binary search in the range $[0,\ub]$ to find $\nu^*$ such that
the ellipsoid method when run for $\nu^*$ (with the above separation oracle) returns
%$(\mu^*,\tht^*,\phi^*)\in
a point in $\K(\nu^*)$, and when run for $\nu^*-\e$ certifies that
$\K'(\nu^*-\e)=\es$. Since $\K(\nu^*)\neq\es$, we have that $\lpsepflopt\leq\nu^*$, 
and $\K'(\nu^*-\e)=\es$ implies that the optimal value of (\sepfllpp) 
is at least $\nu^*-\e$. 
For $\nu^*-\e$, the inequalities returned by the separation oracle 
during the execution of the ellipsoid method together with the inequality
$\sum_{j}\tht_j+h^T\phi\leq\nu^*-\e$ yield a polynomial-size certificate for the
emptiness of $\K'(\nu^*-\e)$. By duality (or Farkas' lemma), this implies that if we
restrict (\sepfllpp) to only use the (polynomially many) $x_{i,\ell,S}$ variables
corresponding to the violated inequalities of type \eqref{dlpcp} returned during
the execution of the ellipsoid method, we obtain a polynomial-size feasible solution
$(\hx,\by)$ to (\sepfllpp) of value at least $\nu^*-\e$. 
If we take $\e$ to be inverse exponential in the input size, this also implies that
$(\hx,\by)$ has objective value at least $\nu^*\geq\lpsepflopt$. 
Finally, setting $\bx=\hx/\rho$, we obtain that $(\bx,\by)$ has the desired properties. 
\end{proofof}

\section{Norm-budgeted packing problems with submodular rewards} \label{subnb}
We now consider norm-budgeted packing problems where the reward function,
%a substantial generalization of \normknap, 
%where instead of having 
instead of being an additive (or modular) function (induced by the
item-rewards) is a {\em monotone, submodular} function $\rewd:2^n\mapsto\R_+$, with
$\rewd(T)$ denoting the reward obtained from a set $T\sse[n]$ of items.
So in a generic {\em submodular norm-budgeted packing problem},
the goal (as before) is to maximize $\rewd(T)$ subject to $T$ being an independent in a
given independence system $([n],\sols)$, and the norm budget constraint
$f\bigl(\svec{T}\bigr)\leq B$.
%Clearly, this generalizes \normknap, wherein $\rewd(.)$ is an additive (or modular)
%function induced by the item rewards.
We assume that $\rewd$ is specified via a value oracle. 
%As before, the specification of $([n],\sols)$ and the size-vector $\svec{T}$ will
%depend on the problem being considered.
%We sometimes refer to this as the submodular generalization of the corresponding
%norm-budgeted packing problem (where we have the same independence system $([n],\sols)$
%and the same definition of $\svec{T}$).

We obtain results for the submodular generalizations of some of the norm-budgeted packing 
problems considered in the paper. 
We develop $O(1)$ approximation guarantees for 
the {\em submodular norm-budgeted knapsack} (\subnbknap) problem
(Section~\ref{submodknap}), using a fundamentally different approach than that used for
(regular) \normknap in Section~\ref{knapsack},
and {\em submodular norm-budgeted \maxgap} (\subnbsched) on identical and related machines 
(Section~\ref{submodlb}). The latter uses our result for \subnbknap (essentially) as a
black-box, capitalizing on the observation that the reduction in Theorem~\ref{lbredn}
still works with submodular rewards.
%We leave the question of devising algorithms
%for other submodular norm-budgeted packing problems, with guarantees that qualitatively
%match the guarantees obtained for their (regular, i.e., additive) counterparts, for
%future work. 
%As discuss in more detail in Section~\ref{submodknap}, this requires us to come up with a
%very different approach

\subsection{Submodular norm-budgeted knapsack} %with submodular rewards} 
\label{submodknap}
Recall that in the submodular norm-budgeted knapsack (\subnbknap) problem, 
%where the setup the same as in \normknap, 
%a substantial generalization of \normknap, wherein instead of having an additive (or
%modular) reward function (induced by the item-rewards) we
we have a {\em monotone, submodular} reward-function $\rewd:2^n\mapsto\R_+$ 
%with $\rewd(T)$ denoting the reward obtained from a set $T\sse[n]$ of items.
and the goal is to maximize $\rewd(T)$ subject to $f\bigl(\svec{T}\bigr)\leq B$, where
$\svec{T}$ is the size-weighted characteristic vector of $T$.
%Clearly, this generalizes \normknap, wherein $\rewd(.)$ is an additive (or modular)
%function induced by the item rewards.
%We assume that $\rewd$ is specified via a value oracle, 
We develop an $O(1)$-approximation algorithm for \subnbknap. The approximation factor we
obtain is in fact $\bigl(1+O(\ve)\bigr)(\beta+1)$, where $\beta$ is the approximation
factor for maximizing a submodular function subject to a cardinality constraint
(Theorem~\ref{improvsubnbknapthm} in Section~\ref{submodknap-improv}).
%partition-matroid independence constraint. 
It is known that $\beta=\frac{e}{e-1}$, 
%(even with a general matroid-independence constraint), 
so this yields a $\bigl(\frac{2e-1}{e-1}+O(\ve)\bigr)$-approximation for \subnbknap.

Before describing our algorithm, we first discuss why our earlier approach leading to a
PTAS for \normknap (Section~\ref{knapsack}) is not amenable to handling \subnbknap. 
In the PTAS for \normknap, we set up reward buckets, comprising items having roughly the
same reward, so that any subset containing a certain number of
items from each reward bucket yields good total reward, and in order to satisfy the norm  
budget constraint, we greedily pick the smallest-size items from each reward bucket. 

With a submodular reward function, we run into some immediate problems with this
approach. It is not at all clear what a reward bucket should be in this setting, since
$\rewd(T)$ is 
not separable across items and the contribution of an item to $\rewd(T)$ depends on the
other items. Moreover, suppose that one could even identify some $\{\itemset_j\}$ item
buckets, target rewards $\{\rest_j\}$ to obtain from these buckets, and guarantee
that there is some set of $\num_j$ items in each $\itemset_j$ bucket such
that: (i) these yield the desired $\rest_j$ reward from the bucket, and (ii)
the resulting item-set (comprising items picked from the various buckets) is a feasible
solution. 
The question still remains: {\em how do we pick the $\num_j$ items from each $\itemset_j$
bucket.}  
Greedily picking the $\num_j$ smallest-size items will ensure feasibility, but this will not
usually yield the desired target reward: obtaining a 
target reward by picking a certain number of items corresponds to a
cardinality-constrained submodular-function maximization problem, and approximation
algorithms for this problem usually use greedy approaches based on {\em reward}, as
opposed to size.

The upshot is that dealing with \subnbknap requires us to fundamentally rethink our
approach. Instead of reward buckets, we will now work by creating, loosely speaking,
certain {\em size buckets} by grouping items based on their size (though
%we emphasize that 
items in a size bucket will not necessarily have similar sizes). 
In the reward-bucketing approach, %things were set up so that 
any solution built from the reward buckets was guaranteed to have good reward and a
greedy choice ensured feasibility. In contrast, we will now set things up so 
%as to guarantee feasibility and will solve an optimization problem to obtain good reward; 
%we will argue 
that any solution that picks at most a certain number of items from
each size bucket is {\em guaranteed to be feasible}, and we will solve a suitable
submodular-function maximization problem to obtain good reward. 

\medskip
We now delve into details. For ease of exposition, we describe here a 
$\bigl(1+O(\ve)\bigr)(2\beta+3)$-approximation algorithm, which will convey all the main
ideas, and  discuss the improvement to $\bigl(1+O(\ve)\bigr)(\beta+1)$-approximation in
Section~\ref{submodknap-improv}. 
%Let $\dt=\min\{\ve,1\}$. 
Let $\POS=\POS_{n,1}$ (see Definition~\ref{posdef}), so $\POS$ consists of all
powers of $2$ up to (and potentially including) $n$. Since 
%$\n,\dt$ (and hence $\POS$) 
$\POS$ will be fixed throughout, we drop $n, \dt$ from $\nxt$ and $\prev$. 
Define $\nxt(0):=1$ and $\prev(0):=0$ for notational convenience.
Let the items be ordered so that $\sz_1\geq\sz_2\geq\ldots\geq\sz_n$.
Recall that $\optset\sse[n]$ denotes some fixed optimal solution, and
$\OPT=\rewd(\optset)$ is the optimal value.
Let $\optset=\{\optind_1,\optind_2,\ldots,\optind_{\nopt}\}$, where
$\optind_1<\optind_2<\ldots<\optind_{\nopt}$. 
Let $\ellfin$ be the largest index in $\POS\cap[\nopt]$.
Let $I^*=\{\optind_\ell: \ell\in\POS\cap[\nopt]\}$, which denotes the
$\ell$-th-largest-size items in $\optset$, where $\ell$ ranges over $\pos\cap[\nopt]$.
Define $\optind_0:=0$ and $\optind_{\nxt(\ellfin)}:=n$ for notational convenience.%
\footnote{For example, suppose $n=1000$, and $|\optset|=\nopt=240$ with $\optind_{240}=480$,
i.e., the last item in $\optset$ is the $480$-th item in our ordering of items. 
%i.e., the item with the $480$-th largest size.
Then, $\ellfin=128$. We have 
$I^*=\{\optind_1,\optind_2,\optind_4,\optind_8,\optind_{16},\optind_{32},\optind_{64},\optind_{128}\}$,
and $\optind_0:=0$ and $\optind_{256}=1000$.}

Let $\POS':=\{0\}\cup(\POS\cap[\ellfin])$.
%For $\ell\in\{0\}\cup(\POS\cap[\ellfin])$, 
For $\ell\in\POS'$, define 
$\sbktopt_\ell:=\{\optind_\ell+1,\ldots,\optind_{\nxt(\ell)}\}$;
we think of $\sbktopt_\ell$ as a size bucket. We emphasize however that (in contrast with
reward buckets): (i) items in a size 
bucket need not have similar size; and (ii) we cannot actually identify these size
buckets, since we do not know $\optset$.
Observe that these size buckets %$\bigl\{\sbkt_\ell\bigr\}_{\ell\in\dbrack{\ellfin}}$ 
partition $[n]$ (but it could be that $\sbktopt_{\ellfin}=\es$, if $\ellfin=n$).
Also, by construction, for all $\ell\in\POS'$, $\ell<\ellfin$, we have that $\optset$
contains exactly $\nxt(\ell)-\ell$ items from $\sbktopt_\ell$: we have
$\optset\cap\sbktopt_\ell=\{\optind_{\ell+1},\optind_{\ell+2},\ldots,\optind_{\nxt(\ell)}\}$.
%(see Fig.~\ref{optset-fig}). 
For the last size bucket, we have that
$|\optset\cap\sbktopt_{\ellfin}|\leq\nxt(\ellfin)-\ellfin$ (and note that if $\nopt\in\POS$
then $\optset\cap\sbktopt_{\ellfin}=\es$). 

Lemma~\ref{sizebnds} states one of our main insights, namely that any set $A\sse[n]$
containing at most a certain number of items from each size bucket satisfies
$f\bigl(\svec{A}\bigr)\leq B$. 
%yields a feasible solution. 

\begin{lemma} \label{sizebnds}
Let $A\sse[n]$.
\begin{enumerate}[label=(\alph*), topsep=0.1ex, itemsep=0ex, leftmargin=*]
\item If $\bigl|A\cap\sbktopt_\ell\bigr|\leq\ell-\prev(\ell)$ for all
$\ell\in\POS'$, then $f\bigl(\svec{A}\bigr)\leq B$.

\item Let $I=\{i_\ell: \ell\in\pos\cap[k]\}$ be an index-set, where $k\leq\ellfin$ and
$i_\ell\geq\optind_\ell$ for all $\ell\in\pos\cap[k]$. Define $i_0:=0$ and $i_{\nxt(k)}:=n$.
Suppose that $\bigl|A\cap\{i_\ell+1,\ldots,i_{\nxt(\ell)}\}\bigr|\leq\ell-\prev(\ell)$ for
all $\ell\in\POS'$, $\ell\leq k$. Then, $f\bigl(\svec{A}\bigr)\leq B$.
\end{enumerate}
\end{lemma}

\begin{proof}
%Part (b) generalizes part (a), but it is slightly simpler to prove part (a), so we prove
%part (a) first. 
We will define a vector $v\in\R^n$, and argue that $f(v)\leq B$ and 
$\svec{A}^{\down}\leq v$. 
Let $v\in\R^n$ be the following vector: 
\[
%v_1=\sz_{\optind_1}; \quad  
\frall \ell\in\pos\cap[\ellfin],\,\text{ and all }j\in\{\prev(\ell)+1,\ldots,\ell\},
\ \ \text{set }v_j=\sz_{\optind_\ell}; \quad \text{$v_j=0$ for all other $j$}.
\]
%For part (a), 
%First, we show that $f(v)\leq B$. 
%Let $T$ be a feasible solution to \normknap that is compatible with $I$. Recall from
%Definition~\ref{validind}, this means that 
%$\bigl|T\cap\{i_\ell+1,i_{2\ell}\}\bigr|=\ell$ for all $\ell\in\pos$, $\ell\leq k/2$.
We argue that $v\leq\svec{\optset}^{\down}$, which implies that 
$f(v)\leq f\bigl(\svec{\optset}\bigr)\leq B$. 
Let $u=\svec{\optset}^\down$.
%We have $v_1=\sz_{\optind_1}=u_1$.
Consider any $\ell\in\pos\cap [\ellfin]$. %with $\ell\geq 2$. 
We have $v_j=\sz_{\optind_\ell}$ for all $j\in\{\prev(\ell)+1,\ldots,\ell\}$. Now consider
the %same coordinates for the 
vector $u$. The number of items in $\optset$ up to and including $\optind_\ell$ is
precisely $\ell$,
%$\bigl|\optset\cap[\optind_\ell]\bigr|=1+\sum_{r=1}^{\ell/2}\bigl|T\cap\{i_r+1,\ldots,i_{2r}\}\bigr|=\ell$, 
and so $u_j\geq\sz_{\optind_\ell}$ for all $j\leq\ell$. It follows that 
$(v_j)_{j\in\{\prev(\ell)+1,\ldots,\ell\}}\leq (u_j)_{j\in\{\prev(\ell)+1,\ldots,\ell\}}$. This
holds for all $\ell\in\pos\cap[\ellfin]$, so together this covers all non-zero
coordinates of $v$, and we obtain that $v\leq u$.

%Now consider part (b). 
Next, we argue that $\svec{A}^{\down}\leq v$.
Let $z=\svec{A}^\down$. 
\begin{comment}
Recall that $\sbktopt_\ell:=\{\optind_\ell+1,\ldots,\optind_{\nxt(\ell)}\}$ for all
$\ell=0,1,\ldots,\ellfin$, and these size buckets partition $[n]$.
Recall that 
$\optset\cap\sbktopt_\ell=\{\optind_{\ell+1},\optind_{\ell+2},\ldots,\optind_{\nxt(\ell)}\}$
for all $\ell=0,\ldots,\ellfin-1$, and 
$\optset\cap\sbktopt_{\ellfin}=\{\optind_{\ellfin+1},\optind_{\ellfin+2},\ldots,\optind_{\nopt}\}$
(which could be the empty set).
\end{comment}
%
For part (a), 
%Since $A\cap[i_1-1]=\es$, we have $z_1\leq\sz_{i_1}=v_1$. 
consider any $\ell\in\pos\cap[\ellfin]$, %with $2\leq\ell\leq k$, 
and consider the coordinates $j\in\{\prev(\ell)+1,\ldots,\ell\}$. 
We have 
\[
\bigl|A\cap[\optind_\ell]\bigr|=
\sum_{r\in\POS':r\leq\prev(\ell)}\bigl|A\cap\sbktopt_r\bigr|
\leq \sum_{r\in\POS':r\leq\prev(\ell)}\Bigl(r-\prev(r)\Bigr)=\prev(\ell).
\]
%consider the coordinates $j\in\{\prev(\ell)+1,\ldots,\ell\}$. 
So for all $j>\prev(\ell)$, we have $z_j\leq\sz_{\optind_\ell}$. On the other hand, as
argued above we have $v_j=\sz_{\optind_\ell}$ for all $j\in\{\prev(\ell)+1,\ldots,\ell\}$.  
So we obtain that 
$(z_j)_{j\in\{\prev(\ell)+1,\ldots,\ell\}}\leq (v_j)_{j\in\{\prev(\ell)+1,\ldots,\ell\}}$. This holds
for all $\ell\in\pos\cap[\ellfin]$, so we obtain that 
$(z_j)_{j\in[\ellfin]}\leq (v_j)_{j\in[\ellfin]}$. Finally, note that 
$|A|=\sum_{r\in\POS'}\bigl|A\cap\sbktopt_r\bigr|\leq\sum_{r\in\POS'}\bigl(r-\prev(r)\bigr)=\ellfin$, 
so $z_j=0$ for all $j>\ellfin$. It follows that $z\leq v$.

For part (b), consider any $\ell\in\pos\cap[k]$.
Similar to above, we have 
\[
\bigl|A\cap[i_\ell]\bigr|=
\sum_{r\in\POS':r\leq\prev(\ell)}\Bigl|A\cap\{i_r+1,\ldots,i_{\nxt(r)}\}\Bigr|
\leq \sum_{r\in\POS':r\leq\prev(\ell)}\Bigl(r-\prev(r)\Bigr)=\prev(\ell).
\]
%consider the coordinates $j\in\{\prev(\ell)+1,\ldots,\ell\}$. 
So $z_j\leq\sz_{i_\ell}$ for all $j>\prev(\ell)$. 
We also have as before $v_j=\sz_{\optind_\ell}$ for all $j\in\{\prev(\ell)+1,\ldots,\ell\}$, 
%(as shown earlier), 
and $\sz_{\optind_\ell}\geq\sz_{i_\ell}$, since $i_\ell\geq\optind_\ell$.  
So $(z_j)_{j\in\{\prev(\ell)+1,\ldots,\ell\}}\leq (v_j)_{j\in\{\prev(\ell)+1,\ldots,\ell\}}$. This holds
for all $\ell\in\pos\cap[k]$, so we obtain that 
$(z_j)_{j\in[\ellfin]}\leq (v_j)_{j\in[\ellfin]}$. Finally, note that 
$|A|=\sum_{r\in\POS':r\leq k}\bigl|A\cap\{i_r+1,\ldots,i_{\nxt(r)}\}\bigr|\leq\sum_{r\in\POS':r\leq k}\bigl(r-\prev(r)\bigr)=k$, 
so $z_j=0$ for all $j>k$. It follows that $z\leq v$.
\end{proof}

For sets $S,T\sse[n]$, let $\rewd_S(T):=\rewd(T\cup S)-\rewd(S)$ denote the incremental
reward that we obtain by adding the item-set $T$ to $S$.
Part (a) of Lemma~\ref{sizebnds} suggests that we build a solution $A$ satisfying the
bounds mentioned therein. In particular, since $\rewd$ is submodular, one can argue that
if we build a solution by considering the $\sbktopt_\ell$ size buckets (in any order),
and finding the $\ell-\prev(\ell)$ items from $\sbktopt_\ell$ to add to our current solution
that maximize the incremental reward obtained, then this yields good reward. Note that the  
subproblem we need to solve here is a cardinality-constrained submodular maximization
problem. The issue with implementing this plan is that 
{\em we do not know the $\sbktopt_\ell$ size buckets}. Part (b) of Lemma~\ref{sizebnds}
suggests a way out of this. We will instead build an index-set $I$ and a
solution $A$ simultaneously satisfying the requirements of part (b) of
Lemma~\ref{sizebnds}. This will still guarantee feasibility, and we will show that this 
can be done so as to obtain $\Omega(\OPT)$ reward.  

It will be convenient to first observe that there is a near-optimal solution satisfying the 
cardinality bounds in Lemma~\ref{sizebnds} (a). This is due to the following simple,
general result.

\begin{lemma} \label{submodmat}
Let $([n],\I)$ be a matroid with rank function $\rk:2^{[n]}\mapsto\Zp$, and
$\Pc=\bigl\{x\in\R_+^n: x(S)\leq\rk(S)\ \ \forall S\sse[n]\bigr\}$ be the corresponding
matroid polytope.
Let $A\sse[n]$ and $\rho\geq 1$ be such that $\chi^A/\rho\in\Pc$. Let
$\gensub:2^{[n]}\mapsto\Rp$ be a monotone, submodular function. Then, there exists a set
$B\sse A$ such that $B\in\I$ and $\gensub(B)\geq\gensub(A)/\rho$.
\end{lemma}

Lemma~\ref{submodmat} follows from the fact that the 
{\em multilinear extension of $\gensub$} drops by a factor of at most $\rho$ when we scale
$\chi^A$ by $\rho$, and one can round a point in a matroid polytope to an integer point
without incurring any rounding loss in the value of the multilinear extension.
To avoid detracting the reader, we defer the proof to the end of this section. 
We utilize the following corollary. 
%which follows from the fact that the cardinality bounds
%in Lemma~\ref{sizebnds} (b) can be encoded by a partition matroid. 
\begin{comment}
scaling the indicator vector of
$\optset-\{optind_1\}$ by $\al$ yields a fractional solution $x$ satisfying the cardinality
bounds, and one can argue that the {\em multilinear extension of $\rewd$} has at least the
stated value at $x$. 
The polytope encoding the cardinality bounds is a (partition) matroid
polytope, for which it is known that rounding $x$ to an integer solution does not incur
any loss in value (compared to the multilinear extension), which yields the lemma. 
\end{comment}

\begin{corollary} \label{noptsoln}
%Let $\al=\max_{\ell\in\POS':\ell\geq 1}\frac{\nxt(\ell)-\ell}{\ell-\prev(\ell)}$. (Note
%that $\al=2$ here.) 
There is a solution $\noptset\sse[n]$ such that
$\bigl|\noptset\cap\sbktopt_\ell\bigr|\leq\ell-\prev(\ell)$ for all $\ell\in\POS'$ and 
$\rewd(\noptset)\geq\rewd\bigl(\optset-\{\optind_1\}\bigr)/2$.
\end{corollary}

\begin{proof}
We apply Lemma~\ref{submodmat} to the partition matroid encoding the stated
cardinality bounds and the set $A=\optset-\{\optind_1\}$. We claim that we can take
$\rho=2$ in Lemma~\ref{submodmat}, which yields the claimed result.
This follows because $A\cap\sbktopt_0=\es$ and 
$\bigl|A\cap\sbktopt_\ell\bigr|\leq\nxt(\ell)-\ell\leq 2\bigl(\ell-\prev(\ell)\bigr)$ for
all $\ell\in\POS'-\{0\}$. 
\end{proof}

Let $\rmax:=\max\,\{\rewd(\{e\}): e\in[n],\ f(\sz_e)\leq B\}$. Since $\rewd$ is monotone,
submodular, we have $\OPT\leq n\cdot\rmax$.
We may assume that we know $\ellfin$, and (as is routine by now) that we
have an estimate $\optval$ such that $\optval\leq\OPT\leq(1+\ve)\optval$. Let
$\Dt=\frac{\ve\cdot\optval}{|\POS|}$. 
For $\ell\in\POS'$, define $\noptset_\ell:=\noptset\cap\sbktopt_\ell$, and
$\noptset_{>\ell}:=\bigcup_{r\in\POS':r>\ell}\noptset_r$. (Note that $\noptset_0=\es$.)
%$\sbktopt_{>\ell}:=\bigcup_{r\in\POS':r>\ell}\sbktopt_r$ and
%$\sbktopt_{\geq\ell}:=\bigcup_{r\in\POS':r\geq\ell}\sbktopt_r$
%let
%$\rnopt_\ell:=\rewd_{\noptset\cap\sbktopt_{>\ell}}(\noptset\cap\sbktopt_\ell)$,
%\rewd\bigl(\noptset\cap(\bigcup_{r\in\POS':r\geq\ell}\sbktopt_r)\bigr)
%-\rewd\bigl(\noptset\cap(\bigcup_{r\in\POS':r>\ell}\sbktopt_r)\bigr)$ 
%which is the incremental reward obtained from $\noptset_\ell$ over the
%reward from the $\noptset\cap\sbktopt_r$ sets for $r\in\POS'$, $r>\ell$.
%(Note that $\rnopt_0=0$.)
Since $|\POS|=O(\log n)$, similar to Section~\ref{knapsack},    
we may also assume that we can estimate the incremental rewards
$\rewd_{\noptset_{>\ell}}(\noptset_\ell)$ within an additive error of $\Dt$, for all
$\ell\in\POS'$: more precisely, we may assume that we know 
$\rnopt_\ell:=\Dt\cdot\floor{\frac{\rewd_{\noptset_{>\ell}}(\noptset_\ell)}{\Dt}}$  
for all $\ell\in\POS'$, since 
$\sum_{\ell\in\POS'}\floor{\frac{\rewd_{\noptset_{>\ell}}(\noptset_\ell)}{\Dt}}\leq\bigl(1+\frac{1}{\ve}\bigr)|\POS|$. 

Consider using these $\rnopt_\ell$-estimates %(up to an additive error of $\Dt$) 
%for all $\ell\in\POS'$
to build the index-set $I$ and solution $A$, roughly speaking, via the following
procedure. 
Let $\alg$ be a $\beta$-approximation algorithm for submodular maximization subject to a
cardinality constraint. 
Set $i_{\nxt(\ellfin)}:=n$, $A\assign\es$, and $I\assign\es$.
Starting at $\ell=\ellfin$, we find the largest index $i<i_{\nxt(\ell)}$ such that $\alg$
can find a subset $T\sse\{i+1,\ldots,i_{\nxt(\ell)}\}$ with $|T|\leq\ell-\prev(\ell)$
satisfying $\rewd_A(T)\geq\rnopt_\ell/\beta$.
The idea is that if $i_{\nxt(\ell)}\geq\optind_{\nxt(\ell)}$, then
$\noptset_\ell$ is a set that yields incremental reward $\rnopt_\ell$, and so since
$\alg$ is a $\beta$-approximation algorithm, we  
will be able to find such a set for some $i\geq\optind_\ell$. Given this, we can set
$i_\ell:=i$, and update
$A\assign A\cup T$, $I\assign I\cup\{i_\ell\}$, and $\ell\assign\prev(\ell)$, and
continue until we have $\ell=0$ at the end of the iteration. 

By construction, $I$ and $A$ would then satisfy Lemma~\ref{sizebnds} (b), which suggests
that $A$ is a feasible solution obtaining reward roughly $\rewd(\noptset)/\beta$.
However, {\em there is a significant piece of fallacious reasoning} here, which poses a
serious impediment. In arguing that $i_\ell\geq\optind_\ell$, we assumed above that 
%$\noptset_\ell$ yields incremental reward (at least) $\rnopt_\ell$ over the set $A$, i.e., 
$\rewd_A(\noptset_\ell)\geq\rnopt_\ell$, but {\em this need not hold}; 
we only have that $\rewd_{\noptset_{>\ell}}(\noptset_\ell)\geq\rnopt_\ell$. Consequently, we
cannot argue that $i_\ell\geq\optind_\ell$, and so $A$ need not be a feasible solution.
%so all bets are off regarding the feasibility of $A$.  
One could avoid this problem if we could estimate the sequence
$\bigl\{\rewd_{A_{\nxt(\ell)}}(\noptset_\ell)\bigr\}_{\ell\in\POS'}$ of incremental rewards,
where $A_{\nxt(\ell)}$ is the set of items included (from $\{i_{\nxt(\ell)}+1,n\}$) when
considering index $\ell$; we call $\bigl\{\rewd_{A_{\nxt(\ell)}}(\noptset_\ell)\bigr\}_{\ell\in\POS'}$
the incremental-reward sequence corresponding to the (nested) set sequence
$\{A_\ell\}_{\ell\in\POS'}$.
But this runs the risk of circular reasoning, since the $A_\ell$-sets are themselves
determined by the incremental-reward sequence used (in the above procedure)!

Essentially, we seek a reward sequence $\{\rest_\ell\}_{\ell\in\POS'}$ that corresponds to
the incremental-reward sequence of a set sequence $\{A_{\ell}\}_{\ell\in\POS'}$,
such that $\{A_\ell\}_{\ell\in\POS'}$ is precisely the set sequence that one would obtain
via the above procedure from the $\{\rest_\ell\}_{\ell\in\POS'}$ reward sequence! 
%(This has the flavor of a fixed point.) 
We argue that this is indeed possible, and that the solution $A=A_1$ constructed from such
a reward-sequence obtains large reward. Roughly speaking this holds because:
(1) one can define a suitable set sequence and target incremental-reward sequence, whose
entries are multiplies of $\Dt$,  
%where $\rest_\ell$
%depends only on $A_r$ for $r\in\POS'$, $r>\ell$, and $A_\ell$ depends on $\rest_\ell$ and
%$A_{\nxt(\ell)}$;
and can argue %after rounding incremental rewards to multiples of $\Dt$, we can 
%isolate 
that this target reward sequence can be isolated within a polynomial-size collection of
reward sequences  
%polynomial number of incremental-reward sequences 
(even though the number of set sequences can be quite large); and 
(2) given this target reward sequence, we can recover the set sequence used to define the
reward sequence. 

\vspace*{-1ex}
\paragraph{The algorithm.} 
Recall that we may assume that we know $\ellfin$ and an estimate $\optval$ such that
$\optval\leq\OPT\leq(1+\ve)\optval$, and we set $\Dt=\frac{\ve\cdot\optval}{|\POS|}$. 
(More precisely, we repeat the procedure below for all possible values of $\ellfin$ and
$\optval$, where $\optval$ is of the form $\rmax(1+\ve)^k$ and lies in $[\rmax, n\rmax]$.)
%\footnote{That is, we repeat the procedure decribed below for all possible values of
%$\ellfin$ and $\optval$ to obtain a collection of item sets, and return the best feasible
%solution in this collection.}
Define the following collection of reward sequences.
\[
\rseqset:=
\Bigl\{\rest\in\Rp^{\POS'-\{0\}}:\ \rest_\ell\text{ is a multiple of $\Dt$}\ \ \forall
\ell\in\POS'-\{0\},
\quad \sum_{\ell\in\POS'-\{0\}}\rest_\ell\leq\Bigl(1+\tfrac{1}{\ve}\Bigr)\Dt\cdot|\POS|\Bigr\}
\]
Note that $|\rseqset|\leq 2^{O(|\POS|/\ve)}=n^{O(1/\ve)}$.
Recall that $\alg$ is a $\beta$-approximation algorithm for cardinality-constrained
submodular maximization. We assume that $\alg$ is deterministic. 

For each $\rest\in\rseqset$, we do the following. 
Set $i_{\nxt(\ellfin)}:=n$. 
Initialize $A_{\nxt(\ellfin)}:=\es$, and $I:=\es$.
For each $\ell\in\POS'-\{0\}$ considered in {\em decreasing} order, we do the following.
%Starting at $\ell=\ellfin$, 
Let $i$ be the largest index smaller than $i_{\nxt(\ell)}$ such that 
for the cardinality-constrained submodular maximization problem with ground set
$\{i+1,\ldots,i_{\nxt(\ell)}\}$, submodular function $\rewd_{A_{\nxt(\ell)}}(\cdot)$, and
cardinality bound $\ell-\prev(\ell)$, 
$\alg$ returns a set $T$
%can find a subset $T\sse\{i+1,\ldots,i_{\nxt(\ell)}\}$ with $|T|\leq\ell-\prev(\ell)$
satisfying $\rewd_{A_{\nxt(\ell)}}(T)\geq\rest_\ell$.
If no such index $i$ exists, then we declare failure; this indicates that $\rest$ is not a
suitable target reward sequence.
Set $A_\ell\assign T\cup A_{\nxt(\ell)}$, $i_\ell:=i$, and $I\assign I\cup\{i_\ell\}$.
We add $A_1$ to our collection of solutions.

Finally, from the collection of solutions computed, we return the feasible solution that
attains maximum reward, or the element achieving reward $\rmax$, whichever yields higher
reward. 

\paragraph{Analysis.}
We prove the following performance guarantee.

\begin{theorem} \label{subnbknapthm}
The above algorithm returns a feasible solution obtaining reward at least 
$\bigl(\frac{1}{2\beta+3}-\ve\bigr)\cdot\OPT$.
\end{theorem}

We assume that we have the the correct $\ellfin$, $\optval$ values,
%It suffices to 
and show that there is some reward sequence in $\rseqset$ for which the algorithm computes
a feasible solution, and the better of $\rmax$ and the solution computed by the algorithm yields 
the desired reward. 
%and a feasible solution with the desired reward. 
To this end, define the following sequence of rewards and nested sets.
Recall that $\noptset_\ell=\noptset\cap\sbktopt_\ell$ for $\ell\in\POS'$.

Set $\bi_{\nxt(\ellfin)}:=n$. 
Initialize $B_{\nxt(\ellfin)}:=\es$, $\bI:=\es$.
For each $\ell\in\POS'-\{0\}$ considered in decreasing order, we do the following.
%Starting at $\ell=\ellfin$, 
Set $\brewd_{\ell}:=\Dt\cdot\floor{\frac{\rewd_{B_{\nxt(\ell)}}(\noptset_{\ell})}{\beta\Dt}}$.
Let $i$ be the largest index smaller than $\bi_{\nxt(\ell)}$ such that 
for the cardinality-constrained submodular maximization problem with ground set
$\{i+1,\ldots,\bi_{\nxt(\ell)}\}$, submodular function $\rewd_{B_{\nxt(\ell)}}(\cdot)$, and
cardinality bound $\ell-\prev(\ell)$, 
$\alg$ returns a set $T$ such that $\rewd_{B_{\nxt(\ell)}}(T)\geq\brewd_\ell$.
Update $B_\ell\assign T\cup B_{\nxt(\ell)}$, $\bi_\ell:=i$, $\bI\assign\bI\cup\{\bi_\ell\}$.

%The performance guarantee of the algorithm then follows immediately from
%Lemmas~\ref{feasset}, \ref{fixedpt}, and~\ref{fixedpt}.
Lemma~\ref{feasset} shows that $B:=B_1$ is a feasible solution, and Lemma~\ref{goodrewd}
shows that $B$ attains good reward.
%$B:=B_1$ and the index set $\bI$ satisfy the bounds in
%Lemma~\ref{sizebnds} (b), and hence that $B$ is a feasible solution. 
Lemma~\ref{goodseq} argues that the $(\brewd_\ell)_{\ell\in\POS'-\{0\}}$ reward sequence lies in
$\rseqset$. Lemma~\ref{fixedpt} shows the key result that, given this reward sequence, our
algorithm produces precisely the sequence of sets $\{B_\ell\}_{\ell\in\POS'-\{0\}}$. 
%Given this, we finally argue that $B$ achieves the desired near-optimal reward.
%Lemma~\ref{feasset} and Lemma~\ref{fixedpt}. 
Theorem~\ref{subnbknapthm} then follows by combining these results.

\begin{lemma} \label{feasset}
The item-set $B$ and index-set $\bI$ satisfy the bounds in Lemma~\ref{sizebnds} (b). Hence
$B$ is a feasible solution.
\end{lemma}

\begin{proof}
By construction, we have
$\bigl|B\cap\{\bi_\ell+1,\ldots,\bi_{\nxt(\ell)}\}\bigr|\leq\ell-\prev(\ell)$ for all
$\ell\in\POS'-\{0\}$. 
We argue that $\bi_{\ell}\geq\optind_\ell$ for all $\ell\in\POS'-\{0\}$. 
This will show that $B$ and $\bI$ satisfy the requirements of Lemma~\ref{sizebnds} (b),
which implies that $B$ is a feasible solution.

We proceed by induction. For any $\ell\in\POS'-\{0\}$, we have that
$\rewd_{B_{\nxt(\ell)}}(\noptset_\ell)\geq\beta\cdot\brewd_\ell$ by design.
So assuming inductively that $\bi_{\nxt(\ell)}\geq\optind_{\nxt(\ell)}$, which holds also
for the base case $\ell=\ellfin$ by definition, we have that for $i=\optind_\ell$,
$\noptset_\ell$ is a feasible solution to the cardinality-constrained
submodular-maximization problem considered in iteration $\ell$ with value at least
$\beta\cdot\brewd_\ell$. So since $\alg$ is a $\beta$-approximation algorithm, certainly
for index $\optind_\ell$, $\alg$ would return a desired set. It follows that
$\bi_\ell\geq\optind_\ell$. 
\end{proof}

\begin{lemma} \label{goodrewd}
We have $\rewd(B)\geq\frac{\rewd(\noptset)}{\beta+1}-\ve\cdot\OPT$.
\end{lemma}

\begin{proof}
We have $\rewd_{B_{\nxt(\ell)}}(B_\ell)\geq\brewd_\ell$ for all $\ell\in\POS'-\{0\}$ by
construction. So
\begin{alignat*}{1}
\rewd(B) & =\sum_{\ell\in\POS'-\{0\}}\rewd_{B_{\nxt(\ell)}}(B_\ell)
\geq\sum_{\ell\in\POS'-\{0\}}\brewd_\ell
\geq\sum_{\ell\in\POS'-\{0\}}\biggl(\frac{\rewd_{B_{\nxt(\ell)}}(\noptset_\ell)}{\beta}-\Dt\biggr)
\\
& \geq \frac{1}{\beta}\cdot\sum_{\ell\in\POS'-\{0\}}\rewd_B(\noptset_\ell)-|\POS'|\cdot\Dt
\tag{\small since $\rewd$ is submodular and $B\supseteq B_{\nxt(\ell)}$} \\
& \geq \frac{1}{\beta}\cdot\rewd_B\Bigl(\bigcup_{\ell\in\POS'-\{0\}}\noptset_\ell\Bigr)
-\ve\cdot\OPT
\tag{\small since $\rewd$ is submodular} \\
& \geq \frac{\rewd(\noptset)-\rewd(B)}{\beta}-\ve\cdot\OPT.
\tag{\small since $\rewd$ is monotone}
\end{alignat*}
It follows that
$\rewd(B)\cdot\frac{\beta+1}{\beta}\geq\frac{\rewd(\noptset)}{\beta}-\ve\cdot\OPT$, which
implies the stated bound.
\end{proof}

\begin{lemma} \label{goodseq}
We have $(\brewd_\ell)_{\ell\in\POS'-\{0\}}\in\rseqset$.
\end{lemma}

\begin{proof}
We only need to show that 
$\sum_{\ell\in\POS'-\{0\}}\brewd_\ell\leq\bigl(1+\frac{1}{\ve}\bigr)|\POS|$ since 
all $\brewd_\ell$ entries are multiples of $\Dt$.
We have $\rewd_{B_{\nxt(\ell)}}(B_\ell)\geq\brewd_\ell$ for all $\ell\in\POS'-\{0\}$. 
%by construction. 
So $\sum_{\ell\in\POS'-\{0\}}\brewd_\ell\leq\rewd(B_1)=\rewd(B)$. Since $B$ is feasible
and $\optval$ is a correct estimate of $\OPT$,
we have $\rewd(B)\leq\OPT\leq(1+\ve)\optval$.
So $\sum_{\ell\in\POS'-\{0\}}\brewd_\ell\leq(1+\ve)\optval$. We also know that the sum on
the left is a multiple of $\Dt$. So we have that 
$\sum_{\ell\in\POS'-\{0\}}\brewd_\ell\leq\Dt\cdot\floor{\frac{(1+\ve)\optval}{\Dt}}
\leq\Dt\cdot\bigl(1+\frac{1}{\ve}\bigr)|\POS|$.
\end{proof}

\begin{lemma} \label{fixedpt} 
On the input sequence $(\brewd_\ell)_{\ell\in\POS'-\{0\}}$, the algorithm returns the set
sequence $\{B_\ell\}_{\ell\in\POS'-\{0\}}$.
\end{lemma} 

\begin{proof} 
This follows simply because inductively, one can argue that each iteration of the
algorithm unfolds in exactly the same way as the corresponding iteration of the procedure
used to construct the $(\brewd_\ell)$ reward sequence. To avoid cumbersome terminology, we use
$\brewd$-procedure to denote the latter procedure.

We show the following. For every $\ell\in\POS'-\{0\}$, assuming that the index-set $I$ and
the sets $\{A_r\}_{r\in\POS':r>\ell}$ computed by the algorithm by the start of iteration
$\ell$, are the same as the index set $\bI$ and the sets $\{B_r\}_{r\in\POS':r>\ell}$
computed by the $\brewd$-procedure by the start of iteration $\ell$, the same holds at the
end of the iteration, i.e., the start of the next iteration.
The stated assumption clearly holds when $\ell=\ellfin$, so if we prove the above
claim, then this implies that $A_\ell=B_\ell$ for all $\ell\in\POS'-\{0\}$.

To see why the claim holds, observe that in iteration $\ell$, both the algorithm and the
$\brewd$-procedure are solving exactly the same cardinality-constrained submodular
maximization problem when considering any index $i$. This is because
$i_{\nxt(\ell)}=\bi_{\nxt(\ell)}$ and $A_{\nxt(\ell)}=B_{\nxt(\ell)}$, since $I=\bI$ at
the start of the iteration and $A_r=B_r$ for all $r\in\POS'$, $r>\ell$, and we are
considering the same target reward $\brewd_\ell$ by design. Since $\alg$ is deterministic,
it follows that both the
algorithm and the $\brewd$-procedure compute the same index $i$ and the same set $T$ to
add to their respective current solutions in this iteration. Hence, the invariant is
maintained. 
\end{proof}

\begin{proofof}{Theorem~\ref{subnbknapthm}}
The reward obtained is at least $\max\bigl\{\rewd(B),\rmax\bigr\}$. 
By Corollary~\ref{noptsoln} and
since $\rewd(\{\optind_1\})\leq\rmax$, we have that
$\rewd(\noptset)\geq\frac{\OPT-\rmax}{2}$. So using Lemma~\ref{goodrewd}, we obtain that 
\begin{equation*}
\begin{split}
\max\bigl\{\rewd(B),\rmax\bigr\} &
\geq \frac{2(\beta+1)}{2\beta+3}\cdot\rewd(B)+\frac{1}{2\beta+3}\cdot\rmax \\
& \geq \frac{2(\beta+1)}{2\beta+3}\cdot\biggl(\frac{\OPT-\rmax}{2(\beta+1)}-\ve\OPT\biggr)
+\frac{1}{2\beta+3}\cdot\rmax 
\geq \biggl(\frac{1}{2\beta+3}-\ve\biggr)\cdot\OPT. \quad \qedhere
\end{split} 
\end{equation*}
\end{proofof}

\subsubsection{Improvement to \boldmath $\bigl(\frac{1}{\beta+1}-O(\ve)\bigr)$-approximation}
\label{submodknap-improv}
To obtain the improved approximation factor, we make a few changes to the earlier
algorithm. First, letting $\dt=\min\{\ve,1\}$, we now take $\POS=\POS_{n,\dt}$, 
so as to avoid the factor-$2$ loss in Corollary~\ref{noptsoln}, since for large
enough $\ell$, we now have that
$\nxt(\ell)-\ell\leq(1+2\ve)\bigl(\ell-\prev(\ell)\bigr)$. However, this savings kicks in
only for sufficiently large $\ell$. We also want to reduce the $\rmax$ loss incurred
earlier because of excluding $\optind_1$ when we move to $\noptset$.  
We handle these considerations as follows. As before, let $\ellfin$ be the largest index
in $\POS\cap[\nopt]$, and let $\POS':=\{0\}\cup(\POS\cap[\ellfin])$.
We guess
$\optind_1,\ldots,\optind_{\ell_1}$, where $\ell_1$ is a sufficiently large constant 
depending on $\ve$, and prove a generalization of Lemma~\ref{sizebnds} 
(Lemma~\ref{nsizebnds}) showing that if
$A$ is such that $A\cap[\optind_{\ell_1}]=\{\optind_1,\ldots,\optind_{\ell_1}\}$, 
%for some $\ell_1\in\POS$, 
$A\cap\sbktopt_{\ell_1}=\es$, and
$\bigl|A\cap\sbktopt_\ell\bigr|\leq\ell-\prev(\ell)$ for all $\ell\in\POS'$,
$\ell>\ell_1$, then $A$ is a feasible solution. 
We can take $\ell_1$ sufficiently large such that
$\optset\cap\sbktopt_{\ell_1}$ contributes (incremental) reward at most $\ve\OPT$ (see
Claim~\ref{goodind}). 
This enables us to argue that there is a solution $\noptset$ satisfying the above
bounds such that $\rewd(\noptset)\geq(1-\ve)\OPT$ (Lemma~\ref{nearopt}). Finally, since we start
by including items $\{\optind_1,\ldots,\optind_{\ell_1}\}$, we modify the algorithm and
the $\brewd$-procedure accordingly to stop when we reach index $\ell_1$.
%instead of moving to a solution $\noptset$ where we exclude $\optind_1$ 
%
Recall that $\sbktopt_\ell:=\{\optind_\ell+1,\ldots,\optind_{\nxt(\ell)}\}$ for all
$\ell=0,1,\ldots,\ellfin$, and these size buckets partition $[n]$.

\begin{lemma} \label{nsizebnds}
Let $A\sse[n]$ be such that 
$A\cap[\optind_{\ell_1}]=\{\optind_1,\optind_2,\ldots,\optind_{\ell_1}\}$
for some $\ell_1\in\POS'$.
\begin{enumerate}[label=(\alph*), topsep=0.1ex, itemsep=0ex, leftmargin=*]
\item If $A\cap\sbktopt_{\ell_1}=\es$ and
$\bigl|A\cap\sbktopt_\ell\bigr|\leq\ell-\prev(\ell)$ for all
$\ell\in\POS'$, $\ell>\ell_1$, then $f\bigl(\svec{A}\bigr)\leq B$.

\item Let $I=\{i_\ell: \ell\in\pos,\,\ell_1<\ell\leq k\}$ be an index-set, 
where $k\leq\ellfin$ and
$i_\ell\geq\optind_\ell$ for all $\ell\in\pos$, $\ell_1<\ell\leq k$. 
Define $i_{\ell_1}:=\optind_{\ell_1}$ and $i_{\nxt(k)}:=n$.
Suppose that $A\cap\{i_{\ell_1}+1,\ldots,i_{\nxt(\ell_1)}\}=\es$ and
$\bigl|A\cap\{i_\ell+1,\ldots,i_{\nxt(\ell)}\}\bigr|\leq\ell-\prev(\ell)$ for
all $\ell\in\POS'$, $\ell_1<\ell\leq k$. Then, $f\bigl(\svec{A}\bigr)\leq B$.
\end{enumerate}
\end{lemma}

\begin{proof}
We mimic the proof of Lemma~\ref{sizebnds}.
%We will define a vector $v\in\R^n$, and argue that $f(v)\leq B$ and 
%$\svec{A}^{\down}\leq v$. 
Let $v\in\R^n$ be the following vector: 
\[
v_r=\sz_{\optind_r}\ \ \forall r\in[\ell_1]; \quad\ \ 
\forall \ell\in\pos',\,\ell>\ell_1,\,\forall j\in\{\prev(\ell)+1,\ldots,\ell\},
\ \ \text{set }v_j=\sz_{\optind_\ell}; \quad \text{$v_j=0$ for all other $j$}.
\]
%We argue that $v\leq\svec{\optset}^{\down}$, which implies that 
%$f(v)\leq f\bigl(\svec{\optset}\bigr)\leq B$. 
Let $u=\svec{\optset}^\down$.
Clearly, $v_j=u_j$ for all $j\in[\ell_1]$.
Consider any $\ell\in\pos'$, $\ell>\ell_1$. %with $\ell\geq 2$. 
We have $v_j=\sz_{\optind_\ell}$ for all $j\in\{\prev(\ell)+1,\ldots,\ell\}$. 
%Now consider the vector $u$. 
The number of items in $\optset$ up to and including $\optind_\ell$ is
precisely $\ell$,
and so $u_j\geq\sz_{\optind_\ell}$ for all $j\leq\ell$. 
Therefore $(v_j)_{j\in\{\prev(\ell)+1,\ldots,\ell\}}\leq (u_j)_{j\in\{\prev(\ell)+1,\ldots,\ell\}}$. This
holds for all $\ell\in\pos'$, $\ell>\ell_1$, so together this covers all non-zero
coordinates of $v$, and we obtain that $v\leq u$.

%Now consider part (b). 
Next, we argue that $\svec{A}^{\down}\leq v$.
Let $z=\svec{A}^\down$. We have $z_j=v_j$ for all $j\in[\ell_1]$. 
\begin{comment}
Recall that 
$\optset\cap\sbktopt_\ell=\{\optind_{\ell+1},\optind_{\ell+2},\ldots,\optind_{\nxt(\ell)}\}$
for all $\ell=0,\ldots,\ellfin-1$, and 
$\optset\cap\sbktopt_{\ellfin}=\{\optind_{\ellfin+1},\optind_{\ellfin+2},\ldots,\optind_{\nopt}\}$
(which could be the empty set).
\end{comment}
%
For part (a), consider any $\ell\in\pos'$, $\ell>\ell_1$,
and consider the coordinates $j\in\{\prev(\ell)+1,\ldots,\ell\}$. 
We have 
\[
\bigl|A\cap[\optind_\ell]\bigr|=
\ell_1+\sum_{r\in\POS':\ell_1<r\leq\prev(\ell)}\bigl|A\cap\sbktopt_r\bigr|
\leq \ell_1+\sum_{r\in\POS':\ell_1<r\leq\prev(\ell)}\Bigl(r-\prev(r)\Bigr)=\prev(\ell).
\]
So for all $j>\prev(\ell)$, we have $z_j\leq\sz_{\optind_\ell}$. On the other hand, as
argued above we have $v_j=\sz_{\optind_\ell}$ for all $j\in\{\prev(\ell)+1,\ldots,\ell\}$.  
So we obtain that 
$(z_j)_{j\in\{\prev(\ell>)+1,\ldots,\ell\}}\leq (v_j)_{j\in\{\prev(\ell)+1,\ldots,\ell\}}$. This holds
for all $\ell\in\pos'$, $\ell>\ell_1$, so 
$(z_j)_{j\in[\ellfin]}\leq (v_j)_{j\in[\ellfin]}$. 
Also, $|A|\leq\ell_1+\sum_{r\in\POS':\ell_1<r\leq k}\bigl(r-\prev(r)\bigr)=k$, 
so $z_j=0$ for all $j>k$. It follows that $z\leq v$.

For part (b), consider any $\ell\in\pos'$, $\ell>\ell_1$.
Similar to above, we have 
\[
\bigl|A\cap[i_\ell]\bigr|=
\ell_1+\sum_{r\in\POS':\ell_1<r\leq\prev(\ell)}\Bigl|A\cap\{i_r+1,\ldots,i_{\nxt(r)}\}\Bigr|
\leq \ell_1+\sum_{r\in\POS':\ell_1<r\leq\prev(\ell)}\Bigl(r-\prev(r)\Bigr)=\prev(\ell).
\]
So $z_j\leq\sz_{i_\ell}\leq\sz_{\optind_\ell}$ for all $j>\prev(\ell)$, 
and $v_j=\sz_{\optind_\ell}$ for all $j\in\{\prev(\ell)+1,\ldots,\ell\}$. 
%and $\sz_{\optind_\ell}\geq\sz_{i_\ell}$, since $i_\ell\geq\optind_\ell$.  
So we have that $(z_j)_{j\in\{\prev(\ell)+1,\ldots,\ell\}}\leq
(v_j)_{j\in\{\prev(\ell)+1,\ldots,\ell\}}$. This holds 
for all $\ell\in\pos'$, $\ell>\ell_1$, so 
$(z_j)_{j\in[\ellfin]}\leq (v_j)_{j\in[\ellfin]}$. Finally, 
$|A|\leq\ell_1+\sum_{r\in\POS':\ell_1<r\leq k}\bigl|A\cap\{i_r+1,\ldots,i_{\nxt(r)}\}\bigr|
\leq\ell_1+\sum_{r\in\POS':\ell_1<r\leq k}\bigl(r-\prev(r)\bigr)=k$, 
so $z_j=0$ for all $j>k$. It follows that $z\leq v$.
\end{proof}

Let $\ell'$ be the smallest index in $\POS$ that is at least $\frac{2}{\dt^2}$, and let
$\ell_0>\ell'$ be such that 
$\bigl|\POS\cap[\ell',\ell_0]\bigr|=\ceil{\frac{1}{\ve}}$. 
%indices $\ell\in\POS$ with $\ell'<\ell\leq\ell_0$. 
Let $N=\ceil{\frac{1}{\ve}}$. 
Since $\nxt(\ell)\leq(1+\dt)\ell+1$ for all $\ell\in\POS$, we have 
\[
\ell_0\leq(1+\dt)^{N-1}\ell'+\sum_{r=1}^{N-1}(1+\dt)^r
\leq(1+\dt)^{N-1}\Bigl(\ell'+\tfrac{(1+\dt)^2}{\dt}\Bigr)
\leq e^{\dt(N-1)}\Bigl(\ell'+2+\dt+\tfrac{1}{\dt}\Bigr)
\leq e(1+\dt)\ell'+e(2+\dt).
\]
%\leq(1+\dt)^{N+1}\ell'\leq e^{2+\dt}\cdot\ell'$.
We may assume that $\nopt=|\optset|\geq\ell_0$ as otherwise we can use brute force
enumeration to find an optimal solution. Then, $\ell',\ell_0\in\POS'$.

\begin{claim} \label{goodind}
There is some index $\ell_1\in\POS\cap[\ell',\ell_0]$ %with $\ell'\leq\ell_1\leq\ell_0$ 
such that $\rewd_{\optset-\sbktopt_{\ell_1}}(\optset\cap\sbktopt_{\ell_1})\leq\ve\OPT$.
\end{claim}

\begin{proof}
For an index-set $I\sse[n]$, let $\sbktopt_I:=\bigcup_{r\in I\cap\POS'}\sbktopt_r$.
There are at least $\frac{1}{\ve}$ indices in $\POS\cap[\ell',\ell_0]$ and 
\begin{equation*}
\begin{split}
\sum_{\ell\in\POS\cap[\ell'.\ell_0]}\rewd_{\optset-\sbktopt_\ell}(\optset\cap\sbktopt_\ell)
& \leq\sum_{\ell\in\POS\cap[\ell',\ell_0]}\rewd_{\optset-\sbktopt_{[\ell',\ell]}}(\optset\cap\sbktopt_\ell) \\
& =\rewd(\optset)-\rewd(\optset-\sbktopt_{[\ell',\ell_0]})\leq\OPT. 
\qedhere
\end{split}
\end{equation*}
\end{proof}

In the sequel, $\ell_1$ is fixed to be the index given by Claim~\ref{goodind}.

\begin{lemma} \label{nearopt}
There is a solution $\noptset\sse[n]$ satisfying the requirements of Lemma~\ref{nsizebnds}
(a) such that $\rewd(\noptset)\geq(1-2\dt)(1-\ve)\OPT\geq(1-3\ve)\OPT$.
\end{lemma}

\begin{proof}
Let $S=\{\optind_1,\ldots,\optind_{\ell_1}\}$, and $T=\optset-S-\sbktopt_{\ell_1}$.
For any $\ell\in\POS$, we have $\nxt(\ell)-\ell\leq\dt\ell+1$, and
$\ell-\prev(\ell)\geq\ell-\frac{\ell}{1+\dt}=\frac{\dt\ell}{1+\dt}$. 
Since $\ell_1\geq\ell'$, for all $\ell\geq\ell_1$, we therefore have 
\[
\nxt(\ell)-\ell\leq(1+\dt)\Bigl(\ell-\prev(\ell)\Bigr)+1\leq(1+2\dt)\Bigl(\ell-\prev(\ell)\Bigr)
\]
where the final inequality follows because 
$\dt\bigl(\ell-\prev(\ell)\bigr)\geq\frac{\dt^2\ell}{1+\dt}\geq\frac{\dt^2\ell'}{2}\geq 1$.
Since $|T\cap\sbktopt_\ell|\leq\nxt(\ell)-\ell$ for all $\ell\in\POS'$, $\ell>\ell_1$,
by Lemma~\ref{submodmat}, we can find $B\sse T$ such that
$|B\cap\sbktopt_\ell|\leq\ell-\prev(\ell)$ for all $\ell\in\POS'$, $\ell>\ell'$ and 
$\rewd_S(B)\geq(1-2\dt)\rewd_S(T)$.
Consider $\noptset=S\cup B$. By construction, $\noptset$ satisfies the requirements of
Lemma~\ref{nsizebnds} (a). Also, we have
\begin{equation*}
\begin{split}
\rewd(\noptset) & =\rewd(S)+\rewd_S(B)\geq\rewd(S)+(1-2\dt)\rewd_S(T)
\geq(1-2\dt)\rewd(S\cup T)\\
&=(1-2\dt)\rewd\bigl(\optset-\sbktopt_{\ell_1}\bigr)
\geq(1-2\dt)\Bigl(\OPT-\rewd_{\optset-\sbktopt_{\ell_1}}\bigl(\optset\cap\sbktopt_{\ell_1}\bigr)\Bigr) \\
& \geq(1-2\dt)(1-\ve)\OPT. \qedhere
\end{split} 
\end{equation*} 
\end{proof}

\vspace*{-1ex}
\paragraph{The algorithm.}
As before, we assume that we know $\ellfin$, and $\optval$ satisfying
$\optval\leq\OPT\leq(1+\ve)\optval$. 
We also assume that we know $\ell_1$ and the elements
$\optind_1,\optind_2,\ldots,\optind_{\ell_1}$. 
Recall that $\Dt=\frac{\ve\cdot\optval}{|\POS|}$ and
$\alg$ is a deterministic $\beta$-approximation algorithm for cardinality-constrained
submodular maximization.  
We now define the following collection of reward sequences.
\[
\rseqset:=
\Bigl\{\rest\in\Rp^{\POS'-\dbrack{\ell_1}}:\ \rest_\ell\text{ is a multiple of $\Dt$}\ \ \forall
\ell\in\POS'-\dbrack{\ell_1},
\quad \sum_{\ell\in\POS'-\dbrack{\ell_1}}\rest_\ell\leq\Bigl(1+\tfrac{1}{\ve}\Bigr)\Dt\cdot|\POS|\Bigr\}
\]
We have $|\rseqset|\leq 2^{O(|\POS|/\ve)}=n^{O(1/\dt^2)}$.
%We assume that $\alg$ is deterministic. 
As mentioned earlier, for each reward-sequence, we run the earlier iterative procedure but
stop when we reach index $\ell_1$. 
For each $\rest\in\rseqset$, we do the following. 
Set $i_{\nxt(\ellfin)}:=n$. 
Initialize $A_{\nxt(\ellfin)}:=\{\optind_1,\ldots,\optind_{\ell_1}\}$, and $I:=\es$.
We repeat the following steps for each $\ell\in\POS'-\dbrack{\ell_1}$ considered in
decreasing order. %we do the following. 
%Starting at $\ell=\ellfin$, 
Let $i$ be the largest index smaller than $i_{\nxt(\ell)}$ such that 
for the cardinality-constrained submodular maximization problem with ground set
$\{i+1,\ldots,i_{\nxt(\ell)}\}$, submodular function $\rewd_{A_{\nxt(\ell)}}(\cdot)$, and
cardinality bound $\ell-\prev(\ell)$, 
$\alg$ returns a set $T$
%can find a subset $T\sse\{i+1,\ldots,i_{\nxt(\ell)}\}$ with $|T|\leq\ell-\prev(\ell)$
satisfying $\rewd_{A_{\nxt(\ell)}}(T)\geq\rest_\ell$.
Set $A_\ell\assign T\cup A_{\nxt(\ell)}$, $i_\ell:=i$, and $I\assign I\cup\{i_\ell\}$.
We add $A_1$ to our collection of solutions.

From the collection of solutions computed, we return the feasible solution that
attains maximum reward.

\paragraph{Analysis.}
The analysis mimics the earlier analysis. 
Assume we have the correct $\ellfin$, $\optval$, $\ell_1$ values, and the set
$\{\optind_1,\ldots,\optind_{\ell_1}\}$. 
The target reward sequence and nested-set
sequence are now defined as follows, mimicking the modified algorithm.
Recall that $\noptset_\ell=\noptset\cap\sbktopt_\ell$ for $\ell\in\POS'$.

Set $\bi_{\nxt(\ellfin)}:=n$. 
Initialize $B_{\nxt(\ellfin)}:=\{\optind_1,\ldots,\optind_{\ell_1}\}$, $\bI:=\es$.
For each $\ell\in\POS'-\dbrack{\ell_1}$ considered in decreasing order, we do the following.
Set $\brewd_{\ell}:=\Dt\cdot\floor{\frac{\rewd_{B_{\nxt(\ell)}}(\noptset_{\ell})}{\beta\Dt}}$.
Let $i$ be the largest index smaller than $\bi_{\nxt(\ell)}$ such that 
for the cardinality-constrained submodular maximization problem with ground set
$\{i+1,\ldots,\bi_{\nxt(\ell)}\}$, submodular function $\rewd_{B_{\nxt(\ell)}}(\cdot)$, and
cardinality bound $\ell-\prev(\ell)$, 
$\alg$ returns a set $T$ such that $\rewd_{B_{\nxt(\ell)}}(T)\geq\brewd_\ell$.
Update $B_\ell\assign T\cup B_{\nxt(\ell)}$, $\bi_\ell:=i$, $\bI\assign\bI\cup\{\bi_\ell\}$.

Let $B:=B_{\nxt(\ell_1)}$. Essentially, Lemmas~\ref{feasset}--\ref{fixedpt} continue to
hold, with the only cosmetic changes in some of the statements and proofs is that we
replace $\POS'-\{0\}$ by $\POS'-\dbrack{\ell_1}$, and we invoke Lemma~\ref{nsizebnds} (b)
in place of Lemma~\ref{sizebnds} (b). We point out these cosmetic changes, but omit
re-proving the statements.

Analogous to Lemma~\ref{feasset}, we now have that $B$ and
$\bI$ satisfy the requirements of Lemma~\ref{nsizebnds} (b). 
%The proof is essentially the same with the only cosmetic change being that we iterate over
%all $\ell\in\POS'-\dbrack{\ell_1}$ as in the above procedure. 
Lemma~\ref{goodrewd} holds as is. Similar to Lemma~\ref{goodseq}, we now have that
$(\brewd_\ell)_{\ell\in\POS'-\dbrack{\ell_1}}\in\rseqset$, and the analogue of
Lemma~\ref{fixedpt} is that the (modified) algorithm returns
$\{B_\ell\}_{\ell\in\POS'-\dbrack{\ell}}$ on the input reward sequence
$\{\brewd_\ell\}_{\ell\in\POS'-\dbrack{\ell_1}}$. Due to the improved bound on
$\rewd(\noptset)$, we therefore obtain the following performance guarantee.

\begin{theorem} \label{improvsubnbknapthm}
The above algorithm returns a feasible solution obtaining reward at least 
$\bigl(\frac{1-3\ve}{\beta+1}-\ve\bigr)\cdot\OPT$
\end{theorem}

\begin{proof}
By Lemma~\ref{goodrewd} and Lemma~\ref{nearopt}, we obtain reward at least 
$\rewd(B)\geq\frac{\rewd(\noptset)}{\beta+1}-\ve\cdot\OPT\geq\frac{(1-3\ve)\OPT}{\beta+1}-\ve\cdot\OPT$. 
\end{proof}

\begin{proofof}{Lemma~\ref{submodmat}}
The multilinear extension of $\gensub$ is defined as the function
$\genmulti:[0,1]^n\mapsto\Rp$ given by 
$\genmulti(x):=\sum_{S\sse[n]}\gensub(S)\prod_{e\in S}x_e\prod_{e\notin S}(1-x_e)$. It is
well-known that any $x\in\Pc$ can be rounded to an integer point $\tx\in\Pc$ 
whose support is a subset of the support of $x$, 
%$\supp(\tx)\sse\supp(x)$ 
such that $\genmulti(\tx)\geq\genmulti(x)$~\cite{CalinescuCPV11}.
Let $\tx=\chi^B$ be the rounded point corresponding to $\chi^A/\rho$, so $B\sse A$.
Note that $\genmulti(\chi^A/\rho)$ corresponds to the expected value of a random subset of
$A$ where we include each element of $A$ independently with probability
$\frac{1}{\rho}$. It is known that this expected value is at least
$\gensub(A)/\rho$; this holds even for subadditive $\gensub$ (see~\cite{Feige09},
Propositions 2.2 and 2.3). Putting these facts together gives
$\gensub(B)=\genmulti(\tx)\geq\genmulti(\chi^A/\rho)\geq\gensub(A)/\rho$.
\end{proofof}

\subsection{Submodular norm-budgeted \maxgap on related machines} \label{submodlb}
Recall that submodular norm-budgeted \maxgap (\subnbsched) on related machines is the
generalization of \normsched on related machines, where the reward $\rewd(T)$ from a set
$T\sse J$ of jobs is given by a monotone, submodular function $\rewd:2^J\mapsto\Rp$ 
%(specified via a value oracle). 
We devise a $12.91$-approximation algorithm for \subnbsched. %(Theorem~\ref{subnblbthm}).
%The approximation factor we obtain is . 
We emphasize that we have not attempted to optimize this factor, preferring simplicity of
exposition instead. 

\begin{theorem} \label{subnblbthm}
There is a
$\bigl(5\cdot\frac{2e-1}{e-1}+\ve\bigr)\approx(12.91+\ve)$-approximation
algorithm for \subnbsched on related machines.
\end{theorem}

%Our algorithm has two components. 
This result is built from two components.
We first observe that the reduction described in
Section~\ref{reduction} also applies with a submodular reward function. 
Specifically Theorem~\ref{lbredn} generalizes to this setting, where the (submodular)
norm-budgeted matching problem that we need to now solve is a constrained version of
\subnbsched where we can assign at most one job per machine. 
%constraint that one needs to solve
Instead of %restating and 
re-proving the version of Theorem~\ref{lbredn} for submodular
rewards, we point out the key place in the proof where the arguments still go through with 
a submodular reward function. 
We assume we have a bicriteria $(\rho,\gm)$-approximation for \subnbsched on related 
machines. For simplicity, we assume here that $\gm$ is an integer, but 
%in Appendix~\ref{append-subnb}, we show 
we remark that with more care, %one can also handle the case where $\gm$ is arbitrary.
the entirety of Theorem~\ref{lbredn} extends to the submodular-rewards setting.
Recall that in the proof of Theorem~\ref{lbredn}, we create another instance $\I'$ of the
same problem (i.e., \subnbsched on related machines), but with a reduced budget $B/\gm$,
and an instance $\I''$ of \subnbsched on related machines, where we have the additional
constraint that at most one job can be assigned to any machine. 
We upper bound the optimal value $\OPT$ of
the original instance by $\gm\OPT_{\I'}+(\gm-1)\OPT_{\I''}$ by partitioning the job-set
$\optset$ corresponding to an optimal solution to $\I$ to create $\gm$ feasible solutions
to $\I'$ and $(\gm-1)$ feasible solutions to $\I''$. 
Observe that {\em if the reward-function $\rewd(.)$ is subadditive}, i.e., 
$\rewd(S\cup T)\leq\rewd(S)+\rewd(T)$, as is the case with a submodular function, we still
obtain the same upper bound $\OPT\leq\gm\OPT_{\I'}+(\gm-1)\OPT_{\I''}$. Given this, the
rest of the proof goes through as is, and we thus obtain a
$\bigl(\gm\rho+(\gm-1)\al\bigr)$-approximation, where $\al$ is the approximation factor
for the at-most-one-job-per-machine variant of \subnbsched on related machines.

\paragraph{Identical machines.} 
We discuss the setting of identical machines separately,
as the underlying arguments are simpler and more direct here. 
Notice that with identical machines, the at-most-one-job-per-machine variant of
\subnbsched is {\em precisely} submodular norm-budgeted knapsack.  
%Observe that this latter problem is simply \subnbknap. 
So as a consequence of the above reduction, we can focus on developing a bicriteria 
approximation for \subnbsched on identical machines. 
%and solving the at-most-one-job-per-machine variant of the problem. 
%The latter problem is simply \subnbknap, so we can use our
%$O(1)$-approximation algorithm for \subnbknap from Section~\ref{submodknap}. 
%So we focus on developing a bicriteria $(\rho,\gm)$-approximation algorithm for
%\subnbsched, where we may violate the norm-budget by a $\gm$-factor.
%
%\medskip

We proceed to describe a simple such bicriteria approximation algorithm.
%We show that this can also be essentially
Let $(J, m, \{p_j\}_{j\in J}, \rewd:2^J\mapsto\Rp, f:\R^m\mapsto\Rp, B)$ be an instance of
\subnbsched on identical machines. 
%Let $\ellfin$ be the last index in $\POS$. 
We may assume that $f$ is normalized so that $f(1,0,\ldots,0)=1$.
Let $\sg^*:\optset\mapsto[m]$ be an optimal solution inducing the load vector
$\lvo=\lvec{\sg^*}$. Since we have identical machines, we may assume that
$\lvo=\lvo^{\down}$.
Let $\POS=\POS_{m,1}$, i.e., it consists of powers of $2$ up to (and potentially
including) $m$. Fix some $\ve>0$, $\ve\leq 0.33$. 
In polynomial time, we can guess a vector $\vt\in\Rp^\POS$ such that 
$\lvo_\ell\leq t_\ell\leq (1+\ve)\lvo_\ell+\kp$ for all $\ell\in\POS$, where 
$\kp=\ve B/m$. 

Our bicriteria approximation algorithm proceeds as follows. 
%Let $\POS'=\POS\cup\{m\}$.
Initialize $A\assign\es$. %$I\assign\es$, and  
Recall that for $S,T\sse J$, 
$\rewd_S(T):=\rewd(S\cup T)-\rewd(S)$ is the incremental reward of adding $T$ to $S$. 
For all $\ell\in\POS$, considered in increasing order, we do the following. We
approximately solve a {\em knapsack-constrained submodular maximization} problem with
ground set $\{j\notin A: p_j\leq t_\ell\}$, submodular function $\rewd_A(.)$, $p_j$'s as the
item sizes, and knapsack budget $\bigl(\nxt(\ell)-\ell\bigr)t_\ell$. It is known that
knapsack-constrained submodular maximization admits a
$\frac{e}{e-1}$-approximation~\cite{Sviridenko04}.   
Let $T$ be the set of jobs returned.
We schedule the jobs in $T$ on machines $\{\ell,\ldots,\nxt(\ell)-1\}$ so that the 
load assigned to each machine is at most $2t_\ell$.
Since we have identical machines and each job in $T$ has $p_j\leq t_\ell$, this is always
possible. We update $A\assign A\cup T$ (and continue with the next index in $\POS$).

%\vspace*{-1ex}
%\paragraph{Analysis.}
\medskip
Let $\beta=\frac{e}{e-1}$.
Let $\sg$ be the assignment so obtained. Note that by design, $\lvec{\sg}\leq 2t^{\exp}$,
so $f(\lvec{\sg})\leq 2f(t^{\exp})\leq 2(1+\ve)f(\lvo)+\ve B$ (Lemma~\ref{toplestim} (b)),
which is at most $3B$.
To lower bound the reward obtained, let $\optset_\ell\sse\optset$ be the jobs assigned by
$\sg^*$ to machines $\ell,\ldots,\nxt(\ell)-1$, for $\ell\in\POS$. Note that
$\{\optset_\ell\}_{\ell\in\POS}$ partitions $\optset$. For any $\ell\in\POS$ and any
$i\in\{\ell,\ldots,\nxt(\ell)-1\}$, we have $\lvo_i\leq t^{\exp}_i=t_\ell$, so
$p(\optset_\ell)\leq\bigl(\nxt(\ell)-\ell\bigr)t_\ell$ and 
$p_j\leq t_\ell$ for all $j\in\optset_\ell$. 
Let $A_{\prev(\ell)}$ be the set $A$ at the start of iteration $\ell$ (with $A_0:=\es$), and
$T_\ell$ be the set added in iteration $\ell$. It follows that $\optset_\ell-A_{\prev(\ell)}$ is a valid
solution to the knapsack-constrained submodular maximization problem considered in
iteration $\ell$, and so   
$\rewd_{A_{\prev(\ell)}}(T_\ell)\geq\frac{1}{\beta}\cdot\rewd_{A_{\prev(\ell)}}(\optset_\ell-A_{\prev(\ell)})
=\frac{1}{\beta}\cdot\rewd_{A_{\prev(\ell)}}(\optset_\ell)$. 
Similar to the proof of Lemma~\ref{goodrewd}, this implies that for the final set $A$
returned, we have $\rewd(A)\geq\frac{\rewd(\optset)}{\beta+1}$.

%\medskip
Thus, the above algorithm is a $\bigl(\frac{2e-1}{e-1},\,3\bigr)$-approximation algorithm for
\subnbsched on identical machines. 
%where recall that $\beta=\frac{e}{e-1}$.
Combining this with the $\bigl(\frac{2e-1}{e-1}+\ve\bigr)$-approximation algorithm for
\subnbknap using the reduction of Theorem~\ref{lbredn}, yields the guarantee stated in
Theorem~\ref{subnblbthm} for identical machines.
%following.

\subsubsection{Related machines}
%Recall that $s_i>0$ is the speed of machine $i$. 
Let $(J, m, \{p_j\}_{j\in J}, \{s_i\}_{i\in[m]}, \rewd:2^J\mapsto\Rp, f:\R^m\mapsto\Rp, B)$, 
be an instance of \subnbsched on related machines,
where $s_i>0$ is the speed of machine $i$, which fixes $p_{ij}=p_j/s_i$ as the time taken
to process job $j$ on machine $i$.
%With related machines, 
We proceed in a roughly similar fashion as with identical machines.  
Order the machines so that $s_1\geq s_2\geq\ldots\geq s_m$.
One basic property that we utilize, which we sometimes refer to as the similarly-ordered property,
is that given any assignment of jobs to machines, %can be 
if we permute the assignment so that the work assigned to a machine is non-decreasing in its
speed, then this does not increase the norm of the resulting load vector;
see Lemma~\ref{wksort}. Thus, we may focus on assignments where
the associated work-vector is similarly ordered as the speed vector. 
%e sometimes refer to this property as the similarly-ordered property. %sorted work property, 
Again, let $\beta=\frac{e}{e-1}$. 

The at-most-one-job-per-machine problem is not quite \subnbknap now, but we nevertheless
observe that {\em the $(\beta+1)$-approximation guarantee of our algorithm for \subnbknap
carries over to this problem} due to the fact that our algorithm returns a solution whose
size-weighted characteristic vector is coordinate-wise at most that of the optimal solution. 
%The upshot is that 

We therefore focus (as before) on obtaining a bicriteria approximation
for \subnbsched on related machines. We do so by proceeding along similar lines as with
identical machines, with two key differences: (1) instead of guessing the load vector
induced by an optimal solution, we guess the {\em work vector} induced by a near-optimal
solution, where the work on a machine is the total processing time of the jobs assigned to
the machine; (2) we utilize a more-refined, coordinate-wise $(1+\ve)$-approximate,
estimate of the work-vector of a structured near-optimal solution. 
%where we guess this work-vector within a  
Given such an estimate, 
as with identical machines, we solve a knapsack-constrained submodular-maximization
problem to identify a set of jobs to assign to each group of machines having the same
work-estimate. 
We argue that this yields a $\bigl(\frac{2e-1}{e-1}+\ve,3\bigr)$-approximation algorithm for 
\subnbsched on related machines.

%Combining this with the $\bigl(\frac{2e}{e-1}+\ve\bigr)$-approximation algorithm for the
%at-most-one-job-per-machine problem 
Combining these two ingredients using the reduction of Theorem~\ref{lbredn} yields the
guarantee stated in Theorem~\ref{subnblbthm} for related machines. 

\paragraph{The at-most-one-job-per-machine problem.}
Let $\I''$ %be the instance obtained from $\I$, where 
denote the instance where we are constrained to assign at most one job to any given machine.
\begin{comment}
By the similarly-ordered property, if we have a feasible solution to $\I''$ that assigns a
set $A\sse J$ of jobs %$\sg:A\mapsto[m]$  
(so $|A|\leq m$), then we know that the assignment where the
$i$-th largest-size job in $A$ is assigned to the $i$-th fastest machine, for all
$i=1,\ldots,|A|$, yields a feasible assignment.
%the second-largest job is assigned to the second-highest speed machine
\end{comment}
Let $\optpp\sse J$ be the jobs assigned by an optimal solution to $\I''$. 
%the version of \subnbsched on
%related machines where we are constrained to assign at most one job to each machine. So
%$|\optpp|\leq m$, and by the sorted work-order property, given

Consider the \subnbknap instance ``defined'' by item-set $J$, $\{\sz_j=p_j\}_{j\in J}$ item-sizes,
and $\rewd:2^J\mapsto\Rp$ reward function; the norm $f$ and the budget $B$ will not be
relevant. Notice that the algorithm we develop for \subnbknap in Section~\ref{submodknap}
does not really depend on $f$ or $B$ in that this information is only used at the end to
select a suitable feasible solution from among a polynomial-size set of candidate
solutions.
%Furthermore, observe 
%We observe that our algorithm for 
Furthermore and more precisely, our algorithm for \subnbknap identifies a polynomial-size
collection $\C\sse 2^J$, and our analysis shows that for {\em any} $O\sse J$, 
%(and any monotone, symmetric norm $f$, budget $B$), 
there is some set $S\in\C$ satisfying
$\rewd(S)\geq\bigl(\frac{1}{\beta+1}-\ve\bigr)\cdot\rewd(O)$ and
$\svec{S}^{\down}\leq\svec{O}^{\down}$.%
\footnote{
%In Lemma~\ref{nearopt}, we move from an item-set $\optset$ to $\noptset\sse\optset$. 
To elaborate, Lemma~\ref{nsizebnds} (and Lemma~\ref{sizebnds}) prescribes some structural
conditions depending on an item-set $\optset$, and actually shows that 
%if an item-set $A$ satisfies these conditions, then 
$\svec{A}^{\down}\leq\svec{\optset}^{\down}$ for any item-set $A$ satisfying these
conditions. Also, we prove that a specific set (depending on $\optset$) in our portfolio
$\C$ obtains good reward relative to $\rewd(\optset)$ and satisfies the stated conditions.} 
So taking $O=\optpp$, this means that we can find in polynomial time some $S\sse J$ such
that $\rewd(S)\geq\bigl(\frac{1}{\beta+1}-\ve\bigr)\cdot\rewd(\optpp)$ and 
$(p_j)_{j\in S}^{\down}\leq(p_j)_{j\in\optpp}^{\down}$ 

Let $\sg'':\optpp\mapsto[m]$ be the assignment where the $i$-th largest-size job in
$\optpp$ is assigned to the $i$-th fastest machine, for all $i=1,\ldots,|\optpp|$. By the
similarly-ordered property, we have $f\bigl(\lvec{\sg''}\bigr)\leq B$.
Since $(p_j)_{j\in S}^{\down}\leq(p_j)_{j\in\optpp}^{\down}$, it follows that assigning %the assignment
the $i$-th largest-size job in $S$ to the $i$-th fastest machine, for
all $i=1,\ldots,|S|$, yields a load vector that is coordinate-wise at most $\lvec{\sg''}$. 
Hence, we obtain a feasible solution to $\I''$ obtaining reward at least
$\bigl(\frac{1}{\beta+1}-\ve\bigr)\cdot\OPT_{\I''}$.

\paragraph{Bicriteria approximation algorithm.}
Recall that we have $s_1\geq \ldots\geq s_m$.
Recall that for an assignment $\sg:A\mapsto[m]$, the associated {\em work vector}
$\wvec{\sg}$ has coordinates 
$\wvec{\sg}_i:=p\bigl(\sg^{-1}(i)\bigr)=\sum_{j\in A:\sg(j)=i}p_j$ for all $i\in[m]$.
Let $\sg^*:\optset\mapsto[m]$ be an optimal solution inducing the work vector
$\wvecopt=\wvec{\sg^*}$. 
As discussed earlier, we may assume that $\wvecopt=\wvecopt^{\down}$.

Fix some $0<\ve\leq 0.33$, and some $0<\dt<1$.
Let $\POS=POS_{m,\dt}$.
It will not be enough to just estimate $\wvecopt_\ell$ for all $\ell\in\POS$ because
having $\wvec{\sg}_\ell=O(\wvecopt_\ell)$ for all $\ell\in\POS$, for an assignment $\sg$,   
does not imply that $f\bigl(\lvec{\sg}\bigr)=O\bigl(f(\lvec{\sg^*})\bigr)$.
So we proceed somewhat differently.

Let $\ell_0$ be the smallest index in $\POS$ that is at least $\frac{2}{\dt^2}$. 
Let $\POSg=\{\ell\in\POS: \ell>\ell_0\}\cup\{m\}$.
Define $\Mc_\ell:=\{\prev(\ell)+1,\ldots,\ell\}$ for all $\ell\in\POSg$. 
As noted in the proof of Lemma~\ref{nearopt}, we have
$|\Mc_\ell|\leq(1+2\dt)\cdot|\Mc_{\prev(\ell)}|$ for all $\ell\in\POSg$, $\ell>\nxt(\ell_0)$. 
%$\nxt(\ell)-\ell\leq(1+2\dt)\bigl(\ell-\prev(\ell)\bigr)$ for all $\ell\geq\ell_0$.
Also, $|\Mc_{\nxt(\ell_0)}|\leq 2\dt\cdot\ell_0$.
Let $\kp=\ve\cdot\wvecopt_1/m$. 
In polynomial time, we can guess a non-increasing vector $\vt\in\Rp^{[\ell_0]\cup\POSg}$
such that $\wvecopt_\ell\leq t_\ell\leq\max\bigl\{(1+\ve)\wvecopt_\ell,\,\kp\bigr\}$ for all
$\ell\in[\ell_0]\cup\POSg$. 

\medskip
%Our bicriteria approximation algorithm proceeds as follows. 
%Let $\POS'=\POS\cup\{m\}$.
\noindent Initialize $A\assign\es$. %$I\assign\es$, and  
%Recall that for $S,T\sse J$, 
%$\rewd_S(T):=\rewd(S\cup T)-\rewd(S)$ is the incremental reward of adding $T$ to $S$. 
The algorithm has two phases.
\begin{enumerate}[label=$\bullet$, topsep=0.1ex, noitemsep, leftmargin=*]
\item For all $\ell\in[\ell_0]$, considered in increasing order, 
approximately solve a knapsack-constrained submodular maximization problem with
ground set $\{j\notin A: p_j\leq t_{\ell}\}$, submodular function $\rewd_A(.)$, $p_j$'s as the
item sizes, and knapsack budget $t_\ell$, to obtain a job-set $T$. 
%It is known that knapsack-constrained submodular maximization admits a 
%$\frac{e}{e-1}$-approximation~\cite{Sviridenko04}.   
%Let $\beta=\frac{e}{e-1}$.

We schedule the jobs in $T$ on machine $\ell$, if $t_\ell>\kp$, and on machine $1$
otherwise.
We update $A\assign A\cup T$ (and continue with the next index).

\item For all $\ell\in\POSg$, considered in increasing order, we 
approximately solve a knapsack-constrained submodular maximization problem with
ground set $\{j\notin A: p_j\leq t_{\ell}\}$, submodular function $\rewd_A(.)$, $p_j$'s as the
item sizes, and knapsack budget $|\Mc_\ell|\cdot t_\ell$, to obtain a job-set $T$. 
%Let $T$ be the set of jobs returned.

If $t_\ell\leq\kp$, we schedule the jobs in $T$ on machine $1$. 
Otherwise, we schedule the jobs in $T$ on machines in $\Mc_\ell$ 
%$\{\prev(\ell)+1,\ldots,\ell\}$ 
so that the work assigned to each machine is at most $2t_\ell$.
Since $p_j\leq t_\ell$ for all $j\in T$, and $p(T)\leq|\Mc_\ell|\cdot t_\ell$,
this is always possible.  
We update $A\assign A\cup T$ (and continue with the next index in $\POSg$).
\end{enumerate}

%\vspace*{-1ex}
%\paragraph{Analysis.}
\medskip
For the analysis, we argue that there is a near-optimal solution that can be used to
exhibit feasible job-sets for the submodular maximization problem solved in each
iteration.

\begin{lemma} \label{struclem}
There is a job-set $\noptset\sse J$ and assignment $\tsg:\noptset\mapsto[m]$ satisfying
the following  properties.
(a) $\rewd(\noptset)\geq(1-2\dt)\OPT$;
%(b) $\wvec{\tsg}_1\leq(1+\ve)\wvecopt_1$;
(b) $\wvec{\tsg}_i\leq\wvecopt_i$ for all $i\in[\ell_0]$; and
(c) $\wvec{\tsg}_i\leq\wvecopt_\ell$ for all $i\in\Mc_\ell$ and all $\ell\in\POSg$.
\end{lemma}

We prove Lemma~\ref{struclem} shortly, but first we show that this implies that the above
algorithm %is a $\bigl(\frac{2e}{e-1},3\bigr)$-approximation algorithm 
has the desired bicriteria approximation guarantee for \subnbsched on
related machines.  
%readily yields the guarantee of Theorem~\ref{subnblbthm} for related machines.

Let $\sg$ be the assignment returned by the algorithm. By design, the total work
assigned to each machine $i\in[m]$ is at most: 
\[
\begin{cases}
t_1+m\kp\leq(1+2\ve)\wvecopt_1; & \text{if $i=1$} \\
t_i\leq(1+\ve)\wvecopt_i; & \text{if $i\in[\ell_0]$ and $t_i>\kp$}; \\
2t_{\ell}\leq 2(1+\ve)\wvecopt_{\ell}\leq 2(1+\ve)\wvecopt_i; \quad & 
\text{if $i\in\Mc_\ell$, $\ell\in\POSg$, and $t_\ell>\kp$} \\
%\text{if $i\in[m]-[\ell_0]$ and $t_{\nxt(i)}>\kp$} \\
%and at most
0; & \text{otherwise}.
%\text{if ($i\in[\ell_0]$ and $t_i\leq\kp$) or ($i\in[m]-[\ell_0]$ and $t_{\nxt(i)}\leq\kp$)}
\end{cases}
\]
%(In particuler, we have the work on $i$ is $0$ if $i\in[\ell_0]\cup\POSg$ and
%$t_i\leq\kp$, or $i\notin([\ell_0]\cup \POSg)$ and $\nxt(i)=m+1$.) 
So $\wvec{\sg}\leq 3\cdot\wvecopt$, and therefore 
$f(\lvec{\sg})\leq 3f(\lvec{\sg^*})\leq 3B$.

We now lower bound the reward obtained. Recall that we have a $\beta$-approximation for
knapsack-constrained submodular maximization, where $\beta=\frac{e}{e-1}$.
For $\ell\in[\ell_0]\cup\POSg$, let 
%$A_{\ell-1}$ be the set $A$ at the start of iteration $\ell$ (with $A_0:=\es$), and 
$T_\ell$ be the set of jobs added to $A$ in iteration $\ell$, i.e., the iteration when we
consider index $\ell$. 
Let $\noptset$, $\tsg$ be as given by Lemma~\ref{struclem}.
%For any set $I\sse[m]$ of machines, let $\noptset_I$ denote the set of jobs assigned by
%$\tsg$ to machines in $I$;  
%we abbreviate $\noptset_{\{i\}}$ to $\noptset_i$ for $i\in[m]$. 

\begin{comment}
For $\ell\in\Mc_0\cup\POSg$, let $\noptset_\ell\sse\noptset$ be the jobs assigned by
$\tsg$ to machine $\ell$, if $\ell\in[\ell_0]$, and to machines
$\prev(\ell)+1,\ldots,\ell$, if $\ell\in\POSg$. Note that the $\noptset_\ell$ sets
partition $\noptset$.
\end{comment}

For $\ell\in[\ell_0]$, let $\noptset_\ell\sse\noptset$ be the jobs assigned by $\tsg$ to
machine $\ell$.  
We have $p(\noptset_\ell)\leq\wvecopt_\ell\leq t_\ell$, so 
if $A'$ is the set $A$ at the start of iteration $\ell$, then $\noptset_\ell-A'$ is a
valid solution to the knapsack-constrained submodular maximization problem considered in
iteration $\ell$. 
For $\ell\in\POSg$, let $\noptset_\ell\sse\noptset$ be the jobs assigned by $\tsg$ to
machines in $\Mc_\ell$.  
For any $i\in\Mc_\ell$, we have $\wvec{\tsg}_i\leq\wvecopt_\ell\leq t_\ell$.
So $p(\noptset_{\ell})\leq|\Mc_\ell|\cdot t_\ell$ and $p_j\leq t_\ell$ for all
$j\in\noptset_{\ell}$.  
So again letting $A'$ be the set $A$ at the start of iteration $\ell$, it follows that
$\noptset_\ell-A'$ is a valid solution to the knapsack-constrained submodular maximization
problem considered in iteration $\ell$. 

In both cases, we therefore obtain that 
$\rewd_{A'}(T_\ell)\geq\frac{1}{\beta}\cdot\rewd_{A'}(\noptset_\ell)$. 
Noting also that the $\noptset_\ell$ sets, where $\ell$ ranges over $[\ell_0]\cup\POSg$,
partition $\noptset$, this implies that that the final set $A$ satisfies 
$\rewd(A)\geq\frac{\rewd(\noptset)}{\beta+1}\geq\frac{1-2\dt}{\beta+1}\cdot\OPT$.
%\medskip
Thus, the above algorithm is a $\bigl(\frac{2e}{e-1}+O(\dt),\,3\bigr)$-approximation algorithm. 
%for \subnbsched on related machines. 
%where recall that $\beta=\frac{e}{e-1}$.

\begin{comment}
Combining this with the $\bigl(\frac{2e}{e-1}+\ve\bigr)$-approximation algorithm for the
at-most-one-job-per-machine problem yields the guarantee stated in
Theorem~\ref{subnblbthm} for related machines.
\end{comment}

We utilize the following claim in proving Lemma~\ref{struclem}.

\begin{claim} \label{xosprop}
Let $\gensub:2^{[n]}\mapsto\Rp$ be a monotone, submodular function. 
Let $S\sse[n]$ be partitioned as $T_1\cup\ldots\cup T_r$ for some $r\geq 1$. 
Let $a\in\{0,1,\ldots,r\}$. There is an index-set $I\sse[r]$ with $|I|=a$  
such that $\gensub\bigl(\bigcup_{q\in I}T_q\bigr)\geq\frac{a}{r}\cdot\gensub(S)$.
\end{claim}

\begin{proof}
If $a=0$, the statement trivially holds for $I=\es$. So suppose $a\geq 1$.
The claim follows from the fractional-cover property of submodular functions (and more 
generally, XOS functions), which states that 
\[
\Bigl(\min\ \sum_{C\sse S}\gensub(C)x_C \quad \text{s.t.} 
\quad \sum_{C\sse S}x_C\geq 1\ \ \forall e\in S, \quad x\geq 0\Bigr) \geq \gensub(S).
\]
It follows that
\[
\frac{1}{\binom{r-1}{a-1}}\cdot\sum_{I\sse[r]: |I|=a}\gensub\bigl(\bigcup_{q\in I}T_q\bigr)\geq\gensub(S)
\]
which implies that there is some $I\sse[r]$ with $|I|=a$ such that 
$\gensub\bigl(\bigcup_{q\in I}T_q\bigr)\geq\gensub(S)\cdot\frac{\binom{r-1}{a-1}}{\binom{r}{a}}
=\gensub(S)\cdot\frac{a}{r}$.
\end{proof}

\begin{proofof}{Lemma~\ref{struclem}}
Let $\Mc_{\ell_0}:=[\ell_0]$.
For any set $I\sse[m]$ of machines, we use $\optset_I$ to denote the
set of jobs assigned by $\sg^*$ to machines in $I$; 
we abbreviate $\optset_{\{i\}}$ to $\optset_i$ for $i\in[m]$. 
Thus, $p(\optset_i)=\wvecopt_i$ for all $i\in[m]$. Recall that
$\wvecopt=\wvecopt^{\down}$.
%Define $\Mc_{\ell_0}=[\ell_0]$, and for $\ell\in\POSg$, define 
%$\Mc_\ell:=\{\prev(\ell)+1,\ldots,\ell\}$. Note that 
%$|\Mc_{\ell)}|\leq(1+2\dt)\cdot|\Mc_{\prev(\ell)}|$ for all $\ell\in\POSg$, and
%$|\Mc_{\nxt(\ell_0)}|\leq 2\dt\cdot|\Mc_{\ell_0}|$.

%We modify the optimal solution, $(\optset, \sg^*:\optset\mapsto[m])$, to obtain
%$(\noptset,\tsg)$. 
We modify $\optset$ and $\sg^*$ by dropping some jobs, and moving jobs from some machines
in $\Mc_\ell$ to machines in $\Mc_{\prev(\ell)}$ for all $\ell\in\POSg$, so as to reduce
the total work assigned to machines in $\Mc_\ell$ to at most $\wvecopt_\ell$.
%We do so via the following iterative process.

\smallskip
\noindent
Initialize $R\assign\es$.
%Repeat the following steps for all $\ell\in\{\ell_0\}\cup\POSg$ considered in increasing
%order. 
Consider indices $\ell\in\{\ell_0\}\cup\POSg$ in increasing order, and do the following. 
%we define $H_\ell\sse\Mc_\ell$ iteratively as follows.  
%
\begin{enumerate}[label=$\bullet$, topsep=0.2ex, itemsep=0.1ex]
\item If $\ell=\ell_0$, let $H_\ell$ be the index-set $I$ obtained by applying
Claim~\ref{xosprop} to the submodular function $\rewd(\cdot)$, taking
$S=\optset_{\Mc_\ell}$ with the partition $\{\optset_i\}_{i\in\Mc_\ell}$, and 
$a=|\Mc_{\ell_0}|-|\Mc_{\nxt(\ell_0)}|\geq(1-2\dt)\cdot|\Mc_{\ell_0}|$.

We drop the jobs assigned by $\sg^*$ to machines in $\Mc_\ell-H_\ell$, thus freeing up
these machines.
%Note that we free up $|\Mc_{\nxt(\ell_0)}|$ machines from $\Mc_{\ell_0}$.

\item If $\ell=\nxt(\ell_0)$, set $H_\ell=\Mc_\ell$. 
For each machine $i\in H_\ell$, we move the jobs in $\optset_i$ to a distinct machine in
$\Mc_{\ell_0}-H_{\ell_0}$; note that this is well defined since 
$|\Mc_{\ell_0}|-|H_{\ell_0}|=|\Mc_{\nxt(\ell_0)}|$. We thus free up all the machines in $\Mc_\ell$. 

\item For all other $\ell$, let $H_\ell$ be the index-set $I$ obtained by applying
Claim~\ref{xosprop} to the submodular function $\rewd_R(\cdot)$, taking
$S=\optset_{\Mc_\ell}$ with the partition $\{\optset_i\}_{i\in\Mc_\ell}$, and 
$a=|\Mc_{\prev(\ell)}|\geq\frac{|\Mc_\ell|}{1+2\dt}\geq(1-2\dt)\cdot|\Mc_\ell|$.

We drop the jobs assigned by $\sg^*$ to machines in $\Mc_\ell-H_\ell$.
For each machine $i\in H_\ell$, we move the jobs in $\optset_i$ to a distinct machine in
$\Mc_{\prev(\ell)}$. This is well defined since $|H_\ell|=|\Mc_{\prev(\ell)}|$ (and note
that all machines in $\Mc_{\prev(\ell)}$ were freed up in the previous iteration).
After this movement, all machines in $\Mc_\ell$ are free.
\end{enumerate} 
We set $R\assign R\cup\optset_{H_\ell}$ and continue to the next index. 

\smallskip
Intuitively, for all $\ell\in\{\ell_0\}\cup\POSg$, $H_\ell$ is a suitable subset of
machines from $\Mc_\ell$ such that the jobs assigned by $\sg^*$ to these machines gather
large incremental reward.

%\smallskip
For $\ell\in\{\ell_0\}\cup\POSg$, let $R_\ell$ be the set $R$ at the end of the iteration
when index $\ell$ was considered. Define $R_{\prev(\ell_0)}:=\es$ for notational convenience.
Let $\noptset=\bigcup_{\ell\in\{\ell_0\}\cup\POSg}\optset_{H_\ell}$, which is the set $R$ after
all indices in $\{\ell_0\}\cup\POSg$ have been considered, and $\tsg$ be the modified
assignment obtained by the above iterative process.
From the guarantee of Claim~\ref{xosprop}, we have 
\begin{equation}
\rewd_{R_{\prev(\ell)}}\bigl(\optset_{H_{\ell}}\bigr)\geq(1-2\dt)\cdot\rewd_{R_{\prev(\ell)}}\bigl(\optset_{\Mc_\ell}\bigr)
\qquad \text{for all }\ell\in\{\ell_0\}\cup\POSg.
\end{equation}
(This holds also for index $\ell=\nxt(\ell_0)$, where do not utilize Claim~\ref{xosprop}.)
Summing up these inequalities yields $\rewd(\noptset)$ on the LHS. For the RHS, observe
that $R_{\prev(\ell)}\sse\bigcup_{r\in\{\ell_0\}\cup\POSg:\, r<\ell}\optset_{\Mc_r}$; so
using submodularity and the fact that $\{\optset_{\Mc_\ell}\}_{\ell\in\{\ell_0\}\cup\POSg}$
partitions $\optset$, we obtain that the RHS is at least $(1-2\dt)\rewd(\optset)$.
Thus, $\rewd(\noptset)\geq(1-2\dt)\OPT$,  proving part (a).

For part (b), under the assignment $\tsg$, a machine $i\in\Mc_{\ell_0}$ either retains the
job-set $\optset_i$, %assigned to it by $\sg^*$, 
or is assigned the job-set $\optset_{i'}$ for some machine $i'\in\Mc_{\nxt(\ell_0)}$. 
In the latter case, since $i'>i$, we have 
$\wvec{\tsg}_i=p(\optset_{i'})\leq p(\optset_i)=\wvecopt_i$. 

For part (c), consider $\ell\in\POSg$ and $i\in\Mc_\ell$. Then, under $\tsg$, machine $i$
is either assigned jobs in $\optset_{i'}$ for some $i'\in\Mc_{\nxt(\ell)}$, or no
jobs at all (for example, if $\ell$ is the last index in $\POSg$). So we have 
$\wvec{\tsg}_i\leq\max_{i'\in\Mc_{\nxt(\ell)}}\wvecopt_{i'}\leq\wvecopt_\ell$.
\end{proofof}

\section{\boldmath Refinements: improved guarantees for $\topl$ norms} \label{refine}
For the special case of $\topl$ norms, we show that one can derive approximation
guarantees in a simpler fashion, by reducing the norm-budgeted packing problem to
one with a sum- budget constraint and/or one with max- budget constraint.
In this section only, we refer to the $\ell_1$-norm as $\ssum$-norm, and the
$\ell_\infty$-norm as $\smax$-norm (to avoid any confusion with the $\ell$ in $\topl$ 
norm). 

\begin{theorem} \label{toplsizechar}
Let $\I=\bigl([n],\sols,\{\rewd_e,\sz_e\}\}_{e\in[n]},\topl,B\bigr)$ be an instance of a
norm-budgeted packing problem involving a $\topl$ norm, where 
the size-vector $\svec{T}$ induced by a solution $T\in\sols$ is the size-weighted
characteristic vector of $T$.
For $t\geq 0$, let $\I^{\ssum}_t$ be the instance  
$\bigl([n],\sols,\{\rewd_e,(\sz_e-t)^+\}\}_{e\in[n]},\sumnorm,B\bigr)$ involving the
%$\ell_1$ 
$\ssum$-norm.
Then, any $\al$-approximation algorithm for solving instances of the form $\I^{\ssum}_t$ 
can be used to obtain an $\al$-approximate solution to $\I$.
\end{theorem}

\begin{proof}
This is an easy consequence of Theorem~\ref{normprops} (a). Letting $\optset\in\sols$
denote an optimal solution to $\I$, we obtain that for some $e^*\in\optset$, taking
$t^*=\sz_{e^*}$, we have $\topl(\svec{\optset})=\ell t^*+\sum_{e\in\optset}(\sz_e-t^*)^+$.
Thus, guessing this element $e^*$, and solving (approximately) the 
sum-budget constrained problem $\I^{\ssum}_{t^*}$ yields the desired solution to $\I$.
\end{proof}

This immediately implies that for $\topl$-norms, there is an FPTAS for \normknap, a PTAS
for \normmwis~\cite{BergerBGS11,GrandoniRSZ14}, %{AradKS24}, 
and a PTAS for \normmatch~\cite{BergerBGS11,GrandoniRSZ14}.

\medskip
For problems such as \normsched, where the size-vector $\svec{T}$ is obtained by
aggregating some $\sz_e$ values, we show that the $\topl$-norm-budgeted packing problem
reduces to sum-budget constrained problem, and a norm-budgeted packing problem with the
$\ell_\infty$ norm.
We assume that the norm-budgeted packing problem satisfies the following aggregation
property: for any $T\in\sols$, each coordinate $i$ of $\svec{T}$ corresponds to  
a set $S_i\sse T$ such that:
(i) $\svec{T}_i=\sz(S_i)$ for every coordinate $i$, and the $S_i$-sets partition $T$;
and (ii) for any subset $I$ of coordinates, %we have 
the size-vector $\svec{\bigcup_{i\in I}S_i}$ of the solution $\bigcup_{i\in I}S_i$ (which
is in $\sols$) is equal to $\{\svec{T}_i\}_{i\in I}$, modulo permutations of
coordinates and padding with $0$s.
%\footnote{Here, $\svec{\bigcup_{i\in I}S_i}$ denotes the size-vector of the solution
%$\bigcup_{i\in I}S_i$, which lies in $\sols$.}
%and suppose $\svec{T}\in\Rp^k$. Then, for any subset
%$I\sse[k]$, there is a solution $S\sse T$ such that $\svec{S}$ is equal to
%$\bigl(\svec{T}_i\bigr)_{i\in I}$ modulo permutations of the coordinates and padding with
%$0$s. 
For instance, \normsched satisfies this property, where $S_i$ is the set of jobs assigned
to machine $i$ under $T$.
%because taking a subset $I$ of
%$[m]$ corresponds to taking the solution where we select the jobs assigned to machines in
%$I$ under $T$.

\begin{theorem} \label{toplgen}
Let $\I=\bigl([n],\sols,\{\rewd_e,\sz_e\}\}_{e\in[n]},\topl,B\bigr)$ be an instance of a 
norm-budgeted packing problem involving a $\topl$ norm, satisfying the above aggregation
property.  
%any subset of the coordinates of $\svec{T}$
Let $\I^{\ssum}$ denote the $\ssum$-norm budget constrained instance
$\bigl([n],\sols,\{\rewd_e,\sz_e\}\}_{e\in[n]},\ssum,B\bigr)$, and 
$\I^{\smax}$ denote the $\smax$-norm budget constrained instance
$\bigl([n],\sols,\{\rewd_e,\sz_e\}\}_{e\in[n]},\smax,B/\ell\bigr)$. 
Given an $\al$-approximate solution $T^{\ssum}$ to $\I^{\ssum}$, and
a $\beta$-approximate solution $T^{\smax}$ to $\I^{\smax}$, one can obtain an
$(\al+\beta)$-approximate solution to $\I$. 
\end{theorem}

\begin{proof}
This reduction is implicit in the work of~\cite{NeogiPS24}.
%Let $T^{\ssum}\in\sols$ be an $\al$-approximate solution to $\I^{\ssum}$, and $T^{\smax}$
%be a $\beta$-approximate solution to $\I^{\smax}$.
Note that $\topl(\svec{T^{\ssum}})\leq\bigl\|\svec{T^{\ssum}}\bigr\|_1\leq B$, and
$\topl(\svec{T^{\smax}})\leq\ell\cdot\topl[1](\svec{T^{\smax}})\leq\ell\cdot\frac{B}{\ell}=B$. 
So $T^{\ssum}$ and $T^{\smax}$ are both feasible solutions to $\I$. We claim that the
better of the two solutions achieves an $(\al+\beta)$-approximation.

Let $\optset$ be an optimal solution to instance $\I$, and $\OPT=\rewd(\optset)$ be the
optimal value for instance $\I$. Let $\vo=\svec{\optset}$ be a vector in $\R^k$, 
and let $\optset_1,\ldots\optset_k$ be the partition of $\optset$ such that
$\vo_i=\sz(\optset_i)$ for all $i\in[k]$.
Let $I=\bigl\{i\in[k]: \vo_i>\frac{B}{\ell}\bigr\}$. Since $\topl(\vo)\leq B$, we must
have $|I|\leq\ell$ and so $\sum_{i\in I}\vo_i\leq B$. So taking $S=\bigcup_{i\in I}\optset_i$,
%$S\sse\optset$ is such that
%$\svec{S}$ corresponds to $(\vo_i)_{i\in I}$, then 
we obtain that $\svec{S}=(\vo_i)_{i\in I}$, and hence $S$ is a feasible solution to
$\I^{\ssum}$. 
%Let $S'\sse\optset$ be such that $\svec{S'}$ corresponds to $(\vo_i)_{i\in[k]-I}$. 
Let $S'=\bigcup_{i\in[k]-I}\optset_i$. Then, $\svec{S'}=(\vo_i)_{i\in[k]-I}$, so
by construction $\topl[1](\svec{S'})\leq\frac{B}{\ell}$, and
$S'$ is a feasible solution to $\I^{\smax}$. 
It follows that
$\OPT=\rewd(\optset)=\rewd(S)+\rewd(S')\leq\OPT_{\I^{\ssum}}+\OPT_{\I^{\smax}}$, where
$\OPT_{\I^{\ssum}}$ and $\OPT_{\I^{\smax}}$ denote respectively the optimum values for the
$\I^{\ssum}$ and $\I^{\smax}$ instances. Then
\begin{equation*}
\begin{split}
\max\Bigl\{\rewd(T^{\ssum}),\rewd(T^{\smax})\Bigr\}
& \geq\frac{\al}{\al+\beta}\cdot\rewd(T^{\ssum})+\frac{\beta}{\al+\beta}\cdot\rewd(T^{\smax})
\\
& \geq\frac{1}{\al+\beta}\cdot\bigl(\OPT_{\I^{\ssum}}+\OPT_{\I^{\smax}}\bigr)\geq\frac{\OPT}{\al+\beta}
\end{split}
\end{equation*}
where the second inequality follows from the approximation guarantees of $T^{\ssum}$ and $T^{\smax}$.
\end{proof}

Theorem~\ref{toplgen} yields an improved approximation guarantee 
%for $\topl$ norms. 
%\begin{enumerate}[label=(\alph*), topsep=0.2ex, noitemsep, leftmargin=*]
%\item 
of $\bigl(\frac{2e}{e-1}+\ve\bigr)$ for \normsched with $\topl$ norms, since the min-sum
budgeted problem is simply the standard knapsack problem, and the min-max budgeted problem
is maximum-GAP for which there is an $\frac{e}{e-1}$-approximation~\cite{FleischerGMS11}.

%\item $(2+\ve)$ for \normsched on identical machines, since in this case the min-max
%budgeted problem is an instance of the uniform multiple knapsack problem, which admits a
%PTAS~\cite{Kellerer99,ChekuriK05}.
%\end{enumerate}

\section{PTAS for \boldmath \normsched on related machines} \label{relatedmc}
We now describe a %improve the approximation guarantee of
%Theorem~\ref{identicalbiptas} to 
PTAS for \normsched on related machines 
(i.e., $(1+\ve)$-approximation for any $\ve>0$).
Recall that in the setting of related machines, %corresponds to the special case where 
we have $p_{ij}=p_j/s_i$ for every machine $i\in[m]$, job $j\in J$, where $s_i>0$ is the
speed of machine $i$ and $p_j$ is the processing requirement of job $j$, %which
leading to $p_{ij}=p_j/s_i$ as the time needed on machine $i$ to process job $j$. 
The setting of identical machines is the further special case where all speeds are $1$.
%i.e., obtain a PTAS, %for \normsched on identical machines 
%assuming that we have some further information about the norm $f$, in particular,
%access to approximate subgradients of the norm; see Theorem~\ref{ident-ptasthm}. 
%For the special case when $f$ is the $\ell_\infty$ norm:
%\begin{enumerate}[label=(\arabic*), nosep, leftmargin=*]
%\item 
Recall that the special case of \normsched on identical machines where $f$ is the
$\ell_\infty$ norm corresponds to the uniform multiple knapsack problem (\mkp), which is
already strongly \nphard; thus, our PTAS yields a tight complexity result for \normsched
on related machines. 
%\item 
Furthermore, \normsched on related machines with the $\ell_\infty$ norm
reduces to {\em non-uniform \mkp}, wherein we seek to pack a maximum-reward set of
jobs into $m$ machines (or knapsacks or bins) with given capacities $\{u_i\}_{i\in[m]}$; 
%with rewards and sizes 
%\footnote{In non-uniform \mkp, we have jobs with rewards and sizes, and $m$
%machines with capacities $\{u_i\}_{i\in[m]}$, and we seek an assignment
%$\sg:A\mapsto[m]$ of some subset $A$ of jobs to machines satisfying
%$\sum_{j:\sg(j)=i}p_j\leq u_i$ for all $i\in[m]$, that maximizes the total reward of
%$A$. 
this can be cast as \normsched on related machines by taking $u_i$ to be the speed of
machine $i$ and requiring that the $\ell_\infty$-norm of the load vector should be at most
$1$. 
%which corresponds to the special case where $f$ is the $\ell_\infty
%Moreover, for the $\ell_\infty$ norm, and more generally, for any
%ordered norm, one can indeed obtain the desired norm information, 
Thus our PTAS for \normsched on related machines substantially generalizes the PTAS for
non-uniform \mkp~\cite{ChekuriK05}. 
%\end{enumerate}

We may assume that $p_j>0$ for all jobs, as we can always consider the smaller instance
involving jobs $j$ with $p_j>0$, and then tag on the jobs with zero processing time (to
any machine).

\subsection{Identical machines} \label{ident-ptas}
\begin{comment}
We now describe how to %improve the approximation guarantee of
%Theorem~\ref{identicalbiptas} to 
obtain a PTAS for \normsched on identical machines 
(i.e., $(1+\ve)$-approximation for any $\ve>0$),
%i.e., obtain a PTAS, %for \normsched on identical machines 
assuming that we have some further information about the norm $f$, in particular,
access to approximate subgradients of the norm; see Theorem~\ref{ident-ptasthm}. 
Recall that even the special case where $f$ is the $\ell_\infty$ norm, namely, the uniform
multiple knapsack problem (\mkp), is strongly \nphard, so our PTAS yields a tight result
for identical machines. Moreover, for the $\ell_\infty$ norm, and more generally, for any
ordered norm, one can indeed obtain the desired norm information, so our PTAS for
\normsched on identical machines substantially generalizes the PTAS for uniform \mkp.
\end{comment}

We begin by considering the setting of identical machines, as the algorithm here is
simpler, and will serve to convey some of the ideas that we build upon in developing the
PTAS for related machines (Section~\ref{rel-ptas}).
%\vspace*{-1ex}
%\paragraph{An overview.}
We first give an overview. Note that if we knew
the load-vector $\lvo$ of an optimal solution, then the problem reduces to selecting a
maximum-reward set of jobs and assigning them to machines while staying within capacity
$\lvo_i$ on machine $i$. This is precisely 
%the {\em multiple knapsack problem with non-uniform capacities} 
non-uniform \mkp, considered 
by~\cite{ChekuriK05}, who devised a PTAS for this problem. 
%Our PTAS draws and builds upon on some of their ideas. 
While in \mkp, one knows the capacities and one can utilize this 
information to gain suitable information about the job assignments, one significant
challenge that we encounter is that we do not have this information, and do not have a
``target vector'' $\lvo$ to work with. 
\begin{comment}
Instead, we will ``guess'' a vector close to $\lvo$, and use its coordinates as proxy
capacities (step~\ref{tguess}). However, using these proxy capacities as is will yield a
solution that violates the norm budget. Instead, we use these proxy capacities to glean,
much more carefully and more conservatively, suitable information about the job
assignments; steps~\ref{sizesparse}--\ref{lmcassign}. 
\end{comment}
Instead, we will ``guess'' suitable features of a near-optimal load vector, %$\lvo$, 
and use this information along
with some further enumeration to identify the assignments of a suitable set of jobs.
%identify suitable information about job assignments. 
We then consider a suitable convex program (see \eqref{extncp}) to %encode the problem of 
find a fractional job assignment that is consistent with this partial assignment and
minimizes the norm of the resulting load vector, 
and the information we have gleaned will guarantee that the optimal value of this convex
program is at most the norm budget. 
%there is a feasible solution whose resulting load vector satisfies the norm budget. 
It turns out that one can
give a (easily-computable) closed-form expression for the optimal value of this convex
program and can find an optimal solution by solving a related LP (see
Lemma~\ref{extnlem}). We then argue that this fractional assignment can be rounded to an
integral assignment without violating the norm budget, and losing only a $(1-\ve)$-factor
in the reward. 

\subsubsection{The algorithm} \label{ident-detail}
%The PTAS involves various steps, which 
We now describe the PTAS in detail. 
%We assume that the $p_j$'s are integers. 
%and will build on numerous ideas. 
Some of the steps of the PTAS involve some partial enumeration to obtain some information
about a  near-optimal solution. To keep exposition simple, we will assume that we have
found (by enumeration) information consistent with a near-optimal solution. 
We will show that the entire enumeration can be done in polynomial time.
%that only a constant number of rounds of enumeration, each over polynomial-size sets, is
%needed, so the overall running time is polynomial. 
To avoid cumbersome notation, we assume
that $\frac{1}{\ve}$ is an integer. We also assume that $\ve\leq 0.25$.

We will utilize two useful results. 
Lemma~\ref{enumlem} states an enumeration lemma that we will
frequently use, which captures the enumeration process that was used for \normknap.
%(and in the proof of Theorem~\ref{ident-biptasthm}).  
%and for identical machines. 
%We will also utilize the following 
Lemma~\ref{extnlem} %states a useful result that 
describes how one can extend an assignment of some jobs to a fractional assignment of other
jobs so as to minimize the norm of the resulting load vector. 
For $x\in\R$, we use $(x)^+$ to denote $\max\{x,0\}$. 

\begin{lemma}[{\bf Enumeration lemma}] \label{enumlem}
Let $a_1,a_2,\ldots,a_k\geq 0$ and $\Gm\geq 0$, be such that $\sum_{r=1}^k a_r\leq\Gm$. 
Suppose we have an estimate $\est$ such that $\Gm\leq(1+\ve)\est$, where $\ve>0$. 
Define $\Dt=\frac{\ve\cdot\est}{k}$.
Consider the set %we can efficiently find $\seqset\sse\Rp^k$ with
\[
\seqset:=\Bigl\{b\in\Rp^k:\quad b_r\text{ is a multiple of }\Dt\ \ \forall r\in[k],
\quad\ \  \sum_{r=1}^kb_r\leq\Bigl(1+\tfrac{1}{\ve}\Bigr)k\Dt\Bigr\}.
\]
Then, (a) $|\seqset|\leq 2^{O(k/\ve)}$ and elements in $\seqset$ can be enumerated in
$2^{O(k/\ve)}$ time; and %such that 
(b) $\seqset$ contains a tuple $(\ta_1,\ldots,\ta_k)$
satisfying $a_r-\Dt\leq\ta_r\leq a_r$ for all $r\in[k]$, and 
$\sum_{r=1}^k\ta_r\geq\sum_{r=1}^k a_r-\ve\cdot\est$.

We refer to this by saying that we can ``guess'' (underestimates of) $a_1,\ldots,a_k$ up
to cumulative error $\ve\cdot\est$ in time $2^{O(k/\ve)}$.  
\end{lemma}

\begin{comment}
We will also utilize the following useful result that shows how one can extend an
assignment of some jobs to a fractional assignment of other jobs so as to minimize the
norm of the resulting load vector. For $x\in\R$, we use $(x)^+$ to denote 
$\max\{x,0\}$. 
\end{comment}

\begin{lemma} \label{extnlem}
Let $J_1,J_2\sse J$ be disjoint job sets. Let $\sg:J_1\mapsto[m]$ and
define $\preload_i:=\sum_{j\in J_1:\sg(j)=i}p_j$ for all $i\in[m]$. 
Let $\mcset\sse[m]$. Consider
the following convex program to assign the jobs in $J_2$ fractionally to machines in
$\mcset$ so as to minimize the norm of the resulting load vector.
\begin{equation}
\min \ \ f(L_1,\ldots,L_m)\quad\ \text{s.t.} \quad\
\sum_{i\in\mcset}x_{ij}\geq 1\ \ \ \forall j\in J_2, \quad 
L_i=\sum_{j\in J_2}p_jx_{ij}+\preload_i \ \ \ \forall i\in[m], \quad 
x\geq 0. \tag{CP}
\label{extncp}
\end{equation}
Let $z^*$ be the unique value $z\in\Rp$ such that $\sum_{i\in\mcset}(z-\preload_i)^+=p(J_2)$. 
The optimal value of \eqref{extncp} is equal to 
$f\bigl(\max\{\preload_i,z^*\}_{i\in[m]}\bigr)$ and any %solution
$x^*\in\Rp^{\mcset\times J_2}$ satisfying $\sum_{i\in I}x^*_{ij}\geq 1$ for all 
$j\in J_2$, and 
$\sum_{j\in J_2}p_jx^*_{ij}\leq(z^*-\preload_i)^+$ for all 
$i\in\mcset$ yields an optimal solution to \eqref{extncp}.
\end{lemma}

We also utilize the following simple claim.

\begin{claim} \label{jobdom}
Let $J_1, J_2\sse J$ be two job-sets such that $\svec[p]{J_1}\leq\svec[p]{J_2}$, and
$\sg_2:J_2\mapsto[m]$ be an assignment such that $f(\lvec{\sg_2})\leq B$. Then, there is an
assignment $\sg_1:J_1\mapsto[m]$ such that $f(\lvec{\sg_1})\leq B$.
\end{claim}

\begin{proof}
Let $\pi:J_1\mapsto J_2$ be a one-to-one mapping such that $p_j\leq p_{\pi(j)}$ for all
$j\in J_1$. Let $\sg_1$ be the assignment that assigns each $j\in J_1$ to the machine $\sg_2(\pi(j))$.
Then the load vector $\lvec{\sg_1}$ is coordinate-wise at most $\lvec{\sg_2}$, since 
$\lvec{\sg_1}_i=\sum_{j\in J_1: \sg_2(\pi(j))=i}p_j\leq\sum_{j\in J_1:\sg_2(\pi(j))=i}p_{\pi(j)}\leq\lvec{\sg_2}_i$.
Therefore, $f(\lvec{\sg_1})\leq B$.
\end{proof}

\vspace*{-1ex}
\paragraph{I: Finding a suitable set of jobs.}

\begin{enumerate}[label=(L\arabic*), topsep=0.5ex, itemsep=0.1ex, labelsep=*]
\item \label{goodset} \label{ptas-start}
The first step is to guess a set of jobs assigned by a near-optimal solution
%with $\rewd(A)\geq (1-\ve)\OPT$
using the same approach as for norm-budgeted knapsack.
%as in the proof of Theorem~\ref{ident-biptasthm} (Appendix~\ref{ident-biptas}). 
%Let $\tsg:A\mapsto[m]$ be an assignment such that $f(\lvec{\tsg})\leq B$. 
%
Recall that this involves the following.
We consider reward buckets $\itemset_0,\ldots,\itemset_{\nbuck}$,
where $\itemset_q=\bigl\{j\in J: \frac{\thresh_q}{1+\ve}<\rewd_j\leq\thresh_q\bigr\}$
%consists of all jobs with rewards in the range
%$\bigl(\frac{\thresh_q}{1+\ve},\thresh_q\bigr]$ 
with $\thresh_q=\frac{\rmax}{(1+\ve)^q}$ and
$\nbuck=O\bigl(\frac{1}{\ve}\log\frac{n}{\ve}\bigr)$.
Let $\optval$ be an estimate such that $\optval\leq\OPT\leq(1+\ve)\optval$.
%(iii) 
We guess the number $\num_q$ of jobs from each $\itemset_q$ bucket assigned by an
optimal solution using Lemma~\ref{enumlem}. %this can be described as follows: 
To elaborate, we apply Lemma~\ref{enumlem} to the sequence
$\bigl\{\rewd(\optset\cap\itemset_q)\bigr\}_{q\in\dbrack{\nbuck}}$ taking $\est=\optval$
to guess these rewards up to a cumulative error of $\ve\optval$, and set $\num_q$ to be
the estimate of $\rewd(\optset\cap\itemset_q)$ divided by $\thresh_q$. 
(This is precisely what we do in the proof of Theorem~\ref{normknapthm}.) 
%~\ref{ident-biptasthm}.)
%within $(1+\ve)$ multiplicative error and
%$\frac{\ve\cdot\optval}{\nbuck\cdot\thresh_q}$ additive error; and 
We then have 
$\num_q\leq|\optset\cap\itemset_q|\leq(1+\ve)\bigl(\num_q+\frac{\ve\cdot\optval}{\nbuck\cdot\thresh_q}\bigr)$
for all $q\in\dbrack{\nbuck}$
We pick the $\ceil{\num_q}$ smallest-size jobs from each $\itemset_q$ bucket. 
Let $A$ denote this collection of jobs. We have $\rewd(A)\geq(1-\ve)^2\OPT$ as argued in
the proof of Theorem~\ref{normknapthm}. %~\ref{ident-biptasthm}.
Also $\svec[p]{A}\leq\svec[p]{\noptset}$, 
%so one can argue (as in the proof of Theorem~\ref{ident-biptasthm}) 
so by Claim~\ref{jobdom}, there is an assignment
$\sg:A\mapsto[m]$ such that $f(\lvec{\sg})\leq B$; we call such an assignment a feasible
assignment. 

%Recall that $A=\bigcup_{q\in\dbrack{\nbuck}}S_q$, where $S_q\sse\itemset_q$ for
%$q\in\dbrack{nbuck}$.
\end{enumerate}

\vspace*{-1ex}
\paragraph{\boldmath II: Simplifying the instance, obtaining information about a near-optimal assignment.}

\begin{enumerate}[resume*, start=2]
\item \label{tguess}
%Next, we utilize an insight from~\cite{IbrahimpurS21}. 
%Say that an assignment $\sg:A\mapsto[m]$ is feasible if $f(\lvec{\sg})\leq B$.
Let $\tsg:A\mapsto[m]$ be a
feasible assignment %with $f(\lvec{\tsg})\leq B$ 
such that $\lvec{\tsg}^{\down}$ is {\em lexicographically smallest} among all feasible
assignments.%
\footnote{Given distinct vectors $u,v\in\R^m$, we say that $u$ is lexicographically
smaller than $v$ if for some $i\in[m]$, we have $u_i<v_i$ and $u_{i'}=v_{i'}$ for
all $i'=1,\ldots,i-1$.}
Let $\tlvec=\lvec{\tsg}$.
%assignments $\sg:A\mapsto[m]$ satisfying 
%Since all machines are identical, let us re-index the machines so that
%$\tlvec^{\down}=\tlvec$, 
Let $\tlvmin:=\min_{i\in[m]}\tlvec_i$ be the least load on any machine under $\tsg$.
Let $A_1\sse A=\{j\in A: p_j>\tlvmin\}$.
%(Clearly, $\tlvmin$ is also the least load on any machine that is not assigned any job
%in $A_1$.)
\begin{comment}
Let $\mcset=\{i\in[m]: \tsg^{-1}(i)\cap A_1=\es\}$ be the machines not assigned a job from
$A_1$ under $\tsg$. (Note that $\mcset$ need not be a consecutive set of indices.)
Let $i^*$ be the smallest index in $\mcset$, so $i^*$ is the most-loaded machine in
$\mcset$ (under $\tsg$).
\end{comment}
Let $\tlvmax:=\max\,\{\tlvec_i: i\in[m],\ \tsg^{-1}(i)\cap A_1=\es\}$ be the maximum load
on any machine that is not assigned any job from $A_1$ under $\tsg$. 
Utilizing an insight from~\cite{IbrahimpurS21}, we can infer the following.
%As observed in~\cite{IbrahimpurS21}, letting $i^*$ be the smallest-index machine that has at
%least two jobs assigned to it under $\tsg$, we can infer the following.

\begin{claim} \label{tsgclaim}
\begin{enumerate*}[label=(\alph*)] %noitemsep, leftmargin=*]
\item %Let $j\in A$ be such that $p_j>\tlvec_m$. Then, 
For every $j\in A_1$, job $j$ is the only job assigned by $\tsg$ to machine $\tsg(j)$. \quad
\item We have $\tlvmax\leq 2\cdot\tlvmin$.
\end{enumerate*}
\end{claim}

%We assume we know $i^*$, and hence the jobs in $A$ assigned to machines $1,\ldots,i^*-1$
Clearly, $A_1$ comprises the largest $|A_1|$ jobs in $A$, so we may assume that we know
$A_1$. 
Let $A_2=A-A_1$. Since all machines are identical, let
us re-index the machines so that jobs in $A_1$ are assigned to the first $|A_1|$ machines
(one job per machine). (Note that the load vector under this indexing need not be
equal to $\tlvec^{\down}$; it corresponds to $\tlvec$ under some permutation of coordinates.) 
Let $\mcset=\{|A_1|+1,|A_1|+2,\ldots,m\}$. 
%and $\mcnum=|\mcset|=m-|A_1|$. 
%so $\mcnum=m-|A_1|$.
So jobs in $A_2$ are assigned by $\tsg$ to machines in $\mcset$. 
Let $\tlvavg=\bigl(\sum_{i\in\mcset}\tlvec_i\bigr)/\mcnum=p(A_2)/\mcnum$ be the average
load on machines in $\mcset$. Note that we know $\tlvavg$.
Since $\tlvavg\in\bigl[\tlvmin,\tlvmax\bigr]$, by Claim~\ref{tsgclaim} (b), we have that  
$\tlvavg/2\leq\tlvmin\leq\tlvavg\leq\tlvmax\leq 2\cdot\tlvavg$.

We call a job $j\in A_2$ large if $p_j\geq\ve\cdot\tlvavg$, and small otherwise.
We now consider two cases, depending on whether $\mcnum\leq\frac{2}{\ve^2}$, which we call
a {\em \sparse} instance, or whether $\mcnum>\frac{2}{\ve^2}$, which we call a 
{\em \dense} instance. 
\end{enumerate}

%\vspace*{-1ex}
\paragraph{III: \sparse[S] instance.} \nopagebreak
\begin{enumerate}[resume*, start=3]
%\item {\bf \sparse[S] machines.} 
\item \label{spinst} \label{sp-start}
We guess the assignments of the $\frac{\mcnum}{\ve}$ largest-reward jobs in $A_2$. 
The time needed for this enumeration is
$\bigl(\frac{\mcnum}{\ve}\bigr)^{\mcnum}=\bigl(\frac{1}{\ve}\bigr)^{\poly(1/\ve)}$. 
%If this covers all of $A_2$, there is nothing more to be done, so suppose otherwise.
Let $J_1$ consist of $A_1$ and these $\frac{\mcnum}{\ve}$ jobs, whose assignments have
been determined, and $J_2$ be the remaining jobs in $A$. (If $J_2=\es$, there is nothing
more to be done, so suppose otherwise.)
%Thus, we have determined the assignments of all jobs in $J_1$.
Observe that by design, we have
$\rewd_j\leq\frac{\rewd(J_1-A_1)}{\mcnum/\ve}\leq\frac{\ve}{\mcnum}\cdot\rewd(A_2)$ for
all $j\in J_2$. 

\item \label{sp-frac}
We first find a fractional assignment $x^*$ of jobs in $J_2$ to machines in $\mcset$
by invoking Lemma~\ref{extnlem}, letting $\sg$ be the pre-determined assignment of jobs in
$J_1$. Let $L^*$ be the corresponding load vector. Since $\tsg$
yields one possible feasible solution to the convex program \eqref{extncp}, we know that
$f(L^*)\leq B$.  

\item \label{sp-end}
We round $x^*$ using \gap rounding~\cite{ShmoysT93}. Let $\hsg:A\mapsto[m]$
denote the assignment obtained by concatenating the resulting integer solution and $\sg$. 
Let $\hJ_i\sse A$ be the jobs assigned to machine $i$ by $\hsg$. By properties of \gap
rounding, we know that for each machine $i\in\mcset$, by removing at most one job
$j^*_i\in\hJ_i\cap J_2$ with $x^*_{ij}>0$, we can ensure that the load on $i$ is at most
$L^*_i$.  
%we have $\sum_{j\in\hJ_i}\tp_{j}\leq L^*_i+p_{j^*_i}$, where 
%$j^*_i\in\hJ_i\cap J^{\rem}$ with $x^*_{ij}>0$; 
%let $j^*_i$ denote such a job, with the convention that $\{j^*_i\}=\es$ 
We adopt the convention that if $p(\hJ_i)\leq L^*_i$, then $\{j^*_i\}=\es$. 
%is ''null'' with $p_{j^*_i}=\rewd_{j^*_i}=0$ and $\{j^*_i\}=\es$.  
%there is no such job.
We discard $\{j^*_i\}_{i\in\mcset}$ to obtain a feasible assignment, and return this
solution. Note that since we discard at most $\mcnum$ jobs from $J_2$, we lose reward at
most $\ve\cdot\rewd(A_2)$.
\end{enumerate}

\vspace*{-1ex}
\paragraph{IV: \dense[D] instance.} The \dense-instance setting is more complicated. Here,
we will first aim to determine the assignments of all large jobs in $A_2$, then assign the
remaining jobs by computing a fractional assignment using Lemma~\ref{extnlem}, and
rounding this fractional assignment using \gap rounding. To implement this plan, we will
need to sparsify the instance to reduce the number of distinct {\em types} of large jobs, 
where type of a job denotes its (reward, size) tuple, and we
will need to set aside some jobs that are assigned to ``extra'' machines. We will
eventually argue that we can drop the extra machines and obtain a feasible assignment without
sacrificing the reward by much.

Let $A^{\lrg}\sse A_2$ be the large jobs in $A_2$, let $A^{\sml}=A_2-A^{\lrg}$ be the
small jobs in $A_2$.
Note that $|A^{\lrg}|\leq\mcnum/\ve$ since the total load assigned to machines in
$\mcset$ is $\mcnum\cdot\tlvavg=p(A_2)\geq p(A^{\lrg})$ and 
%is at least $\sum_{j\in J_I}p_j$ (as $\tsg$ assigns all jobs in $J_I$ to machines in $I$) 
%the total processing time of all jobs assigned by $\sg$ to $I$, 
$p_j\geq\ve\cdot\tlvavg$ for every $j\in A^{\lrg}$. 
(If $|A^{\lrg}|>\mcnum/\ve$, then we declare failure; this can only happen if one of
our guesses is incorrect.)

\begin{enumerate}[resume*, start=6]
\item {\bf Sparsifying job sizes.} \label{sizesparse}
We sparsify the instance by using a shifting idea used for bin packing and the
multiple-knapsack problem~\cite{VegaL81,ChekuriK05} so that there are only  
$O\bigl(\frac{1}{\ve^2}\bigr)$ distinct sizes of large jobs.
%$O\bigl(\frac{\nbuck}{\ve}\bigr)=O\bigl(\frac{1}{\ve^2}\ln\frac{n}{\ve}\bigr)$ distinct
%job-sizes.  

%To elaborate, first 
%and $\sg$ %there is an assignment that 
%assigns jobs in $J'_I$ to machines in $I$ without exceeding capacity $t$. 
We consider jobs in $A^{\lrg}$ in non-increasing order of size,
and divide them into $k=1+\frac{1}{\ve^2}$ groups, where the first $\frac{1}{\ve^2}$
groups contain $\floor{\ve^2|A^{\lrg}|}$ jobs, and the last group contains less than
$\frac{1}{\ve^2}$ jobs. (If $\ve^2|A^{\lrg}|<1$, then the first $\frac{1}{\ve^2}$ groups
are empty, and the last group contains all jobs in $A^{\lrg}$.) 
We increase the size of every job in groups $2,\ldots,k-1$ to
the size of the largest job in that group, and move each job in the first group to a 
separate extra machine. 
This creates at most $\ve^2|A^{\lrg}|\leq\ve\cdot\mcnum$ extra machines. 

Let $A'$ denote the jobs in groups $2,\ldots,k$ with their (potentially) modified sizes.
Note that jobs in $A'$ now have at most $\frac{2}{\ve^2}$ distinct sizes.
Let $\tp_j\geq p_j$ denote the modified size of every $j\in A^{\lrg}$ (where $\tp_j=p_j$
if $j$'s size is unchanged). 
We will work with these modified sizes in order to assign the jobs in $A'$, and then
revert to the original job sizes.
For all $r=2,\ldots,k-1$, the space used by jobs in group $r-1$ under assignment $\tsg$
can be used to accommodate the jobs in group $r$ with their modified sizes. 
Thus, $\tsg$ can be used to obtain a feasible assignment 
$A_1\cup A'\cup A^{\sml}\mapsto[m]$, where jobs in $A_1$ are assigned as before to
the first $|A_1|$ machines and 
jobs in $A'$ %(with their modified sizes) 
and $A^{\sml}$ are assigned to machines in
$\mcset$, and the load on each machine $i$ (even under the modified sizes) is at most
$\tlvec_i$. This may involve changing the $\tsg$-assignments of some jobs in $A'$; to
avoid excessive notation, we continue to use $\tsg$, viewed now as an assignment 
$A_1\cup A'\cup A^{\sml}\mapsto[m]$, to denote this modified feasible assignment. 
%denote the resulting modified assignment of jobs in $(J_I-\tJ_I)\cup J'_I$.

\begin{remark}
Given the constant number of job sizes for jobs in $A'$, it is possible to argue that
by considering a {\em polynomial} number of candidate assignments of jobs in $A'$ to
machines in $\mcset$, one can find an assignment that is consistent with $\tsg$. 
\begin{comment}
This is because $\tsg$ assigns any machine $i\in\mcset$ at most $\frac{2}{\ve}$ jobs from
$A'$, since $\tlvec_i\leq 2\cdot\tlvavg$ and $\tp_j\geq\ve\tlvavg$ for every job $j\in
A'$. Thus,  
\end{comment} 
This is because there are at most a constant number
($\bigl(\frac{2}{\ve}\bigr)^{2/\ve^2}$) of possible $A'$-configurations, where an
$A'$-configuration specifies how many jobs in $A'$ of each size are assigned to a machine
(in $\mcset$) so that the total load assigned is at most $2\cdot\tlvavg$. 
%is at most a constant ($\bigl(\frac{2}{\ve}\bigr)^{2/\ve^2}$). 
To specify the $A'\mapsto\mcset$ assignment of
jobs, we can guess, for each $A'$-configuration, how many machines in $\mcset$ are
assigned that configuration; thus, the number of such candidate assignments is at most
$m^{O(1)}$ (where the $O(1)$ term in the exponent is a function of $\frac{1}{\ve}$). 
While this will suffice to obtain a PTAS here, %in order to obtain an efficient PTAS, and
with an eye towards extending things to the setting of related machines (where we can have
up to $O(\log m)$ speed classes and the above enumeration idea will not work), we proceed
differently and reduce the number of $A'$-job assignments that we need to consider to
a {\em constant} (depending on $\ve$).
%time needed to determine the $A'$-job
%assignments to an expression of the form $g(1/\ve)\cdot\poly(m,n)$.
\end{remark}

\item {\bf Sparsifying job rewards.} \label{rewdsparse}
We next use a similar shifting idea to reduce the number of distinct rewards of jobs in
$A'$ to $O\bigl(\frac{1}{\ve^3}\bigr)$. We consider jobs in $A'$ in non-increasing order
of {\em reward} and now divide them into $k'=1+\frac{1}{\ve^3}$ groups, where the first
$\frac{1}{\ve^3}$ groups contain $\floor{\ve^3|A'|}$ jobs, and the last group contains
less than $\frac{1}{\ve^3}$ jobs. (Again, if $\ve^3|A'|<1$, then the last group is all of
$A'$ and all other groups are empty.)
For every $r=\frac{1}{\ve}+1,\ldots,k'-1$, we reduce the reward of each job in group
$r$ to the smallest reward of a job in that group. We also move each job in the first
$\frac{1}{\ve}$ groups to a separate extra machine.
Let $A''\sse A'$ denote the jobs in groups $\frac{1}{\ve}+1,\ldots,k'$. Let
$\trewd_j\leq\rewd_j$ denote the modified reward of each job $j\in A^{\lrg}$, where
$\trewd_j=\rewd_j$ if $j$'s reward is unchanged.

Observe that we create at most $\ve^2|A'|\leq\ve\cdot\mcnum$ extra machines, and jobs in
$A''$ now have at most $\frac{2}{\ve^3}$ distinct rewards.
Moreover, the following claim shows that this reward-sparsification step does not incur
much loss.

\begin{claim} \label{rwdsparse}
We have $\trewd(A')\geq(1-\ve)\rewd(A')$.
\end{claim} 

\item {\bf Assigning jobs in \boldmath $A''$.} \label{lrgasgn} \label{denseasgn}
At this point, we have assigned jobs in $A^{\lrg}-A''$ to at most $2\ve\cdot\mcnum$ extra
machines, one job per machine. Our next task is to find an assignment of jobs in $A''$ to
machines in $\mcset$ that is consistent with (the unknown assignment) $\tsg$.
By steps~\ref{sizesparse} and~\ref{rewdsparse}, jobs in $A''$ correspond to at most
$\frac{4}{\ve^5}$ distinct $(\trewd_j,\tp_j)$ tuples; we call $(\trewd_j,\tp_j)$ the 
{\em type} of job $j$. (Note that jobs of the same type are %essentially
indistinguishable.)  

We know that every machine in $\mcset$ is assigned load at most $2\cdot\tlvavg$ by $\tsg$.
A {\em type configuration} specifies an assignment of jobs in $A''$ to a
machine such that the total load assigned to it due to these jobs is at
most $2\cdot\tlvavg$, by listing out how many jobs of each type are assigned to
the machine. 
Let $\cfgset$ denote the collection of all type configurations.
Since $\tp_j\geq\ve\cdot\tlvavg$ for every $j\in A''$, we have 
%the number of type configurations is at most a constant 
$|\cfgset|\leq C=\bigl(\frac{2}{\ve}\bigr)^{4/\ve^5}$. Define the reward of a
type configuration $\cfg$ to be the reward obtained from the jobs assigned by that
configuration; abusing notation, we denote this by $\trewd(\cfg)$. 
We say that a machine $i$ uses a type configuration $\cfg$ %(or is assigned $\cfg$) 
to mean that the $A''$-jobs assigned to $i$ conform to $\cfg$, i.e., the
number of $A''$-jobs of each type assigned to $i$ is as specified by $\cfg$.
We say that $\tsg$ uses a type configuration $\cfg$ if it uses $\cfg$ for some
machine in $\mcset$. %is assigned $\cfg$ under $\tsg$.
%for a machine $i$, or assigns $\cfg$
%to machine $i$, if the $A''$-jobs assigned by $\tsg$ to $i$ conform to $\cfg$, i.e., the
%number of $A''$-jobs of each type assigned to $i$ is as specified by $\cfg$.

Ideally, we want to determine all the type configurations used by $\tsg$. 
%to assign jobs in $A''$. 
We will settle for finding a set of type configurations to assign to machines in $\mcset$
that accrue total reward %(from the jobs assigned by these type configurations)  
at least $\trewd(A'')-\ve\OPT$. 
For $\cfg\in\cfgset$, let $N_{\cfg}$ denote the number of times $\tsg$ uses
configuration $\cfg$ for machines in $\mcset$. 
%We call $\{N_{I,\cfg}\}_{\cfg\in\cfgset_I}$, the $\tsg$-configuration-usage sequence for
%$I$. 
We apply Lemma~\ref{enumlem} to guess the sequence 
$\bigl\{N_{\cfg}\cdot\trewd(\cfg)\}_{\cfg\in\cfgset}$ up to cumulative error
$\ve\cdot\optval$ by taking $\est=\optval$. 
%(This satisfies the requirements of Lemma~\ref{enumlem} since
%$\sum_{\cfg\in\cfgset_I}N_{\cfg}\trewd(\cfg)\leq\trewd(\bA_I)$.) 
This takes time $2^{O(C/\ve)}$ (which is a constant). 
Define $\num_{\cfg}$ to be the entry for configuration $\cfg$ in the correct guessed
sequence, divided by $\trewd(\cfg)$. 
Then, $\num_{\cfg}\leq N_{\cfg}$ for all $\cfg\in\cfgset$, and 
$\sum_{\cfg\in\cfgset}\num_{\cfg}\cdot\trewd(\cfg)\geq\sum_{\cfg\in\cfgset}N_{\cfg}\cdot\trewd(\cfg)-\ve\cdot\optval$.

\begin{comment}
To this end, for each type configuration $\cfg$, we 
estimate how many machines in $\mcset$ are assigned $\cfg$ under $\tsg$ by first guessing
the total reward $\Totrewd(\cfg)$ obtained from all uses of $\cfg$ by 
$\tsg$ up to the nearest multiple of $\Dt'=\frac{\ve\cdot\optval}{C}$, and then using this
to obtain the desired estimate.%
\footnote{This is similar to how, in \normknap, we estimated the number of items the
optimal solution picks from each reward bucket.} 
(Recall that $\optval$ is an estimate
such that $\optval\leq\OPT\leq(1+\ve)\optval$.) More precisely, 
%if $\Totrewd(\cfg)$ denotes the total reward from all uses of $\cfg$ by $\tsg$, 
%$\tsg$ uses $\cfg$ for some $m_{\cfg}$ machines in $\mcset$, 
we guess $\bigl\{\floor{\frac{\Totrewd(\cfg)}{\Dt'}}\bigr\}_{\cfg\in\cfgset}$. Since 
$\sum_{\cfg\in\cfgset}\floor{\frac{\Totrewd(\cfg)}{\Dt'}}\leq\bigl(1+\frac{1}{\ve}\bigr)C$,
there are at most $2^{O(C/\ve)}$ choices for the 
$\bigl\{\floor{\frac{\Totrewd(\cfg)}{\Dt'}}\bigr\}_{\cfg\in\cfgset}$ values, so 
we may assume that we know $\floor{\frac{\Totrewd(\cfg)}{\Dt'}}$ for all
$\cfg\in\cfgset$. For $\cfg\in\cfgset$, define 
$\num_{\cfg}:=\floor{\frac{\Totrewd(\cfg)}{\Dt'}}\cdot\frac{\Dt'}{\trewd(\cfg)}$, and note
that $\tsg$ uses $\cfg$ for at least $\num_{\cfg}$ machines (since
$\frac{\Totrewd(\cfg)}{\trewd(\cfg)}\geq\num_{\cfg}$).
\end{comment}

For every $\cfg\in\cfgset$, we choose $\ceil{\num_{\cfg}}$ distinct machines from
$\mcset$, and assign these machines the type configuration $\cfg$, where when we pick
$A''$-jobs as specified by $\cfg$ to assign to a machine, we of course always pick from
the unassigned jobs in $A''$.
%Assigning $\cfg$ to a
%machine $i$ means that, from the unassigned jobs in $A''$, for each job type, we select
%the number of jobs of that type specified by $\cfg$ and assign these jobs to $i$. 
If we run out of machines while doing so, i.e.,
$\sum_{\cfg\in\cfgset}\ceil{\num_{\cfg}}>\mcnum$, then we declare failure. Similarly, if
we run out of jobs of a particular type, then we again declare
failure. These failure events can only happen if one of our guesses is incorrect. 

\item {\bf Assigning jobs in \boldmath $A^{\sml}$.} \label{smlasgn} \label{ptas-end}
Let $J_1=A_1\cup A''$, and $\sg$ be the assignment determined above for jobs in $J_1$. We
extend $\sg$ to a fractional assignment $x^*$ of jobs in $A^{\sml}$ to machines in $\mcset$
using Lemma~\ref{extnlem}, taking $J_2=A^{\sml}$, and the original $p_j$ job sizes.
Let $L^*$ be the resulting load vector.
%As in the \sparse case, since $\sg$ is consistent with $\tsg$, we know that $f(L^*)\leq
%B$.

As with a \sparse instance, we round $x^*$ using \gap rounding. Let 
$\hsg:J_1\cup A^{\sml}\mapsto[m]$ be the assignment obtained by concatenating the resulting
integer solution and $\sg$. Let $\hJ_i\sse J_1\cup A^{\sml}$ be the jobs assigned to machine
$i$ by $\hsg$. Again, we know that for each $i\in\mcset$, there is at most one job
$j^*_i\in \hJ_i\cap A^{\sml}$ with $x^*_{ij}>0$ such that 
$p\bigl(\hJ_i-\{j^*_i\}\bigr)\leq L^*_i$;
as always, we set $\{j^*_i\}=\es$ if $p(\hJ_i)\leq L^*_i$.

We cannot discard the jobs in $\bigcup_{i\in\mcset}\{j^*_i\}$, as these jobs may bring in large
reward. Instead, we create extra machines and pack the jobs in
$\bigcup_{i\in\mcset}\{j^*_i\}$ arbitrarily on these extra machines so that each extra
machine is packed maximally within capacity $\tlvavg/2$.  
This creates at most $1+\floor{\frac{\ve\cdot\tlvavg\cdot\mcnum}{(1/2-\ve)\tlvavg}}$ extra machines,
since every extra machine, save for at most one, has at least 
$\tlvavg/2-\ve\cdot\tlvavg$ load assigned to it, and 
$\sum_{i\in\mcset}\tp_{j^*_i}\leq\ve\cdot\tlvavg\cdot\mcnum$.
Since we have a \dense instance, we have $\ve\cdot\mcnum\geq 1$, so the number of extra
machines created this way is bounded by $5\ve\cdot\mcnum$ (recall that $\ve\leq 0.25$). 
Combined with the $2\ve\cdot\mcnum$ extra machines utilized for jobs in $A^{\lrg}-A''$, we
have created at most $7\ve\cdot\mcnum$ extra machines.

Consider all the machines used for jobs in $A_2=A^{\lrg}\cup A^{\sml}$, i.e., the
regular machines in $\mcset$ and the extra machines. 
We retain the $\mcnum$ largest-reward machines from this collection, and discard the
rest. 
The final assignment is the assignment given by $\sg$ for machines in $[|A_1|]$ together 
with this postprocessed $\hsg$ assignment, for machines in $\mcset$.
\end{enumerate}

\subsubsection{Analysis} \label{ident-analysis}
In the sequel, we assume that we have found, by enumeration, the correct information in
steps~\ref{ptas-start}--\ref{ptas-end}. 
We show that the above algorithm is a PTAS.
%an efficient PTAS, that is, for any $\ve>0$, it
%runs in time $g(\ve)\cdot\poly(m,n)$, and returns a solution with reward $(1-\ve)\OPT$. 

\begin{theorem} \label{ident-ptasthm}
The algorithm described in
steps~\ref{ptas-start}--\ref{ptas-end} is a PTAS for \normsched on identical
machines. 
\end{theorem}

We defer the proofs of Lemmas~\ref{enumlem} and~\ref{extnlem} to
Appendix~\ref{append-relmc}, and begin by proving Claim~\ref{tsgclaim}.
%We begin by proving Lemma~\ref{extnlem} and Claim~\ref{tsgclaim}. 
The following well-known and easy-to-see fact %(see, e.g.,~\cite{ChakrabartyS18}) 
will be useful; we prove this in Appendix~\ref{append-relmc} as well.  

\begin{claim} \label{schur}
Let $\gnorm:\R^m\mapsto\Rp$ be a monotone, symmetric norm. 
Let $v\in\Rp^m$, and 
$i,i'\in[m]$ be such that $v_{i}<v_{i'}$ and $0<\kp<v_{i'}-v_{i}$. Let $u$ be
the vector where $u_\ell=v_\ell$ for all $\ell\in[m]-\{i,i'\}$, $u_{i}\leq v_{i}+\kp$,
$u_{i'}=v_{i'}-\kp$. Then, $u^{\down}$ is lexicographically smaller than $v^{\down}$, and
$\gnorm(u)\leq\gnorm(v)$.
\end{claim}

\begin{proofof}{Claim~\ref{tsgclaim}}
Recall that $\tlvec=\lvec{\tsg}$ and $\tlvec^{\down}$ is lexicographically smallest
among the sorted-load vectors of all feasible assignments of jobs in $A$ to the $m$
machines.  

For part (a), consider any $j\in A_1$. Let $i=\sg(j)$. Suppose $i$ is assigned some job
other than $j$ by $\tsg$. Let $j'$ be some such job with $p_{j'}\leq p_j$. 
Note that $\tlvec_i>\tlvmin$ and $\tlvec_i-\tlvmin\geq p_j+p_{j'}-\tlvmin>p_{j'}$.
Consider the load vector $L$ that results when we move job $j'$ to a machine with load
$\tlvmin$. By Claim~\ref{schur}, we obtain that $f(L)\leq f(\tlvec)$ and $L^{\down}$ is
lexicographically smaller than $\tlvec^{\down}$, which gives a contradiction.

Part (b) follows from a similar argument. Suppose $\tlvmax>2\cdot\tlvmin$. Let
$i,i'\in\mcset$ be such that $\tlvmax=\tlvec_i$ and $\tlvmin-=\tlvec_{i'}$. Since
$i\in\mcset$, it must be that $i$ is assigned at least two jobs, otherwise the job
assigned to it would lie in $A_1$. Consider any job $j$ with $\tsg(j)=i$. We have
$\tlvec_i-p_j>\tlvmin$ since $\tlvec_i>2\cdot\tlvmin$ and $p_j\leq\tlvmin$. So, again, if
we transfer job $j$ from $i$ to $i'$, we still obtain a feasible assignment whose sorted
load vector is lexicographically smaller than $\tlvec^{\down}$, yielding a contradiction.
\end{proofof}

When the instance is \sparse, the proof of the performance guarantee is
fairly straightforward, as alluded to when describing the algorithm.

\begin{lemma} \label{spresult}
If the instance is \sparse, then the assignment returned in step~\ref{sp-end} is feasible
and obtains reward at least $(1-\ve)^3\OPT$.
\end{lemma}

\begin{proof}
We have $\rewd(A)\geq(1-\ve)^2\OPT$, as argued in Appendix~\ref{ident-biptas}.
As noted earlier, if $(x^*,L^*)$ is the solution obtained in step~\ref{sp-frac}, we have
$f(L^*)\leq B$, since $\tsg$ yields one potential feasible solution to \eqref{extncp}. The
load vector of the final assignment is coordinate-wise at most $L^*$ by design, so
feasibility follows. Recall that in steps~\ref{sp-start}--\ref{sp-end}, $J_1$ consists of
$A_1$ and the $\frac{\mcnum}{\ve}$ largest-reward jobs in $A_2$, 
$J_2=A_2-J_1$, and $J_1\cup J_2=A=A_1\cup A_2$. 
Also, $\rewd_j\leq\frac{\ve}{\mcnum}\rewd(A_2)$ for all $j\in J_2$.
The reward obtained is $\rewd(A)-\mcnum\cdot\max_{j\in J_2}\rewd_j$ since we discard at
most $\mcnum$ jobs from $J_2$ in step~\ref{sp-end}. This is at least 
$\rewd(A)-\ve\cdot\rewd(A_2)\geq(1-\ve)\rewd(A)\geq(1-\ve)^3\OPT$.
\end{proof}

When the instance is \dense, the analysis is somewhat more involved. Lemma~\ref{denserewd}
bounds the reward obtained by the final assignment, and Lemma~\ref{densefeas} shows that
the assignment returned is feasible. We first prove Claim~\ref{rwdsparse}.

\begin{proofof}{Claim~\ref{rwdsparse}}
If $\ve^3|A'|<1$, then the rewards do not change, so the claim trivially holds. So suppose
otherwise. Let $B_r$ denote the jobs in group $r$, for $r=1,\ldots,k'$, where recall that
$k'=1+\frac{1}{\ve^3}$, groups $B_1,\ldots,B_{k'-1}$ contain $\floor{\ve^3|A'|}$ jobs and
$B_{k'}|<\frac{1}{\ve^3}$. 

We have $\trewd_j=\rewd_j$ for all $j\in\bigl(\bigcup_{r=1}^{1/\ve}B_r\bigr)\cup B_{k'}$. 
For any $r\in[k'-1]$, and any jobs $j\in B_r$, $j'\in B_{r+1}$, we have
$\rewd_j\geq\rewd_{j'}$. 
So for $r\in\bigl\{\frac{1}{\ve}+1,\ldots,k'-2\bigr\}$, since 
%$\rewd_j\geq\red_{j'}$ for any $j\in B_r$, $j'\in B_{r+1}$, and
$|B_r|=|B_{r+1}|$, we have $\trewd(B_r)\geq\rewd(B_{r+1})$.
Also, 
$\rewd\bigl(B_{1/\ve+1}\bigr)
\leq\bigl|B_{1/\ve+1}\bigr|\cdot\frac{\sum_{r=1}^{1/\ve}\rewd(B_r)}{\sum_{r=1}^{1/\ve}|B_r|}
\leq\ve\cdot\rewd(A')$.
Putting things together, we have
\begin{equation*}
\begin{split}
\trewd(A')=\sum_{r=1}^{k'}\trewd(B_r)
& \geq\sum_{r=1}^{1/\ve}\rewd(B_r)+\sum_{r=1/\ve+1}^{k'-2}\rewd(B_{r+1})+\rewd(B_{k'})
\\ & =\rewd(A')-\rewd(B_{1/\ve+1})\geq(1-\ve)\rewd(A'). \qedhere
\end{split}
\end{equation*}
\end{proofof}

\begin{lemma} \label{denserewd}
Suppose the instance is \dense. Then the reward obtained by the final assignment is at
least $\bigl(1-O(\ve)\bigr)\OPT$. 
\end{lemma}

\begin{proof}
First, by Lemma~\ref{enumlem}, the total reward obtained from the assignment of
type configurations to machines in $\mcset$ computed in step~\ref{lrgasgn} is at least
\[
\sum_{\cfg\in\cfgset}\num_{\cfg}\cdot\trewd(\cfg)
\geq\sum_{\cfg\in\cfgset}N_{\cfg}\cdot\trewd(\cfg)-\ve\optval
\geq\trewd(A'')-\ve\OPT
\]
where the final inequality is because $\sum_{\cfg\in\cfgset}N_{\cfg}\cdot\trewd(\cfg)=\trewd(A'')$, 
by definition.  
So the total reward obtained from the assignment $\hsg$ computed in step~\ref{smlasgn} from
the machines in $\mcset$ and the extra machines 
%before postprocessing it to get rid of extra machines, 
is at least 
$\rewd\bigl(A^{\lrg}-A''\bigr)+\bigl(\trewd(A'')-\ve\OPT\bigr)+\rewd(A^{\sml})$, which is
at least
\begin{equation*}
\begin{split}
\rewd\bigl(A^{\lrg}-A'\bigr)&+\trewd(A'-A'')+\trewd(A'')+\rewd\bigl(A^{\sml}\bigr)-\ve\OPT
\\
& \geq \rewd\bigl(A^{\lrg}-A'\bigr)+(1-\ve)\rewd(A')+\rewd\bigl(A^{\sml}\bigr)-\ve\OPT
\geq (1-\ve)\rewd(A_2)-\ve\OPT.
\end{split}
\end{equation*}
where the first inequality is due to Claim~\ref{rwdsparse} and since $\trewd_j=\rewd_j$
for all $j\in A^{\lrg}-A''$.
So after postprocessing $\hsg$, the total reward obtained is at least 
\begin{equation*}
\begin{split}
\rewd(A_1)+\tfrac{1}{1+7\ve}\cdot\Bigl((1-\ve)\rewd(A_2)-\ve\OPT\Bigr)
& \geq(1-7\ve)(1-\ve)\rewd(A)-\ve\OPT \\
& \geq\Bigl((1-7\ve)(1-\ve)^3-\ve\Bigr)\OPT\geq(1-11\ve)\OPT. \qedhere
\end{split}
\end{equation*}
\end{proof}

\begin{lemma} \label{densefeas}
Suppose the instance is \dense. Then the final assignment computed in step~\ref{ptas-end}
is feasible.
\end{lemma}

\begin{proof}
Recall that in step~\ref{smlasgn}, we invoke Lemma~\ref{extnlem} taking $J_1=A_1\cup A''$
and $J_2=A^{\sml}$. Since the assignment $\sg$ computed for $J_1$ is consistent with
$\tsg$, it follows that $\tsg$ yields a feasible solution to \eqref{extncp}, and so the
solution  
$(x^*,L^*)$ obtained from Lemma~\ref{extnlem} satisfies $f(L^*)\leq B$. Let $z^*$ be the
quantity computed in Lemma~\ref{extnlem}; %which determines $L^*$; 
so we have %$\sum_{i\in\mcset}\bigl(z^*-p(\sg^{-1}(i)\bigr)^+=p(J_2)$ and
$L^*_i=\max\bigl\{p\bigl(\sg^{-1}(i)\bigr),z^*\bigr\}$ for all $i\in[m]$.

Let $\bL$ be the load vector given by $\bL_i=\max\{L^*_i,\tlvmin\}$ for all
$i\in\mcset$, and $\bL_i=L^*_i$ for all other $i$. 
We first show that we also have $f(\bL)\leq B$. 
%and that $\hL^{\down}\leq\bL^{\down}$, which yields $f(\hL)\leq f(\bL)\leq B$. 
If $z^*\geq\tlvmin$, then $\bL=L^*$ since $L^*_i\geq z^*$ for all $i\in\mcset$, so this
holds. So suppose $z^*<\tlvmin$. We claim then that $\bL_i\leq\tlvec_i$ for all $i\in[m]$.
This certainly holds for all $i\in[|A_1|]$. Consider $i\in\mcset$. 
We have $\bL_i=\max\bigl\{p\bigl(\sg^{-1}(i)\bigr),\tlvmin\bigr\}$ and since $\sg$ is
consistent with $\tsg$, we have $p\bigl(\sg^{-1}(i)\bigr)\leq\tlvec_i$. It follows that
$\bL_i\leq\tlvec_i$.

Let $\hL$ be the load vector corresponding to the final assignment, under the $\{p_j\}$
job sizes. For every $i\in[|A_1|]$, we have $\hL_i=L^*_i=\bL_i$. For any ``regular'' machine
in $\mcset$, we ensure by design that the total load on it is at most $L^*_i$. Each extra
machine is assigned exactly one job in $A_2$, and so has load at most $\tlvmin$. Thus, for
any combination of $\mcnum$ machines chosen in step~\ref{smlasgn} while postprocessing
$\hsg$, comprising some regular machines and some extra machines, we can bound the load on
each machine by a distinct $\bL_i$ term. Therefore, $\hL^{\down}\leq\bL^{\down}$, and so
$f(\hL)\leq f(\bL)\leq B$.
\end{proof}

\begin{proofof}{Theorem~\ref{ident-ptasthm}}
The time required for enumeration in step~\ref{goodset} is
$\bigl(\frac{n}{\ve}\bigr)^{O(1/\ve)}$. 
In all other steps that involve some enumeration---steps~\ref{tguess}, \ref{sp-start},
\ref{denseasgn}---the time required for enumeration is bounded by $g(1/\ve)\cdot\poly(m,n)$, 
for some function $g$. So the running time is polynomially bounded for any fixed $\ve>0$.

Feasibility of the solution returned and the performance guarantee follow directly from
Lemma~\ref{spresult} for a \sparse instance, and Lemmas~\ref{denserewd}
and~\ref{densefeas} for a \dense instance.
\end{proofof}

\subsection{Related machines} \label{rel-ptas}
%\input{relatedptas}
%In this section we develop a PTAS for \normsched on related machines.
%We remark that we only use the monotonicity, symmetry, and convexity of $f:\Rp^m\to\Rp$
%in our proof.
\paragraph{An overview.}
The PTAS for \normsched on related machines builds upon the ideas underlying the PTAS for
identical machines, but also needs several new ingredients. At a high level, we eventually
group machines into $O\bigl(\frac{\log m}{\ve}\bigr)$ groups so that, roughly speaking, we
can treat each group as an identical-machines instance, 
%to eventually obtain a collection of $O(\log m)$
%identical-machine instances, 
and extend suitable ideas from the approach used in the PTAS for identical
machines to solve these instances. In contrast to much extant work on related
machines, as also non-uniform \mkp, this grouping of machines is 
{\em not based on machine speed}, but instead is determined by the 
{\em total work assigned to the machine}. 
Throughout this section, when we say ``work'' on a machine, we mean the total
processing time of jobs assigned to the machine. 
For an assignment $\sg:S\mapsto[m]$ of a set $S\sse J$ of jobs, the {\em work-vector}
$\wvec{\sg}$ resulting from $\sg$ is therefore given by 
$\wvec{\sg}_i:=\sum_{j\in S:\sg(j)=i}p_j$, for all $i\in[m]$. 
Note that the load of machine $i$ under $\sg$ is $\lvec{\sg}_i=\wvec{\sg}_i/\spd_i$.
%the load of a machine $i$ is therefore
%(work assigned to $i$)/$\spd_i$. %it divided by its speed. 
%As we shall see this grouping based on the machine work-vector offers several
%benefits. 
%One can argue that a work-vector sorted similarly as the speed

We first prove some structural properties about a near-optimal solution. It is easy to
argue that one may assume that the work-vector is sorted similarly as the speed vector
(Lemma~\ref{wksort}): 
given any assignment of jobs to machines, if we permute the assignment so that
%for any pair of machines $i,i'\in[m]$ $s_i>s_{i'}$ $i$ is assigned at least as much work
%as $i'$, then 
the work assigned to machines is non-decreasing 
%(more precisely, does not increase) 
with their speed, 
%a higher-speed machine is assigned at least as much work than a lower-speed machine, 
then this does not increase the norm of the load vector. 
%This is indeed one of the benefits of
%considering work-vectors, namely that from the work-vector one can easily determine the
%corresponding assignment of jobs to machines.
Sort the machines so that $\spd_1\geq\spd_2\geq\ldots\geq\spd_m$.
We create $\ngrp=O\bigl(\frac{\log m}{\ve}\bigr)$ machine groups, where the zeroth group,
called the {\em fast machines}, consists of the first $m_0=m_{\fast}=\frac{2}{\ve^3}$
machines, the first group consists of the next $m_1=\frac{1}{\ve^2}$ machines, and every
subsequent group $r>1$ consists (roughly) of the next $m_{r-1}/(1-\ve)$ machines. %so that
%\frac{1}{\ve^2(1-\ve)}$ 
%and so on, with 
%The sizes of groups $1,2,\ldots$ increase geometrically. 
We argue that there is a near-optimal solution $\tsg$ %assignment $\tsg:A\mapsto[m]$ 
whose work-vector is bounded by a nicely
{\em smoothed-out} vector $\wksmooth$ that: (1) induces a feasible load vector (i.e., the
load vector has norm at most $B$); and (2) except for the fast machines, assigns all
machines in the same group the same work (Lemma~\ref{strucprop}). 
%except for the fast machines, all machines in the same group
%any group $\ell\in[\numgrp]$ have the same work

We guess the set of jobs assigned by $\tsg$, %this near-optimal solution, 
proceeding exactly as in
step~\ref{goodset} of the PTAS for identical machines.
We handle the fast machines, and the remaining machines separately. Since there are only a
constant number of fast machines, after sparsifying (the guessed set of) jobs into 
$O(\log n)$ job types  
using the same shifting idea used for identical machines, we use enumeration ideas
to find an assignment of jobs to the fast machines consistent with $\tsg$, 
%the near-optimal solution, 
and obtaining almost all the reward accrued by (jobs assigned to) the fast
machines under $\tsg$. %the near-optimal solution.

For the remaining machines, we guess a non-increasing vector
$\wkestim\in\Rp^{[m]-[m_{\fast}]}$ that coordinate-wise estimates
$(\wksmooth_i)_{i\in[m]-[m_{\fast}]}$ within a $(1+\ve)$-factor.
%the coordinates of $\{\wksmooth_i\}_{i\in[m]-[m_{\fast}]}$ within a $(1+\ve)$-factor. 
We can think of $\wkestim_i$ as the (work) {\em capacity} of machine $i$. But (as with
identical machines) because these are only estimates, simply solving a non-uniform \mkp
instance with these capacities will not work: we may violate the norm budget, if $\wkestim$
overestimates $\wksmooth$, or get low reward, if $\wkestim$ underestimates $\wksmooth$.
%\footnote{If $\wkestim$ overestimates $\wksmooth$, the resulting
%assignment may violate the norm budget; if it underestimates $\wksmooth$, then 
%we may lose out on the reward.} 
So we need to proceed more carefully. We choose $\wkestim$ to be an overestimate of
$\wksmooth$. We call a maximal (consecutive) set of machines
having the same $\wkestim_i$ value a machine class. 
Importantly, because of our initial grouping of machines, 
{\em every machine class $I$ is \dense in that we have $|I|\geq\frac{1}{\ve^2}$.}
For a machine class $I$ with capacity $W$, 
we again use enumeration to guess the set of jobs with $p_j>\ve W$ assigned to
machines in $I$; we call such jobs ``large'' for machine class $I$. Once we know the set
of large jobs for class $I$ that are assigned to machines in $I$, we use the
configuration-enumeration approach in step~\ref{lrgasgn} of the PTAS for identical
machines to find the actual assignment of these large jobs to machines in $I$ that
is consistent with $\tsg$. %our near-optimal solution. 
Recall that this enumeration takes constant time for a single machine class, 
so since we have $O(\log m)$ machine classes, we can
obtain these assignments for {\em all} machine classes in polynomial time. As with
identical machines, this step may entail creating ``extra machines'' for a machine class
$I$; but since $|I|$ is sufficiently large, the number of such extra machines will be at
most $O(\ve)|I|$, and so we will be able to discard the extra machines at the end without
sacrificing the reward by much. 

Finally, for the remaining jobs, which are assigned as small jobs for a machine class, we
write an LP (see \eqref{extnlp}) to find a fractional assignment respecting the $\wkestim$
capacities, %this constraint, 
and use \gap rounding to round the LP solution. 
For a machine class with capacity $W$, from each machine in that class, we transfer
roughly $\ve\cdot W$ work due to small jobs assigned to that machine, to extra machines
for that class. We argue that this can be done while creating an additional $O(\ve)|I|$
extra machines, and so that any load vector formed
by taking, for each machine class $I$, any collection of $|I|$ machines from among the
(regular and extra) machines used for that class, is feasible. 
Since the number of extra machines for class $I$ is $O(\ve)|I|$, if we
take the $|I|$ highest-reward machines for class $I$, we only lose a
$\bigl(1-O(\ve)\bigr)$-factor in the reward. This yields our PTAS.

%In Section~\ref{section:structure}, we make structural observations about near-optimal
%solutions.  
%In particular, Claim~\ref{claim:sortedwork} shows that, without loss of generality, we may
%assume that in any optimal assignment, each machine is assigned at least as much total
%processing time as any slower machine. 
%Then, we show in Lemma~\ref{lemma:capacities} that there exists \david{...}
%In Section~\ref{related-detail} we give a precise description of the algorithm.

\paragraph{Preliminaries.}
Recall that we index the machines so that $\spd_1\geq\spd_2\geq\ldots\geq\spd_m>0$. 
%(We may sasume that $\spd_m>0$, as we can simply discard zero-speed machines.)
We assume that $\frac{1}{\ve}$ is an integer and $\ve\leq 0.25$.
Define $m_{\fast}=m_0=\frac{2}{\ve^3}$, $m_1=\frac{1}{\ve^2}$, and for $r>1$, define
$m_r=\floor{\frac{m_{r-1}}{1-\ve}}$. 
Let $\Mc_{\fast}=\Mc_0$ be the first $m_\fast$ machines, $\Mc_1$ be the next $m_1$
machines, $\Mc_2$ be the next $m_2$ machines, and so on, until we exhaust all machines.  
More precisely, let $\ngrp$ be the smallest index $r\geq 0$ such that
$\sum_{\ell=0}^r m_\ell\geq m$.
Then, for $r=0,\ldots,\ngrp-1$, define $\Mc_r$ to be the machines in 
$\bigl[\sum_{\ell=0}^{r}m_\ell\bigr]-\bigl[\sum_{\ell=0}^{r-1}m_\ell\bigr]=
\bigl\{\sum_{\ell=0}^{r-1}m_\ell+1,\sum_{\ell=0}^{r-1}m_\ell+2,\ldots,\sum_{\ell=0}^{r}m_\ell\bigr\}$;
let $\Mc_{\ngrp}=[m]-\bigl[\sum_{\ell=0}^{\ngrp-1}m_\ell\bigr]$ be the remaining
machines. Clearly, $|\Mc_r|=m_r$ for all $r=0,1,\ldots,\ngrp-1$, and 
$|\Mc_{\ngrp}|\leq m_{\ngrp}$. We also have $\ngrp=O\bigl(\frac{\log m}{\ve}\bigr)$.
%\footnote{For $r=2\ln_{1/(1-\ve)}m+1\leq\frac{\4\ln m}{\ve}+1$, we have 
%$m_r\geq\frac{1}{\ve^2}(1-\ve)^{r-1}-(r-1)\geq\frac{m^2}{\ve^2}-(r-1)\geq m$.
(Recall that $\Mc_{\fast}=\Mc_0$.)

Let $i^*=m_{\fast}+1$ be the first machine not in $\Mc_{\fast}$; so $i^*=m+1$ indicates that
$\Mc_{\fast}=[m]$. 
For notational convenience, for any vector $v\in\Rp^m$ and index $i>m$, we define
$v_i:=0$.

\begin{comment}
Throughout the rest of the section we consider the following partition $[m]=M^{\fast}\cup M_1\cup\hdots\cup M_L$ of the machines.
Let $m^\fast=|M^\fast|$ and $m_\ell=|M_\ell|$ for each $\ell\in[L]$.
Let $m^\fast=\tfrac{2}{\ve^3}$, $m_\ell=\lceil\tfrac{1}{\ve^2}\cdot(1+\ve)^\ell\rceil$ for each $\ell\in[L-1]$,
and $\lceil\tfrac{1}{\ve^2}\cdot(1+\ve)^L\rceil<m_L\le \lceil\tfrac{1}{\ve^2}\cdot(1+\ve)^L\rceil$.
Observe that $L\le \tfrac{1}{\ve}\ln{m}$.
Here,
$M^\fast=\{1,2,\hdots,m^\fast\}$ contains the first (eq. the fastest) $\tfrac{2}{\ve^3}$ machines.
$M_1=\{\,m^\fast+1,\hdots,m^\fast+m_1\,\}$ contains the next $m_1$ machines, and in general, $M_\ell=\{\,m^\fast+\sum_{\ell^\prime=0}^{\ell-1}m_{\ell^\prime}+1,\hdots,m^\fast+\sum_{\ell^\prime=0}^{\ell}m_{\ell^\prime}\,\}$ for each $\ell\in[L]$.
\end{comment}

\paragraph{Structural results about (near-) optimal solutions.} %\label{section:structure}
Lemma~\ref{wksort} shows that one may always assume that the work assigned to a machine
is non-decreasing in its speed. Lemma~\ref{strucprop} proves important structural
properties about a near-optimal solution, which drives our approach.

\begin{comment}
From now on we assume that the machines have been sorted in nonincreasing order of speeds.
That is, if $i,i^\prime\in[m]$ and $i<i^\prime$, then $s_i\ge s_{i^\prime}$.
Given an assignment $\sg:S\to[m]$ of a subset $S\subseteq J$ of jobs to the machines, we define the {\em work-vector} 
$\wvec{\sg}$ associated to $\sg$ by $\wvec{\sg}_i=\sum_{j\in S:\sg(j)=i}p_j$.
Observe that $\wvec{\sg}_i=s_i\cdot \lvec{\sg}_i$.
Consider the work-vector $\wvecopt$ of an optimal solution.
A simple, but very useful, observation is that we can assume that if $s_i\ge
s_{i^\prime}$, then $\wvecopt_i\ge \wvecopt_{i^\prime}$. In particular, this implies that
$\wvecopt^\down=\wvecopt$, since we sorted machines by nonincreasing order of speeds.
\end{comment} 

\begin{lemma}\label{claim:sortedwork} \label{wksort}
Let $\sg:S\to[m]$ be an assignment of a subset $S\sse J$ of jobs.
Let $\pi:[m]\to[m]$ be the permutation that sorts the coordinates of $\wvec{\sg}$ in
non-increasing order. 
That is, $\wvec{\pi\circ\sg}=\wvec{\sg}^{\down}$, which means that for all $i\in[m]$,
the $i$-th fastest machine is assigned the $i$-th largest work  
%for each $i\in[m]$ i.e., $i$ is the machine with the $i^{th}$ most assigned work 
under $\pi\circ\sg:S\to[m]$. 
Then $f(\lvec{\pi\circ\sg})\le f(\lvec{\sg})$.
\end{lemma}

\begin{restatable}[{\bf Structured near-optimal solution}]{lemma}{firststrucprop}
%\begin{lemma}[{\bf Structured near-optimal solution}]
\label{lemma:capacities} \label{strucprop}
There is a job-set $\noptset\sse J$, an assignment $\tsg:\noptset\mapsto[m]$, and a
work-vector $\wksmooth\in\Rp^m$ satisfying the following properties.
%Let $\pmin=\min_{j\in J}p_j$ and $\ptot=p(J)$.
%There exists a sequence
%$\theta_1\ge \theta_2\ge \hdots \geq\theta_{m^\fast}\ge \w_1^*\ge\hdots\ge \w^*_L=0$ of
%nonnegative integers,  
%and an assignment $\sg^\prime:A^\prime\to[m]$ of a subset of jobs $A^\prime\subseteq J$, such that

\noindent\,
\begin{enumerate*}[label=(\alph*), topsep=0.2ex, itemsep=0.1ex, leftmargin=*]
\item\label{item:reward}
$\rewd(\noptset)\ge (1-\ve)\OPT$; {\quad}

\item\label{item:capacities}
$\wvec{\tsg}\leq\wksmooth$; {\quad}

\item \label{sorted}
$\wksmooth={\wksmooth}^{\down}$; {\quad}
%for $i,i'\in[m]-\Mc_{\fast}$, we have $\wksmooth_i\geq\wksmooth_{i'}$ if $i<i'$; {\ \ }

%$\wvec{\sg^\prime}_i\le \theta_{i}$ for each $i\in M^\fast$, and
%$\wvec{\sg^\prime}_i\le \w^*_\ell$ for each $\ell\in[L]$, $i\in M_\ell$. 
\end{enumerate*}

\begin{enumerate}[label=(\alph*), start=4, topsep=0.1ex, itemsep=0.1ex, leftmargin=*]
\item\label{ngroup}
if $\ngrp\geq 1$, then $\wksmooth_{\ngrp}=0$;

\item\label{item:wm}
for $i\in[m]-\Mc_{\fast}$, if $\wksmooth_i>0$ then $\wksmooth_i\geq\wksmooth_{i^*}/m$; %{\quad}
%If $\w^*_1\ge m\cdot \w^*_\ell$ then $\w_{\ell}^*=0$.

\item\label{smooth}
for all $r\geq 1$, and any $i,i'\in\Mc_r$, we have $\wksmooth_i=\wksmooth_{i'}$; %{\quad}

\item\label{item:feasible}
$f(L)\le B$, where $L\in \Rp^{[m]}$ is given by $L_i=\wksmooth_i/\spd_i$ for all $i\in[m]$.
%$u_i=\frac{\theta_i}{s_i}$ if $i\in M^\fast$, and $u_i=\frac{\w^*_\ell}{s_i}$ if $i\in
%M_\ell$. (Since $\lvec{\sg^\prime}\le u$, the monotonicity of $f$ implies that
%$f(\lvec{\sg^\prime})\le B$.) 
\end{enumerate}
%\end{lemma}
\end{restatable}

\subsubsection{Algorithm details}\label{related-detail}
Let $\noptset\sse J$, $\tsg:\noptset\mapsto[m]$, and $\wksmooth\in\Rp^m$ be as given by
Lemma~\ref{strucprop}. 

%\david{Algorithm overview}

%\begin{claim*}[Restatement of Claim~\ref{polybnd}] There are at most $(2e)^{\max\{M,k\}}$
%sequences of $k$ nonnegative integers that sum to 
%at most $M$. The same bound applies to the number of non-increasing sequences of $k$
%integers chosen from $\dbrack{M}$.\end{claim*}

\vspace*{-1ex}
\paragraph{I: Finding a suitable set of jobs, sparsifying job rewards and sizes.}

\begin{enumerate}[label=(R\arabic*), topsep=0.5ex, itemsep=0.1ex, widest*=16, leftmargin=*]
\item \label{step:goodset} \label{rel-ptas-start}
We guess a set $A$ of jobs achieving large reward for which $\tsg$ can be used to obtain a feasible
assignment, by 
proceeding exactly as in step~\ref{goodset} of the PTAS for identical machines. 
%the near-optimal assignment $\sg^\prime:A^\prime\to[m]$ from Lemma~\ref{lemma:capacities}.
%Similar to Section~\ref{...},

%We consider reward buckets 
Let $\itemset_0,\ldots,\itemset_{\nbuck}$ be reward buckets,
where $\itemset_q=\bigl\{j\in J: \frac{\thresh_q}{1+\ve}<\rewd_j\leq\thresh_q\bigr\}$
%consists of all jobs with rewards in the range
%$\bigl(\frac{\thresh_q}{1+\ve},\thresh_q\bigr]$ 
with $\thresh_q=\frac{\rmax}{(1+\ve)^q}$ and
$\nbuck=O\bigl(\frac{1}{\ve}\log\frac{n}{\ve}\bigr)$.
We obtain an estimate $\optval\leq\OPT\leq(1+\ve)\optval$, and use this to 
%We can apply Lemma~\ref{enumlem} on the sequence 
obtain $\{\num_q\}_{q\in\dbrack{\nbuck}}$ estimates such that 
$\num_q\leq|\noptset\cap\itemset_q|\leq(1+\ve)\bigl(\num_q+\frac{\ve\cdot\optval}{\nbuck\cdot\thresh_q}\bigr)$
for all $q\in\dbrack{\nbuck}$.
\begin{comment}
\footnote{For instance, we can apply Lemma~\ref{enumlem} to the sequence
$\bigl\{\rewd(\noptset\cap\itemset_q)\bigr\}_{q\in\dbrack{\nbuck}}$ taking
$\est=\optval$ to estimate these rewards within an additive $\ve\optval$, and set $\num_q$
to be the estimate of $\rewd(\noptset\cap\itemset_q)$ divided by $\thresh_q$.
This is precisely what we do in the proof of Theorem~\ref{ident-biptasthm}.}
\end{comment}
%As argued in the proof of Theorem~\ref{ident-biptasthm}, if 
We pick the $\ceil{\num_q}$ smallest-size jobs from each $\itemset_q$ bucket. 

Let $A$ denote this set of jobs. 
We treat the reward of each job in $A\cap\itemset_q$ as $\frac{\thresh_q}{1+\ve}$.
Thus, we now have $\nbuck+1$ distinct rewards for jobs in $A$. 
%and with these  {\em modified rewards}, 
%
Since $\svec[p]{A}\leq\svec[p]{\noptset}$, by Claim~\ref{jobdom}, $\tsg$ can be used to
obtain a feasible assignment for the jobs in $A$.  
To keep notation simple, we continue to use $\tsg$ to denote this feasible assignment, and
$\rewd_j$ to denote the modified reward of $j$. 
Note that the (modified) reward of $A$ %satisfies    
%Again, as in the proof of Theorem~\ref{ident-biptasthm}
%We have
%$\rewd(A)\geq
is at least 
$\frac{1}{1+\ve}\cdot\bigl(\rewd(\noptset)-\ve\cdot\OPT\bigr)\geq (1-3\ve)\OPT$. 
%we have $\rewd(A)\geq\frac{1}{1+\ve}\cdot(1-3\ve)\OPT\geq(1-4\ve)\OPT$. 

\item \label{step:goodsettwo}  
{\bf Sparsifying job sizes.}
We next sparsify the job sizes using the shifting idea for bin packing.
%instance using a shifting idea used for bin packing and the multiple-knapsack
%problem~\cite{VegaL81,ChekuriK05} 
so that there are only 
$O\bigl(\frac{\nbuck}{\ve}\bigr)=O\bigl(\frac{1}{\ve^2}\log\frac{n}{\ve}\bigr)$ distinct 
job-sizes for jobs in $A$.
For each $q\in\dbrack{\nbuck}$, we consider jobs in $S_q=A\cap\itemset_q$ in
non-increasing order of size and group them into $k=1+\frac{1}{\ve}$ sets
$B_q^{(1)},\ldots,B_q^{(k)}$, where $|B_q^{(1)}|=\ldots=|B_q^{(k-1)}|=\floor{\ve|S_q|}$ and
$|B_q^{(k)}|<\frac{1}{\ve}$. (If $\ve|S_q|<1$, then $B_q^{(1)}=\ldots=B_q^{(k-1)}=\es$, and $B_q^{(k)}=S_q$.)
We drop the jobs in $B_q^{(1)}$, and for every $r=2,\ldots,k-1$, we set the size of every 
job in $B_q^{(r)}$ to be the size of the largest job in $B_q^{(r)}$. Let 
$T_q=B_q^{(2)}\cup\ldots\cup B_q^{(k)}$.
This creates at most $k-2+|B_q^{(k)}|\leq\frac{2}{\ve}$ distinct modified job sizes for
jobs in $T_q$. 

Let $A_1=\bigcup_{q\in\dbrack{\nbuck}}T_q$ be the jobs remaining in $A$.
To keep notation simple, we continue to use $p_j$'s to denote the modified sizes of jobs in
$A_1$.
Since $B_q^{(1)}\le\ve |S_q|$, it follows that $\rewd(A_1)\ge(1-\ve)\cdot\rewd(A)$.
Note that $\tsg$ can be used to obtain a feasible assignment of jobs in 
$A_1$ with their modified job sizes, whose resulting
work-vector %(resp. load-vector) 
is at most $\wvec{\tsg}$ %(resp. $\lvec{\tilde{\sg}}$)
(coordinate-wise), because for all $q\in\dbrack{\nbuck}$,
$r=2,\ldots,k$, we can use 
the space previously occupied (under the assignment $\tilde{\sg}$) by the jobs in
$B_q^{(r-1)}$ to accommodate the jobs in $B_q^{(r)}$ with their modified sizes. 
This may involve changing the $\tsg$-assignments of some jobs; again, to avoid
excessive notation, we continue to use $\tsg$, 
viewed now as an assignment $A_1\mapsto[m]$, 
to denote this modified feasible assignment.

%Therefore, from now on we assume that 
We now collect all jobs in $A_1$ with the same (size, reward) tuple, which we call the
type of a job into a bucket; note that jobs are indistinguishable in that one can be
replaced with another without affecting feasibility or reward.  
So at this point, we have a set of jobs $A_1\subseteq J$ partitioned as
$\bigcup_{q\in[\anbuck]}\abkt_q$ into $\anbuck$ type buckets, where 
$\anbuck\leq\frac{2}{\ve}\cdot\nbuck=O\bigl(\tfrac{1}{\ve^2}\cdot\log{\tfrac{n}{\ve}}\bigr)$,
such that:   
\begin{enumerate}[label=(\roman*), ref={\theenumi\,\roman*}, topsep=0.1ex, noitemsep, leftmargin=*]
\item each job in $\abkt_q$ has the same (size, reward) tuple, denoted sometimes
by ($p^{(q)}$, $\rewd^{(q)}$); 
\item $\rewd(A_1)\ge (1-\ve)\rewd(A)\geq(1-4\ve)\OPT$; \label{a1ineq}
\item there is an assignment $\tsg:A_1\to[m]$ satisfying $\wvec{\tsg}\leq\wksmooth$. 
%$\wvec{\tsg}_i\le \theta_i$ for each $i\in M^\fast$, and
%$\wvec{\tilde{\sg}}_i\le \w_\ell^*$ for each $\ell\in[L]$, $i\in M_\ell$, where
%$\theta_1\ge\hdots\ge\theta_{m^\fast}\ge\w_1^*\ge\hdots\ge \w_L^*=0$ are as in the
%statement of Lemma~\ref{lemma:capacities}. 
\end{enumerate}
\end{enumerate}

\vspace*{-1ex}
\paragraph{II: Machine capacities and classes, fast, large, and small jobs.}

\begin{enumerate}[resume*, start=3]
\item {\bf Guessing \boldmath $\wksmooth$.}
Recall that $\Mc_{\fast}$ consists of the first $m_{\fast}$ machines, and $i^*=m_{\fast}+1$.
We guess a non-increasing vector $\wkest\in\R^{[m]-\Mc_{\fast}}$, where each $\wkest_i$ is $0$
or a power of $(1+\ve)$, such that for all $i\in[m]-\Mc_{\fast}$, we have 
$\wksmooth_i\leq\wkest_i<(1+\ve)\wksmooth_i$ if $\wkest_i>0$ and $\wkest_i=0$ otherwise.
Note that for $r\geq 1$, all machines in $\Mc_r$ have the same
$\wkest_i$ guess, since they have the same $\wksmooth_i$ value (Lemma~\ref{strucprop}\,\ref{smooth}).
We can do this in polynomial time since $(\wksmooth_i)_{i\in[m]-\Mc_{\fast}}$ is also a
non-increasing vector (Lemma~\ref{strucprop}\,\ref{sorted}),
every non-zero entry of $\wksmooth_i$, for $i\geq i^*$ is at least $\wksmooth_{i^*}/m$
(Lemma~\ref{strucprop}\,\ref{item:wm}), and $\wksmooth_{i^*}$, if non-zero, lies in
$\bigl[p_{\min},\sum_{j\in J}p_j\bigr]$, where $p_{\min}$ the minimum non-zero processing
time of a job. 
%Note that if $\wksmooth_i=0$, then $\wkest_i=0$.

We sometimes say that $\wkest_i$ is the (work) capacity of machine $i$. We call a maximal
(consecutive) set of machines having the same $\wkest_i$ value, a machine class. 
(Note that this is a collection of machines in $[m]-\Mc_{\fast}$.)
The capacity of a machine class is the common $\wkest_i$ value of machines in that class. 
We can discard machines with capacity $0$, so every machine class contains at least
$m_1=\frac{1}{\ve^2}$ machines. 
Since capacities of different machine classes differ by at least a $(1+\ve)$-factor, 
there are at most $O\bigl(\frac{\log m}{\ve}\bigr)$ machine classes. 
Let $\mclass$ denote the collection of machine classes. 

%We can partition $A_1$ into three sets of jobs. This is for descriptive purpopses, and we
%emphasize that {\em we do not know this partition}.
Let $A^{\fast}$ be the jobs in $A_1$ assigned by $\tsg$ to machines in $\Mc_{\fast}$;
We call a job in $A^{\fast}$ is a {\em fast} job.
%We consider a partition of $A_1$ into three job classes.
%Consider a job $j\in A_1$.
%We say that a job $j$ is {\em fast} if $\tilde{\sg}(j)\in M^\fast$.
%If $\tilde{\sg}(j)\in M_\ell$ for some $\ell\in[L-1]$, we consider two cases.
We say that a job $j\in A_1$ is {\em large} for capacity $W$, if $\ve W<p_j\leq W$;
we say that $j$ is {\em small} for capacity $W$ if $p_j\leq\ve W$.
%Otherwise, we say that $j$ is {\em small} (if $\tilde{\sg}(j)\in M_\ell$ and
%$p_j\le\ve\,\w^*_\ell$). 
Let $A^{\lrg}$ be the set of jobs $j\in A_1-A^{\fast}$ such that $j$ is large for capacity
$\wkest_{\tsg(j)}$, and $A^{\sml}$ consist of jobs $j\in A_1-A^{\fast}$ such
that $j$ is small for capacity $\wkest_{\tsg(j)}$.
%Let $A^\fast$, $A^\lrg$, and $A^\sml$ denote the sets of fast, large, and small jobs,
%respectively. Clearly, $A^\fast\cup A^\lrg\cup A^\sml=A_1$ is a partition of $A_1$. 
Clearly, $A^{\fast}$, $A^{\lrg}$, $A^{\sml}$ partition $A_1$. 
Note that {\em we do not know this partition}.

We next proceed to obtain assignments for these three categories of jobs. 
%that is consistent with $\tsg$. 
The assignment may not quite assign all jobs in a category, but it will be consistent with
$\tsg$, and will earn large-enough reward.
\end{enumerate}

\vspace*{-1ex}
\paragraph{III: Assigning jobs to machines in \boldmath $\Mc^\fast$.}

\begin{enumerate}[resume*, start=4]
\item\label{step:fast} 
We use Lemma~\ref{enumlem} to guess, for each $q\in[\anbuck]$, the number of jobs in
$\abkt_q\cap A^{\fast}$.  
To elaborate, for each $q\in[\anbuck]$, we apply Lemma~\ref{enumlem} to guess the sequence
$\Bigl\{\bigl|\abkt_q\cap\tsg^{-1}(i)\bigr|\Bigr\}_{i\in \Mc_{\fast}}$ up to cumulative error
$\ve|\abkt_q|$ by taking $\est=|\abkt_q|$. This takes time
$2^{O(|\Mc_{\fast}|/\ve)}=2^{O(1/\ve^4)}$. and so the total time for doing this for all
$q\in[\anbuck]$ is $2^{O(\anbuck/\ve^4)}=\bigl(\frac{n}{\ve}\bigr)^{1/\ve^6}$.

Let $\bigl\{\tdh_{q,i}\bigr\}_{i\in \Mc_{\fast}}$ denote the correctly guessed sequence for
$q\in[\anbuck]$. Since jobs in $\abkt_q$ are indistinguishable, for all $i\in \Mc_{\fast}$,
we arbitrarily select $\ceil{\tdh_{q,i}}$ (unassigned) jobs from $\abkt_q$ and assign them
to machine $i$. If we run out of jobs in $\abkt_q$ while doing so,
we declare failure; this can only happen if one of our guesses is incorrect.
We do this for all $q\in[\anbuck]$. 

Let $\tA^{\fast}$ be the set of jobs so assigned to machines in $\Mc_{\fast}$. 
%Since we
%assume we have the correct guess for 
%$\bigl\{\bigl|\abkt_q\cap\tsg^{-1}(i)\bigr|\bigr\}_{i\in \Mc_{\fast}}$, for all
%$q\in[\anbuck]$, 
Since the $\tdh_{q,i}$ values underestimate the $\bigl|\abkt_q\cap\tsg^{-1}(i)\bigr|$
quantities, we have $\tA^{\fast}\sse A^{\fast}$. The following claim 
%(proved in Section~\ref{relptas-analysis}) 
shows that $\tA^{\fast}$ achieves a good amount of reward. 

\begin{claim} \label{rewdfast}
We have $\rewd(\tA^{\fast})\geq\rewd(A^{\fast})-\ve\cdot\rewd(A_1)$.
\end{claim}
\end{enumerate}

\vspace*{-2ex}
\paragraph{IV: Assigning large jobs.}
We next aim to determine the assignments of jobs in $A^{\lrg}$.
%for each machine $i\notin M^{\fast}$, the jobs
%assigned by $\tsg$ to $i$ that are large for capacity $\wkest_i$. 
We will not quite be able to assign all jobs in $A^{\lrg}$, but we will find a
large-reward subset of $A^{\lrg}$ and an assignment of these jobs.
This is the most involved portion of the algorithm. 

\begin{comment}
We will actually refrain from
computing an actual assignment of jobs to machines; instead, for each machine class $I$,
at this point, we will only find a work-vector (and the job-sets corresponding to the
work-vector coordinates) 
%an assignment leading to this work-vector) 
using a small number of extra machines, without committing to the mapping between
coordinates of the work-vector and machines in $I$.
\end{comment}
For each machine class $I$, the assignment we compute to machine in $I$ will end up using
$O(\ve)|I|$ extra machines. At this point, we are only considering the {\em work} assigned to
machines, and we will not concern ourselves with issues such as the speeds of these extra
machines. 
%{\em The fact that we do not need to worry about machine speeds is a significant
%source of savings.}
We will retain these extra machines (and more extra machines for a class $I$ may
get added later) until the very end. As one of the last steps of the
algorithm (step~\ref{rel-postprocess}), we will pare down the collection of 
%(regular and extra)
machines used for class $I$ to $|I|$ machines; %retain $|I|$ machines for class $I$;
%drop these extra machines; 
this yields a corresponding work-vector, and guided by Lemma~\ref{wksort}, we will assign  
jobs to machines in $I$ by assigning the jobs corresponding to the largest work-vector
coordinate %(and the corresponding jobs) 
to the fastest machine in $I$, the jobs corresponding to the second-largest coordinate to
the second-fastest machine in $I$, and so on. 
%that can be used to obtain a suitable assignment for machines in $I$
%To keep terminology succinct, we continue to say ``assign jobs to machines'' below, to
%mean that we are assigning jobs to work-vector coordinates. 

We find the job-assignment in various steps.  
Recall that $\mclass$ is the collection of machine classes. 
We say that a job $j$ is large for a machine class if it is large for the capacity of 
that class. 
First, for each machine class $I\in\mclass$ and $q\in[\anbuck]$, 
%letting $W$ be the capacity of machines in $I$, 
we determine the number of jobs from $\abkt_q$ that are large for class $I$ 
%the capacity of class $I$ 
and assigned by $\tsg$ to machines in $I$ (step~\ref{rel-largenum}). Since jobs in a type
bucket $\abkt_q$ are 
indistinguishable, this also yields a set $\tA_I$ of large jobs for $I$ and
assigned by $\tsg$ to machines in $I$.

We will use the configuration-enumeration approach from step~\ref{lrgasgn} of the PTAS
for identical machines to assign these jobs %an assignment of these jobs 
to machines in $I$. This will require us to sparsify job sizes and rewards for jobs in
$\tA_I$ so that we only have a constant (depending on $\frac{1}{\ve}$) number of 
distinct job-types for class $I$. This sparsification may lead to $O(\ve)|I|$
extra machines being created. 
\begin{comment}
As mentioned above, 
%we will retain these extra machines until the very end, when 
we will argue at the end that we can pare down the collection of %(regular and extra)
machines used for class $I$ to $|I|$ machines, 
%drop the extra machines and argue that this yields 
to obtain a feasible assignment without sacrificing the reward by much. 
\end{comment}
For technical reasons, we will also first isolate some ``giant'' jobs in
$\tA_I$ that will be assigned to separate machines in $I$, without any other large jobs 
assigned to these machines. 

So for each machine class $I$, we first assign the giant jobs in $\tA_I$
(step~\ref{giantasgn}), and then sparsify the job types for the remaining jobs in $\tA_I$
(steps~\ref{rel-sizesparse}, \ref{rel-rewdsparse}). 
Next, we use configuration enumeration to 
find the configurations to use for machines in $I$, for all machine classes $I$
(step~\ref{rel-cenum}). We will argue 
that the time needed for enumeration in step~\ref{rel-cenum}, 
%across all machine classes, 
is polynomially bounded. Finally, we map the configurations obtained for each machine
class $I$ to an assignment of a subset of $\tA_I$ to machines in $I$
(step~\ref{rel-lrgasgn}). 
(As noted earlier, all of this leads to a work-vector for each class $I$ using
$O(\ve)|I|$ extra machines.)

\begin{enumerate}[resume*, start=5]
\item {\bf Guessing the number of large jobs assigned to a machine class.} \label{rel-largenum}
We use an idea similar to that in~\cite{ChekuriK05}.
Consider $q\in[\anbuck]$. A job in $\abkt_q$ can be large for a class $I\in\mclass$ with 
capacity $W$ only if $\ve W<p^{(q)}\leq W$, or equivalently 
$W\in\bigl[p^{(q)},\frac{p^{(q)}}{\ve}\bigr)$. Since capacities of different classes differ
by at least a $(1+\ve)$-factor, this means that a job in $\abkt_q$ may be large for at
most $\log_{1+\ve}\frac{1}{\ve}=O\bigl(\frac{1}{\ve}\log\frac{1}{\ve}\bigr)$ machine
classes. 

Let $\lmc_q\sse\mclass$ be the collection of these machine classes, so
$|\lmc_q|=O\bigl(\frac{1}{\ve}\ln\frac{1}{\ve}\bigr)$. 
%for which jobs in $\abkt_q$
As in step~\ref{step:fast}, we use Lemma~\ref{enumlem} to guess 
$\Bigl(\bigl|\{j\in\abkt_q: \tsg(j)\in I\}\bigr|\Bigr)_{I\in\lmc_q}$ up to cumulative
error $\ve|\abkt_q|$ by taking $\est=|\abkt_q|$. 
This takes time $2^{O(|\lmc_q|/\ve)}=2^{O(\frac{1}{\ve^2}\log\frac{1}{\ve})}$.
Let $\bigl\{\tn_{q,I}\bigr\}_{I\in\lmc_q}$ denote the correctly-guessed sequence. 
For each $I\in\lmc_q$, we select an arbitrary set $\tA_{q,I}$ of $\ceil{\tn_{q,I}}$ unassigned
jobs from $\abkt_q$ to assign as large jobs to machines in $I$. 
Again, if we run out of jobs in $\abkt_q$ while doing so, we declare failure. 
%this can only happen if one of our guesses is incorrect.
We do this for all $q\in[\anbuck]$.

The total time required for this enumeration for all $q\in[\anbuck]$ is
$2^{O(\frac{\anbuck}{\ve^2}\log\frac{1}{\ve})}=\bigl(\frac{n}{\ve}\bigr)^{\frac{1}{\ve^4}\log\frac{1}{\ve}}$.
For a machine class $I$, we now have a set of jobs 
$\tA_I=\bigcup_{q\in[\anbuck]: I\in\lmc_q}\tA_{q,I}$ to be assigned (as large jobs) to
machines in $I$. Since the $\tn_{q,I}$ values underestimate the
$\bigl|\{j\in\abkt_q: \tsg(j)\in I\}\bigr|$ quantities, we may assume that 
$\tA_I\sse A^{\lrg}$ and that jobs in $\tA_I$ are assigned by $\tsg$ (as large jobs) to
machines in $I$. 

Note that $|\tA_I|\leq |I|/\ve$ since if $W$ is the capacity of class $I$, we have
$p_j\geq\ve W$ for all $j\in\tA_I$, and 
$p(\tA_I)\leq p\bigl(\tsg^{-1}(I)\bigr)\leq W\cdot|I|$. If $|\tA_I|>|I|/\ve$, then we
declare failure; this can only happen if one of our guesses is incorrect.
Similar to Claim~\ref{rewdfast}, we have the following.

\begin{claim} \label{rewdlarge}
We have $\sum_{I\in\mclass}\rewd(\tA_I)\geq\rewd(A^{\lrg})-\ve\cdot\rewd(A_1)$.
\end{claim}
\end{enumerate}

As noted above, to assign the designated large jobs for each machine class, we 
sparsify job types and use %sparsification and 
configuration enumeration.
%
\begin{comment}
\item {\bf Assigning the designated large jobs.}
We now consider each machine class $I\in\mclass$, and assign the jobs in $\tA_I$ to machines
in $I$. This will be via the configuration-enumeration approach described in
step~\ref{lrgasgn} 
of the PTAS for identical machines. To facilitate this, we need to sparsify jobs in
$\tA_I$ as in steps~\ref{sizesparse}, \ref{rewdsparse} of that algorithm, so that the
enumeration for a specific 
machine class can be done in {\em constant} time (depending on $\frac{1}{\ve}$), and hence
the enumeration for all $O\bigl(\frac{\log m}{\ve}\bigr)$ machine classes takes polynomial
time. 
\end{comment}
Here, we significantly depart from the approach in~\cite{ChekuriK05}, because
our task is not just to find some assignment that respects the capacity of the
machine class, but rather to find an assignment that is {\em consistent} with the 
{\em unknown} assignment $\tsg$.
%respects the {\em unknown} load-vector $\tlvec$ (and hence the norm budget). 
This stronger condition ensures that the load vector from this partial
assignment satisfies the norm budget despite the fact that the $\wkest_i$ capacities are 
only estimates of $\wksmooth_i$.
It also ensures that we can extend this partial assignment
to one that assigns also the small jobs while respecting the norm budget.
%will be {\em crucial} in order to argue that the resulting partial assignment can be
%extended to one that assigns the
%remaining jobs in %$A_3-\bigcup_{I\in\lmc}J_I=
%$\bigl(\bigcup_{I\in\mclass-\lmc}J_I\bigr)\cup\bigl(\bigcup_{\text{small size $s$}}A_{3,s}\bigr)$ while
%respecting the norm budget.  

\smallskip
We execute the following steps for each machine class $I\in\mclass$ with corresponding  
capacity $W$.  
%Let $W$ be the  capacity of class $I$.

\begin{enumerate}[resume*, start=6]
%[label=(\theenumi.\arabic*), leftmargin=2ex, labelsep=*]
\item {\bf Assigning giant jobs.} \label{rel-giant} \label{giantasgn}
We call a job $j\in\tA_I$ with $p_j\geq (1-\ve)W$ a {\em giant} job for $I$; let
$A^{\giant}_I$ be the set of giant jobs in $\tA_I$.
We assign every giant job to a separate machine in  
$I$, and do not assign any other jobs in $\tA_I$ to this machine.
Note that $\tsg$ must also do this, since otherwise the total work assigned by $\tsg$ to a
machine $i\in I$ would exceed $(1-\ve)W+\ve W\geq\wksmooth_i$.
So if $|A^{\giant}_I|>|I|$, we declare failure; this only happens if one of
our guesses is incorrect.

Let $A'_I=\tA_I-A^{\giant}_I$ be the remaining jobs in $\tA_I$.

\item {\bf Sparsifying sizes of jobs in \boldmath $A'_I$.} \label{rel-sizesparse}
We mimic step~\ref{sizesparse}.
Consider jobs in $A'_I$ in non-increasing order of size,
and divide them into $k'=1+\frac{1}{\ve^2}$ groups, where the first $\frac{1}{\ve^2}$
groups contain $\floor{\ve^2|A'_I|}$ jobs, and the last group contains less than
$\frac{1}{\ve^2}$ jobs. 
%(If $\ve^2|\tA_I|<1$, then the first $\frac{1}{\ve^2}$ groups
%are empty, and the last group contains all jobs in $A'_I$.) 
We increase the size of every job in groups $2,\ldots,k'-1$ to
the size of the largest job in that group, and move each job in the first group to a 
separate extra machine. 
This creates at most $\ve^2|A'_I|\leq\ve|I|$ extra machines. 

Let $A''_I$ denote the jobs in groups $2,\ldots,k'$ with their (potentially) modified
sizes, so jobs in $A''_I$ have at most $\frac{2}{\ve^2}$ distinct modified sizes.
Again, to keep notation simple, we continue to use $p_j$'s to denote the modified sizes.
%We will work with these modified sizes in order to assign the jobs in $A'$, and then
%revert to the original job sizes.
For all $r=2,\ldots,k'-1$, the space used by jobs in group $r-1$ under assignment $\tsg$
can be used to accommodate the jobs in group $r$ with their modified sizes. 
So $\tsg$ can be used to obtain a feasible assignment of jobs with their modified sizes,
%$A_1\cup A'\cup A^{\sml}\mapsto[m]$, where jobs in $A_1$ are assigned as before to
%the first $|A_1|$ machines and 
%jobs in $A'$ %(with their modified sizes) 
%and $A^{\sml}$ are assigned to machines in
%$\mcset$, and 
where the work assigned to a machine $i\in I$ (under the modified sizes) continues to be
at most $\wksmooth_i$. 
This may involve changing the $\tsg$-assignments of some jobs in $A'_I$;
%avoid excessive notation, 
as before, we continue to use $\tsg$ to denote this modified feasible assignment. 
%denote the resulting modified assignment of jobs in $(J_I-\tJ_I)\cup J'_I$.

\item {\bf Sparsifying rewards of jobs in \boldmath $A''_I$.} \label{rel-rewdsparse}
We mimic step~\ref{rewdsparse}.
%We next use a similar shifting idea to reduce the number of distinct rewards of jobs in
%$A'$ to $O\bigl(\frac{1}{\ve^3}\bigr)$. 
Consider jobs in $A''_I$ in non-increasing order
of reward and divide them into $k''=1+\frac{1}{\ve^3}$ groups, where the first
$\frac{1}{\ve^3}$ groups contain $\floor{\ve^3|A''_I|}$ jobs, and the last group contains
less than $\frac{1}{\ve^3}$ jobs. 
%(Again, if $\ve^3|A'|<1$, then the last group is all of
%$A'$ and all other groups are empty.)
For every $r=\frac{1}{\ve}+1,\ldots,k''-1$, we reduce the reward of each job in group
$r$ to the smallest reward of a job in that group. We also move each job in the first
$\frac{1}{\ve}$ groups to a separate extra machine.

Let $\bA_I\sse A''_I$ denote the jobs in groups $\frac{1}{\ve}+1,\ldots,k''$. 
%We continue to use $\rewd_j$ to denote the modified rewards.
Let $\trewd_j\leq\rewd_j$ denote the modified reward of each job $j\in\tA_I$, where
$\trewd_j=\rewd_j$ if $j$'s reward is unchanged.
We create at most $\ve^2|A''_I|\leq\ve|I|$ extra machines, and jobs in
$\bA_I$ now have at most $\frac{2}{\ve^3}$ distinct rewards.
%Moreover, the following claim shows that this reward-sparsification step does not incur
%much loss.
Mimicking Claim~\ref{rwdsparse}, we have the following.

\begin{claim} \label{rel-rwdsparse}
We have $\trewd(A''_I)\geq(1-\ve)\rewd(A''_I)$.
\end{claim} 
%\end{enumerate}

%\begin{enumerate}[resume*, start=8]
\medskip
\item {\bf Enumerating configurations.} \label{rel-cenum}
We proceed as in step~\ref{lrgasgn}. 
%but use Lemma~\ref{enumlem} to simplify the description. 
%We describe the enumeration used for 
Consider a machine class $I\in\mclass$ with capacity $W$. 
Jobs in $\bA_I$ are of at most $\frac{4}{\ve^5}$ distinct types, where recall
that type of a job is its (size, reward) tuple.
A configuration for $I$ lists how many jobs from $\bA_I$ of each type are assigned to a
machine so that the total work assigned is at most $W$. 
(Recall that $W$ is the capacity of class $I$.)
Let $\cfgset_I$ denote the collection of all such configurations for class $I$.
Since every job in $\bA_I$ is large for $I$, we have 
$|\cfgset_I|\leq C=\bigl(\frac{1}{\ve}\bigr)^{4/\ve^5}$. 
The reward of a configuration $\cfg$, denoted $\trewd(\cfg)$, is the reward obtained from
the jobs assigned by $\cfg$. 
%abusing notation, we denote this by $\trewd(\cfg)$. 

%We say that a machine $i$ uses a configuration $\cfg$ 
%to mean that the jobs from $\bA_I$ assigned to it conform to $\cfg$, i.e., the
%number of $\bA_I$-jobs of each type assigned to $i$ is as specified by $\cfg$.
Say that $\tsg$ uses a configuration $\cfg$ for $I$, if some machine in $I$ is 
assigned $\cfg$ under $\tsg$.
Call a tuple in $\Zp^{\cfgset_I}$ specifying how many times each configuration in
$\cfgset_I$ is used for machines in $I$ a {\em configuration-usage sequence for $I$}.
For $\cfg\in\cfgset_I$, let $N_{I,\cfg}$ denote the number of times $\tsg$ uses
configuration $\cfg$ for machines in $I$. 
We call $\{N_{I,\cfg}\}_{\cfg\in\cfgset_I}$, the $\tsg$-configuration-usage sequence for
$I$. We apply Lemma~\ref{enumlem} to guess the sequence 
$\bigl\{N_{I,\cfg}\cdot\trewd(\cfg)\}_{\cfg\in\cfgset_I}$ up to cumulative error
$\ve\cdot\trewd(\bA_I)$ by taking $\est=\trewd(\bA_I)$. 
%(This satisfies the requirements of Lemma~\ref{enumlem} since
%$\sum_{\cfg\in\cfgset_I}N_{\cfg}\trewd(\cfg)\leq\trewd(\bA_I)$.) 
This takes time $2^{O(C/\ve)}$ (which is a constant). Recall that this means,
more precisely, that we identify a set of size $2^{O(C/\ve)}$ that contains a sequence
close to the desired sequence.
%configuration-usage sequence for $I$ that is close to the $\tsg$-configuration-usage
%sequence for $I$. 
Define $\num_{I,\cfg}$ to be the entry for configuration $\cfg$ in the correct guessed
sequence, divided by $\trewd(\cfg)$. 
Then, $\num_{I,\cfg}\leq N_{I,\cfg}$ for all $\cfg\in\cfgset_I$, and
\begin{equation}
\sum_{\cfg\in\cfgset_I}\num_{I,\cfg}\cdot\trewd(\cfg)\geq
\sum_{\cfg\in\cfgset_I}N_{I,\cfg}\cdot\trewd(\cfg)-\ve\cdot\trewd(\bA_I)
%=\trewd(\bA_I)-\ve\cdot\trewd(\bA_I)
=(1-\ve)\trewd(\bA_I). \label{aiclass}
\end{equation} 

\smallskip
Our goal is to find an assignment consistent with $\tsg$ of a large-reward subset of
$A^{\lrg}$ to machines in $[m]-\Mc_{\fast}$, so we need to consider all possible combinations
of configuration-usage sequences for the different machine classes. So the overall time
needed to find a suitable configuration-usage sequence for {\em every} machine class, 
%configurations for {\em all} machine classes, 
i.e., the size of the search set that we need to consider such that this set contains, for
every $I\in\mclass$, a configuration-usage sequence for $I$ that is close to the
$\tsg$-configuration-usage sequence for $I$, is 
$\bigl(2^{O(C/\ve)}\bigr)^{|\mclass|}=m^{O(C/\ve^2)}$, which is polynomially bounded.

\item {\bf Mapping the configuration-usage-sequences to a job-assignment.}
\label{rel-lrgasgn}
%{\bf Assigning jobs in \boldmath $\tA_I$ to machines in $I$ (including extra machines).}
We do the following for each machine class $I$.
Recall %at this point we have, for each machine class $I$, 
that we have a vector $\bigl(\num_{I,\cfg}\bigr)_{\cfg\in\cfgset_I}$ specifying the number
of times each $\cfg\in\cfgset_I$ is used, and (under the correct guesses) this is
component-wise at most the $\tsg$-configuration-sequence for $I$,
$\{N_{I,\cfg}\}_{\cfg\in\cfgset_I}$, which specifies the number of times $\tsg$ uses each
configuration for machines in $I$. 

For every $\cfg\in\cfgset_I$, we choose $\ceil{\num_{I,\cfg}}$ machines from $I$ that have
not been assigned any jobs yet, 
%are not assigned any jobs in $A^{\giant}_I$, 
and assign these machines configuration $\cfg$, where when we pick $\bA_I$-jobs as
specified by $\cfg$ to assign to a machine, we always pick from the unassigned jobs in
$\bA_I$. 
%Assigning $\cfg$ to a
%machine $i$ means that, from the unassigned jobs in $A''$, for each job type, we select
%the number of jobs of that type specified by $\cfg$ and assign these jobs to $i$. 
If we run out of machines while doing so, i.e.,
$\sum_{\cfg\in\cfgset}\ceil{\num_{I,\cfg}}+|A^{\giant}_I|>|I|$, then we declare
failure. Similarly, if we run out of jobs of a particular type, then we again declare 
failure. These failure events can only happen if one of our guesses is incorrect. 
Let $\hA_I\sse\bA_I$ denote the set of jobs assigned to machines in $I$ via this process. 
\end{enumerate}

\vspace*{-1ex}
\paragraph{V: Assigning small jobs.}
Recall that $A^{\fast}$ denotes the jobs in $A_1$ assigned by $\tsg$ to machines in
$\Mc_{\fast}$. 
Also, $A^{\lrg}$ and $A^{\sml}$ denote respectively the jobs in $A_1$ that are
assigned by $\tsg$ as large and small on their respective machines.  

At this point, we have a partial assignment $\sg$ consistent with $\tsg$ that assigns jobs
in  $\tA^{\fast}\sse A^{\fast}$ to machines in $\Mc_{\fast}$, and for each machine class
$I\in\mclass$, assigns jobs in $A^{\giant}_I\cup\hA_I\sse\tA_I\sse A^{\lrg}$ to machines
in $I$. Also, for each machine class $I\in\mclass$, jobs in 
$A'_I-\bA_I\sse\tA_I\sse A^{\lrg}$ are assigned to extra machines for $I$.
Let $S=\tA^{\fast}\cup\bigl(\bigcup_{I\in\mclass}(A^{\giant}_I\cup\hA_I)\bigr)$ be the
jobs assigned by $\sg$, and let 
$A^{\rem}=A_1-S-\bigcup_{I\in\mclass}(A'_I-\bA_I)$ be the jobs remaining in $A_1$, which are
not assigned to regular or extra machines. Note that $A^{\rem}\supseteq A^{\sml}$.

\begin{enumerate}[resume*, start=11]
\item We extend $\sg$ by writing an LP to find a large-reward subset of $A^{\rem}$ to assign to
machines in $[m]-\Mc_{\fast}$. Since we have taken care of large jobs, we insist that all
these jobs are assigned as small jobs on their respective machines, encoded
by \eqref{forbid}. We also ensure that we do not exceed capacity $\wkest_i$ on any machine
$i$, taking into account both the jobs assigned by $\sg$ to $i$ and the jobs assigned by
the LP to $i$, encoded by \eqref{workmc}. Also, we revert to the original
$\rewd_j$-rewards. This yields the following LP.
\begin{alignat}{3}
\max & \quad & \sum_{j\in A_2,i\in[m]-\Mc_{\fast}}\rewd_jx_{ij} & \tag{LP} \label{extnlp} \\
\text{s.t.} & \quad & \sum_{i\in[m]-\Mc_{\fast}}x_{ij} & \leq 1 \qquad && \forall j\in A^{\rem}
\label{extnlp-jasgn} \\
&& x_{ij} = 0 \quad \text{if }p_j&>\ve\cdot\wkest_i \qquad && 
\forall i\in[m]-\Mc_{\fast},\,j\in A^{\rem} \label{forbid} \\
&& p\bigl(\sg^{-1}(i)\bigr)+\sum_{j\in A^{\rem}}p_jx_{ij} & \leq \wkest_i 
\qquad && \forall i\in[m]-\Mc_{\fast} \label{workmc} \\
&& x & \geq 0. \notag
\end{alignat}

\begin{claim} \label{extnlpopt}
Under the correct guesses in previous steps, we have
$\OPT_{\text{\ref{extnlp}}}\geq\rewd(A^{\sml})$. 
\end{claim}

\item Let $x^*$ be an optimal solution to \eqref{extnlp}.
We round $x^*$ using \gap rounding~\cite{ShmoysT93}. By properties of \gap rounding, the
resulting integer solution has objective value at least that of $x^*$.
Let $\hsg$ denote the assignment obtained by concatenating the resulting integer solution
and $\sg$.   
\end{enumerate}

\vspace*{-1ex}
\paragraph{VI. Postprocessing to obtain a feasible, near-optimal solution.}
We now clean things up to obtain a feasible solution without sacrificing the reward by
much. There are three interrelated issues we encounter: (1) we have used some
extra machines for each machine class $I$; (2) the assignment $\hsg$ may exceed the
$\wkest_i$ work-capacity on some machine $i\in[m]-\Mc_{\fast}$; (3) even if we respect the
$\wkest_i$ work capacities, since $\wkest_i$ overestimates $\wksmooth_i$, this need not
yield a feasible solution. 

Consider any machine $i\in[m]-\Mc_{\fast}$. 
By properties of \gap rounding, we know that there is at most one job $j$ with $\hsg(j)=i$
and $x^*_{ij}>0$ whose removal will reduce the work assigned to $i$ to at most $\wkest_i$;
also, since $x^*_{ij}>0$, we have $p_j\leq\ve\cdot\wkest_i$. 
So we have $p\bigl(\hsg^{-1}(i)\bigr)\leq (1+\ve)\wkest_i$. Moreover, if
$p\bigl(\hsg^{-1}(i)\bigr)>p\bigl(\sg^{-1}(i)\bigr)$, then this work difference is due to
small jobs for $\wkest_i$ assigned by $\hsg$ to $i$.

\smallskip
For every machine class $I\in\mclass$ with capacity $W$, we execute the following steps. 
\begin{enumerate}[resume*, start=13]
\item \label{rel-movejobs}
For each machine $i\in I$, let $R_i\sse\hsg^{-1}(i)-\sg^{-1}(i)$ be a minimal set of
jobs such that  
$p\bigl(\hsg^{-1}(i)\bigr)-p(R_i)\leq\max\bigl\{p(\sg^{-1}(i)),(1-\ve)W\}$. From the above
observations, we have that $p_j\leq\ve W$ for all $j\in R_i$ and 
$p(R_i)\leq 3\ve W$. 

We create extra machines and pack the jobs $\bigcup_{i\in I}R_i$ arbitrarily on these
extra machines so that each extra machine is packed maximally within capacity $(1-\ve)W$.
This creates at most $1+\floor{\frac{3\ve W\cdot|I|}{(1-2\ve)W}}$ extra machines,
since every extra machine, save for at most one, has at least 
$(1-\ve)W-\ve W$ load assigned to it, and $p\bigl(\bigcup_{i\in I}R_i\bigr)\leq 3\ve W\cdot|I|$.
We have $|I|\geq\frac{1}{\ve^2}$, so the number of extra machines created this
way is bounded by $7\ve|I|$ (recall that $\ve\leq 0.25$). Combined with the $2\ve|I|$ extra
machines created in steps~\ref{rel-sizesparse} and~\ref{rel-rewdsparse}, we have created
at most $9\ve|I|$ extra machines for class $I$.

\item \label{rel-dropmc} \label{rel-postprocess}
Considering all the machines used for class $I$, i.e., 
$I\cup\{\text{extra machines for $I$}\}$. 
We retain the $|I|$ largest-reward machines from this collection, and discard the rest.

\item \label{rel-finalsort}
This yields a work-vector for class $I$ consisting of the work, and the
corresponding jobs, assigned to the $|I|$ retained machines for class $I$. For the actual 
assignment of jobs to machines in $I$, we sort the work-vector coordinates, and assign the
the jobs corresponding to the $\ell$-th-largest work-vector coordinate to the $\ell$-th
fastest machine in $I$, for all $\ell=1,\ldots,|I|$.
%work-vector coordinate and assign these to the fastest machine in $I$, i.e., the
%smallest-index machine in $I$

\smallskip
\item \label{rel-ptas-end}
The final assignment is the assignment given by $\sg$ for machines in $\Mc_{\fast}$ 
together with the assignment computed above for machines in $[m]-\Mc_{\fast}$. (Note that
a machine $i$ with $\wksmooth_i=\wkest_i=0$ does not have any jobs assigned to it.)
\end{enumerate}

\subsubsection{Analysis} \label{relptas-analysis}
We prove the following.

\begin{theorem} \label{rel-ptasthm}
The algorithm described in steps~\ref{rel-ptas-start}--\ref{rel-ptas-end} is a PTAS
for \normsched on related machines. 
\end{theorem}

\begin{proof}
%\begin{lemma} \label{rel-runtime}
%The total time required for all the enumeration described in
%steps~\ref{rel-ptas-start}--\ref{rel-lrgasgn} is bounded $(mn)^{g(1/\ve)}$ for a suitable
%function $g$. Hence, for any fixed $\ve>0$, the running time is polynomially
%bounded. %$\poly\bigl(\size(\I),(mn)^{\poly(1/\ve)}\bigr)$. 
%\end{lemma}
%\begin{proof}
Lemmas~\ref{rel-rewd} and~\ref{rel-feas} show that we obtain a feasible solution with
reward $\bigl(1-O(\ve)\bigr)\OPT$.
As discussed when describing the algorithm, the time required for enumeration in any of
the individual steps in~\ref{rel-ptas-start}--\ref{rel-cenum} is of the form
$(mn)^{g(1/\ve)}$ for some function $g$. There are only a constant number of steps, so we
have polynomial running time for any fixed $\ve>0$.
\end{proof}

We prove Lemmas~\ref{enumlem}--\ref{strucprop} at the end of this section.
We begin the proof of the performance guarantee by proving
Claim~\ref{rewdfast}, Claim~\ref{rewdlarge} and Claim~\ref{extnlpopt} that were stated
while describing the algorithm. 
We omit the proof of Claim~\ref{rel-rwdsparse} as this simply duplicates the proof of
Claim~\ref{rwdsparse} by making the appropriate notational changes.

\begin{proofof}{Claim~\ref{rewdfast}}
Recall that $\tA^{\fast}$ consists of $\sum_{i\in \Mc_{\fast}}\ceil{\tdh_{q,i}}$ jobs from
$\abkt_q$, for each $q\in[\anbuck]$, and we have (by Lemma~\ref{enumlem})
that $\sum_{i\in \Mc_{\fast}}\tdh_{q,i}\geq\sum_{i\in \Mc_{\fast}}\bigl|\abkt_q\cap\tsg^{-1}(i)\bigr|-\ve|\abkt_q|$. 
Therefore,
\begin{equation*}
\begin{split}
\rewd(\tA^{\fast}) & =\sum_{q\in[\anbuck]}\rewd^{(q)}\sum_{i\in \Mc_{\fast}}\tdh_{q,i}
\ge \sum_{q\in[\anbuck]}\rewd^{(q)}\biggl(\sum_{i\in M^\fast}\bigl|\abkt_q\cap\tsg^{-1}(i)\bigr|-\ve|\abkt_q|\biggr)
\\
& =\rewd(A^\fast)- \ve\cdot\sum_{q\in[\anbuck]}\rewd^{(q)}\cdot|\abkt_q|
= \rewd(A^\fast)-\ve\cdot\rewd(A_1). \qedhere
\end{split}
\end{equation*}
\end{proofof}

\begin{proofof}{Claim~\ref{rewdlarge}}
Recall that $\tA_I=\bigcup_{q\in[\anbuck]: I\in\lmc_q}\tA_{q,I}$, 
and for each $q\in[\anbuck]$, we have $|\tA_{q,I}|=\ceil{\tn_{q,I}}$ and
the $\bigl(\tn_{q,I}\bigr)_{I\in\lmc_q}$
sequence satisfies (by Lemma~\ref{enumlem})
$\sum_{I\in\lmc_q}\tn_{q,I}\geq\sum_{I\in\lmc_q}\bigl|\{j\in\abkt_q: \tsg(j)\in I\}\bigr|-\ve\cdot|\abkt_q|$.
So we have
\begin{equation*}
\begin{split}
\sum_{I\in\mclass}\rewd(\tA_I)
& =\sum_{I\in\mclass}\sum_{q\in[\anbuck]:I\in\lmc_q}\rewd(\tA_{q,I})
=\sum_{q\in[\anbuck]}\rewd^{(q)}\sum_{I\in\lmc_q}\tn_{q,I} \\
& \geq\sum_{q\in[\anbuck]}\rewd^{(q)}
\biggl(\sum_{I\in\lmc_q}\Bigl|\bigl\{j\in\abkt_q:\tsg(j)\in I\bigr\}\Bigr|-\ve|\abkt_q|\biggr) \\
& =\rewd(A^{\lrg})-\ve\cdot\sum_{q\in[\anbuck]}\rewd^{(q)}\cdot|\abkt_q|
=\rewd(A^{\lrg})-\ve\cdot\rewd(A_1). \qedhere
\end{split}
\end{equation*}
\end{proofof}

\begin{proofof}{Claim~\ref{extnlpopt}}
This follows simply because $\tsg$ yields a feasible integer solution to \eqref{extnlp} of
objective value $\rewd(A^{\sml})$. For every $j\in A^{\sml}$, we set $x_{ij}=1$ if
$i=\tsg(j)$ and $0$ otherwise. This clearly satisfies constraints \eqref{extnlp-jasgn}, and
satisfies \eqref{forbid} by the definition of $A^{\sml}$. We satisfy \eqref{workmc}
because under the correct guesses, $\sg$ is consistent with $\tsg$, and the total work
assigned by $\tsg$ to a machine $i\in[m]-\Mc_{\fast}$ is at most
$\wksmooth_i\leq\wkest_i$. 
\end{proofof}

\begin{lemma} \label{rel-rewd}
The assignment returned obtains reward at least $\bigl(1-O(\ve)\bigr)\OPT$.
\end{lemma}

\begin{proof}
We first lower bound the reward obtained before dropping the extra machines in
step~\ref{rel-dropmc}.
%Before dropping the extra machines, including also the jobs assigned to extra machines for
%a machine class, 
This reward (where we are including jobs assigned to extra machines) is 
\begin{equation}
\rewd(\tA^{\fast})+\sum_{I\in\mclass}\Bigl(\rewd(A^{\giant}_I)+\rewd(A'_I-\bA_I)+\rewd(\hA_I)\Bigr)
+\rewd(A^{\sml}). \label{totrewdineq1}
\end{equation}
In the above expression, For a machine class $I$, the $\rewd(A'_I-\bA_I)$ term is the
reward from jobs assigned to extra machines for $I$ in
steps~\ref{rel-sizesparse}, \ref{rel-rewdsparse}; 
the $\rewd(\hA_I)$ term is the reward obtained from the large jobs assigned to machines in $I$ in
step~\ref{rel-lrgasgn} as a result of the configurations obtained for $I$. The last term
$\rewd(A^{\sml})$ is the reward from the small jobs assigned by \gap rounding
(Claim~\ref{extnlpopt}).

Consider a machine class $I$. We have 
\begin{equation}
\rewd(\hA_I)=\sum_{\cfg\in\cfgset_I}\ceil{\num_{I,\cfg}}\cdot\rewd(\cfg) 
\geq\sum_{\cfg\in\cfgset_I}\num_{I,\cfg}\cdot\trewd(\cfg)\geq(1-\ve)\trewd(\bA_I)
\label{aiineq1}
\end{equation}
where the last inequality is due to \eqref{aiclass}. So 
\begin{equation*}
\begin{split}
\rewd(A'_I-\bA_I)+\rewd(\hA_I)&\geq\rewd(A'_I-A''_I)+\rewd(A''_I-\bA_I)+(1-\ve)\trewd(\bA_I)
\\ & \geq\rewd(A'_I-A''_I)+(1-\ve)\trewd(A''_I)
\geq\rewd(A'_I-A''_I)+(1-\ve)^2\rewd(A''_I) \\
& \geq(1-\ve)^2\rewd(A'_I)
\end{split}
\end{equation*}
The first two inequalities use the fact that $\bA_I\sse A''_I\sse A'_I$; the first
inequality also uses \eqref{aiineq1}. The third is due to
Claim~\ref{rel-rwdsparse}. 

Plugging the above in \eqref{totrewdineq1}, and since
$\tA_I=A^{\giant}_I\cup A'_I$, we obtain that the total reward is at least
$\rewd(\tA^{\fast})+\sum_{I\in\mclass}(1-\ve)^2\rewd(\tA_I)+\rewd(A^{\sml})$.
Now using Claims~\ref{rewdfast} and~\ref{rewdlarge}, and since 
$A_1=A^{\fast}\cup A^{\lrg}\cup A^{\sml}$, the total reward obtained is at least 
$(1-\ve)^2\rewd(A_1)-2\ve\cdot\rewd(A_1)\geq(1-4\ve)\rewd(A_1)$.
We have $\rewd(A_1)\geq (1-4\ve)\OPT$ (see~\ref{a1ineq}), so the total reward obtained
before dropping extra machines is at least $(1-8\ve)\OPT$.

When we drop machines, for a machine class $I$, we retain the $|I|$ largest-reward
machines from the (at most) $(1+9\ve)|I|$ machines used for class $I$, so we obtain at
least a $\frac{1}{1+9\ve}$-fraction of the total reward before dropping machines. So the
reward of the final assignment is at least
$\frac{1-8\ve}{1+9\ve}\cdot\OPT\geq(1-17\ve)\OPT$. 
\end{proof}

\begin{lemma} \label{rel-feas}
The assignment returned is feasible, i.e., $f(\text{resulting load vector})\leq B$.
\end{lemma}

\begin{proof}
Let $\twvec$ be the work-vector resulting from the assignment returned.
We argue that %the sorted work-vector resulting from the assignment is coordinate-wise at
$\twvec\leq\wksmooth$. By Lemma~\ref{strucprop}\,\ref{item:feasible}, this implies that
the resulting load vector has norm at most $B$.

Recall that the partial assignment $\sg$ computed in
steps~\ref{rel-ptas-start}--\ref{rel-lrgasgn} is consistent with $\tsg$ and
$\wvec{\tsg}\leq\wksmooth$ (Lemma~\ref{strucprop}\,\ref{item:capacities}).
Machines in $\Mc_{\fast}$ are only assigned jobs by $\sg$, so this implies that
$\twvec_i\leq\wvec{\tsg}_i\leq\wksmooth_i$ for all $i\in\Mc_{\fast}$.

Next, consider a machine class $I\in\mclass$ with capacity $W$. 
Note that all through steps~\ref{rel-largenum}--\ref{rel-dropmc}, we are concerned with
the work-vector of machines in $I$.
So $\sg$ being consistent with $\tsg$ means more precisely that there is a permutation
$\pi:I\mapsto I$ such that 
$\wvec{\sg}_i=p(\sg^{-1}(i))\leq p\bigl(\tsg^{-1}(\pi(i))\bigr)=\wvec{\tsg}_{\pi(i)}$. 
When we move jobs in $\tA_I$ to extra machines in
steps~\ref{rel-sizesparse}, \ref{rel-rewdsparse}, one job per extra machine, we ensure
that these are not giant jobs for $I$; so the work on an extra machine is at most $(1-\ve)W$.
In step~\ref{rel-movejobs}, by design, we ensure that the work assigned to any machine
$i\in I$ is at most 
$\max\bigl\{p(\sg^{-1}(i)),(1-\ve)W\bigr\}$, and
the work assigned to any extra machine for $I$ is at most $(1-\ve)W$. Note that
$(1-\ve)W\leq\wksmooth_i$ for every $i\in I$, since $W$ is the common value of $\wkest_i$
for all $i\in I$. So $(1-\ve)W\leq\min_{i\in I}\wksmooth_i$. 
%We can also say that if $p(\sg^{-1}(i))>(1-\ve)W$, then 

This implies that for every machine retained for $I$ in step~\ref{rel-dropmc}, one can
bound the work assigned to that machine by a {\em distinct} $\wvec{\tsg}_i$ term, and
hence a distinct $\wksmooth_i$ term, for some $i\in I$. 
It follows that the sorted work-vector of the machines retained for $I$ is coordinate-wise
at most $(\wksmooth_i)_{i\in I}$. Due to the sorting performed in step~\ref{rel-finalsort},
this implies that $(\twvec_i)_{i\in I}\leq(\wksmooth)_{i\in I}$. 

This holds for every $I\in\mclass$, so combined with
$(\twvec_i)_{i\in\Mc_{\fast}}\leq(\wksmooth_i)_{i\in\Mc_{\fast}}$, we obtain that $\twvec\leq\wksmooth$.
\end{proof}

\subsubsection*{Proofs of Lemmas~\ref{wksort} and~\ref{strucprop}}

\begin{proofof}{Lemma~\ref{wksort}}
If $\pi$ is the identity permutation, there is nothing to be shown.
Let $\w=\wvec{\sg}$.
So suppose there are machines $i,i'\in[m]$ with $i<i'$ (so $\spd_i\geq\spd_{i'}$) and
$\w_i<\w_{i'}$. Consider the assignment $\sg'$, where we switch the
assignments of machines 
$i$ and $i'$. That is, for $j\in S$, we set $\sg'(j)=\sg(j)$ if $\sg(j)\notin\{i,i'\}$; we
set $\sg'(j)=i$ if $\sg(j)=i'$, and $\sg'(j)=i'$ if $\sg(j)=i$.
%Let $i,i^\prime\in[m]$ be two machines (indices) with $i<i^\prime$ (which implies that $s_i\ge s_{i^\prime}$).
%Let $w_i=\wvec{\sg}_i$ and $w_{i^\prime}=\wvec{\sg}_{i^\prime}$.
%Suppose that $w_i<w_{i^\prime}$.
%Consider the assignment $\sg^\prime:S\to[m]$ 
%with $\sg^\prime(j)=i$ if $j\in\sg^{-1}(i^\prime)$, $\sg^\prime(j)=i^\prime$ if $j\in\sg^{-1}(i)$, and $\sg(j)=\sg^\prime(j)$ for all other $j\in S$.
It suffices to show that $f(\lvec{\sg^\prime})\le f(\lvec{\sg})$, since via a series of
pairwise interchanges, one can move from the assignment $\sg$ to the assignment
$\pi\circ\sg$.

We utilize Claim~\ref{schur}. Let $v=\lvec{\sg}$, $u=\lvec{\sg'}$, and
$\kp=\bigl(\w_{i'}-\w_i\bigr)/\spd_{i'}$. 
Then $u_\ell=v_\ell$ for all $\ell\in[m]-\{i,i'\}$. We have 
\begin{alignat}{1}
u_i & =\frac{\sum_{j\in S:\sg'(j)=i}p_j}{\spd_i}=\frac{\w_{i'}}{\spd_i}
\leq\frac{\w_i}{\spd_i}+\frac{\w_{i'}-\w_i}{\spd_{i'}}=v_i+\kp, \label{uiineq} \\
\text{and} \quad u_{i'} & = \frac{\w_i}{\spd_{i'}}=\frac{\w_{i'}}{\spd_{i'}}-\kp=v_{i'}-\kp \notag
\end{alignat}
where the inequality in \eqref{uiineq} is because $\spd_i\geq\spd_{i'}$.
So by Claim~\ref{schur}, we have $f(u)\leq f(v)$.
\end{proofof}

We restate Lemma~\ref{strucprop} for convenience.
\firststrucprop*

\begin{proof} %of}{Lemma~\ref{strucprop}}
%Consider the work-vector $\wvecopt$ of 
Consider an optimal assignment $\sg^*:\optset\to[m]$, %where $A^*\sse J$ is some subset of jobs.
and let $\wvecopt=\wvec{\sg^*}$.
By Claim~\ref{claim:sortedwork}, we can assume that $\wvecopt=\wvecopt^{\down}$.
%$\wvecopt_i\ge\wvecopt_{i^\prime}$ for any $i>i^\prime$.
%(Recall that $\Mc_0=\Mc_{\fast}$.)
For the boundary case where $\ngrp=0$, which means that $\Mc_{\fast}=[m]$, we take
$\noptset=\optset$, $\tsg=\sg^*$ and $\wksmooth=\wvecopt$. It is easy to verify that this
satisfies~\ref{item:reward}--\ref{item:feasible}.

So assume that $\ngrp\geq 1$.
We take $\wksmooth_i=\wvecopt_i$ for all $i\in\Mc_{\fast}$. 
For all $i\in\Mc_{\ngrp}$, set $\wksmooth_i=0$.
%If $\ngrp>1$, then for all $i\in\Mc_1$, set $\wksmooth_i=\wvecopt_{i'}$, where $i'$ is the
%last machine (i.e., slowest machine) in $\Mc_1$. 
Consider an index $r\in[\ngrp-1]$. Let $i'$ be the last machine in $\Mc_r$. 
For all $i\in\Mc_r$, set $\wksmooth_i=\wvecopt_{i'}$, if this value is at least
$\wksmooth_{i^*}/m$, and $0$ otherwise. 

This satisfies properties~\ref{sorted}--\ref{smooth} by
construction. Property~\ref{item:feasible} also holds by construction, since
$\wksmooth\leq\wvecopt$, and so $f(L)\leq f(\lvec{\sg^*})\leq B$.

We next define $\noptset$ and $\tsg$ so as to
satisfy~\ref{item:reward}, \ref{item:capacities}. 
Define the reward of a machine to be the total reward of the jobs assigned to it under
$\sg^*$. Let $H_{\fast}\sse\Mc_{\fast}$ be the $|\Mc_{\fast}|-\frac{1}{\ve^2}-1$
largest-reward machines in $\Mc_{\fast}$.
Let $H_1=\Mc_1$. 
%to be the $\frac{1}{\ve^2}-\frac{1}{\ve}$ largest-reward machines in $\M_1$. 
For $r=2,\ldots,\ngrp$, define $H_r\sse\Mc_r$ to be the $m_{r-1}=|\Mc_{r-1}|$
largest-reward machines in $\Mc_r$. 

Let $\noptset$ be the jobs assigned to machines in
$H_{\fast}\cup\bigcup_{r=1}^{\ngrp}H_r$. Note that $\frac{|H_{\fast}|}{|\Mc_{\fast}|}\geq(1-\ve)$
and $\frac{|H_r|}{|\Mc_r|}\geq(1-\ve)$ for
all $r=1,\ldots,\ngrp$. This is clearly true for $r=1$. For $r\geq 2$, this follows
because $m_r\leq\frac{m_{r-1}}{1-\ve}$. Hence, $\rewd(\noptset)\geq (1-\ve)\OPT$.

The assignment $\tsg$ is defined as follows. 
%Let $\ell$ be the smallest index for which $\wksmooth_\ell=0$
We drop all jobs assigned to machines in $\Mc_{\fast}-H_{\fast}$.
Loop through indices $r=1,\ldots,\ngrp$ in that order, and perform the following steps. 

\begin{enumerate}[label=$\bullet$, noitemsep, leftmargin=*]
\item If $r=1$, we move the jobs assigned to machines in $H_1=\Mc_1$, to machines in
$\Mc_{\fast}-H_{\fast}$ (thus freeing up all machines in $\Mc_1$). 
That is, for each machine $i\in H_1$, we pick a distinct machine
$i'\in\Mc_{\fast}-H_{\fast}$ and move all jobs assigned to $i$, to machine $i'$. This is
always possible since $|H_1|\leq|\Mc_{\fast}|-|H_{\fast}|-1$. Also, after this movement,
we still have at least one machine in $\Mc_{\fast}-H_{\fast}$ that does not have any jobs
assigned to it. Let $\bi\in\Mc_{\fast}-H_{\fast}$ be such a free machine.

\item If $r>1$, we drop the jobs assigned to machines in $\Mc_r-H_r$, and
%move the jobs assigned to machines in $H_r$ to machines in $\Mc_{r-1}$ (thus freeing up
%all machines in $\Mc_r$). 
%More precisely, 
for each machine $i\in H_r$, pick a distinct machine $i'\in\Mc_{r-1}$
and move all the jobs assigned to $i$, to machine $i'$. Again, this is always possible
since $|H_r|\leq|\Mc_{r-1}|$ for all $r=2,\ldots,\ngrp$. After this movement all machines
in $\Mc_r$ are free.
%(Recall that $\Mc_0=\Mc_{\fast}$.)

\item If $r>1$ and $\wksmooth_i=0$ for machines in $\Mc_r$, then we 
%pick a machine in $\Mc_{\fast}-H_{\fast}$ that does not have any jobs assigned to it, and 
move all jobs assigned to machines in $H_{r+1}\cup H_{r+2}\cup\ldots\cup H_{\ngrp}$ to the
free machine $\bi\in\Mc_{\fast}-H_{\fast}$. Also terminate the loop here.
\end{enumerate}

It is clear that $\tsg$ assigns precisely the jobs in $H_{\fast}\cup\bigcup_{r=1}^{\ngrp}H_r$.
We argue that~\ref{item:capacities} holds. When we move jobs from a machine $i\in H_r$ to
a machine $i'\in\Mc_{r-1}$ that does not have any jobs assigned to it, we are assigning
$\wvecopt_i$ work to machine $i'$. But note that $\wksmooth_{i'}$ is the work assigned by
$\sg^*$ to the slowest machine in $\Mc_{r-1}$, which is therefore at least $\wvecopt_i$,
since $\wvecopt=\wvecopt^{\down}$. So in this case, we have
$\wvec{\tsg}_{i'}\leq\wksmooth_{i'}$. 
The other type of movement happens when, for some index $r>1$, we assign all jobs assigned
to machines in
$H_{r+1}\cup\ldots\cup H_{\ngrp}$ to the free machine in $\bi\in\Mc_{\fast}-H_{\fast}$. But this
happens because $\wksmooth_i=0$, for all machines $i\in\Mc_r$, which implies
that the first machine in $\Mc_{r+1}$ has at most $\wvecopt_{i^*}/m$ work assigned to it
under $\sg^*$. So machine $\bi\in\Mc_{\fast}-H_{\fast}$ %that is assigned all these jobs
is assigned at most $\wvecopt_{i^*}\leq\wvecopt_{\bi}=\wksmooth_{\bi}$ units of work.
Thus, we have shown that $\wvec{\tsg}\leq\wksmooth$.
\end{proof} %of}

\bibliographystyle{abbrv}
\bibliography{stoc_ref}

\appendix \label{appstart}

\section{Proofs omitted from Sections~\ref{prelim} and~\ref{knapsack}} 
\label{append-prelim} \label{append-knapsack}

Recall that for a non-increasing vector $v\in\Rp^{\POS_{N,\dt}}$, its expansion $v^\exp\in\Rp^N$ 
is the vector: $v^\exp_i=v_i$ for $i\in\POS_{N,\dt}$ and $v^\exp_i=v_{\prev(i)}$ for
$i\in[N]\sm\POS_{N,\dt}$.  
Recall also that for $u\in\R^N$ and $\tht\in\R$, we define 
$Q^{>\tht}(u):=\bigl|\{i\in[N]: u_i>\tht\}\bigr|$. 
The following results will be useful.

\begin{theorem}[Claim 2.3 and Theorem 2.4 in~\cite{IbrahimpurS21}] \label{normprops}
Let $u,v\in\R_+^N$, and $\ell\in[N]$.
\begin{enumerate}[label=(\alph*), topsep=0.2ex, itemsep=0.1ex, leftmargin=*]
\item We have 
$\topl(u)=\min_{t\geq 0}\bigl(\ell t+\sum_{i\in[M]}(u_i-t)^+\bigr)
=\ell u^{\down}_\ell+\sum_{i\in[M]}(u_i-u^{\down}_\ell)^+\bigr)
=\int_0^{\infty}\min\bigl\{\ell,Q^{>\tht}(u)\bigr\}d\tht$.

\item If $\topl(u)\leq\al\topl(v)+\beta$ for all $\ell\in[N]$, then
$\gnorm(u)\leq\al\cdot\gnorm(v)+\beta\cdot\gnorm(1,0,\ldots,0)$ for every monotone,
symmetric norm $\gnorm:R^N\mapsto\R_+$.
\end{enumerate}
\end{theorem}

%For a vector $v\in\R^M$ and $\tht\in\R$, define $Q^{>\tht}(v):=\bigl|\{i\in[M]:
%v_i>\tht\}\bigr|$. 

\begin{comment}
\begin{lemma}[Lemma 2.8 in~\cite{IbrahimpurS21} paraphrased] \label{toplestim}
Let $\ve,\dt,\gm>0$.
Let $u\in\Rp^M$, and $v\in\Rp^{\POS_{M,\dt}}$ be a non-increasing vector. 
Let $\gnorm:\R^M\mapsto\Rp$ be a monotone, symmetric norm. %Let $\ve,\gm>0$.
%
\begin{enumerate}[(a), topsep=0.25ex, itemsep=0.15ex, leftmargin=*]
%\item If $u^{\down}_\ell\leq v_\ell$ for all $\ell\in\POS$, then
%$\topl[i](u)\leq\topl[i](v^\exp)$ for all $i\in[m]$, and hence,
%$\gnorm(u)\leq\gnorm(v^\exp)$. 
%
\item If $v_\ell\leq(1+\ve)u^{\down}_\ell+\gm$ for all $\ell\in\POS_{M,\dt}$, 
%then $\topl[i](v^\exp)\leq(1+\dt)(1+\ve)\topl[i](u)+i\gm$ for all $i\in[m]$, and hence,
then $\gnorm(v^\exp)\leq(1+\dt)(1+\ve)\gnorm(u)+M\gm\cdot\gnorm(1,0,\ldots,0)$. 

\item Let $v$ be as in part (a). Let $\dt\leq 1$. %(in $\POS=\POS_{m,\dt}$).
Let $\al\in\R^M$ be such that $\al^{\down}_1\leq v_1$ and
%$N^{>v_\ell}(\al):=\bigl|\{i\in[M]:\al_i>v_\ell\}\bigr|
$Q^{>v_\ell}(\al)\leq(1+\dt)\ell-1$ for all $\ell\in\POS_{M,\dt}$.
%Then, $\topl[i](\al)\leq(1+4\dt)(1+\ve)\topl[i](u)+5i\kp$ for all
%$i\in[m]$. Hence, 
Then
$\gnorm(\al)\leq(1+4\dt)(1+\ve)\gnorm(u)+5M\gm\cdot\gnorm(1,0,\ldots,0)$.
\end{enumerate}
\end{lemma}
\end{comment}

\begin{proofof}{Lemma~\ref{toplestim}}
As noted earlier, part (a) is precisely Lemma 2.8 (b) in~\cite{IbrahimpurS21}.

For part (b), we mimic the proof of Lemma 2.8 (c) in~\cite{IbrahimpurS21}. 
We drop the subscripts $N,\dt$ from $\POS$, $\prev$, $\nxt$.
Consider any index $i\in[N]$. By Lemma~\ref{normprops} (a), we have 
$\topl[i](\al)=\int_0^{\infty}\min\bigl\{i,Q^{>\tht}(\al)\bigr\}d\tht$. Note that
$Q^{>v_1}(\al)=0$, so we can cap the limit of integration at $v_1$.
Let $\bell=i$ if $i\in\POS$ and $\prev(i)$ otherwise. Observe that $i\leq(1+\dt)\bell$. 
We have
\begin{equation*}
\begin{split}
\topl[i](\al) & \leq \int_0^{v_{\bell}}id\qt
+\sum_{\ell\in\POS:1<\ell\leq\bell}\int_{v_\ell}^{v_{\prev(\ell)}}Q^{>\qt}(\al)d\qt
\leq i\cdot v_{\bell}
+\sum_{\ell\in\POS:1<\ell\leq\bell}(v_{\prev(\ell)}-v_\ell)Q^{>v_\ell}(\al) \notag \\
& \leq i\cdot v_{\bell}
+\sum_{\ell\in\POS:1<\ell\leq\bell}(v_{\prev(\ell)}-v_\ell)(1+\dt)(\ell-1)
\leq i\cdot v_{\bell}
+\sum_{\ell\in\POS:1<\ell\leq\bell}(v_{\prev(\ell)}-v_\ell)(1+\dt)^2\prev(\ell).
\end{split}
\end{equation*}
The second inequality is because $Q^{>\qt}(\al)$ is non-increasing in $\qt$; the third
follows from the conditions in the lemma statement; and the final inequality is because
%\eqref{pfineq1} follows since 
$\ell-1\leq(1+\dt)\prev(\ell)$.
Recall that $\prev(1)=0$. Since $i\leq(1+\dt)\bell\leq(1+\dt)^2\bell$, we can upper bound
the final expression above by 
$(1+\dt)^2\sum_{\ell\in\POS:\ell\leq\bell}v_\ell\bigl(\ell-\prev(\ell)\bigr)$. 

Finally,
observe that, since $v$ is non-increasing, 
$\sum_{\ell\in\POS:\ell\leq\bell}v_\ell\bigl(\ell-\prev(\ell)\bigr)$ is the
$\topl[\bell]$-norm of the vector $w\in\Rp^N$, where $w_j=v_j$ if $j\in\POS$, and
$w_j=v_{\nxt(j)}$ if $j\in[N]\sm\POS$. We have $w\leq v^{\exp}$ and $\bell\leq i$, so 
$\sum_{\ell\in\POS:\ell\leq\bell}v_\ell\bigl(\ell-\prev(\ell)\bigr)\leq\topl[i](v^{\exp})$.
So we obtain that $\topl[i](\al)\leq(1+\dt)^2\topl[i](v^{\exp})$ for all $i\in[N]$. By
Lemma~\ref{normprops} (b), this shows that $\gnorm(\al)\leq(1+\dt)^2\gnorm(v^{\exp})$, and we
have $(1+\dt)^2\leq 1+3\dt$ since $\dt\leq 1$.
\end{proofof}

\begin{proofof}{Theorem~\ref{scaling}}
We have
$\rewd_e\leq\rewd'_e\cdot\frac{\ve\cdot\rmax}{n}\leq\rewd_e+\frac{\ve\cdot\rmax}{n}$
for all $e\in[n]$.
So for any $T\sse[n]$, we have 
$\rewd(T)\leq\rewd'(T)\cdot\frac{\ve\cdot\rmax}{n}\leq\rewd(T)+\ve\cdot\rmax$.

Applying the first inequality to the oprimal solution with $\{\rewd_e\}_{e\in[n]}$
rewards yields $\OPT\leq\frac{\ve\cdot\rmax}{n}\cdot\OPT'$. Applying the second inequality
to the optimal solution with $\{\rewd'_e\}_{e\in[n]}$ rewards yields 
$\frac{\ve\cdot\rmax}{n}\cdot\OPT'\leq\OPT+\ve\cdot\rmax\leq(1+\ve)\OPT$. This proves part
(a). 

For part (b), using the second inequality above and part (a), we have
\begin{equation*}
\begin{split}
\rewd(T) & \geq\rewd'(T)\cdot\frac{\ve\cdot\rmax}{n}-\ve\cdot\rmax
\geq\frac{1}{\rho}\cdot\OPT'\cdot\frac{\ve\cdot\rmax}{n}-\ve\cdot\OPT \\
& \geq\frac{1}{\rho}\cdot\OPT-\ve\cdot\OPT. \qedhere
\end{split}
\end{equation*}
\end{proofof}

\begin{proofof}{Claim~\ref{polybnd}} 
%We reproduce the argument from~\cite{ChakrabartyS19a}, for completeness.
A sequence $a_1,a_2,\ldots,a_k$ of nonnegative integers such that 
$\sum_{i\in[k]}a_i\leq M$ can be mapped bijectively to the set of $k+1$ integers 
$a_1,a_2,\ldots,a_k,M-a_k$ from $\{0\}\cup [M]$ that sum to $M$. The number
of such sequences of $k+1$ integers is equal to the coefficient of $x^M$ in the generating
function $(1+x+\ldots+x^M)^{k+1}$. This is equal to the coefficient of $x^M$ in $(1-x)^{-(k+1)}$,
which is $\binom{M+k}{M}$ using the binomial expansion. 
Let $U=\max\{M,k\}$. 
We have $\binom{M+k}{M}=\binom{M+k}{U}\leq\bigl(\frac{e(M+k)}{U}\bigr)^U\leq(2e)^U$.

If we have a non-increasing sequence $a_1\geq a_2\geq \ldots \geq a_k$ of $k$ integers
from $\dbrack{M}$, then we can map this bijectively to the set of $k+1$ integers
$M-a_1,a_1-a_2,a_2-a_3,\ldots,a_k$ from $\dbrack{M}$ that sum to $M$.
\end{proofof}

\section{Bicriteria guarantees} \label{append-budgviol}
We briefly discuss how the machinery developed for minimum-norm covering problems can be
used to obtain bicriteria guarantees for norm-budgeted packing problems where the norm
budget is violated by a $(1+\ve)$-factor, for any $\ve>0$. 

For \normsepfl, we already obtain such a guarantee by combining Lemma~\ref{constfrp} and
Theorem~\ref{normsepfl-bi}: given a $\beta$-approximation algorithm for \onefrp, we obtain
a bicriteria $\bigl(O(\beta),1+O(\ve)\bigr)$-approximation algorithm for \normsepfl.
So we focus on \normknap and \normmwis, %(and \normmatch), 
wherein $\sz(T)$ is the
size-weighted characteristic vector of $T$.

%As in the proof of Theorem~\ref{normsepfl-bi}, 
We utilize the machinery from Section~\ref{msn-prelim}. Let $\dt=\min\{\ve,1\}$. We take
$\POS=\POS_{n,\dt}$. Recall that $\optset$ denotes some fixed optimal solution. 
%Let $\ovec=\svec{\optset}$. 
For $u\in\R^n$ and $\tht\in\R$, recall the notation 
$Q^{>\tht}(u):=\bigl|\{i\in[n]: u_i>\tht\}\bigr|$.
We may assume that $f$ is normalized so that $f(1,0,\ldots,0)=1$.
As in the proof of Theorem~\ref{normsepfl-bi}, we can identify a polynomial-size set
$\T\sse\Rp^{\POS}$ containing a non-increasing vector $\vt$ such that
$\ovec^{\down}_\ell\leq t_\ell\leq(1+\ve)\ovec^{\down}_\ell+\frac{\ve\cdot B}{n}$. Given
$\vt$, we aim to find a maximum-reward solution $T\in\sols$ satisfying
$Q^{>t_\ell}\bigl(\svec{T}\bigr)\leq\ell-1$ for all $\ell\in\pos$. 
Lemma~\ref{toplestim} coupled with the bounds on the $t_\ell$'s then implies that
$f\bigl(\svec{T}\bigr)\leq\bigl(1+O(\ve)\bigr)B$. 

Note that the constraints $Q^{>t_\ell}\bigl(\svec{T}\bigr)\leq\ell-1$ for all
$\ell\in\pos$ can be captured by a matroid. So incorporating these constraints 
in: %This yields the following.
\begin{enumerate}[label=$\bullet$, topsep=0.2ex, itemsep=0.1ex, leftmargin=*]
\item \normknap, %incorporating these constraints 
yields an (standard) \mwis problem on a matroid, which can be solved exactly, so we obtain
a $(1,1+\ve)$-approximation here; 
\item \normmwis on a matroid, yields a weighted matroid-intersection problem, which can be
  solved exactly, so we obtain a $(1,1+\ve)$-approximation here;
\item \normmwis on a $k$-set system, yields an (standard) \mwis problem on a $(k+1)$-set
system, which admits (slightly better than) a $(k+1)$-approximation, so we obtain a
$\bigl(k+1,1+\ve)$-approximation here.
\end{enumerate}

\section{Proof of Lemma~\ref{constfrp}} \label{append-constfrp}
%\begin{proofof}{Lemma~\ref{constfrp}}
Fix a facility $i$.
Recall that $\M_i=(\C,\sols_i)$, we have client-rewards $v\in\Rp^\C$ and client-weights 
$\wt\in\Rp^C$, and a budget $t\in\Rp$. 
Define $\sols_{i,t}:=\{S\in\sols_i: \wt(S)\leq t\}$. 
In \constr \onefrp, we seek a maximum-reward
solution $A\in\sols_{i,t}$. %such that $\wt(A)\leq t$.

Let $\alg^{\onefrp}$ be the given $\beta$-approximation algorithm for \onefrp.
%Consider any facility $i\in\F$, $t\in\Rp$, and rewards $v\in\Rp^\C$. 
Let $\vopt=\max_{S\in\sols_{i,t}}v(S)$, and $S^*\in\sols_{i,t}$ be such that $v(S^*)=\vopt$.
Let $\targ$ be a given target value. We will show that as long as $\targ\leq\vopt$, we can
find $R\in\sols_i$ such that $v(R)\geq\frac{\targ}{\beta+1}$ and $\wt(R)\leq(1+\ve)t$;
%$\sum_{j\in R}c_{ij}\leq 2t$; 
call this a  ``success''.
We can then do binary search in the range $[0,v(\C)]$ to
find an interval $[\targ,\targ+\gm]$, such that we have success
for $\targ$, and we do not have success for $\targ+\gm$. It must therefore be that
$\vopt<\targ+\gm$ and 
by taking $\gm$ sufficiently small, but still such that $\log\bigl(\frac{1}{\gm}\bigr)$ is
polynomially bounded in the input size, we obtain that $\vopt\leq\targ$.

So suppose we have a target value $\targ\leq\vopt$. We may assume that $\ve\leq 1$.
%By combining the above with 
%an enumeration step that enumerates all elements in
%$S^*$ with $c_{ij}\geq\ve t$, one can obtain a bicriteria $(\beta+1,1+\ve)$ approximation.  
We may assume that we know $Q^*=\{j\in S^*: \wt_j\geq\ve t\}$ since $|Q^*|\leq\frac{1}{\ve}$.
%Suppose we have a set $Q\in\sols_{i,t}$ with $c_{ij}\geq\ve t$ for all $j\in Q$. 
(More precisely, we run the steps below for all $Q\in\sols_{i,t}$ with
$|Q|\leq\frac{1}{\ve}$ and $\wt_j\geq\ve t$ for all $j\in Q$, and return the best
solution found.)  
%We perform a similar computation as above to find, 
%For a given target value $\targ$, 
Let $\C'=\{j\in\C: \wt_j<\ve t\}$.
We use $\alg^{\onefrp}$ to compute a set
$Z\in\sols_i$ with $Z\sse Q^*\cup \C'$ %\{j\in\C: \wt_j<\ve t\}$ 
that approximately maximizes
$v(S)-\frac{\targ}{(\beta+1)t}\cdot\wt(S)$ over all $S\in\sols_i$ satisfying 
$S\sse Q^*\cup\C'$. %\{j\in\C: \wt_j<\ve t\}$. 
This can be cast as an \mwis problem on the independence system $\M_i$ as
follows. Let $\ld_j:=v_j-\tfrac{\targ}{(\beta+1)t}\cdot \wt_j$ for all 
$j\in Q^*\cup\C'$, %\{j\in\C: \wt_j<\ve t\}$, 
and $\ld_j:=0$ otherwise. 
Define $\tld_j:=\ld_j$ if $\ld_j\geq 0$, %and $\wt_j\leq t$, 
and $\tld_j:=0$ otherwise.
Then the \mwis problem on $\M_i$ with $\tld$ client rewards is the same as maximizing
$v(S)-\frac{\targ}{(\beta+1)t}\cdot\wt(S)$ over all $S\in\sols_i$ with 
$S\sse Q^*\cup\C'$. %\{j\in\C: \wt_j<\ve t\}$. 
This is because we can always take a solution $S$ 
to this \mwis problem and discard clients with $\tld_j=0$ to obtain a set $T\in\sols_i$ 
with $T\sse Q^*\cup\C'$, without affecting the $\tld$-value.

%Consider client rewards $\tld_j=\max\bigl\{\ld_j:=v_j-\tfrac{\targ}{(\beta+1)t}\cdot c_{ij},0\bigr\}$.
We have $\ld(S^*)\geq v(S^*)-\frac{\targ}{(\beta+1)t}\cdot\wt(S^*)
\geq \vopt-\frac{\targ}{\beta+1}\geq\frac{\beta}{\beta+1}\cdot\targ$, so since
$\alg^{\onefrp}$ is a $\beta$-approximation algorithm for (in particular) the \mwis
problem on $\M_i$, we obtain that $\tld(Z)\geq\ld(Z)\geq\frac{\targ}{\beta+1}$.
As noted above, we may assume that $\ld_j>0$ for all $j\in Z$, because otherwise,
we can simply delete $j$ from $Z$ without decreasing the $\tld$-value of the set.

%This corresponds to \mwis with the
%ground set $\C-\{j\in\C-Q^*: c_{ij}\geq\ve t\}$.)
%for all $j\in S$. 
Now let $R$ be a minimal subset of $Z$ containing $Z\cap Q^*$, with 
$\wt(R)\geq t$. Note that this is well defined, since $Q^*\in\sols_{i,t}$.  
%to obtain a set $R$ as above.
Since $Z-Q^*\sse\C'$, %\{j\in\C: c_{ij}<\ve t\}$, 
we have $\wt(R)\leq(1+\ve)t$.
If $R=Z$, then $v(R)\geq\ld(R)=\ld(Z)=\tld(Z)\geq\frac{\targ}{\beta+1}$. 
Otherwise, since $\ld_j>0$ for all $j\in Z$, we have 
$v(R)\geq\frac{\targ}{(\beta+1)t}\cdot\wt(R)\geq\frac{\targ}{\beta+1}$. So we
always have $v(R)\geq\frac{\targ}{\beta+1}$. \hfill \qed

\begin{remark} 
When the $\M_i$s are matroids, \constr \onefrp amounts to finding a maximum-weight
independent set subject to a knapsack constraint, which admits a
PTAS~\cite{BergerBGS11,GrandoniRSZ14}, i.e. a $(1+\ve,1)$-approximation, for any $\ve>0$.
%$\frac{e}{e-1}$-approximation~\cite{Sviridenko04}. 
There is also a PTAS when each $\M_i$ corresponds to the intersection of two
matroids~\cite{BergerBGS11,GrandoniRSZ14}. 
%\end{enumerate}
\end{remark}

\section{Proof of Theorem~\ref{identicalbiptas}} \label{append-identbiptas}
%\begin{proofof}{Theorem~\ref{identicalbiptas}} 
\label{identical} \label{ident-biptas}
For the $(1+\ve,1+\ve)$-approximation algorithm, which we also refer to as a bicriteria
PTAS, we identify a set $A$ of jobs with $\rewd(A)\geq(1-\ve)^2\OPT$ that admits an
assignment satisfying the norm-budget constraint {\em exactly}. Given this, one can use 
the PTAS for minimum-norm load-balancing on identical machines from~\cite{IbrahimpurS21}
to find an assignment for $A$ that violates the norm=budget constraint by a
$(1+\ve)$-factor. Combining the bicriteria PTAS with the PTAS for \normknap using
Theorem~\ref{lbredn}, yields the 
%we therefore obtain an improved 
$(2+\ve)$-approximation for \normsched on identical machines. 

We find $A$ using the same approach as for norm-budgeted knapsack. 
%supplementing this with one additional insight. 
Let $\sg^*:\jopt\mapsto[m]$ be an optimal solution, and 
$\lvecopt=\lvec{\sg^*}$ be the load-vector induced by $\sg^*$. 
As always, we may assume that all rewards are integers bounded by $\frac{n}{\ve}$.
We may also assume that we know the maximum reward $\rmax$ of a job in $\jopt$, and that
$\rmax$ is the maximum reward among all jobs. 
%since we can discard all jobs with higher rewards. 
For an integer $q\geq 0$, define $\thresh_q:=\frac{\rmax}{(1+\ve)^q}$, and 
$\itemset_q:=\bigl\{j\in J:\frac{\thresh_q}{1+\ve}<\rewd_j\leq\thresh_q\bigr\}$.
Let $\nbuck\leq O\bigl(\log\frac{n}{\ve}\bigr)$ be the number of reward buckets that
together cover all jobs with non-zero reward.
As with norm-budgeted knapsack, assume that we have an estimate $\optval$ such that
$\optval\leq\OPT\leq(1+\ve)\optval$. 
By enumeration over a polynomial-size set, we may assume that we know
$\rnum_q:=\floor{\frac{\rewd(\jopt\cap\itemset_q)}{\Dt}}$ for all
$q\in\dbrack{\nbuck}$, where $\Dt=\frac{\ve\cdot\optval}{\nbuck}$. 
Define $\num_q=\rnum_q\cdot\frac{\Dt}{\thresh_q}$ for $q\in\dbrack{\nbuck}$. 

Now, we claim that if we select, for each $q\in\dbrack{\nbuck}$, the set $S_q$ of
$\ceil{\num_q}$ smallest-size jobs in $\itemset_q$, then the union $A$ of these $S_q$-sets
has the desired properties: we have $\rewd(A)\geq(1-\ve)^2\OPT$, and there is an assignment
$\sg:A\mapsto[m]$ such that $f(\lvec{\sg})\leq B$.  

The analysis in the proof of Theorem~\ref{normknapthm} shows that the quality
of the $\rnum_q$ and $\num_q$ estimates is good enough to yield
$\rewd(A)\geq(1-\ve)^2\OPT$, since %we have
$\rewd(\jopt\cap\itemset_q)<(1+\ve)\rewd(S_q)+\Dt$ %\frac{\ve\optval}{\nbuck}$ 
for every $q\in\dbrack{\nbuck}$. 

We also have $\svec[p]{S_q}\leq\svec[p]{\jopt\cap\itemset_q}$ for every
$q\in\dbrack{\nbuck}$, and so $\svec[p]{A}\leq\svec[p]{\jopt}$. Let $\pi:A\mapsto\jopt$ be
a one-to-one mapping such that $p_j\leq p_{\pi(j)}$ for all $j\in A$. So, if we consider the
assignment $\sg(j)=\sg^*(\pi(j))$ for every $j\in A$, then the load vector
$\lvec{\sg}$ is coordinate-wise dominated by $\lvo$. This is because  
$\lvec{\sg}_i=\sum_{j\in A: \sg^*(\pi(j))=i}p_j\leq\sum_{j\in A:\sg^*(\pi(j))=i}p_{\pi(j)}\leq\lvec{\sg^*}_i$.
Therefore, $f(\lvec{\sg})\leq B$. \hfill \qed 
%\end{proofof}

\section{Proofs omitted from Section~\ref{relatedmc}} \label{append-relmc}

\begin{proofof}{Lemma~\ref{enumlem}}
Define $\Dt=\frac{\ve\cdot\est}{k}$. 
Let 
\[
\seqset:=\biggl\{b\in\Rp^k:\ b_r\text{ is a multiple of }\Dt\ \forall r\in[k],
\quad \sum_{r=1}^kb_r\leq\Bigl(1+\tfrac{1}{\ve}\Bigr)k\Dt\biggr\}.
\]
By Claim~\ref{polybnd}, we have $|\seqset|\leq 2^{O(k/\ve)}$ since
$\bigl(\frac{b_r}{\Dt}\bigr)_{r\in[k]}$ is a sequence of $k$ nonnegative integers that
sum up to at most $\bigl(1+\frac{1}{\ve}\bigr)k$. It is easy to give a recursive procedure
that enumerates all sequences in $\seqset$, where each leaf of the recursion tree is
labeled by a distinct sequence in $\seqset$, so elements in $\seqset$ can be enumerated in
time $2^{O(k/\ve)}$.

Define $\ta_r=\floor{\frac{a_r}{\Dt}}\cdot\Dt$ for all $r\in[k]$. Then, 
$a_r-\Dt\leq\ta_r\leq a_r$ and is a multiple of $\Dt$ for all $r\in[k]$, 
and $\sum_{r=1}^k\ta_r\leq\Gm\leq(1+\ve)\est=\bigl(1+\frac{1}{\ve}\bigr)k\Dt$. Therefore,
$\ta\in\seqset$. Also
$\sum_{r=1}^k\ta_r\geq\sum_{r=1}^ka_r-k\Dt=\sum_{r=1}^ka_r-\ve\cdot\est$. 
\end{proofof}

\begin{proofof}{Claim~\ref{schur}}
We have: (i) $u_\ell=v_\ell$ for all $\ell\in[m]-\{i,i'\}$, 
(ii) $\max\{u_i,u_{i'}\}<v_{i'}$ and (iii) $v_i<v_{i'}$. 
Let $\pi:[m]\mapsto[m]$ be the permutation 
%where each index $i\in[m]$ is mapped to its corresponding position in 
corresponding to $v^{\down}$,
i.e., $v_{\pi(1)}\geq v_{\pi(2)}\geq\ldots\geq v_{\pi(m)}$. Let $\ell=\pi^{-1}(i')$. Then
$\ell<\pi^{-1}(i)$ since $v_{i'}>v_i$. Also, for any $j<\ell$, we have
$\pi(j)\in[m]-\{i,i'\}$ and $u_{\pi(j)}=v_{\pi(j)}\geq v_{i'}>\max\{u_i,u_{i'}\}$. 
It follows that $u^{\down}_j=v^{\down}_j$ for all $j=1,\ldots,\ell-1$.
Also $u^{\down}_\ell\leq\max\{u_i,u_{i'}\}<v_{i'}=v^{\down}_\ell$. Hence, $u^{\down}$ is
lexicographically smaller than $v^{\down}$.

To show that $\gnorm(u)\leq\gnorm(v)$, we argue that $\topl[j](u)\leq\topl[j](v)$ for all
$j\in[m]$. Consider any $j\in[m]$, and let $S\sse[n]$ be the coordinates of $u$
corresponding to the first $j$ coordinates of $u^{\down}$; in particular, we have
$\topl[j](u)=\sum_{r\in S}u_r$. If $S$ contains $i,i'$, or $S\sse[m]-\{i,i'\}$, then we
have $\sum_{r\in S}u_r=\sum_{r\in S}v_r\leq\topl[j](v)$. If $S$ contains $i'$ but not $i$,
then $\sum_{r\in S}u_r<\sum_{r\in S}v_r\leq\topl[j](v)$. If $S$ contains $i$ but not $i'$, 
then taking $T=S-\{i\}\cup\{i'\}$, we have 
$\sum_{r\in S}u_r<\sum_{r\in T}v_r\leq\topl[j](v)$. 
\end{proofof}

\begin{proofof}{Lemma~\ref{extnlem}}
Let $L^*$ be the vector $\bigl(\max\{\preload_i,z^*\}\bigr)_{i\in[m]}$, where recall that
$\preload_i=p\bigl(\sg^{-1}(i)\bigr)$ for all $i\in[m]$. Clearly,
$(x^*,L^*)$ is a feasible solution to \eqref{extncp}. (Note that the inequality
$\sum_{j\in J}p_jx^*_{ij}\leq (z^*-\preload_i)^+$ is tight for all $i\in\mcset$, since
adding these inequalities gives the same quantity $p(J_2)$ on the LHS and RHS.)
Since the feasible region of \eqref{extncp} is closed and bounded, and $f$ is continuous,
\eqref{extncp} has an optimal solution. Let $(x, L)$ be an optimal solution to
\eqref{extncp} such that $L^{\down}$ is lexicographically smallest among the sorted
load-vectors of all optimal solutions to \eqref{extncp}.
%We show that $\topl(L^*)\leq\topl(L)$ for all
%$\ell\in[m]$, which implies that $f(L^*)\leq f(L)$.

If $L^*\leq L$, then $f(L^*)\leq f(L)$, so $(x^*,L^*)$ is also an optimal solution. 
%due to monotonicity of $f$.
So suppose $L^*_i>L_i$ for some $i\in\mcset$. In this case, we arrive at a contradiction.
%Consider some $i\in\mcset$ where $L^*_i>L_i$. Then, 
It must be that $L^*_i=z^*>\preload_{i}$. There must 
be some $i'\in\mcset$ such that $L_{i'}>\max\{\preload_{i'},z^*\}>L_i$. Otherwise, 
\[
p(J_2)=\sum_{r\in\mcset}(L_r-\preload_r)<\sum_{r\in\mcset}\bigl(\max\{\preload_r,z^*\}-\preload_r\bigr)
=\sum_{r\in\mcset}(z^*-\preload_r)^+=p(J_2)
\]
which yields a contradiction. The first equality above is because
$L_r-\preload_r=\sum_{j\in J_2}p_jx_{rj}$ and $\sum_{r\in\mcset}x_{rj}=1$ for every 
$j\in J_2$; the second inequality is because $L_i<\max\{\preload_i,z^*\}$; the last
equality is from the definition of $z^*$.
Since $i'\in\mcset$ and $L_{i'}>\max\{\preload_{i'},z^*\}$, there is some job $j\in J_2$
with $x_{i'j}>0$. We can now decrease $x_{i'j}$ by some $\tht>0$ and increase $x_{ij}$ by
$\tht$, which has the effect of decreasing $L_{i'}$ and increasing $L_i$ by
$\kp=p_j\tht$. We can choose $\tht$ so that $\kp<L_{i'}-L_i$.
So if $L'$ denotes the new load vector, by Claim~\ref{schur}, we have that
$f(L')\leq f(L)$, so $L'$ corresponds to the load vector of an optimal solution.  
But also $(L')^{\down}$ is lexicographically smaller than $L^{\down}$ (by
Claim~\ref{schur}), which contradicts the choice of $(x,L)$.
\end{proofof}

\section{\boldmath Dynamic-programming based FPTAS for \normknap with ordered norms}
\label{dpknap}
An ordered norm is a nonnegative linear combination of $\topl$ norms. Equivalently, an
ordered 
norm $f:\R^n\mapsto\Rp$ is specified by a non-increasing vector $\mu\in\Rp^n$, which
defines $f(v):=\sum_{i\in[n]}\mu_iv^{\down}_i$ for $v\geq 0$.%
\footnote{For $v\in\R^n-\Rp^n$, we define $f(v)=f\bigl(\{|v_i|\}_{i\in[n]}\bigr)$.}

Let $\bigl([n],\{\sz_i,\rewd_i\}_{i\in [n]},f:\R^n\mapsto\Rp,B\bigr)$ be an instance of
\normknap, where $f$ is an ordered norm specified as above by a non-increasing vector
$\mu\in\Rp^n$. 
We first apply Theorem~\ref{scaling} (scaling and rounding) to obtain that $\rewd_e$ is
an integer and $\rewd_e\leq\ceil{\frac{n}{\ve}}$ for all $e\in[n]$, losing a
$(1-\ve)$-factor in the optimal value.
Let $\rmax$ be the maximum (scaled) reward of an item.
The key observation that enables the dynamic program (DP) is that if we know that item $e$
is the $j$-th largest-size item in our solution, then we can determine its norm
contribution to be $\mu_j\sz_e$ since we have an ordered norm. 

This motivates the following DP. 
Order items so that $\sz_1\geq\ldots\geq\sz_n$. 
For $i,j\in\dbrack{n}$ and integers $\targ\in[0,n\cdot\rmax]$, define 
\[
\dptab(i,j,\targ)\, :=\, 
\min\ \bigl\{f(\svec{T}):\ \ T\sse[i], \quad |T|=j, \quad \rewd(T)\geq\targ\bigr\}.
\]
We set $\dptab(i,j,\targ)=\infty$ to denote that there is no feasible solution to the
underlying problem, and use the convention that $\infty+x=\infty$ for any $x\in\R$.
We can calculate the $\dptab(\cdot)$ entries via the following DP. 
%The base cases are: 
\begin{alignat*}{2}
\text{Base cases:} \quad
& \text{for all}\ i\in\dbrack{n}, \quad
&& \dptab(i,0,\targ)=0\ \,\text{if $\targ=0$},\quad \dptab(i,0,\targ)=\infty\ \text{otherwise}; 
\\
& && \dptab(i,j,\targ)=\infty\ \, \text{if }i<j. \\
%to denote that there is no feasible solution to the underlying problem. 
%& \frall j\in[n],\ \text{all}\ \targ, \quad \dptab(0,j,\targ)=\infty.
\text{DP recurrence:} \quad
& \text{for other}\ i,j,\targ, \quad
&&
\dptab(i,j,\targ)=\min\bigl\{\dptab(i-1,j-1,\targ-\rewd_i)+\mu_j\sz_i,\,\dptab(i-1,j,\targ)\bigr\}.
%\label{dprec}
\end{alignat*}

The two terms in the RHS of the DP recurrence correspond to including item $i$ in the
solution, which means that it is the $j$-th largest item in the solution, or not including
item $i$. Clearly, we can calculate the $\dptab(\cdot)$ entries in
$O(n^3\cdot\rmax)=O\bigl(\frac{n^4}{\ve}\bigr)$ time. The optimum value is then found by
considering the largest value $\targ$ for which $\dptab(n,j,\targ)$ is at most $B$ for some
$j\in\dbrack{n}$, and the optimal solution can be computed by tracing back to see how
$\dptab(n,j,\targ)$ is computed.  
By Theorem~\ref{scaling}, the optimal solution to the scaled instance yields reward
at least $(1-\ve)\cdot\OPT$. 

\medskip
We note that this DP approach does not seem amenable to handle more-general monotone
symmetric norms, even the case where the norm is the maximum of {\em two} ordered
norms. For this setting, one can come up with a pseudopolynomial time
algorithm, by keeping track of the budget consumption under each individual ordered norm
in the DP state; but this does not translate to an approximation scheme. A monotone,
symmetric norm may in general be the pointwise maximum of an arbitrary (even uncountable)
collection of ordered norms, so an extension of the DP-based approach to handle this
general setting seems rather infeasible. 

\end{document}